\documentclass[12pt,oneside]{memoir}
\usepackage[a4paper]{geometry}

\usepackage{bm}

\usepackage{algorithm}
\usepackage{algorithmicx}
\usepackage{algpseudocode}
\algrenewcommand\algorithmicrequire{\textbf{Input:}}
\algrenewcommand\algorithmicensure{\textbf{Output:}}
\algdef{SE}[FOR]{NoDoFor}{EndFor}[1]{\algorithmicfor\ #1}{\algorithmicend\ \algorithmicfor}%
\usepackage{amsmath, amssymb,amsfonts,mathtools}
\usepackage{tensor}
\usepackage{lmodern}

\usepackage{textcomp} % Copyright symbol

\usepackage{tabularx,multirow,float}% http://ctan.org/pkg/multirow
\usepackage{hhline}% http://ctan.org/pkg/hhline

\usepackage{cite}
\usepackage{bibentry}

\usepackage{dsfont}
\expandafter\let\csname equation*\endcsname\relax
\expandafter\let\csname endequation*\endcsname\relax
\usepackage{amsmath}
\usepackage{microtype}

\DeclareMathOperator{\Det}{det}
\DeclareMathOperator{\spn}{span}
\newcommand*{\D}{\mathcal{D}}
\newcommand*{\e}{\mathrm{e}}

\usepackage{graphicx,youngtab}
\usepackage{subfig}
\usepackage{caption}
\usepackage{physics}

\usepackage{amsthm}

\newtheorem{thm}{Theorem}
\newtheorem{lem}[thm]{Lemma}

\newtheorem{defn}[thm]{Definition}

\newcommand{\BigO}[1]{\ensuremath{\operatorname{O}\left(#1\right)}}

\newcommand\defeq{\mathrel{\overset{\makebox[0pt]{\mbox{\normalfont\tiny\sffamily def}}}{=}}}

\usepackage{longtable}

\usepackage[toc,sort=standard]{glossaries}
\makeglossaries
\newglossaryentry{USTC}{name = USTC,description = {University of Science and Technology of China}}
\newglossaryentry{HWS}{name = HWS,description = {Highest weight state}}
\newglossaryentry{loqc}{name = LOQC,description = {Linear optical quantum computation}}
\newglossaryentry{irrep}{name = irrep,description = {Irreducible representation}}
\newglossaryentry{qip}{name = QIP,description = {Quantum information processing}}
\newglossaryentry{sun}{name = \ensuremath{\mathfrak{su}(n)},description = {Generating algebra of special unitary group of degree $n$}}
\newglossaryentry{Sn}{name = \ensuremath{\mathrm{S}_{n}},description = {The symmetric group, whose elements are all the permutation operations on $n$ distinct objects, and whose group operation is the composition of permutation operations}}
\newglossaryentry{un}{name = \ensuremath{\mathfrak{u}(n)},description = {Generating algebra of unitary group of degree $n$}}
\newglossaryentry{SUN}{name = \ensuremath{\mathrm{SU}(n)},description = {Special unitary group of degree $n$}}
\newglossaryentry{UN}{name = \ensuremath{\mathrm{U}(n)},description = {Unitary group of degree $n$}}
\newglossaryentry{glnc}{name = \ensuremath{\mathrm{GL}(n,\mathds{C})},description = {General linear group of degree $n$ over the field $\mathds{C}$ of complex numbers}}
\newglossaryentry{gl2c}{name = \ensuremath{\mathrm{GL}(2,\mathds{C})},description = {General linear group of degree $2$ over the field $\mathds{C}$ of complex numbers}}
\newglossaryentry{dof}{name = DOF,description = {Degree of freedom}}
\newglossaryentry{SharpPC}{name = \ensuremath{\#\mathrm{P}}-complete,description = {$\#\mathrm{P}$ is the set of the counting problems associated with the decision problems in the set $\mathrm{NP}$ of nondeterministic polynomial time problems. A problem is $\#\mathrm{P}$-complete if and only if it is in $\#\mathrm{P}$, and every problem in $\#\mathrm{P}$ can be reduced to it by a polynomial-time counting reduction. In other words, $\#\mathrm{P}$-complete problems are the hardest problems in $\#\mathrm{P}$}}
\newglossaryentry{LI}{name = LI,description = {Linearly independent}}
\newglossaryentry{LD}{name = LD,description = {Linearly dependent}}
\newglossaryentry{csd}{name = CSD,description = {Cosine-sine decomposition}}
\newglossaryentry{DFun}{name = \ensuremath{\mathcal{D}}~functions,description = {Matrix elements of irreducible representations of group elements}}
\newglossaryentry{csm}{name = CS matrix, plural = CS, description = {Cosine-sine matrix. Defined in Equation~(\ref{Eq:CSMatrix})}}
\newglossaryentry{pdf}{name = pdf, description = {Probability density function}}
\newglossaryentry{B2}{name =\ensuremath{\mathcal{B}_{2}},description = {$2\times 2$ matrix representing a balanced beam splitter. Defined in Equation~\eqref{Eq:BVartheta}}}
\newglossaryentry{LieG}{name = Lie group, plural = Lie group, description = {A group that is also a finite-dimensional smooth manifold and in which the group operations of multiplication and inversion are smooth maps}}
\newglossaryentry{LieA}{name = Lie algebra, plural = Lie algebras, description = {A vector space over some field together with a binary operation called the Lie bracket, which satisfies the following bilinarity, alternatively, anticommutativity and the Jacobi identity}}
\newglossaryentry{Fock}{name = Fock state, description = {Fock states have well-defined particle numbers. Informally, the eigenstates of the number operator $\sum_{i}a_{i}^{\dagger}a_{i}$}. Linear optical implementations of QIP tasks such as BosonSampling and LOQC employ Fock state inputs and measurement in the Fock-state basis}
\newglossaryentry{FND}{name = Frobenius distance, description = {The Frobenius norm of a matrix $A_{m \times m}$ is $\|A\|_F=\sqrt{\sum_{i,j=1}^m |a_{ij}|^2}$.
The Frobenius distance between matrices $U$ and~$V$ is defined as
$\mathrm{dist}(U,V) \defeq \|U-V\|_F=\sqrt{\sum_{i,j=1}^m |U_{ij}-V_{ij}|^2}$,
and is a symmetric, positive-definite and subadditive distance function on the set of matrices}}
\newglossaryentry{QRD}{name = QR-decomposition, description = {The QR decomposition of a matrix $M$ is its decomposition into a product $M = QR$ of the orthogonal matrix $Q$ and the upper triangular matrix $R$.}}

\usepackage{setspace}

\usepackage[colorlinks]{hyperref}
\usepackage{memhfixc}

\begin{document}
\pagenumbering{gobble}% Remove page numbers (and reset to 1)
\clearpage
\thispagestyle{empty}
\let\cleardoublepage\clearpage

\frontmatter
\begin{center}\renewcommand{\baselinestretch}{1.9} 

UNIVERSITY OF CALGARY \\
\vfill
   {Multi-Photon Multi-Channel Interferometry\\ for Quantum Information Processing \bigskip \par}        % Set title in size.
   \vfill                  % Vertical space after title.
by \\
\vfill
   {Ish Dhand \par }              % Set author 
\vfill
{A THESIS \par                % Allow the word dissertation to be used here
SUBMITTED TO THE FACULTY OF GRADUATE STUDIES \par
IN PARTIAL FULFILLMENT OF THE REQUIREMENTS FOR THE \par
DEGREE OF DOCTOR OF PHILOSOPHY}
\vfill
{GRADUATE PROGRAM IN PHYSICS AND ASTRONOMY \par}
\vfill
{CALGARY, ALBERTA \par}
{{December}, {2015} \par}
\vfill
\copyright {Ish Dhand}~~{2015} \par
\end{center} 
\renewcommand{\baselinestretch}{1} 

\chapter{Abstract}
This thesis reports advances in the theory of design, characterization and simulation of multi-photon multi-channel interferometers.
I advance the design of interferometers through an algorithm to realize an arbitrary discrete unitary transformation on the combined spatial and internal degrees of freedom of light.
This procedure effects an arbitrary $n_{s}n_{p}\times n_{s}n_{p}$~unitary matrix on the state of light in $n_{s}$~spatial and $n_{p}$~internal modes.
The number of beam splitters required to realize a unitary transformation is reduced as compared to existing realization by a factor $n_{p}^2/2$.
I thus enable the optical implementation of higher dimensional unitary transformations.

I devise an accurate and precise procedure for characterizing any multi-port linear optical interferometer using one- and two-photon interference.
Accuracy is achieved by estimating and correcting systematic errors that arise due to spatiotemporal and polarization mode mismatch.
Enhanced accuracy and precision are attained by fitting experimental coincidence data to a curve simulated using measured source spectra.
The efficacy of our characterization procedure is verified by numerical simulations.

I develop group-theoretic methods for the analysis and simulation of linear interferometers.
I devise a graph-theoretic algorithm to construct the boson realizations of the canonical SU$(n)$ basis states, which reduce the canonical subgroup chain, for arbitrary $n$.
The boson realizations are employed to construct $\D$-functions, which are the matrix elements of arbitrary irreducible representations, of SU$(n)$ in the canonical basis.
I show that immanants of principal submatrices of a unitary matrix $T$ are a sum $\sum_{t} \mathcal{D}^{(\lambda)}_{tt}(\Omega)$ of the diagonal $\mathcal{D}$-functions of group element $\Omega$,
with $t$ determined {by} the choice of submatrix, and the irrep $(\lambda)$ determined by the immanant under consideration.
The algorithm for $\mathrm{SU}(n)$ $\mathcal{D}$-function computation and the results connecting these functions with immanants open the possibility of group-theoretic analysis and simulation of linear optics.

\chapter{Acknowledgements}
I am indebted to my adviser Barry C.~Sanders for his expert scientific guidance, patient writing advice and generous support.
I thank my other collaborators and mentors Hubert de Guise, Sandeep K.~Goyal and He Lu whom I have learned much from and whose company I have enjoyed immensely.
I am grateful to
Dominic W.~Berry,
David Feder,
Gilad Gour,
David W.~Hobill,
Alexander I.~Lvovsky,
Christoph Simon and
Urbasi Sinha
for wise professional counsel.

My time at Calgary was pleasant thanks to my friends and my colleagues
Gabriel Aguilar,
Jobin George,
Joydip Ghosh,
Mark Girard,
Hannah Gordon,
Hon-Wai Lau,
Pascal Lefebvre,
Kady Lyons,
Jonathan Johannes,
Farokh Mivehvar,
Varun Narasimhachar,
Christopher O'Brien,
Pantita Palittapongarnpim,
Marcel.li Grimau Puigibert,
Rahul Raut,
Ambrish Raghoonundun,
Nafiseh Sang-Nourpour,
Arashdeep Singh,
Priyaa Varshinee Srinivasan,
Raju Valivarthi,
Lucile Veissier,
Navid Yousefabadi and
Ehsan Zahedinejad.
Thanks to
Jan Blume,
Luc Couturier,
Yaxiong Liu and
Ingo Nosske
for great company during my time in Shanghai.
I thank
Tracy Korsgaard,
Nancy Jing Lu,
Lucia Wang and
Gerri Zannet
for their patient support and kind words.
Above all, I thank my family for their unwavering love and support.

\newpage
\tableofcontents

\newpage
\listoffigures
\newpage
\listoftables

\printglossary[title={List of symbols, abbreviations and nomenclature},toctitle={List of symbols, abbreviations and nomenclature}]

\chapter{Thesis content previously published}

The results reported in this thesis are published in peer-reviewed journals or are under peer-review.
The content of these articles are used in my thesis either verbatim or with some modifications as required.
My publications are:
\begin{enumerate}
\item[\cite{Dhand2015b}]
Ish~Dhand and Sandeep~K.~Goyal.
Realization of arbitrary discrete unitary transformations using spatial and internal modes of light.
{\em Physical Review A}, 92:043813, 2015. 
(\href{http://arxiv.org/abs/1508.06259}{arXiv:1508.06259}).

\item[\cite{Dhand2015a}]
Ish~Dhand, Abdullah~Khalid, He~Lu, and Barry~C.~Sanders.
{A}ccurate and precise characterization of linear optical interferometers, 2015,
{\em Journal of Optics}, 18(3):035204, 2016. 
(\href{http://arxiv.org/abs/1508.00283}{arXiv:1508.00283}).

\item[\cite{Dhand2015d}]
Ish~Dhand, Barry~C.~Sanders, and Hubert~de~Guise.
{A}lgorithms for {SU}$(n)$ boson realizations and $\mathcal{D}$-functions, 
{\em Journal of Mathematical Physics}, 56(11):111705, 2015.
(\href{http://arxiv.org/abs/1507.06274}{arXiv:1507.06274}).

\item[\cite{deGuise2015}]
Hubert~de~Guise, Dylan~Spivak, Justin~Kulp, and Ish~Dhand.
 $\mathcal{D}$-functions and immanants of unitary matrices and submatrices.
\textit{Journal of Physics A: Mathematical and Theoretical}, 49(9):09LT01, 2016.
(\href{http://arxiv.org/abs/1511.01851}{arXiv:1511.01851}).
\end{enumerate}

\noindent
The following article reports the work done during my PhD but is different in scope from the main topic of my PhD. Hence, this article is not included in my PhD thesis.
\begin{enumerate}
\item[\cite{Dhand2014}] Ish Dhand and Barry C. Sanders. Stability of the Trotter-Suzuki decomposition. \textit{Journal of Physics A: Mathematical and Theoretical}, 47(26):265206, 2014.
(\href{http://arxiv.org/abs/1403.3469}{arXiv:1403.3469}).
\end{enumerate}

The following changes are made to the material taken from my published papers for use in this thesis:
\begin{itemize}
\item
Each section of Chapter~1 contains a few sentences verbatim (but not explicitly marked) from the respective introduction sections of my papers~\cite{Dhand2015b,Dhand2015a,Dhand2015d,deGuise2015}.
\item
Section~\ref{Sec:LinearOptics} contains material from the background section of my NJP Submission~\cite{Dhand2015a} in addition to new material on the definition of linear optics.
\item
The material in Section~\ref{Sec:OneTwoPhoton} is copied from the background section of my NJP Submission.
\item
Sections~\ref{Sec:SUnDefinitions} and~\ref{Sec:BosonRealizations} comprise material published in my JMP paper~\cite{Dhand2015d}. A paragraph with additional notational details~(\ref{Eq:CompressedNotation}) has been added.
\item
The definition of immanants in Section~\ref{Sec:Immanants} is copied verbatim from my JPA submission~\cite{deGuise2015}.
\item
Section~\ref{Sec:CSD} is copied verbatim from the background section of my PRA paper.
\item
Sections~\ref{Sec:Algorithm}, \ref{Sec:Cost} and~\ref{Sec:DesignConclusion} are copied from Sections~III,~IV and~V respectively of the PRA paper.
Certain phrases (beam splitter, phase shifter) were changed in the interest of consistency.
\item
Sections~\ref{Sec:Procedure}--\ref{Sec:Scattershot} contain material from the revised version of my NJP submission. The ordering of the sections has been changed for coherence.
\item
Section~\ref{Sec:Instability} is an edited version of Appendix~A of my revised NJP submission.
\item
Section~\ref{Sec:SunAlgorithms} comprises material from my JMP paper. The introduction has been changed for better context and the notation made consistent with that in the rest of this thesis.
\item
Section~\ref{Sec:SunImmanantResult} comprises material from my JPA submission. The notation made consistent with that in the rest of this thesis.
\item
Appendix~\ref{Appendix:Construction} is copied from Appendix~A my article in PRA.
\item
Appendix~\ref{Sec:CurveFitting} is copied from Appendix~B of the revised version of my NJP submission.
\item
Appendices~\ref{Appendix:SubAlgebraChoice} and~\ref{Appendix:Connection} are taken from my article in JMP.
\end{itemize}
\clearpage
\pagenumbering{arabic}

\mainmatter

%=================%
\chapter{Introduction}
%=================%

%------------------------------%
\section{Research problem and objectives}
\label{Sec:IntroProblems}
%------------------------------%
Here I define my research problem and objectives in multi-photon multi-channel interferometry for quantum information processing~(\gls{qip}).
QIP advances computation and communication by exploiting quantum mechanics.
Efficient quantum algorithms can solve problems for which no efficient classical algorithm is known~\cite{Shor1997,Lloyd1996,Harrow2009,Nielsen2010}.
Quantum communication protocols are computationally secure~\cite{Bennett1984,Gisin2002}.
Numerous physical systems are being investigated for implementing QIP;
nonclassical light is a strong candidate for QIP implementations because it promises long coherence times and ease of transmission.

QIP protocols that require nonlinear media have been proposed~\cite{Chuang1995,Clausen2002,Howell2000,Milburn1989} but existing natural or electromagnetically induced nonlinearities are too weak and too noisy to be useful~\cite{Kok2007}.
In contrast, linear optics is important for implementing QIP tasks because of its relative ease of implementation.

Numerous QIP tasks can be implemented on linear optics.
The problem of sampling the output coincidence distribution of a linear optical interferometer, i.e., the BosonSampling problem, is hard to simulate on a classical computer~\cite{Aaronson2013}.
BosonSampling involves sampling from the photon-coincidence distribution at the output of an interferometer when single photons are incident at each input port.
Sampling from this distribution is computationally hard classically but is easy with a linear-optical interferometer~\cite{Aaronson2013,Aaronson2013a}.
Single-photon detectors and linear optical interferometers allow for efficient universal quantum computation via linear optical quantum computing (\gls{loqc})~\cite{Knill2001}.
Linear optics can simulate the quantum quincunx~\cite{Do2005} and quantum random walks~\cite{Jeong2004}.
Linear optics coupled with laser-manipulated atomic ensembles enables long-distance quantum communication~\cite{Duan2001}.
A wide class of communication protocols can be realized with coherent states and linear optics~\cite{Arrazola2014}.

On the experimental front, recent advances in photonic technology including photonic circuits on silicon chips~\cite{Politi2008,Matthews2009,Crespi2011,Shadbolt2012,Metcalf2013}, noise-free high-efficiency photon number-resolving detectors~\cite{GolTsman2001,Rosenberg2005,Gansen2007,Lita2008,Najafi2015}, high-fidelity single-photon sources~\cite{Santori2002,URen2004,Faraon2008,Sipahigil2012} have engendered the experimental implementation of multi-photon multi-channel linear optical interferometry.
Reconfigurable interferometers that can perform arbitrary linear transformations on the spatial modes of light have been demonstrated~\cite{Mower2014,Harris2015,Carolan2015}.

Despite advances in the experimental implementations, the theory of design, characterization and simulation of linear optics is still in its nascent stages.
The scalability of implementing linear optics on spatial modes is limited because aligning and stabilizing beam splitters in multi-channel interferometers is challenging.
Characterization of optical processes using classical light requires that the device being characterized be stable with respect to the probing set-up on a sub-wavelength scale.
A linear optical interferometer can be characterized using one- and two-photon statistics but the current characterization procedure lacks the accuracy and precision required for application to QIP.
The classical simulation of a linear optical interferometer under indistinguishable single-photon inputs (i.e., photons with identical spectra arriving simultaneously at the interferometer) is well studied, but methods for simulating partially distinguishable photons (i.e., photons with different spectra or those with different arrival times) are oversimplified and sub-optimal.
I aim to make linear optics a viable system for QIP by overcoming these challenges in the theory of multi-photon multi-channel linear optics.

In my PhD thesis, I report advances to the theory of design, characterization and simulation of linear optics.
In collaboration with Sandeep K.~Goyal, I tackle the inadequacy of current interferometer design procedures by devising a realization of arbitrary discrete unitary transformations using spatial and internal modes of light, thereby reducing the beam splitter requirement and improving scalability.
My collaborators and I construct a procedure to characterize a linear optical interferometer accurately and with known precision and demonstrate the efficacy of the procedure by simulations and experiments.
Finally, I contribute to the simulation of interferometry under partially distinguishable single-photon inputs by means of two results that enable a deeper analysis and faster simulation of photon measurement probabilities.
The results include
(i)~an algorithm to compute matrix elements (\gls{DFun}) of the irreducible representations (\gls{irrep}) of the special unitary group (\gls{SUN}) in the canonical basis and
(ii)~results connecting $\mathcal{D}$-functions to immanants of the interferometer transformation matrix.

The remainder of this chapter summarizes the results reported in this thesis.
Section~\ref{Sec:IntroProblems} details my research problem and objective regarding the design, characterization and simulation of multi-photon multi-channel interferometry.
Section~\ref{Sec:IntroRealization} elucidates our results on the design of arbitrary linear optical interferometers.
We describe our procedure for the accurate and precise interferometer characterization in Section~\ref{Sec:IntroAccurate}.
Section~\ref{Sec:IntroSimulation} details our contribution to group-theoretic methods in the context of interferometer simulation.
The chapter concludes with an overview of the thesis in Section~\ref{Sec:IntroOverview}.

%------------------------------%
\section{Realization of linear optics in spatial and internal degrees of freedom}
\label{Sec:IntroRealization}
%------------------------------%
This section overviews our procedure for realizing arbitrary discrete unitary transformations in spatial and internal degrees of freedom~(\gls{dof}s) of light.
Linear optical transformations can be realized on various DOFs of light.
For instance, any $2\times 2$ unitary transformation on the polarization DOF can be decomposed into elementary operations that are implemented using quarter- and half-wave plates~\cite{Simon1989,Simon1990,Simon2012}.
Any unitary transformation on an arbitrary number of spatial modes can be realized as an arrangement of beam splitters, phase shifters and mirrors~\cite{Reck1994,Rowe1999,Guise2001} and of temporal modes using nested fiber loops or dispersion~\cite{Motes2014,Motes2015,Pant2015}.
Finally, unitary transformations on orbital-angular-momentum modes of light can be realized using beam splitters, phase shifters, holograms and extraction gates~\cite{Garcia-Escartin2011}.

Current experimental implementations choose the spatial DOF to perform quantum walks~\cite{Peruzzo2010,Crespi2013,Poulios2014}, BosonSampling~\cite{Broome2013,Spring2013,Metcalf2013,Crespi2013a,Bentivegna2015}, bosonic transport simulations~\cite{Harris2015} and photonic quantum gates~\cite{Politi2008,Pooley2012,Meany2015}.
Implementing linear optical transformations on $n$ spatial modes requires aligning $\BigO{n^{2}}$ (see footnote\footnote{
For functions $f$ and $g$ defined on some subset of the real numbers, I write~\cite{Knuth1976}
\begin{equation}
f(x)=\BigO{g(x)}
\end{equation}
if and only if there exist positive constant $M$ and real number $x_{0}$ such that
\begin{equation}
|f(x)| \le \; M |g(x)|\text{ for all }x \ge x_0,
\end{equation}
i.e., $g(x)$ grows faster than $f(x)$ for asymptotically large values of $x$.
} for a definition of big-O notation) beam splitters~\cite{Reck1994}; this requirement poses the key challenge to the scalability of linear optical implementation of QIP protocols.

One approach to realizing larger unitary transformations is to use internal DOFs, such as polarization, arrival time and orbital angular momentum in addition to the spatial DOF.
Specifically, any lossless transformation on $n_{s}$~spatial and $n_{p}$~internal modes is described by an $n_{s}n_{p}\times n_{s}n_{p}$~unitary transformation.
However, there was no known method to effect an arbitrary $n_{s}n_{p}\times n_{s}n_{p}$~unitary transformation on the state of light in $n_{s}$ spatial and $n_{p}$ internal modes.

We aimed to devise an efficient realization of an arbitrary unitary transformation using spatial and internal DOFs.
By efficient I mean that the cost of realizing the transformation, as quantified by the number of required spatial and internal optical elements, scales no faster than a polynomial in the dimension of the transformation.
Specifically, we construct an algorithm to decompose an arbitrary $n_{s}n_{p}\times n_{s}n_{p}$ unitary transformation into a sequence of $\BigO{n_{s}^{2}}$ beam splitters and $\BigO{n_{s}^{2}}$ internal transformations, each of which acts only on the $n_{p}$ internal modes of light in one spatial mode.

In contrast to the Reck \emph{et al.}~approach, which allows the realization of any discrete unitary transformation in spatial modes, our approach enables the realization into spatial and internal modes.
The Reck \emph{et al.}~procedure decomposes arbitrary $n\times n$ unitary matrices into a product of $2\times 2$ unitary matrices, which are realized as beam splitters and phase shifters.
On the other hand, incorporating $n_{p}$-dimensional internal DOFs requires decomposing into $2n_{p}\times 2n_{p}$ beam splitter matrices and $n_{p}\times n_{p}$ unitary matrices representing internal transformations.
Thus, the Reck \emph{et al.}~procedure cannot incorporate internal DOFs.

Although the Reck \emph{et al.}~approach cannot be used to design interferometers that transforms the spatial and internal modes of light, we can realize arbitrary $n_{s}n_{p}\times n_{s}n_{p}$ transformations exclusively on the spatial modes.
Such a realization requires $n_{s}n_{p}$ spatial modes and $\BigO{n_{s}^{2}n_{p}^{2}}$ beam splitters.
At the cost of increasing the required number of internal optical elements by a factor of two, we reduce the required number of beam splitters by a factor of $n_{p}^{2}/2$ as compared to the Reck \emph{et al.}~method.
Another difference between our method and the Reck \emph{et al.}~method is that our method requires only balanced beam splitters, which are easier to construct accurately~\cite{Huisman2014}.

Reducing the required number of beam splitters at the cost of increasing the number of optical elements is desirable both in free-space and in on-chip implementations of linear optical transformations.
Free-space implementations of linear optics require beam splitters to be stable with respect to each other at sub-wavelength length scales.
On-chip beam splitters rely on evanescent coupling~\cite{Szameit2007}, which requires overcoming the challenge of aligning different optical channels.
On the other hand, operations on internal elements do not require mutual stability and are typically easier to align.
For these reasons, operations on internal elements are preferred over beam splitters both in free-space and in on-chip implementations of linear optical transformations.

Moreover, our approach is advantageous experimentally because of its flexibility in the choice of $n_{p}$ and $n_{s}$.
For instance, consider the realization of a $6\times 6$ unitary matrix.
The Reck \emph{et al.}~approach allows for a realization of this transformation on an interferometer with six spatial modes.
Depending on experimental requirements, our procedure allows for a realization of the $6\times 6$ transformations ($n_{s}n_{p}=6$) using either
(i)~six spatial modes ($n_{s} = 6, n_{p}= 1$),
(ii)~three spatial and two internal modes, for instance polarization ($n_{s} =3, n_{p}= 2$),
(iii)~two spatial and three internal modes ($n_{s} =2, n_{p}= 3$) or
(iv)~one spatial and six internal modes ($n_{s} =1, n_{p}= 6$).

In summary, our procedure enables the realization of arbitrary $n_{s}n_{p}\times n_{s}n_{p}$ linear optical interferometers on $n_{s}$~spatial and $n_{p}$~internal DOFs thereby reducing the beam splitter requirement by a factor of $n_{p}^{2}/2$.
Chapter~\ref{Chap:Design} details this procedure for realizing arbitrary unitary transformation on the spatial and internal modes of light.

%------------------------------%
\section{Accurate and precise characterization of linear optics}
\label{Sec:IntroAccurate}
%------------------------------%

This section details our procedure for the accurate and precise characterization of linear optics and summarizes the numerical and experimental evidence of the superiority of the procedure over existing procedures.
The accurate and precise characterization of linear optics is important in quantum information processing tasks such as BosonSampling, LOQC and quantum walks.
The classical hardness of the BosonSampling problem crucially depends on bounds on the error in the implemented interferometer~\cite{Arkhipov2014}.
The proposed practical applications of BosonSampling, in quantum metrology and in the computation of molecular vibronic spectra, rely on the accurate implementation and characterization of linear optics~\cite{Motes2015,Huh2015}.
Accurate and precise characterization is important in LOQC because a high success probability of the employed non-deterministic linear-optical gates relies on implementing the desired gates with high fidelity~\cite{Kok2007}.
Furthermore, linear interferometers used in photonic quantum walks require accurate characterization especially if quantum walks are employed for solving classically hard problems~\cite{Gamble2010,Peruzzo2010,Sansoni2012}.
In other words, the accurate and precise characterization of interferometers enables a verifiable quantum speedup of linear-optical protocols over classical computers.

Classical-light procedures~\cite{Lobino2008,Keshari2013} for linear optics characterization are unsuitable for \gls{Fock} based experiments because the interferometer parameters change when classical light sources and homodyne detectors are coupled to and decoupled from the interferometer ports.
This change could result from a drift of interferometer parameters in the time required to couple (decouple) sources and detectors or as the result of the mechanical process of coupling (decoupling) itself.
Characterization procedures that rely on Fock-state (rather than classical-light) inputs enable interferometer characterization without altering the experimental setup if the implemented QIP task employs Fock states.
Thus, Fock-state characterization procedures would thus be accurate in BosonSampling and LOQC implementations

The Laing-O'Brien procedure~\cite{Laing2012} uses Fock states (one and two photons) for characterizing linear optical interferometers and does not require sub-wavelength stability.
This procedure assumes perfect matching in source field and large-number statistics on the detected photons.
Hence, implementations of this procedure are inaccurate due to spatiotemporal and polarization mode mismatch in the source field and imprecise due to shot noise.

We aimed to devise an accurate and precise procedure that uses one and two photons for the characterization of linear optical interferometers and to devise a rigorous method to estimate the standard deviation in the interferometer parameters~\cite{Altepeter2005}.
Furthermore, we aimed to provide a correct alternative to the $\chi^2$-test, which has been used to estimate the confidence in the characterized interferometer parameters in current BosonSampling implementations~\cite{Crespi2013,Metcalf2013,Tillmann2015}\footnote{The $\chi^2$-test~\cite{Pearson1900,Plackett1983,Greenwood1996} is used to quantify the goodness of fit between probability distribution functions of two categorical variables, which can take a fixed number of values.
Coincidence-count curves and visibilities are not probability distribution functions of categorical variables, but rather are collections of many categorical variables (variables that can take on one of a fixed finite number of possible values), one variable corresponding to each time-delay value chosen in the experiment.
Hence, quantifying the goodness of fit between two coincidence curves using the $\chi^2$-test is incorrect.
This incorrectness undermines the claim that the data are consistent with quantum predictions and disagree with classical theory~\cite{Metcalf2013,Tillmann2015} and leaves the choice of unitary matrices~\cite{Crespi2013} unjustified.
}.

The above aims were attained via a procedure to characterize a linear optical interferometer accurately and precisely using one- and two-photon interference.
Five strengths of our approach over the Laing-O'Brien procedure~\cite{Laing2012} are that our procedure
(i)~accounts for and corrects systematic error from spatial and polarization source-field mode mismatch via a calibration procedure;
(ii)~increases accuracy and precision by fitting experimental coincidence data to curve simulated using measured source spectra;
(iii)~accurately estimates the error bars on the characterized interferometer parameters via a bootstrapping procedure;
(iv)~employs maximum-likelihood estimation to determine the unitary transformation matrix that best represents the characterization data and
(v)~reduces the experimental time required to characterize interferometers using a scattershot procedure.
Chapter~\ref{Chap:Procedure} details the characterization procedure.

The efficacy of our characterization procedure has been verified both numerically and experimentally.
Experimentally, beam splitters (two-channel interferometers) were characterized using our procedure and using the Laing-O'Brien procedure, and these reflectivities were compared with the correct values obtained from single-photon measurement.
Reflectivities obtained from our procedure match those obtained from single-photon measurements within 95\% confidence intervals whereas those obtained from the Laing-O'Brien procedure do not.
Numerically, we simulated $1000$ characterization experiments using measured spectra with varying shot-noise and mode mismatch.
Our procedure yields one to two orders of magnitude improvement over existing procedures in the accuracy as measured by the trace distance between the expected and observed unitary.
Chapter~\ref{Chap:Verification} details the numerical and experimental verification of our accurate and precise characterization procedure.

%------------------------------%
\section{$\mathrm{SU}(n)$ and $\mathrm{S}_{n}$ group theory for simulation of linear optics}
\label{Sec:IntroSimulation}
%------------------------------%

This section elaborates on my contribution to the theory of the special unitary and the symmetric groups for application to linear optics.
I outline our algorithms for computing irreducible representations of $\mathrm{SU}(n)$ and our result on the connection between immanants and $\mathcal{D}$-functions of $\mathrm{SU}(n)$ matrices and submatrices.
The special unitary group $\mathrm{SU}(n)$, whose elements represent all $n$-channel interferometers and the permutation group~\gls{Sn}, which manifests the bosonic exchange symmetries, enable a deep analysis of photons interfering at a linear interferometer.

Motivated by recent progress in linear optics implementations, we aimed to develop group-theoretic methods for a realistic theory of photon coincidences that accommodates multimode photon pulses, multimode detection and non-simultaneous arrival of photons.
My contribution to $\mathrm{SU}(n)$ and $\mathrm{S}_{n}$ group theory for interferometer simulation comprises
(i)~an algorithm to compute $\mathrm{SU}(n)$ $\mathcal{D}$-functions in the canonical basis and
(ii)~results connecting $\mathrm{SU}(n)$ $\mathcal{D}$-functions to immanants of fundamental representations.
These two results pave the way for a complete group-theoretic analysis of multi-photon multi-channel interferometry for arbitrary numbers of photons and channels.

Recent application of $\mathrm{SU}(3)$ group theory to three-photon interferometry inspire our $\mathrm{SU}(n)$ $\mathcal{D}$-function calculation algorithm.
$\mathrm{SU}(3)$ $\mathcal{D}$-functions enable a symmetry-based interpretation of the action of a three-channel linear interferometer on partially distinguishable single-photon inputs~\cite{Tan2013,Tillmann2015}.
Exploiting the permutation symmetries manifest in multi-photon systems reduces the cost of computing interferometer outputs in comparison to brute-force techniques~\cite{Guise2014}.

The $\mathcal{D}$-function calculator relies on boson realizations, which map operators and states of groups to transformations and states of bosonic systems.
We devise a graph-theoretic algorithm%
\footnote{
A graph is a mathematical structure that captures pairwise relations between objects.
Specifically, a graph is defined as an ordered pair $G = (\mathcal{V}, \mathcal{E})$ comprising a finite set $\mathcal{V}$ of vertices or nodes or points together with a set $\mathcal{E}$ of edges or arcs or lines, which are $2$-element subsets of $\mathcal{V}$.
Graph theory is a branch of discrete mathematics the deals with the study of graphs.
}
to construct the boson realizations of the canonical SU$(n)$ basis states, which reduce the canonical subgroup chain, for arbitrary $n$.
The boson realizations are employed to construct $\mathcal{D}$-functions, which are the matrix elements of arbitrary irreducible representations, of SU$(n)$ in the canonical basis.
The algorithm, which I detail in Section~\ref{Chap:Sun}, comprises the following key steps
(i)~a mapping of the weights of an irrep to a graph and
(ii)~a graph-theoretic algorithm to compute boson realizations of the canonical basis states of SU$(n)$ for arbitrary $n$.
The algorithm offers a significant advantage over the two competing procedures, namely factorization and exponentiation as I demonstrate in Section~\ref{Sec:SunAlgorithms}.

My second result is a theorem relating $\mathrm{SU}(n)$ $\mathcal{D}$-functions with immanants of the fundamental representation of $\mathrm{SU}(n)$.
My collaborators and I expand a result of Kostant~\cite{Kostant1959} on immanants of an arbitrary $n\times n$ unitary matrix $T\in \mathrm{SU}(n)$ to the submatrices of $T$.
Specifically, we show that immanants of principal submatrices of a unitary matrix $T$ are a sum of the diagonal $\mathcal{D}$-functions of a group element $\Omega$, with $t$ determined by the choice of submatrix, and the irrep $(\lambda)$ determined by the immanant under consideration.
This result connects photon output probabilities, which depend on immanants of submatrices of the interferometer matrix, with $\mathrm{SU}(n)$ $\mathcal{D}$-functions.
The theorem is stated and proved in Section~\ref{Sec:SunImmanantResult}.

%------------------------------%
\section{Overview of chapters}
\label{Sec:IntroOverview}
%------------------------------%
Chapters~\ref{Chap:Background} and~\ref{Chap:Back2} present the relevant background for the results reported in this thesis.
I define linear optics and detail the action of a linear optical interferometer on one- and two-photon inputs in Chapter~\ref{Chap:Background}.
One- and two-photon inputs are employed in our procedure for the characterization of linear optics.
Chapter~\ref{Chap:Procedure} details the characterization procedure.
Chapter~\ref{Chap:Verification} presents the numerical and experimental evidence of the accuracy and precision of our characterization procedure.

Chapter~\ref{Chap:Back2} includes the group theory of $\mathrm{SU}(n)$ group and its algebra, and I define determinant, immanants and permanents, which are relevant to my results on the group theory of linear optics.
My results on $\mathrm{SU}(n)$ and $\mathrm{S}_{n}$ group-theoretic method for simulation of linear optics are presented in Chapter~\ref{Chap:Sun}.
Chapter~\ref{Chap:Back2} also presents the cosine-sine decomposition of unitary matrices.
The cosine-sine decomposing is the key building block for our procedure for the relations of linear optics in spatial and internal modes, which I detail in Chapter~\ref{Chap:Design}.
I conclude this thesis with a summary and a list of open problems in Chapter~\ref{Chap:Summary}.

%=================%
\chapter{Background: Linear Optics}
\label{Chap:Background}
%=================%
This chapter presents relevant definitions and background on linear optical transformations.
The action of a multi-mode linear optical interferometer on single photons entering one or two input ports and vacuum entering the other ports is detailed.

Section~\ref{Sec:LinearOptics} defines linear optics as transformations performed by materials in which the electric polarization is linearly dependent on the incoming electric field and describes linear optical transformations as unitary operations.
Section~\ref{Sec:OneTwoPhoton} presents expressions for the probability of detecting single photons at given output ports when single photons are incident at given input ports and of coincident photon detections when two controllably delayed photons are incident on the interferometer.

%------------------------------%
\section{Definition of linear optics}
\label{Sec:LinearOptics}
%------------------------------%
Here I define linear optical media by their linear response to light.
I parameterize the discrete unitary transformation effected by an interferometer and present a treatment of losses and dephasing at the interferometer ports.

\begin{defn}[Linear optics~\cite{Mandel1995}]
Linear optics is defined as the set of transformations effected by media whose response to electromagnetic fields is linear.
In other words, the electric polarization
\begin{equation}
\bm{P}\defeq \bm{D}-\epsilon_{0}\bm{E} = \epsilon_{0}\bm{\chi} \bm{E}
\end{equation}
is linear in the electric field $E$, where $\epsilon_{0}$ is the vacuum permittivity, $D$ the electric displacement and $\bm{\chi}$ is the electric susceptibility tensor.
\end{defn}
I consider the propagation of light in one-dimensional non-magnetic medium, which is a medium with zero magnetic susceptibility.
The Hamiltonian describing the energy of the system is~\cite{Mandel1995} %Pg1069
\begin{align}
H = & \iint \mathrm{d}\omega\,\mathrm{d}x\,\left[\frac{1}{2\mu_{0}}\bm{B}^{2}(x,\omega) + \frac{\epsilon_{0}}{2} \bm{E}^{2}(x,\omega)\right] \nonumber\\
&+ \iint \mathrm{d}\omega\,\mathrm{d}\omega'\,\mathrm{d}x\, \left[\frac{1}{2}\,\bm{\chi}(\omega,\omega')\bm{E}(x,\omega)\bm{E}(x,\omega')\right],
\end{align}
where $x$ is the spatial coordinate, $\omega$ refers to frequency, $\bm{B}$ is the magnetic field and $\mu_{0}$ is the vacuum permeability.
I consider electromagnetic fields with finite spatial and temporal extent, i.e.,
\begin{align}
\hat{\bm{E}}(x,t)= \hat{\bm{E}}_{0}(x,t)\mathrm{e}^{\mathrm{i}kx-\mathrm{i}\omega t} + \mathrm{c.c.}
\end{align}
for complex valued envelope $\hat{\bm{E}}_{0}$, wavenumber $k$, $\mathrm{c.c.}$~representing complex conjugate and the quantities with caret denoting the Fourier transform of the respected quantities without caret.
Assuming bandwidth (standard deviation of $\bm{E}_{0}(\omega)$) narrow as compared to the central frequency (mean frequency of $\bm{E}_{0}(\omega)$ ) and performing canonical quantization by replacing the~$\bm{E}$ and~$\bm{B}$ fields with the corresponding free-field Hilbert space operators gives us~\cite{Kok2007}
\begin{equation}
H =\int\mathrm{d}\omega\, \sum_{jk}A_{jk} a_{j}^{\dagger}(\omega) a_{k}(\omega),
\label{Eq:Hamiltonian}
\end{equation}
which is bilinear in the creation and annihilation operators for complex $\{A_{jk}\}$.

The Hamiltonian~(\ref{Eq:Hamiltonian}) effects photon-number preserving transformations on the state of the incoming light.
The interferometer transforms the photonic creation and annihilation operators according to
\begin{equation}
a_{j}^\dagger(\omega) \rightarrow \sum_{i=1}^m V_{ij}(\omega) a_{i}^\dagger(\omega)
\label{Eq:interferometeraction}
\end{equation}
and its complex conjugate, where $V(\omega)$ is the transformation matrix of the interferometer.
In general, the elements $\left\{V_{ij}(\omega)\right\}$ of the transformation matrix depend on the frequency of transmitted light.
I assume that the spectral functions $\bm{E}_{0}$ of the incoming light are narrow compared to frequencies over which the entries $\{V_{ij}\}$ change noticeably and thus treat $V$ to be frequency-independent.
Under this assumption, photon-number conservation imposes unitarity
\begin{equation}
V^\dagger(\omega) V(\omega) = \mathds{1}
\end{equation}
of the transformation matrix $V(\omega)$ for all real $\omega$.
Thus, linear optical interferometers effect unitary transformations on the incoming state of light.

Following~\cite{Laing2012}, we parameterize the unitary matrix $V$ to aid the clarity of our characterization procedure (Chapter~\ref{Chap:Procedure}).
If only Fock states are incident at the interferometer and only photon-number-counting detection is performed on the outgoing light, then the measurement outcomes are invariant under phase shifts at each input and output port.
That is, interferometer $\hat{V}=D_1VD_2^\dagger$ produces the same measurement outcome as $V$ for any diagonal unitary matrices $D_1$ and $D_2$.
%, i.e., diagonal matrices whose non-zero matrix elements are all complex numbers of unit magnitude,
Mathematically, if $D_1, D_2$ are diagonal unitary matrices, then
\begin{equation}
V\sim\hat{V}\iff \hat{V}= D_1VD_2^\dagger
\end{equation}
is an equivalence relation.
Members of the same equivalence class defined by this equivalence relation produce the same number-counting measurement outcomes on receiving Fock-state inputs.

\begin{figure}[h]
\begin{centering}
\includegraphics[width=0.9\textwidth]{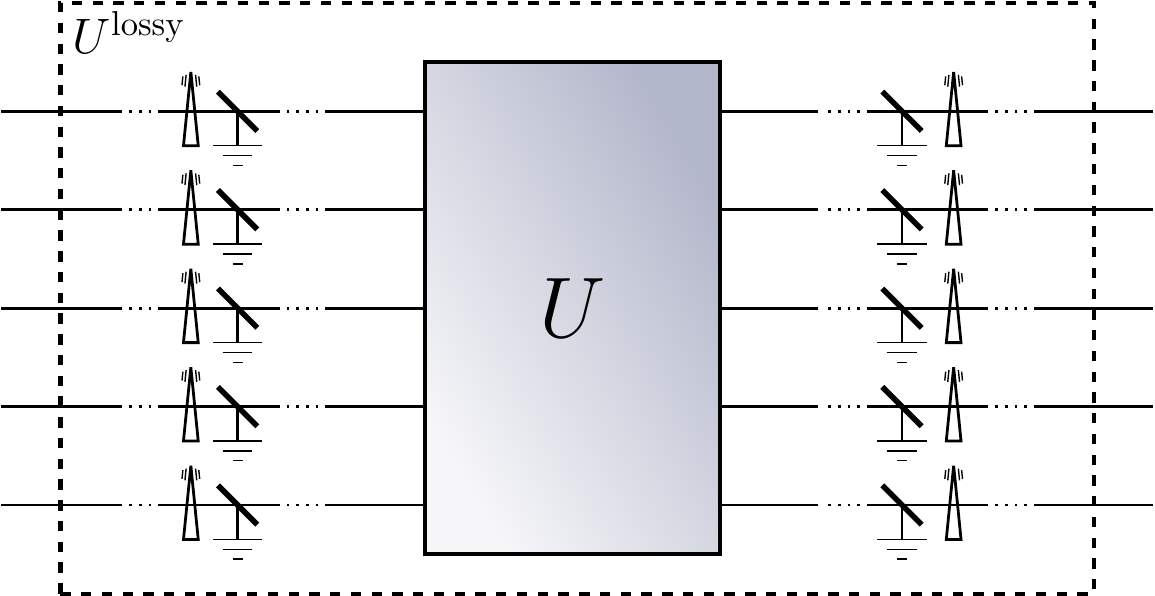}
\caption{Schematic diagram of the interferometer.
$U$ effects a unitary transformation on a multimode state of light.
The dotted lines represent the couplings of the interferometer with light sources and detectors.
The beam splitters at the input and output modes model the linear losses because of imperfect coupling and detector inefficiency.
The vacuum input to these beam splitters is not shown.
One of the beam splitter outputs enters the interferometer whereas the other one is lost.
The triangles represent the random dephasing at the input and output ports.
The dashed box labelled $U^{\mathrm{lossy}}$ represents the combined effect of the dephasing, the losses and the unitary interferometer.
}
\label{Figure:Interferometer}
\end{centering}
\end{figure}

Each equivalence class can be represented by a unique matrix~$U$ whose first row and first column consist of real elements.
The complex matrix entries of the {class representative} $U\sim V$ are
\begin{equation}
U_{ij} = t_{ij}\e^{\mathrm{i} \theta_{ij}}:\,
t_{ij}\in \mathds{R}^+,\,
\theta_{ij}\in(-\pi,\pi],\,
\theta_{i1} \equiv 0, \theta_{1j} = 0\,
\forall i,j\in\{1,2,\dots,m\}.
\label{Eq:Utt}
\end{equation}
The constraints $\theta_{i1} \equiv 0, \theta_{1i} \equiv 0\, \forall i\in\{1,2,\dots,m\}$ on the input and output phases of the transformation matrix are obeyed in the following parameterization of $U$
\begin{align}
U &=
L\times A \times M, \nonumber\\
&\defeq \begin{pmatrix} 1 & 0 & \cdots& 0 \\
0 & \sqrt{\lambda_2} & \cdots & 0\ \\
\vdots & \vdots & \ddots & 0 \\
0 & 0 & \cdots & \sqrt{\lambda_m} \\
\end{pmatrix} \begin{pmatrix}
1 & \cdots & 1 \\
1 & \cdots & \alpha_{2m} \mathrm{e}^{\mathrm{i} \theta_{2m}} \\
\vdots & \ddots & \vdots \\
1 & \cdots & \alpha_{mm} \mathrm{e}^{\mathrm{i} \theta_{mm}}
\end{pmatrix}
\begin{pmatrix} \sqrt{\mu_1} & 0 & \cdots& 0 \\
0 & \sqrt{\mu_2} & \cdots & 0\ \\
\vdots & \vdots & \ddots & 0 \\
0 & 0 & \cdots & \sqrt{\mu_m} \\
\end{pmatrix}. \label{Eq:RMS0}
\end{align}
Thus, the values~$\{\lambda_i\}, \{\alpha_{ij}\},\{\theta_{ij}\}, \{\mu_{j}\}$ completely parametrize the class representative matrix~$U$.
%In other words, the amplitudes $t_{ij}$ are related to $\alpha_{ij}$ according to
%\begin{equation}
%t_{ij} = \sqrt{\lambda_i}\alpha_{ij}\mathrm{e}^{\mathrm{i}\theta_{ij}}\sqrt{\mu_j},
%\end{equation}
%where $\alpha_{k1},\alpha_{1k}=1~\forall k\in\{1,2,\dots,m\}$.
%Our characterization procedure calculates the values of $\{\alpha_{ij}\}$ and $\{\theta_{ij}\}$ and the respective standard errors in these values.

Next, I model the losses at the input and output ports of the interferometer.
I assume time-dependent linear loss and dephasing at each interferometer port.
I model losses using parameters $\nu_{j}$ and $\kappa_{i}$, which are the respective probabilities of transmission at the input mode $j$ and output mode $i$.
%The loss probabilities at the input and output ports $j$ and $i$ are $1-s_{j}$ and $1-r_{i}$ respectively.
Dephasing is modelled using parameters $\xi_j$ and $\phi_i$, which are the arbitrary multiplicative phases at the input and output ports.
Hence, the actual transformation effected by the interferometer is given by the matrix $U^\mathrm{lossy}$, which has matrix elements
\begin{align}
U^\mathrm{lossy}_{ij} &= \e^{\mathrm{i} \phi_i}\sqrt{\kappa_{i}}\, U_{ij} \sqrt{\nu_{j}}\e^{\mathrm{i} \xi_j}\nonumber\\
&= \e^{\mathrm{i} \phi_i}\sqrt{\kappa_{i}} \sqrt{\lambda_{i}}\, \alpha_{ij}\e^{\mathrm{i} \theta_{ij}} \sqrt{\mu_{j}}\sqrt{\nu_{j}}\e^{\mathrm{i} \xi_j}.
\label{Eq:Probability}
\end{align}
Figure~\ref{Figure:Interferometer} depicts the relation between the representative matrix $U$ and the actual transformation $U^\mathrm{lossy}_{ij}$ that is effected by the interferometer.

This completes the definition and parameterization of the linear optical interferometer.
Our characterization procedure (Chapter~\ref{Chap:Procedure}) employs one- and two-photon inputs to estimate the values of parameters $\{\lambda_{i}\},\{\alpha_{ij}\},\{\theta_{ij}\},\{\mu_{j}\}$ of~(\ref{Eq:Probability}).
In the next section, I recall the expectation values of measurements performed on interferometer outputs when one- and two-photon states are incident at the input ports.

%------------------------------%
\section{One- and two-photon inputs to linear optical interferometer}
\label{Sec:OneTwoPhoton}
%------------------------------%
The section details the action of an $m$-mode interferometer on one- and two-photon inputs.
Our characterization procedure employs single-photon counting to estimate the complex amplitudes~$\{\alpha_{ij}\}$ of the representative matrix~$U$ entries.
The complex arguments~$\{\theta_{ij}\}$ of $U$ are estimated using two-photon coincidence counts.

Consider a single photon entering the $i$-th mode of an $m$-mode interferometer.
The monochromatic\footnote{Two monochromatic photons are distinguishable based on the ports that they occupy and on their respective frequencies $\omega_1$ and $\omega_2$.} photonic creation and annihilation operators acting on the $i$-th and the $j$-th ports obey the canonical commutation relation
\begin{equation}
\left[a_i(\omega_1),a_j^\dagger(\omega_2)\right] = \delta_{ij}\delta(\omega_1-\omega_2)\mathds{1}
\end{equation}
for positive real frequencies $\omega_{1},\omega_{2}$.
\begin{defn}[State of single photon]
The state of a single photon entering the $i$-th mode is
\begin{equation}
\ket{1}_i = \int_{-\infty}^{\infty} \mathrm{d}\omega f_{i}(\omega) a_i^\dagger(\omega)\ket{0},
\label{Eq:Single}
\end{equation}
where $f_{i}(\omega)$ is the normalized square integrable {spectral function}, $\ket{0}$ is the $m$-mode vacuum state.
\end{defn}

The state of two photons entering modes $i$ and $j\ne i$ of the interferometer is
\begin{equation}
|11\rangle_{ij} = \int_{-\infty}^{\infty} \mathrm{d}\omega_1\int_{-\infty}^{\infty}\mathrm{d}\omega_2\,f_i(\omega_1) f_j(\omega_2) a_i^\dagger(\omega_1)a_j^\dagger(\omega_2)|0\rangle
\label{Eq:twophotonstate}
\end{equation}
with exchange symmetry holding if $f_{i}(\omega) = f_{j}(\omega)$.
One- and two-photon states are transformed into superpositions of one- and of two-photon states respectively under the action of the linear interferometer.

\begin{figure}[h]
\includegraphics[width=\textwidth]{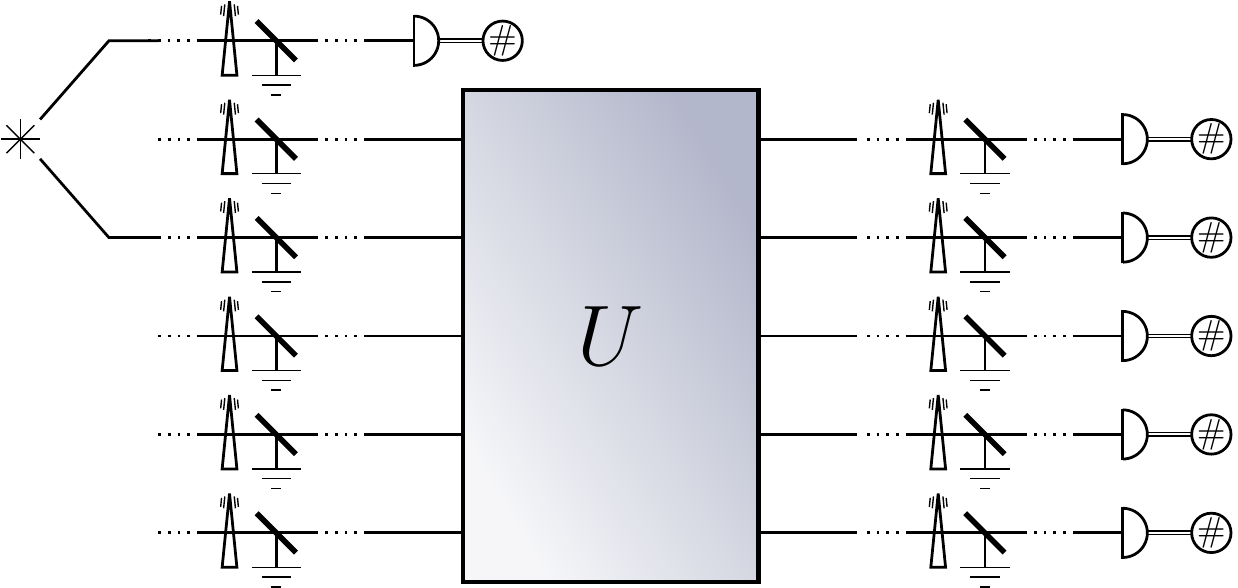}
\caption{Schematic diagram of single-photon counting at the output of an interferometer when single photons are incident at one input port.
The star symbol represents a source of single-photon pairs.
Single photons are incident at one of the input ports whereas vacuum state is input to the remaining input ports (not shown in figure).
The semicircles at the output ports represent single-photon detectors and the circles with the included $\#$ represent the photon-counting logic connected to the detectors.
}\label{Figure:1Photon}
\end{figure}

I consider the case of single-photon transmission.
The interferometer transforms the single-photon input state~(\ref{Eq:Single}) to the state at the output ports according to~(\ref{Eq:Probability}).
A photon is detected at the $i$-th output port with a probability
\begin{equation}
P_{ij} = \left|U^\mathrm{lossy}_{ij} \right|^2 = \kappa_{i} \lambda_{i} \alpha_{ij}^2 \mu_{j} \nu_{j}
 \label{Eq:PSinglePhotons}
\end{equation}
when a single-photon is incident on the $j$-th input port.

%Instead of single photons~(\ref{Eq:Single}), current experimental realizations employ two-mode squeezed light generated by pumping nonlinear crystals~\cite{Ling1986,Lvovsky2014}.
%The non-vacuum component of this state is dominated by single-photon pairs; one photon of each pair is impinged at the input ports.
%The state also includes contribution from two- and more-photon pairs.
%The expression for $P_{ij}$, if two-mode squeezed light is incident at the input ports, is presented in~\ref{Sec:SinglePhoton}.

Whereas the values of $\{\alpha_{ij}\}$ are estimated using single photon counting, $\{\theta_{ij}\}$ values are estimated using two-photon coincidence measurement.
Now I present probabilities of detecting two-photon coincidence at the interferometer outputs when controllably delayed pairs of photon are incident at the input ports.
If a controllably delayed photon pair is incident at input ports $j$ and $j'$, then the probability $C_{ii'jj'}(\tau)$ of coincidence measurement at detectors placed at output ports $i$ and~$i'$ is
\begin{align}
C_{ii'jj'}(\tau) =&\,
\kappa_{i}\kappa_{i'}\nu_{j}\nu_{j'}\Big[\left(t_{ij}^2t_{i'j'}^2+t_{ij'}^2t_{i'j}^2\right)\int \mathrm{d}\omega_1\mathrm{d}\omega_2 \left|f_{j}(\omega_1)f_{j'}(\omega_2)\right|^2 \nonumber
\\ &+2 t_{ij}t_{ij'}t_{i'j}t_{i'j'} \int \mathrm{d}\omega_1\mathrm{d}\omega_2 f_{j}(\omega_1)f_{j'}(\omega_2)f_{j}(\omega_2)f_{j'}(\omega_1)
\nonumber
\\&\times \cos\left(\omega_2\tau-\omega_1\tau+\theta_{ij}-\theta_{ij'}-\theta_{i'j}+\theta_{i'j'}\right)\Big].
\end{align}
On substituting according to~(\ref{Eq:Utt}), we obtain~\cite{Rohde2006}
\begin{align}
C_{ii'jj'}(\tau) =&\,\kappa_{i}\kappa_{i'}\lambda_{i}\lambda_{i'}\mu_{j}\mu_{j'}\nu_{j}\nu_{j'}\Big[\left(\alpha_{ij}^2 \alpha_{i'j'}^2+ \alpha_{ij'}^2 \alpha_{i'j}^2\right)\int \mathrm{d}\omega_1\mathrm{d}\omega_2 |f_{j}(\omega_1)f_{j'}(\omega_2)|^2 \nonumber
\\ &+2 \gamma \alpha_{ij} \alpha_{ij'} \alpha_{i'j} \alpha_{i'j'} \int \mathrm{d}\omega_1\mathrm{d}\omega_2 f_{j}(\omega_1)f_{j'}(\omega_2)f_{j}(\omega_2)f_{j'}(\omega_1)
\nonumber
\\&\times \cos\left(\omega_2\tau-\omega_1\tau+\theta_{ij}-\theta_{ij'}-\theta_{i'j}+\theta_{i'j'}\right)\Big].
\label{Eq:CoincidenceRate}
\end{align}
 where $\tau$ is the time delay between the two photons, $f_j(\omega), f_{j'}(\omega)$ describe source-light spectrum and~$\gamma$ is the mode-matching parameter, which I describe in the remainder of this section.

\begin{figure}[h]
\includegraphics[width=\textwidth]{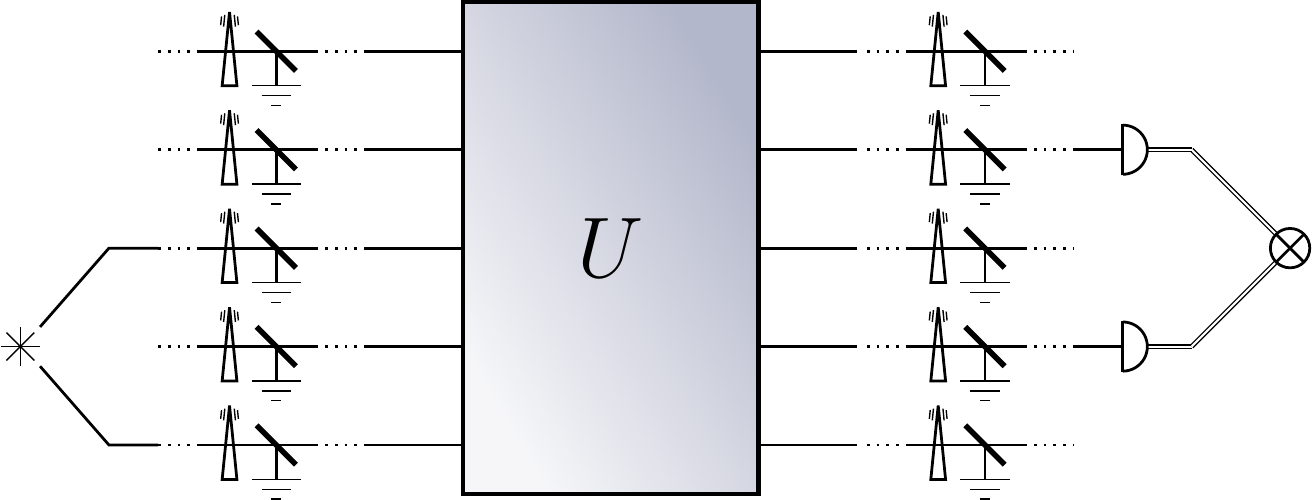}
\caption{Schematic diagram for coincidence measurement the interferometer output when single-photon pairs are incident on two different input ports of an interferometer.
The star symbol represents a source of single-photon pairs and the semicircles at the output ports represent single-photon detectors.
The coincidence logic, which is depicted by $\otimes$, counts two-photon coincidence events at the detectors.
}
\label{Figure:2Photons}
\end{figure}

Two-photon coincidence probabilities~(\ref{Eq:CoincidenceRate}) depend on the mode matching in the source field.
Spatial and polarization mode mismatch is quantified by the mode-matching parameter~$\gamma$~\cite{Rohde2006}.
Perfectly indistinguishable light sources, such as light from a single-mode fibre, have relative mode matching $\gamma=1$ whereas $\gamma=0$ indicates that the sources are completely distinguishable.

%----------------------------------------%
\section{Summary}
\label{Sec:Back1Summary}
%----------------------------------------%

In summary, I have defined linear optics and have presented a description of linear optical transformations as unitary transformations~(\ref{Eq:interferometeraction}) acting on the incoming states of light.
I have defined the representative matrix~(\ref{Eq:Utt}) of the equivalence class of unitary transformations that produce identical number-counting measurements on Fock-state inputs and have presented a treatment~(\ref{Eq:Probability}) of linear loss and dephasing at the input and output ports of the interferometer.

I have defined the state~(\ref{Eq:Single}) of a photon and detailed the probability~(\ref{Eq:PSinglePhotons}) of single-photon transmission from an input port to an output port of a linear interferometer.
Finally, I have presented expressions for the coincidence probability~(\ref{Eq:CoincidenceRate}) of obtaining coincidence measurement when two-controllably delayed photons of a given mode mismatch are incident at a linear interferometer.

%=================%
\chapter{Background: $\mathrm{SU}(n)$ and $\mathrm{S}_{n}$ methods in linear optics}
\label{Chap:Back2}
%=================%
This chapter presents relevant background in $\mathrm{SU}(n)$ and $\mathrm{S}_{n}$ group theory.
I introduce relevant $\mathrm{SU}(n)$ and $\mathrm{S}_{n}$ methods before presenting their connection with three-photon interferometry.
The chapter is structured as follows.
Section~\ref{Sec:SUnDefinitions} comprises definitions regarding the special unitary group $\mathrm{SU}(n)$ and its algebra $\mathfrak{su}(n)$.
In Section~\ref{Sec:BosonRealizations}, I define and provide expressions of the boson realizations of $\mathfrak{su}(n)$ operators.
Section~\ref{Sec:Immanants} includes definitions of determinants, immanants and permanents of matrices and examples for two- and three-dimensional matrices.
Section~\ref{Sec:ThreePhotons} presents the relevant
Section~\ref{Sec:CSD} presents the relevant background of the cosine-sine decomposition, which is a key building block of the realization procedure that is presented in Chapter~\ref{Chap:Design}.

%----------------------------------------%
\section{The special unitary group and its algebra}
\label{Sec:SUnDefinitions}
%----------------------------------------%

Here I recall the relevant properties of special-unitary group $\mathrm{SU}(n)$ and its algebra \gls{sun}.
I explain how the $\mathfrak{su}(n) \supset \mathfrak{su}(n-1) \supset \dots \supset \mathfrak{su}(2)$ subalgebra chain is used to label the basis states of the unitary irreps of $\mathrm{SU}(n)$.
I present the background for $n =2$ before dealing with $\mathrm{SU}(n)$ for arbitrary $n$.

Consider the special unitary group
\begin{equation}
\mathrm{SU}(2) = \{V: V\in \gls{gl2c},\,V^\dagger V = \mathds{1},\, \Det{V} = 1\}
\end{equation}
of $2\times 2$ special unitary matrices.
Each element of $\mathrm{SU}(2)$ can be parametrized by three angles $\Omega=(\alpha,\beta,\gamma)$.
The defining $2\times 2$ representation of an element $V(\Omega)$ of $\mathrm{SU}(2)$ is given by
\begin{align}
\mathrm{V}(\Omega)
&=
\begin{pmatrix}
	\mathrm{e}^{-\frac{1}{2}\mathrm{i}(\alpha +\gamma )} \cos \frac{\beta}{2} & -\mathrm{e}^{-\frac{1}{2}\mathrm{i}(\alpha -\gamma )} \sin \frac{\beta}{2} \\
	\mathrm{e}^{\frac{1}{2}\mathrm{i}(\alpha -\gamma )} \sin \frac{\beta}{2} & \mathrm{e}^{\frac{1}{2}\mathrm{i}(\alpha +\gamma )} \cos \frac{\beta}{2}
\end{pmatrix}.
\label{Eq:2x2Rmatrix}
\end{align}
The \gls{LieA} corresponding to \gls{LieG} $\mathrm{SU}(2)$ is denoted by $\mathfrak{su}(2)$ and is spanned by the operators $J_x,J_y,J_z$, which satisfy the angular momentum commutation relations
\begin{equation}
[J_x,J_y] = iJ_z\, ,\quad
[J_y,J_z] = iJ_x\, ,\quad
[J_z,J_x] = iJ_y.
\label{Eq:Jxyz}
\end{equation}

I transform the basis~(\ref{Eq:Jxyz}) of $\mathfrak{su}(2)$ to the complex combinations
\begin{equation}
C_{1,2} = {J_x + \mathrm{i} J_y}\, ,\qquad
C_{2,1} = {J_x - \mathrm{i} J_y}\, ,\qquad
H_1 = 2J_z,
\label{Eq:PlusMinusZ}
\end{equation}
which satisfy the commutation relations
\begin{equation}
[H_1,C_{1,2}] = 2C_{1,2} \, ,\quad
[H_1,C_{2,1}] = -2C_{2,1}\, ,\quad
[C_{1,2},C_{2,1}] = H_1.
\label{Eq:SU2RaisingLoweringCommutations}
\end{equation}
These commutation relations~(\ref{Eq:SU2RaisingLoweringCommutations}) facilitate the construction of a $(2J+1)$-dimensional irrep with carrier space spanned by basis states $\{\ket{J,M}:\, -J\le M\le J\}$~\cite{Littlewood1950}.
The integer $2M$ is the weight of the eigenstate $\ket{J,M}$ for
\begin{equation}
H_1\ket{J,M} = 2M\ket{J,M}.
\end{equation}
The operators $C_{1,2}$ and $C_{2,1}$ act on eigenstates of $H_1$ by raising or lowering the weight $2M$ of the states
\begin{align}
C_{1,2} \ket{J,M} = \sqrt{J(J+1)-M(M+1)}\ket{J,M+1},\\
C_{2,1} \ket{J,M} = \sqrt{J(J+1)-M(M-1)}\ket{J,M-1},
\end{align}
where $2J$ is the highest eigenvalue of $H_1$.

Each basis state of a finite-dimensional irrep of $\mathrm{SU}(2)$ is labelled by integral weight $2M\in\{-2J,-2J+2,\dots,2J-2,2J\}$.
The unique basis state $\ket{J,J}$ is called the highest-weight state (\gls{HWS}) and is annihilated by the action of the raising operator $C_{1,2}$.
The representation is labelled by the largest eigenvalue $2J$ of $H_1$.
Basis states of an $\mathrm{SU}(2)$ irrep are visualized as collections of points on a line with the location of each point related to the weight of the state.
Figure~\ref{Figure:SU2} gives a geometrical representation of the action of $\mathfrak{su}(2)$ operators and illustrative examples of $\mathrm{SU}(2)$ irreps.

\begin{figure}	
\centering
\subfloat[]{\label{Figure:SU2Algebra}\includegraphics[width=0.31\textwidth]{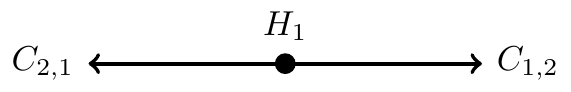}}
\qquad
\subfloat[]{\label{Figure:SU2Irrep2}\includegraphics[width=0.54\textwidth]{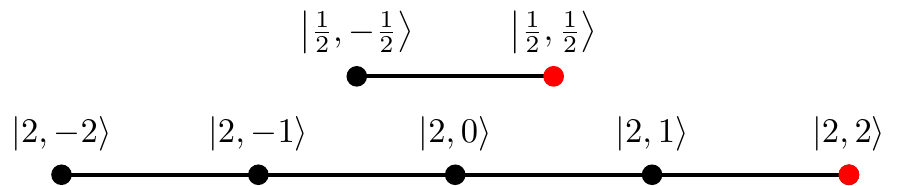}}
\caption{(a) Generators of the $\mathfrak{su}(2)$ Algebra.
The action of the raising and lowering operators $C_{1,2}, C_{2,1}$ on the basis states is represented by the directed lines.
The basis states are invariant under the action of the Cartan operator $H_1$, which is represented by the dot at the centre.
(b) $\mathrm{SU}(2)$ irreps labelled by highest weights $2M=1$ and $2M=4$ respectively.
The dots represent the basis states whereas the lines connecting the dots represent the transformation from one basis state to another by the action of the $\mathfrak{su}(2)$ raising and lowering operators.
The red dot represents the HWS, which is annihilated by the action of the raising operator $C_{1,2}$.}\label{Figure:SU2}
\end{figure}

Next I consider the case of arbitrary $n$.
The unitary group \gls{UN} is the Lie group of $n\times n$ unitary matrices
\begin{equation}
\mathrm{U}(n) \defeq \{V\colon \,V\in \gls{glnc}, V^\dagger V = \mathds{1}\}.
\label{Eq:UnitaryGroupDefinition}
\end{equation}
The corresponding Lie algebra is denoted by \gls{un}.
The complex extension of $\mathfrak{u}(n)$ is spanned by $n^2$ operators $\{C_{i,j} \colon i,j \in 1,2, \dots n\}$
satisfying the canonical commutation relations
\begin{equation}
 [C_{i,j},C_{k,l}] = \delta_{j,k}C_{i,l} - \delta_{i,l}C_{k,j}.
 \label{Eq:Commutation}
\end{equation}
The group $\mathrm{SU}(n)$ is the subgroup of those U$(n)$ transformations that satisfy the additional property $\Det V=1$; i.e.,
\begin{equation}
\label{Eq:Determinant}
\mathrm{SU}(n) \defeq \{V\colon V \in \mathrm{U}(n),\Det V = 1\}.
\end{equation}
The U$(n)$ $\mathcal{D}$-functions differ from the $\mathrm{SU}(n)$
$\mathcal{D}$-functions by at most
a phase, and I concentrate here on the $\mathrm{SU}(n)$ case.

The operator $N= C_{1,1}+C_{2,2}+\dots + C_{n,n}$ is in the centre%
\footnote{The centre of an algebra $\mathfrak{u}$ comprises those elements $x$ of $\mathfrak{u}$ such that $xu = ux$ for all $u\in\mathfrak{u}$.}
of $\mathfrak{u}(n)$.
The Lie algebra $\mathfrak{su}(n)$ is obtained from $\mathfrak{u}(n)$ by eliminating the operator $N$.
The $n-1$ operators
\begin{equation}
H_i = C_{i,i} - C_{i+1,i+1} \quad \forall i \in \{1,2,\dots,n-1\}
\label{Eq:Cartan}
\end{equation}
commute with each other and span the Cartan subalgebra of $\mathfrak{su}(n)$.
Hence, I have the following definition of the $\mathfrak{su}(n)$ algebra.
\begin{defn}[$\mathfrak{su}(n)$ algebra~\cite{Littlewood1950}]
The algebra $\mathfrak{su}(n)$ is the span of the operators $\{C_{i,j} \colon i,j \in \{1,2, \dots, n\},\, i\ne j\}$ and $\{H_i: H_i = C_{i,i} - C_{i+1,i+1},\, i \in \{1,2,\dots,n-1\}\}$ where the operators $\{C_{i,j}\}$ obey the commutation relations
\begin{equation}
[C_{i,j},C_{k,l}] = \delta_{j,k}C_{i,l} - \delta_{i,l}C_{k,j}.
\end{equation}
\label{Defn:Algebra}
\end{defn}
\noindent The linearly independent (\gls{LI}) $\mathfrak{su}(n)$ basis states span the carrier space of $\mathfrak{su}(n)$ representations.
Each basis state is associated with a weight, which is the set of integral eigenvalues of the Cartan operators.
\begin{defn}[Weight of $\mathfrak{su}(n)$ basis states~\cite{Littlewood1950}]
The weight of a basis state is the set $\Lambda = (\lambda_1,\lambda_2,\dots,\lambda_{n-1})$ of $n-1$ integral eigenvalues of the Cartan operators $\{H_1,H_2,\dots,H_{n-1}\}$.
$\mathfrak{su}(n)$ basis states have well defined weights.
\end{defn}
%\noindent Thus, the weights of basis states of $\mathfrak{su}(n)$ can be viewed as geometric structures in $n-1$ dimensions.

Of the $n^2-1$ elements, $n-1$ Cartan operators generate the maximal Abelian subalgebra of $\mathfrak{su}(n)$.
The remaining operators satisfy the commutation relation
\begin{equation}
[H_{i},C_{j,k}] =
\begin{cases}
\beta_{i,jk} C_{j,k}, &\forall j < k,\\
-\beta_{i,jk} C_{j,k}, &\forall j > k,
\end{cases}
\end{equation}
for Cartan operators $H_i$ of Definition~\ref{Defn:Algebra} and for positive integral roots $\beta_{i,jk}$.
The operators $\{C_{j,k}\colon j<k\}$ define a set of raising operators.
The remaining off-diagonal operators $\{C_{j,k}\colon j>k\}$ are the $\mathfrak{su}(n)$ lowering operators.
Each irrep contains a unique state that has nonnegative integral weights $K = (\kappa_1,\dots,\kappa_{n-1})$ and is annihilated by all raising operators.
This state is the HWS of the irrep.
\begin{defn}[Highest-weight state]
The HWS of an $\mathrm{SU}(n)$ irrep is the unique state that is annihilated according to
\begin{equation}
C_{i,j}\ket{\psi^{K}_\mathrm{hws}} = 0 \quad \forall i<j,\, i,j \in \{1,2,\dots,n\}
\end{equation}
by the action of all the raising operators.
\end{defn}
\noindent
The weight of the HWS also labels the irrep; i.e., two irreps with the same highest weight are equivalent and two equivalent representations have the same highest weight.
Hence, I label an irrep by $K = (\kappa_1,\kappa_2,\dots,\kappa_{n-1})$ if the HWS of the irrep has weight $\Lambda = K$.

Whereas in $\mathrm{SU}(2)$ the weight $2M$ and the representation label $J$ are enough to uniquely identify a state in the representation, this is not so for $\mathrm{SU}(n)$ representations.
In general, more than one $\mathrm{SU}(n)$ basis state of an irrep could share the same weight.
For example, certain states of the $K= (2,2)$ irrep of $\mathrm{SU}(3)$ irrep have the same weight~(Fig.~\ref{Figure:SU3}).
The number of basis states that share the same $\mathrm{SU}(n)$ weight $\Lambda = (\lambda_1,\lambda_2,\dots,\lambda_{n-1})$ is the multiplicity $M({\Lambda})$ of the weight~\cite{Kostant1959}.
Hence, uniquely labelling the $\mathrm{SU}(n)$ basis states requires a scheme to lift the possible degeneracy of weights.

\begin{figure}[h]
\centering
\subfloat{\includegraphics[width=0.31\textwidth]{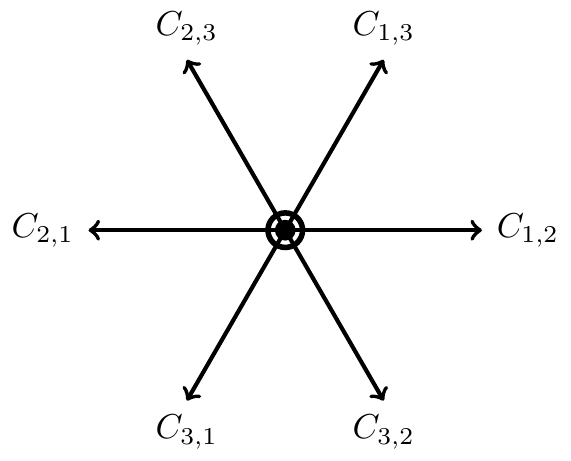} \label{Figure:SU3Algebra}}
\qquad
\subfloat{\includegraphics[width=0.54\textwidth]{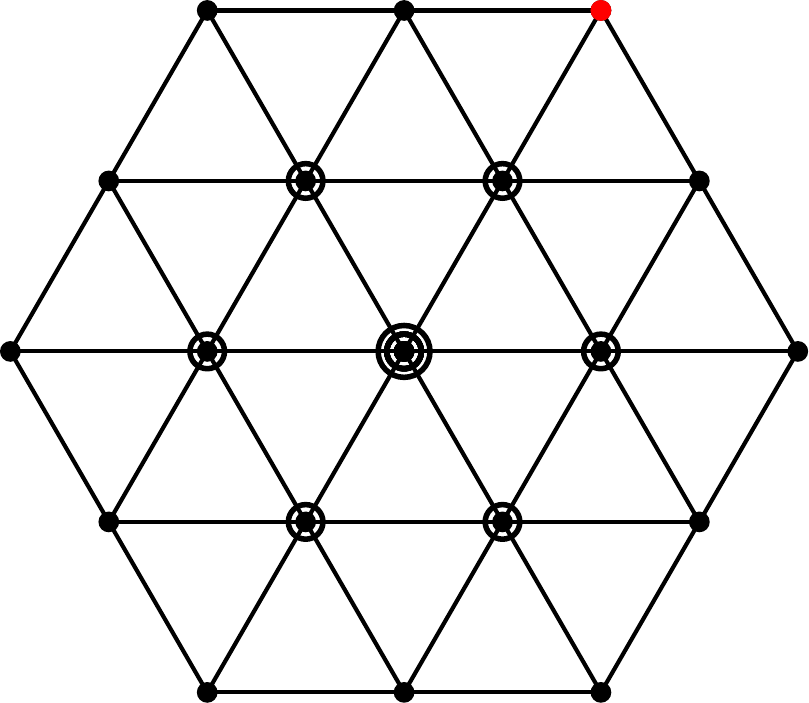} \label{Figure:SU3Irrep}}
\caption{(a) Generators of the $\mathfrak{su}(3)$ algebra.
The action of the raising operators $\{C_{1,2}, C_{1,3}, C_{2,3}\}$ and lowering operators $\{C_{2,1},C_{1,3},C_{2,3}\}$ on the canonical basis states and their linear combinations is represented by the directed lines.
%The basis states are invariant under the action of the Cartan operators $H_1, H_2$, which are represented by the dot at the centre.
(b) The $\mathrm{SU}(3)$ irrep labelled by its highest weight $(\kappa_1,\kappa_2) = (2,2)$.
The dots and circles represent the canonical basis states.
The dimension of the space of states at a given vertex is the sum of the number of dots and the number of circles at the vertex, for instance weights associated with dimension two are represented by one dot and one circle.
The lines connecting the dots represent the transformation from states of one weight to those of another by the action of $\mathrm{SU}(3)$ raising and lowering operators.
The red dot represents the highest weight of the irrep.
A unique HWS occupying this weight is annihilated by the action of each of the raising operator.}
\label{Figure:SU3}
\end{figure}

One approach to labelling the $\mathrm{SU}(n)$ basis states involves specifying the transformation properties under the action of the subalgebras of $\mathfrak{su}(n)$.
We restrict our attention to the canonical subalgebra chain
\begin{equation}
\mathfrak{su}_{1,2,\dots, n}(n)\supset \mathfrak{su}_{1,2,\dots, n-1}(n-1) \supset \dots \supset \mathfrak{su}_{1,2}(2),
\label{Eq:SubalgebraChain}
\end{equation}
where $\mathfrak{su}_{1,2,\dots, m}(m)$ is the subalgebra generated by the operators $\{C_{i,j}\colon i,j\in\{1,2,\dots,m\}\,,i\ne j\}$ and $\{H_k\colon k\in\{1,2,\dots,m-1\}\}$.
Details about the choice of subalgebra chain are presented in~\ref{Appendix:SubAlgebraChoice}.
Henceforth, I drop the subscript and denote $\mathfrak{su}_{1,2,\dots, m}(m)$ by $\mathfrak{su}(m)$.

The canonical basis comprises the eigenstates of the $\mathfrak{su}(m)$ generators for all $m\le n$ according to the following definition.
\begin{defn}(Canonical basis states)
\label{Definition:CanonicalBasisStates}
The canonical basis states of $\mathrm{SU}(n)$ irrep $K^{(n)}$ are those states
\begin{equation}
\Big|\tensor*{\psi}{*^{K^{(n)}}_{\Lambda^{(n)}}^{,\dots,}_{,\dots,}^{K^{(3)},}_{\Lambda^{(3)},}^{K^{(2)}}_{\Lambda^{(2)}}}\Big\rangle
\end{equation}
that have well defined values of
\begin{enumerate}
\item
irrep labels $K^{(m)}$ for $\mathfrak{su}(m)$ algebras for all $\{m:2\le m\le n\}$ and
\item
$\mathfrak{su}(m)$ weights $\Lambda^{(m)}$, i.e., eigenvalues of the Cartan operators of $\mathfrak{su}(m)$ algebras for all $\{m:2\le m\le n\}$.
\end{enumerate}
\end{defn}
\noindent
Consider the example of the $(\kappa_1,\kappa_2) = (1,1)$ irrep of $\mathrm{SU}(3)$.
There are two basis states with the weight $(\lambda_1,\lambda_2) = (0,0)$.
We can identify these two states by specifying
\begin{enumerate}
\item
the $\mathfrak{su}(3)$ irrep label $K^{(3)} = (\kappa_1,\kappa_2) = (1,1)$ and the $\mathfrak{su}(2)$ irrep label $K^{(2)} = (\kappa_1) = (0)$ or $K^{(2)} = (\kappa_1) = (1)$.
\item
the $\mathfrak{su}(3)$ weights $\Lambda^{(3)} = (\lambda_1,\lambda_2) = (0,0)$ and $ \mathfrak{su}(2)$ weight $\Lambda^{(2)} = (\lambda_1) = (0)$.
\end{enumerate}
The connection between our labelling of canonical basis states of Definition~\ref{Definition:CanonicalBasisStates} and the Gelfand-Tsetlin patterns~\cite{Gelfand1988} is detailed in~\ref{Appendix:Connection}.
The canonical basis state $\Big|\tensor*{\psi}{*^{K^{(n)}}_{\Lambda^{(n)}}^{,\dots,}_{,\dots,}^{K^{(3)},}_{\Lambda^{(3)},}^{K^{(2)}}_{\Lambda^{(2)}}}\Big\rangle
$ for which $K^{(m)} = \Lambda^{(m)}$ for all $m\in\{2,\dots,n\}$ is the highest weight of the irrep $K^{(n)}$.

The relative phases between the canonical basis states are fixed by comparing with the phase of the HWS~\cite{Gelfand1988}.
Matrix elements of the simple raising operators $C_{\ell,\ell+1},\, \ell \in\{1,\dots,n-1\}$ are set as positive~\cite{Barut1986}.
Thus, I impose the following additional constraint on the canonical basis states
\begin{equation}
\mel{\psi_\mathrm{hws}}{c_{1,2}^{p_{1,2}} c_{2,3}^{p_{2,3}}\cdots c_{n-1,n}^{p_{n-1,n}}}{\tensor*{\psi}{*^{K^{(n)}}_{\Lambda^{(n)}}^{,\dots,}_{,\dots,}^{K^{(3)},}_{\Lambda^{(3)},}^{K^{(2)}}_{\Lambda^{(2)}}}} \ge 0,
\label{Eq:PhaseConvention}
\end{equation}
for all canonical basis states, for positive integers $p_{\ell,\ell+1}$.
%If depend on the weight~$\Lambda^{(n)}$ of the irrep.

%To finalize the standardization of our construction, I need to fix the relative phases between states in different $\mathfrak{su}(m-1)$ irreps of $\mathfrak{su}(m)$.
%The sign of various terms in the polynomial expression of the highest weight state is fixed by the ordering of rows and columns in the determinental expression of this state, and I use the highest weight as reference.
%We fix the overall phase of a state relative to the highest weight so the matrix elements of the simple raising operators between any two states is positive.
%This is a recursive procedure since, in general, the simple raising operators do not connect arbitrary states with the highest weight state.
%In practice, this phase convention implies that, for any state $\ket{\psi}$, I have
%$\bra{\psi_\mathrm{hws}} c_{1,2}^{p_{1,2}} c_{2,3}^{p_{2,3}}... c_{n-1,n}^{p_{n-1,n}}\ket{\psi} > 0$ for those values of the exponents $p_{k,k+1}$ (which are determined completely by the weight of $\ket{\psi}$ ) which produce a non-zero matrix elements.

$\mathcal{D}$-functions are the matrix elements of $\mathrm{SU}(n)$ irreps.
The rows and columns of $\mathrm{SU}(n)$ matrix representations are labelled by $\mathrm{SU}(n)$ basis states.
The expression for $\mathrm{SU}(n)$ $\mathcal{D}$-functions generalize those of the $\mathrm{SU}(2)$ $\mathcal{D}$-functions~(\ref{Eq:D}) with $M,M^\prime$ replaced by suitable labels for weights and $J$ replaced by suitable subalgebra labels.
\begin{defn}[$\mathcal{D}$-functions]
$\mathcal{D}$-functions of an $\mathrm{SU}(n)$ transformation $V(\Omega)$ are
\begin{equation}
\tensor*{\D}{*^{K^{(n)}}_{\Lambda^{(n)}}^{,\dots,}_{,\dots,}^{K^{(3)},}_{\Lambda^{(3)},}^{K^{(2)}}_{\Lambda^{(2)}}^;_;^{K^{\prime(n)}}_{\Lambda^{\prime(n)}}^{,\dots,}_{,\dots,}^{K^{\prime(3)},}_{\Lambda^{\prime(3)},}^{K^{\prime(2)}}_{\Lambda^{\prime(2)}}}
(\Omega)
\defeq \Big\langle{\tensor*{\psi}{*^{K^{(n)}}_{\Lambda^{(n)}}^{,\dots,}_{,\dots,}^{K^{(3)},}_{\Lambda^{(3)},}^{K^{(2)}}_{\Lambda^{(2)}}}}\Big|V(\Omega)\Big|{\tensor*{\psi}{*^{K^{\prime(n)}}_{\Lambda^{\prime(n)}}^{,\dots,}_{,\dots,}^{K^{\prime(3)},}_{\Lambda^{\prime(3)},}^{K^{\prime(2)}}_{\Lambda^{\prime(2)}}}\Big\rangle
},
\label{Eq:Defined}
\end{equation}
where $\Omega = \{\omega_1,\omega_2,\dots,\omega_{n^2-1}\}$ is the set of $n^2-1$ independent angles that parameterize an $\mathrm{SU}(n)$ transformation~\cite{Reck1994}.
\end{defn}
\noindent Note that $\mathrm{SU}(n)$ $\mathcal{D}$-functions~(\ref{Eq:Defined}) are non-zero only if the left and the right states belong to the same $\mathrm{SU}(n)$ irrep, i.e.,
\begin{equation}
K^{(n)} \ne K^{\prime(n)}\implies \,\tensor*{\D}{*^{K^{(n)}}_{\Lambda^{(n)}}^{,\dots,}_{,\dots,}^{K^{(3)},}_{\Lambda^{(3)},}^{K^{(2)}}_{\Lambda^{(2)}}^;_;^{K^{\prime(n)}}_{\Lambda^{\prime(n)}}^{,\dots,}_{,\dots,}^{K^{\prime(3)},}_{\Lambda^{\prime(3)},}^{K^{\prime(2)}}_{\Lambda^{\prime(2)}}}
(\Omega) = 0.
\label{Eq:OrthogonalityDs}
\end{equation}
$\mathcal{D}$-functions of an irrep $K$ refer to those $\mathcal{D}$-functions for which $K^{(n)} = K^{\prime(n)} = K$.
The notation~(\ref{Eq:Defined}) is employed in our $\mathcal{D}$-function algorithm presented in Section~\ref{Sec:SunAlgorithms}.

%In principle the $\mathcal{D}$-functions of $\mathrm{SU}(n)$ can obtained by exponentiating the representation of the algebra.
%Multiple techniques have been devised to obtain matrix elements of the algebra elements because observables and Hamiltonians are often expressed as elements of an algebra
%However, exponentiation is rarely practical: for one matrix realizations typically grow exponentially in the size of the problem; for another one might require only specific matrix elements of matrix representing a transformation rather than the entire matrix itself.
We approach the task of constructing $\mathrm{SU}(n)$ $\mathcal{D}$-functions by using boson realizations of $\mathrm{SU}(n)$ states.
In the next section, I define boson realizations and illustrate the construction of $\mathrm{SU}(2)$ $\mathcal{D}$-functions using $\mathrm{SU}(2)$ boson realizations.

%----------------------------------------%
\section{Boson realizations of $\mathrm{SU}(n)$}
\label{Sec:BosonRealizations}
%----------------------------------------%
In this section, I describe boson realizations, which map $\mathfrak{su}(n)$ operators and carrier-space states to operators and states of a system of $n-1$ species of bosons on $n$ sites respectively.
We first present the mapping for $n=2$ and illustrate $\mathrm{SU}(2)$ $\mathcal{D}$-functions calculation using the $\mathrm{SU}(2)$ boson realization.
The section concludes with a discussion of boson realizations of $\mathrm{SU}(n)$ for arbitrary $n$.

The commutation relations~(\ref{Eq:SU2RaisingLoweringCommutations}) of $\{C_{1,2},C_{2,1},H_1\}$ are reproduced by number-preserving bilinear products of creation and annihilation operators that act on a two-site bosonic system.
Specifically, the $\mathfrak{su}(2)$ operators have the boson realization
\begin{equation}
 C_{1,2} \mapsto c_{1,2} \defeq a_1^\dagger a_2\, ,\quad
C_{2,1} \mapsto c_{2,1} \defeq a_2^\dagger a_1 \, ,\quad
H_1 \mapsto h_1 \defeq a_1^\dagger a_1 - a_2^\dagger a_2,
\label{Eq:su2bosonmap}
\end{equation}
where the bosonic creation and annihilation operators obey the commutation relations
\begin{equation}
\left[a_i,a_j^\dagger\right] = \delta_{ij}\mathds{1}, \quad \left[a_i,a_j\right] = \left[a_i^\dagger,a_j^\dagger\right] = 0.
\label{Eq:ccr}
\end{equation}
Here and henceforth, I use lower-case symbols for boson realizations of the respective upper-case symbols.
Explicitly,
\begin{equation}
[h_1,c_{1,2}] = 2c_{1,2}\,\qquad
[h_1,c_{2,1}] = - 2c_{2,1}\,\qquad
[c_{1,2},c_{2,1}] = h_1.
\end{equation}
The operators $\{c_{1,2},c_{2,1},h_1\}$ also span the complex extension of the $\mathfrak{su}(2)$ Lie algebra.

Boson realizations map the states in the carrier space of $\mathrm{SU}(2)$ to the states of a two-site bosonic system.
Specifically, each basis state of the $(2J+1)$-dimensional $\mathrm{SU}(2)$ irrep maps
\begin{equation}
\ket{J,M} \mapsto \frac{(a_1^\dagger)^{J+M} (a_2^\dagger)^{J-M} }{\sqrt{(J+M)!(J-M)!}}
\ket{0}
\label{Eq:su2polynomialstate}
\end{equation}
to the state of a two-site system with $J+M$ and $J-M$ bosons in the two sites respectively.

The $(2J+1)$-dimensional irreps of $\mathrm{SU}(2)$ map to number-preserving transformations on a two-site system of $2J$ bosons in the basis of Equation~(\ref{Eq:su2polynomialstate}).
The elements of these $(2J+1)\times(2J+1)$ matrices are the $\mathrm{SU}(2)$ $\mathcal{D}$-functions
\begin{equation}
\mathcal{D}^J_{M^\prime M}(\Omega)
		\defeq \bra{J,M^\prime}V(\Omega)\ket{J,M}
\label{Eq:D}
\end{equation}
for irrep $J$ and row and column indices $M^\prime, M$.
The expression for $\mathcal{D}$-functions~(\ref{Eq:D}) of $\mathrm{SU}(2)$ element $V(\Omega)$ can be calculated by noting that the creation operators transform under the action of $V$ of Equation~(\ref{Eq:2x2Rmatrix}) according to
\begin{align}
a_1^\dagger\to V_{11}a_1^\dagger + V_{12}a_2^\dagger, \nonumber \\
a_2^\dagger\to V_{21}a_1^\dagger + V_{22}a_2^\dagger,
\end{align}
where $V$ is the $2\times 2$ fundamental representation of $V(\Omega)$.
The state $\ket{J,M}$~(\ref{Eq:su2polynomialstate}) thus transforms to
\begin{equation}
\ket{J,M}\to \frac{\left(V_{11}a_1^\dagger + V_{12}a_2^\dagger\right)^{J+M}\left(V_{21}a_1^\dagger + V_{22}a_2^\dagger\right)^{J-M}}
{\sqrt{(J+M)!(J-M)!}}\ket{0}\,
\label{Eq:transformedsu2polynomialstate}
\end{equation}
as the vacuum state $\ket{0}$ is invariant under the action $V$.
Using Equations~(\ref{Eq:su2polynomialstate}) and (\ref{Eq:transformedsu2polynomialstate}), I obtain
\begin{equation}
\mathcal{D}^J_{M^\prime M}(\Omega) =\Bigg\langle0\Bigg|\frac{a_1^{J+M'}a_2^{J-M'}\left(V_{11}a_1^\dagger + V_{12}a_2^\dagger\right)^{J+M}\left(V_{21}a_1^\dagger + V_{22}a_2^\dagger\right)^{J-M}}{\sqrt{(J+M')!(J-M')!}\sqrt{(J+M)!(J-M)!}}\Bigg|0\Bigg\rangle,
\label{Eq:DFunction}
\end{equation}
which can be evaluated using the commutation relations of the creation and annihilation operators~(\ref{Eq:ccr}).%
\footnote{
A useful computational shortcut involves the map
\begin{equation}
a_k^\dagger\to x_k\, ,\quad a_\ell\to \frac{\partial}{\partial x_\ell}\,,\quad k,\ell \in\{1,2\},
\label{Eq:PolynomialMap}
\end{equation}
which preserves the boson commutation relations.
The map~(\ref{Eq:PolynomialMap}) transforms the vector~$\ket{J,M}$~(\ref{Eq:su2polynomialstate}) into a formal polynomial and the corresponding dual vector~$\bra{J,M}$ into a linear differential operator in the dummy variables $x_1,x_2$.
The $\mathcal{D}$-function~(\ref{Eq:DFunction}) is thus evaluated as the action of a linear differential operator on a polynomial in $x_1,x_2$.
}

In Chapter~\ref{Chap:Sun}, our objective is to generalize Equations~(\ref{Eq:su2bosonmap}) and (\ref{Eq:su2polynomialstate}) systematically from $n = 2$ to arbitrary $n$.
In the remainder of this section, I define boson realizations of operators and carrier-space states of $\mathfrak{su}(n)$.
Furthermore, I construct the boson realization for the HWS of arbitrary $\mathrm{SU}(n)$ irreps.

$\mathrm{SU}(n)$ boson realizations map $\mathrm{SU}(n)$ states $\Big|{\tensor*{\psi}{*^{K^{(n)}}_{\Lambda^{(n)}}^{,\dots,}_{,\dots,}^{K^{(3)},}_{\Lambda^{(3)},}^{K^{(2)}}_{\Lambda^{(2)}}}}\Big\rangle$ and $\mathfrak{su}(n)$ operators to states and operators of a system of bosons on $n$ sites.
Bosons are labelled based on the site $i\in \{1,2,\dots,n\}$ at which they are situated and by an internal DOF, which is denoted by an additional subscript on the bosonic operators.
The bosonic creation and annihilation operators on this system are
\begin{align}
a^\dagger_{i,j}\colon i\in \{1,2,\dots,n\}, j \in \{1,2,\dots,n-1\}&, \quad \text{(Creation)}\\
a_{k,l}\colon k\in \{1,2,\dots,n\}, l\in \{1,2,\dots,n-1\}&,\quad \text{(Annihilation)},
\end{align}
where the first label in the subscript is the usual index of the site occupied by the boson.
The second index refers to the internal degrees of freedom of the boson.
Each boson can have at most $n-1$ possible internal states to ensure that basis states can be constructed for arbitrary irreps.
In photonic experiments, this internal DOF could correspond to the polarization, frequency, orbital angular momementum or the time of arrival of photons.

The $\mathfrak{su}(n)$ operators are mapped to number-preserving bilinear products of boson creation and annihilation operators.
Specifically, raising and lowering operators $C_{i,j}$ of $\mathfrak{su}(n)$ map to bosonic operators $c_{i,j}$ according to
\begin{equation}
C_{i,j} \mapsto c_{i,j} \defeq \sum_{k=1}^{n-1} a^\dagger_{i,k} a_{j,k}.
\end{equation}
Operators $\{c_{i,j}\}$ make bosons hop from site $j$ to site $i$.
The operators $h_i$ are the image of the Cartan operators $H_i$:
\begin{equation}
H_{i} \mapsto h_{i} \defeq a^\dagger_i a_i - a^\dagger_{i+1} a_{i+1}.
\end{equation}
Operators $\{h_i\}$ count the difference in the total number of bosons at two sites and commute among themselves.
As usual, I used the upper-case symbols to denote the $\mathfrak{su}(n)$ elements and the corresponding lower-case symbols for the respective boson operators.
%one can decompose any irreducible representation of the unitary group {U(n)} into a canonical basis (the Gelfand-Tsetlin basis), indexed by integer-valued Gelfand-Tsetlin patterns, by first decomposing this representation into irreducible representations of {U(n-1)}, then {U(n-2)} and so forth.

The boson realizations of the basis states of $\mathrm{SU}(n)$ are obtained by the action of polynomials in creation operators $\left\{a^\dagger_{i,j}\colon i\in \{1,2,\dots,n\},j\in \{1,2,\dots,n-1\}\right\}$ on the $n$-site vacuum state $\ket{0}$.
Each term in the polynomial is a product of
\begin{equation}
N_K= \kappa_1 + 2\kappa_2 + \dots + (n-1)\kappa_{n-1}
\label{Eq:N}
\end{equation}
boson creation operators for basis states in irreps $K = (\kappa_1,\kappa_2,\dots,\kappa_{n-1})$.
Therefore, an $\mathrm{SU}(n)$ basis state is specified by the coefficient of a polynomial consisting of terms that are products of $N_K$ creation operators.

Now I introduce a compressed notation for $\mathcal{D}$-functions.
The compressed notation is based on the boson realization of $\mathrm{SU}(n)$ states and on the orthogonality~(\ref{Eq:OrthogonalityDs}) of $\mathcal{D}$-functions.
In this notation, the $\mathcal{D}$-function
\begin{equation}
\tensor*{\D}{*^{K^{(n)}}_{\Lambda^{(n)}}^{,\dots,}_{,\dots,}^{K^{(3)},}_{\Lambda^{(3)},}^{K^{(2)}}_{\Lambda^{(2)}}^;_;^{K^{\prime(n)}}_{\Lambda^{\prime(n)}}^{,\dots,}_{,\dots,}^{K^{\prime(3)},}_{\Lambda^{\prime(3)},}^{K^{\prime(2)}}_{\Lambda^{\prime(2)}}}
(\Omega); \quad K^{(n)} = K^{\prime(n)}
\end{equation}
is represented by
\begin{equation}
\mathcal{D}^{K^{(n)}}_{\nu_{1}\dots \nu_{n},K^{(n-1)},\dots,K^{(2)};\nu^{\prime }_{1}\dots \nu^{\prime }_{n},K^{\prime (n-1)},\dots,K^{\prime (2)}}
\label{Eq:CompressedNotation}
\end{equation}
where $\nu_{1}\dots \nu_{n}$ are the bosonic occupation numbers of the respective states, suffice to uniquely identuify each of the weights $\{\lambda^{(m)}; 1<m\le n\}$.
Sections~\ref{Sec:ThreePhotons} and~\ref{Sec:SunImmanantResult} employ the compressed notation~(\ref{Eq:CompressedNotation}).

The HWS of a given $\mathrm{SU}(n)$ irrep can be explicitly constructed in the boson realization (as polynomials in creation and annihilation operators) according to the following lemma.
\ignorespaces
\begin{lem}[Boson realization of HWS~\cite{Moshinsky1962a,Moshinsky1962}]
\label{Lemma:HWS}
The bosonic state
\begin{equation}
\ket{\psi^K_\mathrm{HWS}}
=
\Det \begin{pmatrix}a_{1,1}^\dagger&\dots&a_{1,n-1}^\dagger\\
\vdots&\ddots&\vdots\\
a_{n-1,1}^\dagger&\dots &a_{n-1,n-1}^\dagger\end{pmatrix}^{\kappa_{n-1}}\cdots\,
\Det \begin{pmatrix}a_{1,1}^\dagger&a_{1,2}^\dagger\\a_{2,1}^\dagger&a_{2,2}^\dagger\end{pmatrix}^{\kappa_2}
\Det \left({a_{1,1}^\dagger}\right)^{\kappa_1}\ket{0}
\label{Eq:Hws}
\end{equation}
is a HWS for a given $\mathrm{SU}(n)$ irrep $K = (\kappa_1,\kappa_2,\dots,\kappa_{n-1})$.
\end{lem}
\noindent \ignorespaces
One can verify that the state $\ket{\psi^K_\mathrm{hws}}$~(\ref{Eq:Hws}) is annihilated
\begin{equation}
c_{j,k} \ket{\psi^K_\mathrm{hws}} = 0 \quad \forall j<k
\end{equation}
by the action of any of the raising operators.

Thus, the HWS of any irrep can be constructed analytically using Lemma~\ref{Lemma:HWS}.
In the Section~\ref{Sec:SunAlgorithms}, I present an algorithm to construct each of the basis states of arbitrary $\mathrm{SU}(n)$ irreps.
Furthermore, I present an algorithm to compute expressions for $\mathrm{SU}(n)$ $\mathcal{D}$-functions in terms of the entries of the fundamental representation.
The next section shows how $\mathcal{D}$-functions and immanants are connected to outputs of linear interferometry.
%----------------------------------------%
\section{Determinants, immanants and permanents of a matrix}
\label{Sec:Immanants}
%----------------------------------------%
This section presents relevant definitions and background on the immanants of matrices, which include determinants and permanents as special cases.
The immanants of the interferometer transformation matrix are important in multi-photon interferometry because they manifest the permutation symmetries of the interfering photons~\cite{Guise2014}.
I define determinants, permanents and immanants before detailing the connection between immanants and interferometer output probabilities in the next section.

The determinant and permanent of a matrix are defined respectively as follows.
\begin{defn}[Determinant and permanent of a matrix]
The determinant of an $n\times n$ matrix $T$ is the antisymmetric sum
\begin{equation}
\operatorname{det}(T)\defeq \sum_{\sigma}\operatorname{sgn}(\sigma)T_{1\sigma(1)}T_{2\sigma(2)}\ldots T_{m\sigma(m)}\, ,
\end{equation}
where the sum is over all permutations $\sigma$ over $\{1,2,\dots, n\}$ and $\operatorname{sgn}(\sigma)$ represents the parity of the permutation.
Likewise, the permanent of $T$ is the symmetric sum
\begin{equation}
\operatorname{per}(T)\defeq \sum_{\sigma}T_{1\sigma(1)}T_{2\sigma(2)}\ldots T_{m\sigma(m)}
\label{Eq:Permanent}
\end{equation}
over all permutations $\sigma$ over $\{1,2,\dots, n\}$.
\end{defn}
\noindent Although the determinant of an $n\times n$ complex matrix can be computed or approximated in time polynomial in $n$, the approximation of a matrix permanent is a~\gls{SharpPC} problem~\cite{Bunch1974,Valiant1979}.
In fact, the proof~\cite{Aaronson2013} of hardness of the BosonSampling problem relies on the hardness of approximating the permanent.

Immanants are matrix functions whose hardness is intermediate between determinants and permanents~\cite{Burgisser2000}.
\begin{defn}[Immanant of a matrix]
The immanant~$\operatorname{imm}^{\{\tau\}}(T)$ of the $m\times m$ matrix
$T$, associated with the partition $\{\tau\}$, is defined as~\cite{Littlewood1950}
\begin{equation}
\operatorname{imm}^{\{\tau\}}(T)\defeq \sum_{\sigma}\chi^{\{\tau\}}(\sigma)\,T_{1\sigma(1)}T_{2\sigma(2)}\ldots T_{m\sigma(m)}\, ,
\label{defineimmanant}
\end{equation}
where
$\sigma\in \mathrm{S}_m$ permutes $k$ to $\sigma(k)$, and $\chi^{\{\tau\}}(\sigma)$ is the character%
\footnote{The character of an element $\sigma$ of an $\mathrm{S}_{n}$ representation $\{\lambda\}$ is the trace of the matrix representing~$\sigma$.
For instance, the trivial representation $\Yboxdim{6pt}\Yvcentermath1 \yng(3)$ of $\mathrm{S}_{3}$ represents each element of $\mathrm{S}_{3}$ as a $1\times 1$ matrix with entry unity. Thus, $\chi^{\Yboxdim{3pt}\Yvcentermath1 \yng(3)}(\sigma)= 1$ for all $\sigma\in\mathrm{S}_{3}$. The character functions for $\mathrm{S}_{2}$ and $\mathrm{S}_{3}$ are presented in Tables~\ref{Tab:S2} and~\ref{Tab:S3}~\cite{Littlewood1950}.
}
 of $\sigma$ in the irrep~$\{\tau\}$ of~$S_m$.
\end{defn}
\noindent The determinant and permanent of a matrix are the immanants associated with the alternating ($\chi^{\{\tau\}}(\sigma) = \operatorname{sgn}(\sigma)$) and the trivial $(\chi^{\{\tau\}}(\sigma) = 1)$ characters of $\mathrm{S}_{n}$ respectively.

\begin{table}
\centering
\begin{tabular}{|c|c|c|c|c|}
\hline
&$\mathds{1}$	& $\sigma_{ab} \defeq\{P_{12}\}$\\
\hline
$\lambda$ & $\chi^\lambda(\mathds{1})$ & $\chi^\lambda(\sigma_{ab})$\phantom {\Yvcentermath1\Yboxdim{8pt}$\yng(1,1)$} \\
\hline
\Yboxdim{6pt}\yng(2) & 1 & 1 \\
\Yboxdim{6pt}\yng(1,1)& 1 &-1 \\
\hline
\end{tabular}
\caption{Character table for $\mathrm{S}_{2}$. The first row labels the different elements of the permutation group and the first column comprises the $\mathrm{S}_{2}$ irreps.}
\label{Tab:S2}
\end{table}

\begin{table}
\centering
\begin{tabular}{|c|c|c|c|c|}
\hline
&$\mathds{1}$	& $\sigma_{ab} \defeq\{P_{12},P_{13},P_{23}\}$&$\sigma_{abc}\defeq\{P_{123},P_{132}\}$\\
\hline
$\lambda$ & $\chi^\lambda(\mathds{1})$ & $\chi^\lambda(\sigma_{ab})$\phantom {\Yvcentermath1\Yboxdim{8pt}$\yng(1,1)$}
&$\chi^\lambda(\sigma_{abc})$ \\
\hline
\Yboxdim{6pt}\yng(3) & 1 & 1 &1 \\
\Yboxdim{6pt}\yng(2,1) & 2 & 0 & -1 \\
\Yboxdim{6pt}\yng(1,1,1)& 1 &-1 &1 \\
\hline
\end{tabular}
\caption{Character table for $\mathrm{S}_{3}$. The first row labels the different elements of the permutation group. The first column comprises the three $\mathrm{S}_{3}$ irreps, which are identified with the permanent, immanant and determinant respectively}
\label{Tab:S3}
\end{table}
For $\mathrm{S}_{2}$, the determinant, which is labelled by \Yboxdim{6pt}\Yvcentermath1 \yng(1,1), and the permanent $\Yboxdim{6pt}\Yvcentermath1 \yng(2)$ are the only two immanants
\begin{align}
\Yboxdim{6pt}\Yvcentermath1 \yng(1,1) =&\,U_{11} U_{22} -U_{12} U_{21} \\
\Yboxdim{6pt}\Yvcentermath1 \yng(2) =&\,U_{11} U_{22} +U_{12} U_{21}.
\end{align}
In the case of $\mathrm{S}_{3}$, there are three immanants, including the permanent $\Yboxdim{6pt}\Yvcentermath1 \yng(3)$, determinant $\Yboxdim{6pt}\Yvcentermath1 \yng(1,1,1) $ and another immanants labelled by the partially symmetric representation~$\Yboxdim{6pt}\Yvcentermath1\yng(2,1)$ of the permutation group.
The three immanants of $3\times 3$ matrix~$U$ can be expressed in terms of matrix elements as
\begin{align}
\Yboxdim{6pt}\Yvcentermath1 \yng(1,1,1) =&\,U_{11} U_{22} U_{33}-U_{11} U_{23} U_{32}-U_{12} U_{21} U_{33}+U_{12} U_{23} U_{31}+U_{13}U_{21} U_{32}-U_{13} U_{22} U_{31}\\
\Yboxdim{6pt}\Yvcentermath1 \yng(2,1) =&\,2U_{11} U_{22} U_{33}-U_{12}U_{23}U_{31}-U_{13}U_{21}U_{32}\\
\Yboxdim{6pt}\Yvcentermath1 \yng(3) =&\,U_{11} x_{22} U_{33}+U_{11} U_{23} U_{32}-U_{12} U_{21} U_{33}+U_{12} U_{23} U_{31}+U_{13}U_{21} U_{32}+U_{13} U_{22} U_{31}.
\end{align}
In summary, I have defined the immanants of an $n\times n$ matrix and have given examples of immanants for two- and three-channel interferometers.

The probability of detecting photon coincidence permanent value of the transformation submatrix if indistinguishable single-photons are incident~\cite{Troyansky1996}.
If mutual time delay is introduced between the photons, other immanants also arise in the coincidence expression, as I detail in the next section for the three-photon case.

%----------------------------------------%
\section{$\mathrm{SU}(3)$ and $\mathrm{S}_{3}$ methods for three-photon interferometry}
\label{Sec:ThreePhotons}
%----------------------------------------%
Here I present the connection between three-photon coincidence probabilities and the immanants and $\mathcal{D}$-functions of the transformation matrix.
Specifically, I present expressions for three-fold coincidence probabilities when three controllably-delayed single photons are incident at an interferometer.

Consider three photons with square-integrable spectra $f(\omega)$ incident at a three-channel interferometer $U$ at arrival times $\tau_{1},\tau_{2}$ and $\tau_{3}$ respectively.
The probability of obtaining a three-fold coincidence at the output of the interferometer is
\begin{align}
	\wp
		=&\int \text{d}\omega_{1}\,\text{d}\omega_{2}\,\text{d}\omega_{3}\,
		\left|f(\omega_1)\right|^2
		\left|f(\omega_2)\right|^2
		\left|f(\omega_3)\right|^2\nonumber\\&
		\Big|\,U_{11}U_{22}U_{33}\mathrm{e}^{\mathrm{i}(\omega_1\tau_1+\omega_2\tau_2+\omega_3\tau_3)}
		+U_{11}U_{23}U_{32} \mathrm{e}^{\mathrm{i}(\omega_1\tau_1+\omega_3\tau_2+\omega_2\tau_3)}
		+U_{12}U_{21}U_{33} \mathrm{e}^{\mathrm{i}(\omega_2\tau_1+\omega_1\tau_2+\omega_3\tau_3)}\nonumber\\&
		+U_{12}U_{23}U_{31} \mathrm{e}^{\mathrm{i}(\omega_2\tau_1+\omega_3\tau_2+\omega_1\tau_3)}
		+U_{13}U_{21}U_{32} \mathrm{e}^{\mathrm{i}(\omega_3\tau_1+\omega_1\tau_2+\omega_2\tau_3)}
		+U_{13}U_{22}U_{31} \mathrm{e}^{\mathrm{i}(\omega_3\tau_1+\omega_2\tau_2+\omega_1\tau_3)}\Big|^2\nonumber.
\end{align}
This coincidence probability can be expressed in terms of sums of specific $\mathrm{SU}(3)$ $\mathcal{D}$-functions using the Schur-Weyl duality~\cite{Guise2014}
\begin{align}
	U_{1i}U_{1j}U_{1k}\equiv&\, \mathcal{D}^{(1,0)}_{i,(100)}\mathcal{D}^{(1,0)}_{j,(010)}\mathcal{D}^{(1,0)}_{k,(001)}\nonumber\\
		=&\,c^{\Yboxdim{4pt}\yng(3)}_{ijk}\mathcal{D}^{\Yboxdim{4pt}\yng(3)}_{(111)1;(111)1}
		+c^{\Yboxdim{4pt}\yng(2,1)}_{ijk,(11)}\mathcal{D}^{\Yboxdim{4pt}\yng(2,1)}_{(111)1;(111)1}+c^{\Yboxdim{4pt}\yng(2,1)}_{ijk,(00)}\mathcal{D}^{\Yboxdim{4pt}\yng(2,1)}_{(111)0;(111)0}
						\nonumber\\&
		+c^{\Yboxdim{4pt}\yng(2,1)}_{ijk,(10)}\mathcal{D}^{\Yboxdim{4pt}\yng(2,1)}_{(111)1;(111)0}
		+c^{\Yboxdim{4pt}\yng(2,1)}_{ijk,(01)}\mathcal{D}^{\Yboxdim{4pt}\yng(2,1)}_{(111)0;(111)1}
		\label{Eq:SWD}
		+c^{\Yboxdim{4pt}\yng(1,1,1)}_{ijk}\mathcal{D}^{\Yboxdim{4pt}\yng(1,1,1)}_{(111)0;(111)0}
\end{align}
for constants~$c^{\{\lambda\}}_{ijk}$ where I have used the compressed $\mathcal{D}$-function notation~(\ref{Eq:CompressedNotation}).
Similarly immanants.
Although expressions for the case of arbitrary time delay $\bm{\tau} \defeq (\tau_{1},\tau_{2},\tau_{3})$ depend on each of the six $\mathcal{D}$-functions of Equation~(\ref{Eq:SWD}), certain time-delay values contain fewer $\mathcal{D}$-functions.

For instance, consider the case of two photon arriving simultaneously and one photon delayed by time $\tau$ with respect to the other two and assuming Gaussian spectra for simplicity.
The coincidence probability is
\begin{equation}
\label{eq:2indenticalphotons}
	\wp=\vert A\vert^2 + \vert B\vert^2+ \vert C\vert^2+\text{e}^{-\sigma^2\tau^2}\big[(A^*+B^*)C+(A^*+C^*)B+(B^*+C^*)A \big],
\end{equation}
where
the functions $A$, $B$ and $C$ can be expressed in terms of $\mathcal{D}^{\{\lambda\}}$-functions by
\begin{align}
\label{eq:23symmetry}
	A=&\,\frac{1}{3}\left(\mathcal{D}^{\Yboxdim{4pt}\Yvcentermath1\yng(3)}_{(111)1;(111)1}
	+2 \mathcal{D}^{\Yboxdim{4pt}\Yvcentermath1\yng(2,1)}_{(111)1;(111)1}\right) \nonumber\\
	B=&\,\frac{1}{3}\left(\mathcal{D}^{\Yboxdim{4pt}\Yvcentermath1\yng(3)}_{(111)1;(111)1}- \mathcal{D}^{\Yboxdim{4pt}\Yvcentermath1\yng(2,1)}_{(111)1;(111)1}+\sqrt{3}\mathcal{D}^{\Yboxdim{4pt}\Yvcentermath1\yng(2,1)}_{(111)0;(111)1}\right)\\
	C=&\,\frac{1}{3}\left(\mathcal{D}^{\Yboxdim{4pt}\Yvcentermath1\yng(3)}_{(111)1;(111)1}
	- \mathcal{D}^{\Yboxdim{4pt}\Yvcentermath1\yng(2,1)}_{(111)1;(111)1}-\sqrt{3}\mathcal{D}^{\Yboxdim{4pt}\Yvcentermath1\yng(2,1)}_{(111)0;(111)1}\right).\nonumber
\end{align}
Alternatively, $A$, $B$ and $C$ can be expressed in terms of immanants by using the following identity~\cite{Kostant1995}
\begin{align}
{\Yboxdim{6pt}\Yvcentermath1 \yng(3)}&=D^{(3)}_{111(1);111(1)}(\Omega)\, ,\nonumber \\
{\Yboxdim{6pt}\Yvcentermath1 \yng(2,1)}&=D^{(11)}_{111(1);111(1)}(\Omega)
+D^{(11)}_{111(0);111(0)}(\Omega)\, ,\\
{\Yboxdim{6pt}\Yvcentermath1 \yng(1,1,1)}&=D^{(0)}_{000(0);000(0)}(\Omega)=1\, ,\nonumber
\end{align}
which connects the $\mathcal{D}$-functions and immanants of the fundamental matrix representation of the interferometer transformation.
Permuting the columns of the interferometer matrix, the functions $A$, $B$ and $C$ are related to immanants by~\cite{Guise2014}
\begin{align}
\label{eq:ABC}
A&=U_{11}U_{22}U_{33} +U_{11}U_{23}U_{32}=\frac{1}{3}\left({\Yboxdim{6pt}\Yvcentermath1\yng(3)}
	+{\Yboxdim{6pt}\Yvcentermath1 \yng(2,1)}_{123}
	+{\Yboxdim{6pt}\Yvcentermath1 \yng(2,1)}_{132} \right),\nonumber\\
B&=U_{12}U_{21}U_{33}+U_{12}U_{23}U_{31}=\frac{1}{3}\left({\Yboxdim{6pt}\Yvcentermath1\yng(3)}
	+{\Yboxdim{6pt}\Yvcentermath1 \yng(2,1)}_{213}
	+{\Yboxdim{6pt}\Yvcentermath1 \yng(2,1)}_{231} \right),\\
C&=U_{13}U_{22}U_{31}+U_{13}U_{21}U_{32}=\frac{1}{3}\left({\Yboxdim{6pt}\Yvcentermath1\yng(3)}
	+{\Yboxdim{6pt}\Yvcentermath1 \yng(2,1)}_{312}
	+{\Yboxdim{6pt}\Yvcentermath1 \yng(2,1)}_{321}\right),\nonumber
\end{align}
where the subscripts refer to the permutation operation acted upon the columns of interferometer matrix.

Note that the immanants of the form $\Yboxdim{6pt}\Yvcentermath1 \yng(1,1,1)$ or the $\mathcal{D}^{\Yboxdim{4pt}\Yvcentermath1\yng(1,1,1)}$~functions do not arise in these expressions.
Symmetry arguments justify the only certain immanants and certain $\mathcal{D}$-functions arise in the expression.

In conclusion, I have presented a treatment of three-photon three-channel interferometry in terms of immanants and $\mathcal{D}$-functions of the interferometer transformation.
Chapter~\ref{Chap:Sun} presents methods that enable the generalization of these group-theoretic methods to the multi-photon multi-channel case.

%----------------------------------------%
\section{Cosine-sine decomposition}
\label{Sec:CSD}
%----------------------------------------%
This section presents the relevant background for our procedure to realize arbitrary discrete unitary transformation on the spatial and internal modes of light.
The procedure, which is presented in Chapter~\ref{Chap:Design}, relies on iteratively performing the cosine-sine decomposition~\gls{csd}~\cite{Stewart1977,Stewart1982,Sutton2009}.
In this section, I detail the CSD and present a CSD-based realization of a $4\times 4$ unitary matrix using two spatial and two polarization modes of light.

The CSD factorizes an arbitrary unitary $(m+n)\times (m+n)$ unitary matrix into a product of three block-diagonal matrices as follows.
\begin{thm}\label{Thm:CSD}
For each $(m+n)\times (m+n)$ unitary matrix $U_{m+n}$, there exist unitary matrices $\mathds{L}_{m+n},\mathds{S}_{m+n},\mathds{R}_{m+n}$, such that
\begin{equation}
U_{m+n}= \mathds{L}_{m+n} \left(\mathds{S}_{2m}\oplus \mathds{1}_{n-m}\right)\mathds{R}_{m+n},
\label{Eq:csd}
\end{equation}
where $\mathds{L}_{m+n}$ and $\mathds{R}_{m+n}$ are block-diagonal
\begin{equation}
\mathds{L}_{m+n} =
\left(\begin{array}{c|c}
L_{m}& {0} \\
\hline
0 & L_{n}'
\end{array}\right),~
\mathds{R}_{m+n} =
\left(\begin{array}{c|c}
R^{\dagger}_{m}& {0} \\
\hline
0 & R^{\prime\dagger}_{n}
\end{array} \right)
\end{equation}
and $\mathds{S}_{2m}$ is an orthogonal cosine-sine (\glspl{csm}) matrix
\begin{align}
\mathds{S}_{2m} &\equiv \mathds{S}_{2m}(\theta_{1},\dots,\theta_{m})\nonumber\\
&\defeq\left(\arraycolsep=1pt\def\arraystretch{0.9}\begin{array}{ccc|ccc}
\cos\theta_1 & & &\sin\theta_1 & \\
%&\cos\theta_2 & & && &\sin\theta_2 & &\\
& \ddots & && \ddots &\\
& & \cos\theta_m &&&\sin\theta_m \\
\hline
-\sin\theta_1 & & & \cos\theta_1 & & \\
%&-\sin\theta_2 & & & &\cos\theta_2 & & &\\
& \ddots & && \ddots & \\
& & -\sin\theta_m &&&\cos\theta_m
\end{array}\right).
\label{Eq:CSMatrix}
\end{align}\end{thm}
The decomposition of $U_{m+n}$ into $\mathds{L}_{m+n}$, $\mathds{S}_{2m}$ and $\mathds{R}_{m+n}$ is depicted in Fig.~\ref{Fig:CSD}.
Here and in the remainder of this section, the subscripts of the matrix symbols denote the respective dimensions of the matrix.

\begin{figure}[h]\begin{center}
 \includegraphics[width=\textwidth]{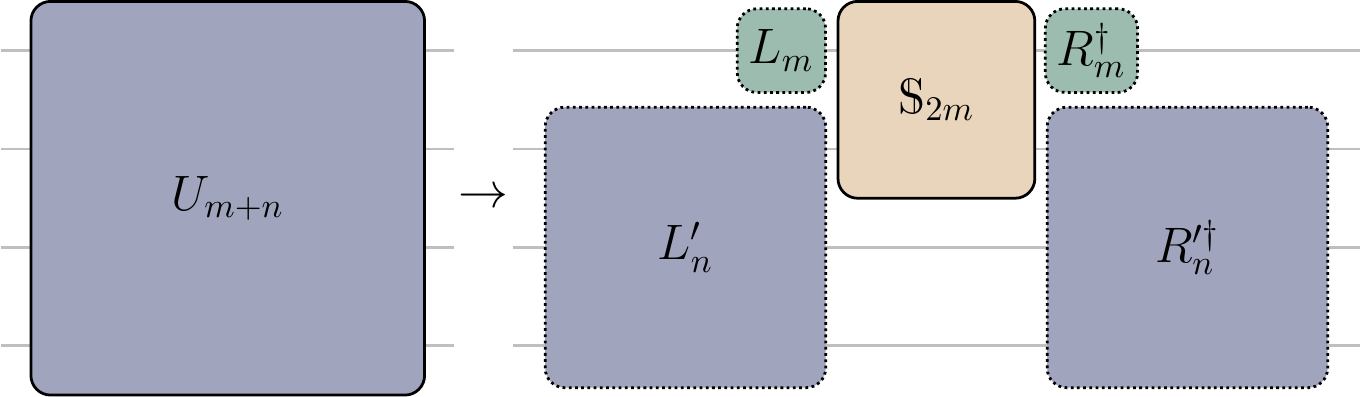}
\end{center}
 \caption{Depiction of the CSD.
 $U_{m+n}$ is an $(m+n)\times (m+n)$ unitary matrix.
 The CSD factorizes $U_{m+n}$ into the block diagonal matrices $\mathds{L}_{m+n}$, $\mathds{S}_{2m}$ and $\mathds{R}^{\dagger}_{m+n}$.
 The boxes labelled $L_{m}$ and $L'_{n}$ represent the block diagonal matrix $\mathds{L}_{m+n} = L_{m}\oplus L'_{n}$. Likewise for $\mathds{R}^{\prime\dagger}_{m+n} = R^{\prime\dagger}_{m}\oplus R^{\prime\dagger}_{n} $ and $\mathds{S}_{m+n} = \mathds{S}_{m+n}\oplus \mathds{1}_{n-m}$.
 }
 \label{Fig:CSD}
\end{figure}

A constructive proof of Theorem~\ref{Thm:CSD} is presented in Appendix~\ref{Appendix:Construction}.
The matrices $\mathds{L}_{m+n}$, $\mathds{S}_{2m}$ and $\mathds{R}_{m+n}$ can be constructed using the singular value decomposition as follows.
In order to perform CSD on $U_{m+n}$, I express it as a $2\times 2$ block matrix
\begin{equation}
 U_{m+n} \equiv\left(\begin{array}{c|c}
 A&B\\\hline
 C&D
 \end{array}\right),
\end{equation}
where $A$ and $D$ are square complex matrices of dimension $m\times m$ and $n\times n$ respectively, and $B$ and $C$ are rectangular with respective dimensions $m\times n$ and $n\times m$.
Each row of the matrix $L_{m}$ ($R_{m}$) is a left-singular (right-singular) vector of $A$, as is proved in Appendix~\ref{Appendix:Construction}.
Similarly, $L^{\prime}_{n}$ and $R_{n}^{\prime}$ are the left- and right-singular vectors of $D$.
Finally, $\{\cos\theta_{i}\}$ is the set of singular values of $A$.
The singular vectors and values of any complex matrix can be computed efficiently using established numerical techniques~\cite{Golub1965,Klema1980,Anderson1992,Press1996}.

Now I illustrate the realization of an arbitrary $4\times 4$ unitary matrix as a linear optical transformation on two spatial and two polarization modes~\cite{Goyal2015}.
The realization is enabled by the CSD, which decomposes the given matrix $U_{4}$ according to
\begin{equation}
 U_4 = \left(
\begin{array}{c|c}
 L_2 & \\
\hline
& L'_2
\end{array}\right)
\mathds{S}_4\left(
\begin{array}{c|c}
 R_2^\dagger & \\
\hline
& R^{\prime\dagger}_2
\end{array}\right)
\label{Eq:CSD4}
\end{equation}
for $m = n = 2$.
The decomposition of unitary matrix $U_{4}$ is depicted in Fig.~\ref{Fig:ns2np2}(a).
By definition, $U_{4}$ acts on the four-dimensional space $\mathcal{H}_{4}$, which I identify with the combined space
\begin{equation}
\mathcal{H}_{4} = \mathcal{H}^{(s)}_{2}\otimes \mathcal{H}^{(p)}_{2}
\end{equation}
of spatial and polarization modes.
Thus, the $2\times 2$~matrices $L_{2}$ and $R^{\dagger}_{2}$ are identified with transformations acting on the two polarization modes of light in the first spatial mode.
Likewise, $L^{\prime}_{2}$ and $R^{\prime\dagger}_{2}$ correspond to transformations on polarization in the second spatial mode.
Each of these operators $L_{2}, L^{\prime}_{2}, R^{\dagger}_{2},R^{\prime\dagger}_{2}$ can be realized with two quarter-wave plates, one half-wave plate and one phase shifter~\cite{Simon1989,Simon1990}.

\begin{figure}[h]
\centering
 \subfloat{\includegraphics[width=0.8\textwidth]{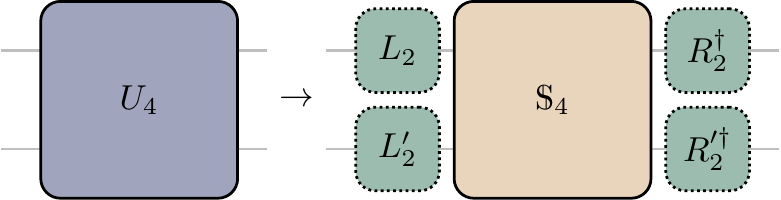}}\\
 \subfloat{\includegraphics[width=0.8\textwidth]{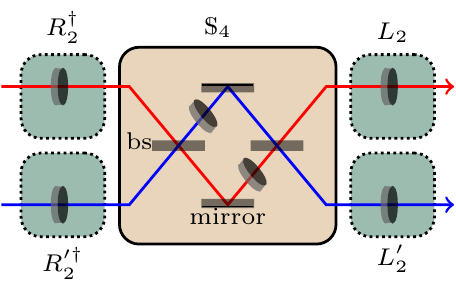}}
 \caption{Realization of a $4\times 4$ unitary matrix $U_{4}$ as a transformation on two spatial and two polarization modes of light.
(a) The CSD factorizes $U_{4}$ into the left and right matrices $L_{2}, L'_{2},R^{\dagger}_{2},R^{\prime\dagger}_{2}$ and the CS matrix $\mathds{S}_{4}$.
(b) The left and right matrices are realized as combinations of quarter- and half-wave plates, and the CS matrix is realized using two beam splitters and a half-wave plate.
 }
 \label{Fig:ns2np2}
\end{figure}

The matrix $\mathds{S}_4$ in Equation~\eqref{Eq:CSD4} is a CS matrix of the form
\begin{equation}
 \mathds{S}_4(\theta_{1},\theta_{2}) = \left(\begin{array}{cc|cc}
\cos\theta_1 & & ~\sin\theta_1 &\\
&\cos\theta_2 & &~\sin\theta_2\\
\hline
-\sin\theta_1 & & ~\cos\theta_1 & \\
&-\sin\theta_2 & & ~\cos\theta_2
\end{array}\right).
\label{Eq:CSMatrix22}
\end{equation}
This matrix can be decomposed further according to
\begin{equation}
 \mathds{S}_4(\theta_{1},\theta_{2}) = (\mathcal{B}_{2}\otimes \mathds{1}_2)(\Theta_{2}\oplus \Theta_{2}^\dagger)(\mathcal{B}_{2}^\dagger\otimes \mathds{1}_2),
\label{Eq:SineCosineDecomp}
\end{equation}
where
\begin{align}
\gls{B2} & \defeq\frac{1}{\sqrt{2}}\begin{pmatrix}
1 & \mathrm{i}\\
\mathrm{i} & 1\\
\end{pmatrix},\label{Eq:BVartheta}\\
 \Theta_{2} &\defeq \begin{pmatrix}
\mathrm{e}^{\mathrm{i}\theta_1} & 0\\
0 & \mathrm{e}^{\mathrm{i}\theta_2}
\end{pmatrix}.
\end{align}
The transformation $(\mathcal{B}_{2}\otimes \mathds{1}_2)$ in Equation~\eqref{Eq:SineCosineDecomp} represents balanced beam splitters, % of reflectivity $\sin^{2}\theta_{\pm}$, which is independent of the polarization of the incoming light.
whereas, the transformations $\Theta_{2}\oplus \Theta_{2}^\dagger$ can be realized using wave plates acting separately on the polarization of light in the two spatial mode.
Figure~\ref{Fig:ns2np2}(b) depicts the optical circuit for the realization of $U_{4}$ using beam splitters, phase shifters and wave-plates.

Although the realization of arbitrary $4\times 4$ transformations on two spatial and two polarization modes was known~\cite{Goyal2015}, there was no known realization of an arbitrary $n_{s}n_{p}\times n_{s}n_{p}$ transformation on $n_{s}$ spatial and $n_{p}$ internal modes.
Such a decomposition is presented in the next chapter.

%----------------------------------------%
%----------------------------------------%
\chapter{Realization of arbitrary discrete unitary transformations on spatial and internal modes of light}
\label{Chap:Design}
%----------------------------------------%
%----------------------------------------%

%Our algorithm is based on the iterative use of the cosine-sine decomposition (CSD).
%The relevant background of the CSD is presented in Section~\ref{Sec:Background}.
This chapter details our procedure for the realization of arbitrary discrete unitary transformations.
I present the decomposition algorithm in Section~\ref{Sec:Algorithm}.
The cost of realizing an arbitrary unitary matrix is discussed in Section~\ref{Sec:Cost}.
I conclude with a discussion of our decomposition algorithm in Section~\ref{Sec:DesignConclusion}.

The majority of the material in this chapter is taken from my article published in Physical~Review~A~\cite{Dhand2015b}.
New material is added or existing material is shifted or eliminated to improve presentation.
Those parts that are reproduced verbatim from our journal paper are listed in ``Thesis content previously published''.

%------------------------------%
\section{Algorithm to design efficient realization}
\label{Sec:Algorithm}
%------------------------------%

Here I describe the algorithm to decompose an arbitrary unitary matrix into beam-splitter and internal transformations.
This section is structured as follows.
Subsection~\ref{Subsec:Inputs} details the inputs and outputs of the decomposition algorithm.
The first of the two stages of the algorithm is a step-by-step decomposition of the unitary into internal and CS matrices and is presented in Subsection~\ref{Subsec:Algorithm}.
The next stage involves factorization of the CS matrices into beam-splitter and internal transformations as described in Subsection~\ref{Subsec:Realization}.

%------------------------------%
\subsection{Inputs and outputs of algorithm}
\label{Subsec:Inputs}
%------------------------------%
In this section, I present the inputs and outputs of our decomposition algorithm.
The algorithm receives an $n_{s}n_{p}\times n_{s}n_{p}$ unitary matrix as an input.
The algorithm returns a sequence of matrices, each of which describes either a beam splitter acting on two-spatial modes or an internal unitary operation, which acts on the internal DOF in one spatial modes whereas leaving the other modes unchanged.
The remainder of this subsection describes the basis and the form of the matrices yielded by our algorithm.

The operators returned by the algorithm act on the combined space
\begin{equation}
\mathcal{H} = \mathcal{H}_{s}\otimes \mathcal{H}_{p},
\end{equation}
where $\mathcal{H}_{s}$ and $\mathcal{H}_{p}$ are spanned
\begin{align}
\mathcal{H}_{s} &= \operatorname{span}\{\ket{s_{1}},\ket{s_{2}},\dots,\ket{s_{n_{s}}}\}, \\
\mathcal{H}_{p} &=\operatorname{span}\{\ket{p_{1}},\ket{p_{2}},\dots,\ket{p_{n_{p}}}\}
\end{align}
by the $n_{s}$ spatial modes and the $n_{p}$ internal modes respectively for positive integers $n_{s}$ and $n_{p}$.
Each operator acting on the combined state of light can be represented by an $n_{s}n_{p}\times n_{s}n_{p}$ matrix in the combined basis
\begin{equation}
\{\ket{c_{k\ell}}\defeq \ket{s_{k}}\otimes\ket{p_{\ell}}: k \in \left\{1,\dots,n_{s}\right\},~\ell \in \{1,\dots,n_{p}\}\}
\end{equation}
of the spatial and the internal modes.
Our algorithm returns the matrix representations of the operators in this combined basis $\{\ket{c_{k\ell}}\}$.

The matrices returned by the algorithm represent either internal or beam-splitter transformations.
Each internal transformation acts on the internal state of light in a spatial mode but not on the light in the other spatial modes.
In the composite basis, the internal transformations acting on the $k$-th spatial mode are represented as
\begin{equation}
U^{(k)}_{n_{p}} \defeq \mathds{1}_{n_{p}(k-1)}\oplus U_{n_{p}}\oplus\mathds{1}_{n_{p}(n_{s}-k)}
\label{Eq:InternalMatrix}
\end{equation}
for $n_{p}\times n_{p}$ unitary matrix $U_{n_{p}}$.

The algorithm also returns beam-splitter matrices, which mix each of the corresponding internal modes of light in two spatial modes.
The matrix representation of this operator in the composite basis is given by
\begin{equation}
\mathcal{B}^{(k)}_{2n_{p}}\defeq \mathds{1}_{n_{p}(k-1)}\oplus \left(\mathcal{B}_{2}\otimes\mathds{1}_{n_{p}}\right)\oplus
 \mathds{1}_{n_{p}(n_{s}-k-1)}
\label{Eq:BSMatrix1}
\end{equation}
for $\mathcal{B}_{2}$ as defined in Equation~\eqref{Eq:BVartheta} representing a balanced beam splitter.
To summarize, the algorithm returns a sequence of matrices, each of which is an internal transformation in the form~of Equation~\eqref{Eq:InternalMatrix} or is a balanced beam-splitter transformation in the form of Equation~\eqref{Eq:BSMatrix1}.

%------------------------------%
\subsection{Decomposition of unitry matrix into internal and CS matrices}
\label{Subsec:Algorithm}
%------------------------------%
In this subsection, I present the first stage of our algorithm.
This stage decomposes the given unitary matrix into matrices representing internal transformations~\eqref{Eq:InternalMatrix} and CS transformations
 \begin{equation}
\begin{split}
 \mathds{S}^{(k)}_{2n_{p}}(\theta_{1},\dots,\theta_{n_{p}})\defeq & \mathds{1}_{n_{p}(k-1)}\oplus \mathds{S}_{2n_{p}}(\theta_{1},\dots,\theta_{n_{p}})\\&\oplus
 \mathds{1}_{n_{p}(n_{s}-k-1)},
 \label{Eq:CSMatrixx}
 \end{split}
 \end{equation}
 which enact the CS matrix $\mathds{S}_{2n_{p}}\equiv\mathds{S}_{2n_{p}}(\theta_{1},\dots,\theta_{n_{p}})$~\eqref{Eq:CSMatrix} on the internal degrees of light in two spatial modes without affecting the light in other modes.

The first stage comprises $n_{s}-1$ iterations.
Of these, the first iteration factorizes the given $n_{s}n_{p}\times n_{s}n_{p}$ unitary matrix into a sequence of internal and CS matrices and one $(n_{s}-1)n_{p}\times (n_{s}-1)n_{p}$ unitary matrix.
This smaller unitary matrix is factorized in the next iteration.
Figure~\ref{Fig:FirstStep} depicts the first of the $n_{s}-1$ iterations that comprise the first stage.

In general, the $j$-th iteration receives an $(n_{s}+1-j)n_{p}\times (n_{s}+1-j)n_{p}$ unitary matrix.
This iteration decomposes the received unitary matrix into a sequence of internal and CS matrices and a smaller $(n_{s}-j)n_{p}\times (n_{s}-j)n_{p}$ unitary matrix which is decomposed in the next iteration.

Now I describe the $j$-th iteration of the decomposition algorithm in detail.
First, the given unitary matrix $U_{(n_{s}+1-j)n_{p}}$ is CS decomposed by setting $m = n_{p}$ and $n = (n_{s}-j)n_{p}$ in the CSD.
This CSD yields the following sequence of matrices
\begin{align}
U_{ (n_{s}+1-j)n_{p}}=&\, \mathds{L}_{n_{p} + (n_{s}-j)n_{p}} \left(\mathds{S}_{2n_{p}}\oplus \mathds{1}_{(n_{s}-1-j)n_{p}}\right)\nonumber\\
&\times\mathds{R}_{n_{p}+ (n_{s}-j)n_{p}},
\label{Eq:csd2}
\end{align}
for block diagonal unitary matrices
\begin{align}
\mathds{L}_{n_{p} + (n_{s}-j)n_{p}} &=
\left(\begin{array}{c|cc}
L_{n_{p}}& \multicolumn{2}{c}{0} \\
\hline
\multirow{2}{*}{0} & \multicolumn{2}{c}{\multirow{2}{*}{$~L^{\prime}_{(n_{s}-j)n_{p}}$}} \\
& \multicolumn{2}{c}{} \\
\end{array}\right),\nonumber\\
\mathds{R}_{n_{p} + (n_{s}-j)n_{p}} &=
\left(\begin{array}{c|cc}
R^{\dagger}_{n_{p}}& \multicolumn{2}{c}{0} \\
\hline
\multirow{2}{*}{0} & \multicolumn{2}{c}{\multirow{2}{*}{$~R^{\prime\dagger}_{(n_{s}-j)n_{p}}$}} \\
& \multicolumn{2}{c}{}\\
\end{array}\right),
\end{align}
and orthogonal CS matrix~$\mathds{S}_{2n_{p}}$.

In other words, the first CSD of the $j$-th iteration factorizes the received unitary transformation acting on $n_{s}+1-j$ spatial modes into
(i)~a $2n_{p}\times 2n_{p}$ CS matrix $\mathds{S}_{2n_{p}}$ acting on the $j$-th and $(j+1)$-th spatial modes,
(ii)~internal unitary matrices $L_{n_{p}}$ and $R_{n_{p}}^{\dagger}$, each of which act on the internal degrees of the $j$-th spatial mode and
(iii)~left and right unitary matrices $L_{(n_{s}-j)n_{p}}^{\prime}$ and $R_{(n_{s}-j)n_{p}}^{\prime\dagger}$ acting on the remaining $n_{s}-j$ spatial modes.
Figure~\ref{Fig:FirstStep}(a) depicts this first CSD for the first iteration.

Next the matrix $L_{(n_{s}-j)n_{p}}^{\prime}$ is CS decomposed.
The resultant $R_{(n_{s}-j-1)n_{p}}^{\prime\dagger}$ from this second CSD commutes with CS matrix $\mathds{S}_{2n_{p}}$ yielded by the first CSD\footnote{The transformations $R_{(n_{s}-k-1)n_{p}}^{\prime\dagger}$ and $\mathds{S}_{2n_{p}}$ act on mutually exclusive spatial modes so their action is independent of the order of enacting the transformations.}.
Hence, the operators $R_{(n_{s}-j-1)n_{p}}^{\prime\dagger}$ and $\mathds{S}_{2n_{p}}$ can be swapped, following which I multiply $R_{(n_{s}-j-1)n_{p}}^{\prime\dagger}$ by $R_{(n_{s}-j)n_{p}}^{\prime\dagger}$.
Figure~\ref{Fig:FirstStep}(b) depicts this second round of CSD and of the multiplication of the two right matrices.

The left unitary matrices thus obtained are repeatedly factorized using the CSD.
The resultant right unitary matrices are absorbed into the initial right unitary matrix $R_{(n_{s}-1)n_{p}}^{\prime\dagger}$.
Thus, we are left with internal and CS matrices and with a unitary matrix
\begin{equation}
U_{(n_{s}-j)n_{p}}= \prod_{\ell = 0}^{n_{s}-j-1} R_{(n_{s}-j-\ell)n_{p}}^{\prime\dagger}
\end{equation} obtained by multiplying each of the right unitary matrices.
This completes a description of the $j$-th iteration of the algorithm.

In summary, at the end of the $j$-th iteration, the algorithm decomposes the received $U_{(n_{s}+1-j)n_{p}}$ transformation into internal and CS matrices and $U_{(n_{s}-j)n_{p}}$ as depicted in Fig.~\ref{Fig:FirstStep}(c).
The $(j+1)$-th iteration of the algorithm receives this smaller $U_{(n_{s}-j)n_{p}}$ unitary matrix and decomposes it into internal and CS matrices and an even smaller unitary matrix.
The algorithm iterates over integral values of $j$ ranging from $1$ to $n_{s}-1$.
Figure~\ref{Fig:Final} depicts the output of the algorithm at the end of the final, i.e., $(n_{s}-1)$-th, iteration.
This completes a description of the first stage of the algorithm.

At the end of the first stage, the given unitary matrix has been factorized into a sequence of internal~\eqref{Eq:InternalMatrix} and CS matrices~\eqref{Eq:CSMatrix}.
The internal matrices can be implemented using optical elements if a suitable realization is known for the internal DOF;
such realizations are known for polarization~\cite{Simon1989,Simon1990}, temporal~\cite{Motes2014} and orbital-angular-momentum~\cite{Garcia-Escartin2011} DOFs.
In the next subsection, I present a realization of the CS matrix using beam splitters acting on spatial modes and internal transformations.

\begin{figure}[p]
\centering
\subfloat[]{\includegraphics[width=0.8\textwidth]{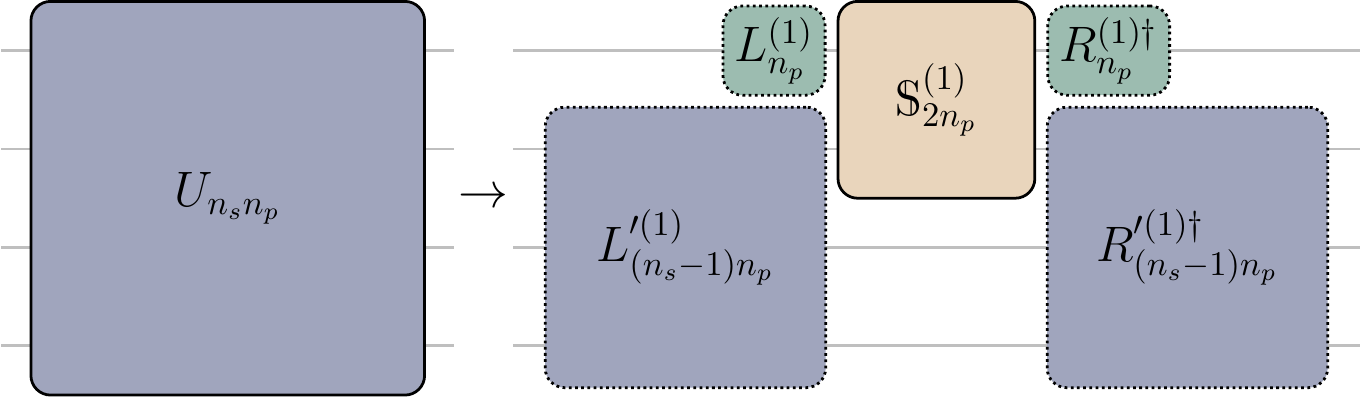}}\\
\subfloat[]{\includegraphics[width=0.8\textwidth]{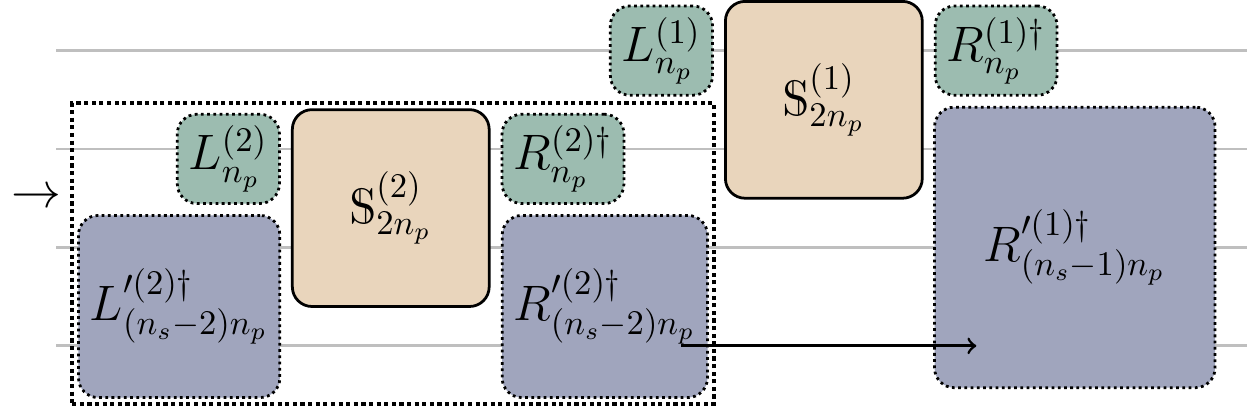}}\\
\subfloat[]{\includegraphics[width=0.8\textwidth]{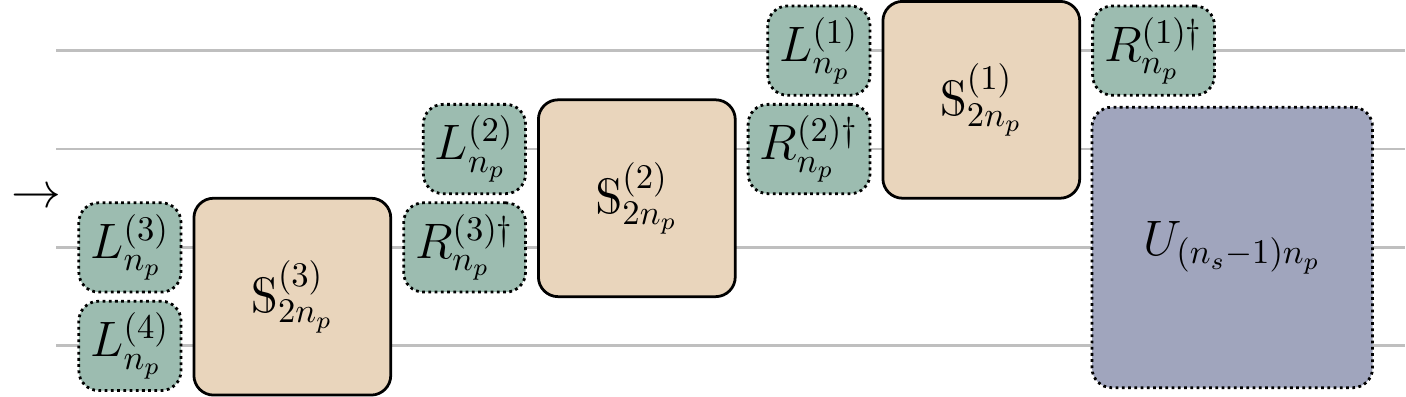}}
\caption{A depiction of the first iteration of the algorithm for the decomposition of a given unitary $U_{n_{s}n_{p}}$ into internal (green) and CS (brown) matrices.
(a) First, the $U_{n_{s}n_{p}}$ unitary matrix is CS decomposed into (i)~a $2n_{p}\times 2n_{p}$ CS matrix $\mathds{S}^{(1)}_{2n_{p}}$ acting on the first two spatial modes,
(ii)~internal unitary matrices $L_{n_{p}}^{(1)}$ and $R_{n_{p}}^{(1)\dagger}$, each of which act on the internal degrees of the first spatial mode and
(iii)~left and right unitary matrices $L_{n_{p}(n_{s}-1)}^{\prime(1)}$ and $R_{n_{p}(n_{s}-1)}^{\prime(1)\dagger}$ acting on the remaining $n_{s}-1$ spatial modes.
 (b) The matrix $L_{n_{p}(n_{s}-1)}^{\prime(1)}$ is further CS decomposed.
 The resultant $R_{n_{p}(n_{s}-2)}^{\prime(2)\dagger}$ from the second decomposition commutes with CS matrix $\mathds{S}^{(1)}_{2n_{p}}$ and can thus be absorbed into $R_{n_{p}(n_{s}-1)}^{\prime(1)\dagger}$.
(c) The algorithm repeatedly decomposes the left unitary matrices.
The resultant right unitary matrices are absorbed into the initial right unitary matrix.
 At the end of one iteration, the algorithm decomposes $U_{n_{s}n_{p}}$ unitary operation into CS matrices, internal unitary matrices and the matrix $U_{n_{p}(n_{s}-1)}$.
 The next iteration of the algorithm decomposes the smaller $U_{n_{p}(n_{s}-1)}$ unitary matrix.
 }
 \label{Fig:FirstStep}
\end{figure}

\begin{figure}[h]
\includegraphics[width = \textwidth]{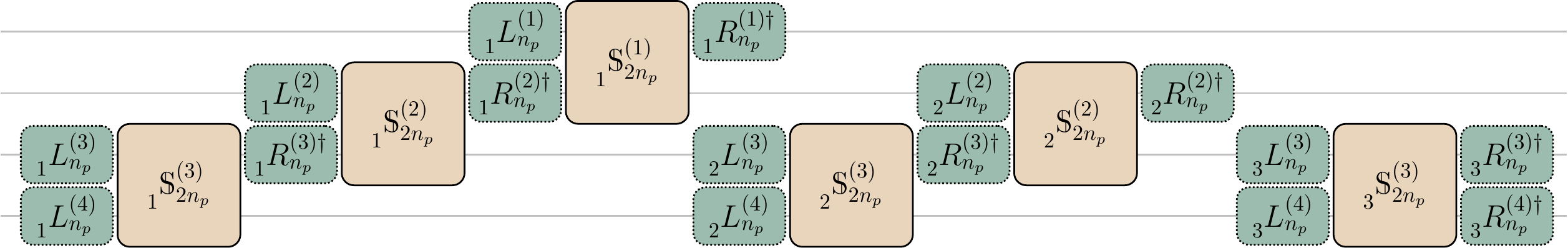}
\caption{A depiction of the output of the first stage of our decomposition algorithm (Subsection~\ref{Subsec:Algorithm}) for the case of $n_{s} = 4$~spatial modes and $n_{p}$~internal modes.
The given $4n_{p}\times 4n_{p}$~unitary matrix is decomposed into $4^{2} = 16$~internal matrices (green) and $n_{s}(n_{s}-1)/2 = 6$~CS matrices (brown). As usual, the right subscript of the matrices is the dimension of the space that the respective operators act on. The right superscript represents the spatial mode that the operators act on. The left subscript specifies the index of iteration that constructed the respective matrices.
}
\label{Fig:Final}
\end{figure}

%------------------------------%
\subsection{Decomposition of CS unitary matrix into elementary operators}
\label{Subsec:Realization}
%------------------------------%
Here I show how the CS matrices can be decomposed into a sequence of beam-splitter transformations and internal unitary matrices.
Specifically, we construct a factorization of any $2n_{p}\times 2n_{p}$ CS matrix $\mathds{S}_{2n_{p}}$, which is in the form of Equation~\eqref{Eq:CSMatrix}, into a sequence of two balanced beam-splitter matrices and two internal-transformation matrices.

Our decomposition of the CS matrix relies on the following identity
\begin{equation}
\mathds{S}_{2n_{p}}(\theta_{1},\dots,\theta_{n_{p}}) = \left(\mathcal{B}_{2}\otimes \mathds{1}_{n_{p}}\right)\left(\Theta_{n_{p}}\oplus \Theta_{n_{p}}^\dagger\right)\left(\mathcal{B}_{2}^\dagger\otimes \mathds{1}_{n_{p}}\right),
\label{Eq:SineCosineDecompNp}
\end{equation}
where $\mathcal{B}_{2}\otimes \mathds{1}_{n_{p}}$ represents a balanced beam splitter~\eqref{Eq:BVartheta}
and
\begin{equation}
 \Theta_{n_{p}} \defeq \begin{pmatrix}
\mathrm{e}^{\mathrm{i}\theta_1} & &\\
& \ddots &\\
& & \mathrm{e}^{\mathrm{i}\theta_{n_{p}}}
\end{pmatrix}.
\label{Eq:ThetaDefine}
\end{equation}
is a transformation on the internal modes.
Thus, any CS matrix can be realized using two balanced beam splitters and two internal transformations.

To summarize, the first stage of the algorithm decomposes the given unitary matrix into internal~\eqref{Eq:InternalMatrix} and CS matrices~\eqref{Eq:CSMatrixx}.
The next stage factorizes the CS matrices returned by the first stage into internal and beam splitter~\eqref{Eq:BSMatrix1} transformations, thereby completing the algorithm.
\textsc{matlab} code for the CSD and for the decomposition algorithm is available online~\cite{Dhand2015}.

%------------------------------%
\section{Cost Analysis: Number of optical elements in realization}
\label{Sec:Cost}
%------------------------------%
Here I discuss the cost of realizing an arbitrary $n_{s}n_{p}\times n_{s}n_{p}$ unitary matrix using our procedure, where the cost is quantified by the number of optical elements required to implement the matrix.
Optical elements required by the decomposition algorithm include balanced beam splitters, phase shifters and elements acting on internal modes.
We conclude this section with a specific example of decomposing a $2n\times 2n$ transformation into spatial and polarization DOFs.
In this case, this decomposition reduces the required number of beam splitters to half with the additional requirement of wave plates as compared to using only spatial modes.

Consider the decomposition of an arbitrary $n_{s}n_{p}\times n_{s}n_{p}$ unitary transformation.
Realization of this transformation using the Reck~\emph{et~al.} method requires $n_{s}n_{p}$ spatial modes and $n_{s}n_{p}(n_{s}n_{p}-1)/2$ biased beam splitters~\cite{Reck1994}.
In comparison, our decomposition requires $n_{s}(n_{s}-1)$ beam splitters.
Thus, we reduce the number of beam splitters required to realize an $n_{s}n_{p}\times n_{s}n_{p}$ transformation by a factor of
\begin{equation}
\eta = \frac{n_{s}n_{p}(n_{s}n_{p}-1)/2}{n_{s}(n_{s}-1)} > n_{p}^{2}/2.
\end{equation}
%which is half the square of the dimension of the internal DOF.

Although our decomposition reduces the required number of beam splitters, the number of optical elements required for internal transformations increases by a factor of $2$.
The Reck~\emph{et al.}~approach requires $n_{s}n_{p}(n_{s}n_{p}+1)/2$ phase shifters to effect an $n_{s}n_{p}\times n_{s}n_{p}$ unitary transformation on spatial modes.

Our approach relies on decomposing to beam splitter and internal unitary transformations, so we count the number of internal optical elements required in the transformation.
Realizing an $n_{p}\times n_{p}$ internal transformation typically requires $n_{p}^{2}$ internal optical elements~\cite{Simon1990,Garcia-Escartin2011,Motes2014}.
Our decomposition requires $n_{s}^{2}$ arbitrary internal transformations, which are represented by matrices $\{L_{n_{p}},L^{\prime}_{n_{p}},R_{n_{p}},R^{\prime}_{n_{p}}\}$ in the output.
These arbitrary transformations can be realized using a total of $n_{s}^{2}n_{p}^{2}$~internal optical elements.
Furthermore, our decomposition also requires $n_{s}(n_{s}-1)$~internal transformations in the form of $\Theta_{n_{p}}$~\eqref{Eq:ThetaDefine}.
Each of these transformations can be realized using $n_{p}$ optical elements for the polarization, temporal and orbital angular momentum modes%%%%
%%%%
\footnote{%
For the polarization DOF the $\Theta_{n_{p}=2}$ matrix can be constructed using two elements: a quarter-wave plate and a phase shifter.
Similarly, for the temporal DOF, the matrix $\Theta_{n_{p}}$ can be realized by setting the reflectivity of the variable beam splitter to zero and the transmission amplitude to $\mathrm{e}^{\mathrm{i}\theta_{j}}$ at an appropriate time~\cite{Motes2014}.
The matrix $\Theta_{n_{p}}$ for the orbital-angular-momentum DOF of light can be constructed using a spatial light modulator (hologram)~\cite{Flamm2013}.
In all these realizations of the matrix $\Theta_{n_{p}}$ no more than $n_{p}$ optical components are required.}%%
%%%%
.
In summary, our decomposition requires a total of $n_{s}n_{p}(n_{s}n_{p}+ n_{s}-1)$, which is an increase by a factor
\begin{equation}
\xi = \frac{n_{s}n_{p}(n_{s}n_{p}+ n_{s}-1)}{n_{s}n_{p}(n_{s}n_{p}+1)/2} = 2 + \BigO{1/n_{p}}
\end{equation}
over the cost of the Reck~\emph{et~al.} approach.

Now we consider the example of using polarization as the internal DOF.
Specifically, we compare the cost of realizing an arbitrary $2n\times 2n$ transformation using (i)~the Reck~\emph{et~al.} approach on only spatial modes and (ii)~our decomposition on the spatial and polarization modes of light, i.e., $n_{s} = n$ and $n_{p} = 2$.
The Reck~\emph{et~al.} decomposition requires $2n$ spatial modes, $n(2n-1)$ beam splitters and $n(2n+1)$ phase shifters.
In comparison,my approach requires $n(n-1)$ balanced beam splitters, $n^{2}$ phase shifters and $3n(n-1)/2$ wave plates.
Thus, our decomposition reduces the required number of beam splitters and phase shifter by a factor of $2$ each at the expense of an additional $3n(n-1)/2$ wave plates.

To summarize this section, our realization of an arbitrary $n_{s}n_{p}\times n_{s}n_{p}$ unitary matrix reduces the number of beam splitters required by a factor of $n_{p}^{2}/2$.
This completes the analysis of the cost of our decomposition.

%------------------------------%
\section{Conclusion}
\label{Sec:DesignConclusion}
%------------------------------%
In conclusion, we devise a procedure to efficiently realize any given $n_{s}n_{p}\times n_{s}n_{p}$ unitary transformation on $n_{s}$ spatial and $n_{p}$ internal modes of light.
Our realization uses interferometers composed of beam splitters and optical devices that act on internal modes to effect the given transformation.
Such interferometers can be characterized by using existing procedures~\cite{Laing2012,Dhand2015a} based on one- and two-photon interference on spatial and internal DOFs~\cite{Walborn2003,Schuck2006,Nagali2009,Karimi2014a}.
We thus enable the design and characterization of linear optics on multiple degrees of freedom.

We overcome the problem of decomposing the given unitary transformation into internal transformations by performing the CSD iteratively.
We also open the possibility of using an efficient iterative CSD in problems where the single-shot CSD is currently used~\cite{Bullock2004,Khan2006,Shende2006}.

By employing $n_{p}$ internal modes, the number of beam splitters required to effect the transformation is reduced by a factor of $n_{p}^{2}/2$ at the cost of increasing the number of internal elements by a factor of $2$.
Our procedure facilitates the realization of higher dimensional unitary transformations for quantum information processing tasks such as linear optical quantum computation, BosonSampling and quantum walks.

%=================%
\chapter{Characterization of linear optical interferometer}
\label{Chap:Procedure}
%=================%

This chapter details our procedure for the characterization of linear optical interferometers.
Section~\ref{Sec:Procedure} details the accurate and precise characterization using one- and two-photon measurements and the bootstrapping-based procedure to estimate the precision of the characterized interferometer parameters.
Section~\ref{Sec:Scattershot} presents a scattershot approach to reducing the experimental time by performing all one- and two-photon measurement in parallel.
I present our approach to removing the numerical instability that adversely effect the accuracy and precision of existing characterization procedures in section~\ref{Sec:Instability}.
Section~\ref{Sec:CharDiscussions} comprises a nontechnical summary of the characterization and comparison with existing procedure.

The majority of the material in this chapter is taken from my article~\cite{Dhand2015a} that I co-authored with Abdullah Khalid, He Lu and Barry C.~Sanders.
New material is added or existing material is shifted or eliminated to improve presentation.
Those parts that are reproduced verbatim from our journal paper are listed in ``Thesis content previously published''.

%------------------------------%
\section{Characterization procedure}
\label{Sec:Procedure}
%------------------------------%
In this section, I describe our procedure to characterize linear optical interferometers.
The outline of this section is as follows.
Subsection~\ref{Subsec:Experiment} describes the experimental data required by our characterization procedure.
This experimental data are processed by various algorithms to determine the transformation matrix~(\ref{Eq:RMS0}).
The algorithm to determine the amplitudes $\{\alpha_{ij}\}$ of the transformation-matrix elements is presented in Subsection~\ref{Subsec:Amplitudes}.
In Subsection~\ref{Subsec:Calibration}, I describe the calibration of the source field by determining the mode-matching parameter
~$\gamma$.
The estimation of $\{\theta_{ij}\}$ using two-photon interference is detailed in Subsection~\ref{Subsec:Phases}.
Maximum-likelihood estimation is employed to find the unitary matrix $U$ that best fits the calculated $\{\alpha_{ij}\}, \{\theta_{ij}\}$ values and serves as the representative matrix~(\ref{Eq:RMS0}).
We discuss the calculation of the best-fit unitary representative matrix in Subsection~\ref{Subsec:MaxLikelyUnitary}.

%------------------------------%
\subsection{Experimental procedure and inputs to algorithms}
\label{Subsec:Experiment}
%------------------------------%
Our characterization procedure relies on measuring
(i)~the spectral function $f_j$ of the source light, (ii)~single-photon detection counts,
(iii)~two-photon coincidence counts from a beam splitter and~(iv)~two-photon coincidence counts from the interferometer.
The measurement data constitute the inputs to our algorithms, which then yield the representative matrix.
Before presenting the algorithms, I detail the experimental procedure and the inputs received by the algorithm in this subsection.

%The light entering the interferometer is first passed through band-pass filters to improve the frequency-mode matching~\cite{Lu2000}.
We characterize the spectral function $f(\omega_i)$ of the incoming light for a discrete set $\Omega=\{\omega_1,\omega_2,\dots,\omega_k\}$ of frequencies.
The integer $k$ of frequencies at which the spectral function is characterized is commonly equal to the ratio of the bandwidth to the frequency step of the characterization device.
The characterized spectral function $f(\omega_i)$ is used to calculate the coincidence probabilities as detailed in Algorithm~\ref{Alg:Coincidence}.

%------------------------------%
\begin{algorithm}[h]
	\begin{algorithmic}[1]	
	\Require{\Statex
	\begin{itemize}
		\item $k, \Omega = \{\omega_1,\omega_2,\dots\omega_{k-1},\omega_k\} \in \left(\mathds{R}^+\right)^{k}$ \Comment{Frequencies at which $f_1,f_2$ are given.}
		\item $f_1,f_2:\Omega \to \mathds{R}^+$\Comment{measured spectra.}
		\item $\ell, T= \{\tau_1,\tau_2,\dots,\tau_\ell\}\in(\mathds{R}\cup 0)^\ell$\Comment{Time delay values.}
		\item $A \leftarrow \{\alpha_{ij},\alpha_{ij'},\alpha_{i\j},\alpha_{i'j'}\}\in \left(\mathds{R}^+\cup\,0\right)^{4}$ \Comment{Amplitudes of $2\times 2$ submatrix of $A$~(\ref{Eq:RMS0}).}
		\item $\Theta \leftarrow \theta_{ij},\theta_{ij'},\theta_{i'j},\theta_{i'j'}\in (-\pi,\pi]$\Comment{Phases of $2\times 2$ submatrix of $A$~(\ref{Eq:RMS0}).}
		\item $\gamma\in[0,1]$ \Comment{Mode-matching parameter
 of photon source.}
		\end{itemize}
	}
	\Ensure{\Statex
	\begin{itemize}
	\item $C:T\to \mathds{R}^+$ \Comment{Two-photon coincidence probabilities correct up to multiplicative factor.}
	\end{itemize}
	}

	\Procedure{Coincidence}{$k,\Omega,f_1,f_2,\ell,T,A,\Theta,\gamma$}
 \For{$\tau$ in $T$}
 \State $C(\tau) \leftarrow\textsc{Integrate}\left[F(A,\Theta, f_1,f_2,\gamma,\omega_i,\omega_j,\tau),\{\omega_i\in \Omega,\omega_j\in\Omega\}\right].$
 \Statex \Comment{Numerically integrate RHS of (\ref{Eq:CoincidenceRate}) over $\omega_i,\omega_j$ with $\kappa_i = \kappa_{i'} = \nu_{j} = \nu_{j'} = 1$.}
 \EndFor
	\State \Return $C$
	\caption{\textsc{Coincidence}: Calculates the expected coincidence rate for two-photon interference for a given $2\times 2$ submatrix of an arbitrary $\mathrm{SU}(m)$ transformation.}
	\EndProcedure
	\label{Alg:Coincidence}
	\end{algorithmic}
\end{algorithm}
%------------------------------%

The amplitudes $\{\alpha_{ij}\}$ are determined by impinging single photons at the interferometer and counting single-photon detections at the outputs.
Single-photon counting is repeated multiple $(B\in\mathds{Z}^+)$ times in order to estimate the precision of the obtained $\{\alpha_{ij}\}$ values.
Specifically, the number
\begin{equation}
N_{ijb_j}\in \mathds{Z}^{+}:\, i,j\in \{1,\dots ,m\}, b_{j}\in\{1,\dots,B\}
\label{Eq:Nijbdef}
\end{equation}
of single-photon detection events are counted at all $m$ output ports $\{i\}$ for single photons impinged at the $j$-th input ports in the $b_{j}$-th repetition.
The counting is then performed for each of the input ports $j\in\{1,\dots,m\}$ of the interferometer.
Algorithm~\ref{Alg:AmplitudeEstimation} uses $\left\{N_{ijb_j}, b_{j}\in\{1,\dots,B\}\right\}$ values to estimate $\alpha_{ij}$ and the standard deviation of the estimate.
% and the respective standard deviations $\left\{\delta\alpha_{ij}\right\}$.
The experimental setup for $\{\alpha_{ij}\}$ measurement is depicted in Figure~\ref{Figure:1Photon}.

Arguments $\{\theta_{ij}\}$ are calculated by fitting curves of measured coincidence counts to curves calculated using measured spectra according to~(\ref{Eq:CoincidenceRate}).
Appendix~\ref{Sec:CurveFitting} elucidates the inputs and outputs of the curve-fitting procedure, such as the Levenberg-Marquardt algorithm~\cite{Levenberg1944,Marquardt1963}, employed by our algorithms.
Before calculating $\{\theta_{ij}\}$, we calibrate the source field for imperfect mode matching by measuring coincidence counts on a beam splitter of known reflectivity.
Controllably delayed single-photon pairs are incident at the two input ports of the beam splitter and coincidence counting is performed on the light exiting from its two output ports.
Algorithm~\ref{Alg:Calibration} details the estimation of~$\gamma$ using coincidence counts $C^{\mathrm{cal}}(\tau)$ for time delay $\tau$ between the incoming photons.

The absolute values and the signs of the arguments $\{\theta_{ij}\in(-\pi,\pi]\}$ are calculated separately.
To estimate the absolute values $\{|\theta_{ij}|\}$ of the arguments, pairs of single photons are incident at two input ports $1$ and~$j\in\{2,\dots,m\}$ and coincidence measurement is performed at two output ports $1$ and~$i\in\{2,\dots,m\}$.
The choice of the input and output ports labelled by index $1$ is arbitrary.
The signs
\begin{equation}
\operatorname{sgn}\theta_{ij} \defeq \begin{cases}
-1 & \mathrm{if}\, \theta_{ij} < 0, \\
0 & \mathrm{if}\, \theta_{ij} = 0, \\
1 & \mathrm{if}\, \theta_{ij} > 0 \end{cases}
\label{Eq:SignDef}
\end{equation}
of the arguments are estimated using an additional $(m-1)^2$ coincidence measurements.
Algorithm~\ref{Alg:PhaseCalcNChannel} details the choice of input and output ports for estimating $\{\operatorname{sgn}\theta_{ij}\}$.
A schematic diagram of the experimental setup for $\{\theta_{ij}\}$ estimation is presented in Figure~\ref{Figure:2Photons}.

%------------------------------%
\subsection{Single-photon transmission counts to estimate $\{\alpha_{ij}\}$ (Algorithm~\ref{Alg:AmplitudeEstimation})}
\label{Subsec:Amplitudes}
%------------------------------%
Now I present our procedure to estimate $\{\alpha_{ij}\}$ values using single-photon counting.
Single-photon transmission probabilities are connected to the amplitudes $\{\alpha_{ij}\}$ according to the relation $P_{ij} = \kappa_{i} \lambda_{i} \alpha_{ij}^2 \mu_{j} \nu_{j}$ (\ref{Eq:PSinglePhotons}).
Although the $\{\alpha_{ij}\}$ values can be calculated from single-photon transmission counts, the factors $\{\lambda_{i}\},\{\mu_{j}\}$ cannot.
The transmission probabilities depend on the products of the factors $\{\lambda_{i}\},\{\mu_{j}\}$ and the loss terms $\{\kappa_{i}\},\{\nu_{j}\}$, so $\{\lambda_{i}\},\{\mu_{j}\}$ cannot be measured without prior knowledge of the losses.
The loss terms are usually unknown and can change between experiments.
Hence, we calculate the values of $\{\alpha_{ij}\}$ from single-photon measurements and choose $\{\lambda_{i}\}$ and~$\{\mu_{j}\}$ such that $U = LAM$ is unitary.

The amplitudes $\{\alpha_{ij}\}$ are determined by estimating transmission probabilities.
The probabilities $P_{11},P_{i1},P_{1j},P_{ij}$ of single-photon detection at output ports $1,i$ when single photons are incident at input ports $1,j$ are expresses in terms of the $\alpha_{ij}$ values according to
\begin{equation}
\frac{P_{11}P_{ij}}{P_{1j}P_{i1}}= \frac{\left|r_1\lambda_1\alpha_{11}\mu_1 s_1\right|^2}{\left| r_1\lambda_1\alpha_{1j} \mu_{j} s_{j}\right|^2}\frac{\left| r_{i}\lambda_{i}\alpha_{ij}\mu_{j} s_{j}\right|^2 }{\left| r_i\lambda_i\alpha_{i1}\mu_1 s_1\right|^2}= \left|\frac{\alpha_{11}\alpha_{ij}}{\alpha_{1j}\alpha_{i1}}\right|^2.
\end{equation}
The probabilities $P_{11},P_{i1},P_{1j},P_{ij}$ are estimated by counting transmitted photons.
The definition~(\ref{Eq:RMS0}) of $\alpha_{ij}$ implies that $\alpha_{11}=\alpha_{i1}=\alpha_{1j}=1$.
Hence, the values of $\alpha_{ij}$ are connected to the single-photon transmission probabilities according to
\begin{equation}
\alpha_{ij} = \sqrt{\frac{P_{11}P_{ij}}{P_{1j}P_{i1}}},
\label{Eq:AlphaMeasure0}
\end{equation}
which is independent of the losses at the input and the output ports.

The transmission probabilities $P_{ij}$ are estimated by counting transmitted photons as follows.
The estimated values of $\{\alpha_{ij}\}$ are random variables that are amenable to random error from under-sampling and experimental imperfections.
Thus, data collection is repeated multiple times.
For accurate estimation of $\alpha_{ij}$ and its standard deviation $\delta\alpha_{ij}$, the number $B$ of repetitions is chosen such that the standard deviation of $\{N_{ijb_j}:b_j\in\{1,\dots,B\}\}$ converges in $B$ for all $i,j\in\{1,\dots,m\}$.
The mean and standard deviation of $\{N_{ijb_j}:b_j\in\{1,\dots,B\}\}$ converge for large enough $B$ if the cumulants of the distribution are finite~\cite{James2006}.

% - - - - - - - - - - - - %
%------------------------------%
\begin{algorithm}[h]
	\begin{algorithmic}[1]
	\Require{\Statex
	\begin{itemize}	
		\item $m \in \mathds{Z}^+$, \Comment Number of modes of interferometer.
		\item $N_{ijb_{j}}: \{1,\dots,m\}\times\{1,\dots,m\} \times \{1,\dots,B\}\to \mathds{Z}^+$ 	
		\Statex\Comment{Single-photon detection counts.}
			\item $B \in \mathds{Z}^+$ \Comment Number of times single-photon counting is repeated .
	\end{itemize}
	}
		
	\Ensure{\Statex
	\begin{itemize}
		\item $\{\tilde{\alpha}_{ij}\} \in \left(\mathds{R}^+\cup 0\right)^{m^2}$\Comment{Estimate of $\{\alpha_{ij}\}$~(\ref{Eq:RMS0}).}
%		\item $\left\{\sigma(\tilde\alpha_{ij})\right\}\in \left(\mathds{R}^+\cup 0\right)^{m^2}$\Comment{Estimate of standard error in $\{\alpha_{ij}\}$.}
	\end{itemize}
	}

	\Procedure{AmplitudeEstimation}{$m,N_{ijb_{j}},B$}
	\For{$i,j \in \{1,\dots,m\}\times\{1,\dots,m\}$}
	%\State $\widetilde{N}_{ij} \leftarrow \sum_{b=1}^{B} N_{ijb}/B$\vspace{3pt}
	\State $\tilde{\alpha}_{ij} \leftarrow$ {\textsc{Mean}}$\left(\sqrt{{N}_{11b_1}{N}_{ijb_j}/{N}_{1jb_j}{N}_{i1b_1}}: b_1,b_j \in \{1,\dots,B\}\right)$
%	\State $\delta\tilde{\alpha}_{ij} \leftarrow$ {\textsc{StdDev}}$\left(\sqrt{{N}_{11b_1}{N}_{ijb_j}/{N}_{1jb_j}{N}_{i1b_1}}: b_1,b_j \in \{1,\dots,B\}\right)$
	\EndFor
	\State \Return $\left\{\tilde\alpha_{ij}\right\}$
	\EndProcedure
	\caption{\textsc{AmplitudeEstimation}: Uses single-photon detection counts to calculate the amplitudes of the complex entries of the transformation matrix. $\tilde\bullet$ represents our estimate of $\bullet$.}
	\label{Alg:AmplitudeEstimation}
	\end{algorithmic}
\end{algorithm}
%------------------------------%

The probabilities $P_{ij}$ are estimated by counting single-photon detection events.
Suppose $N_{ijb_j}$ photons are transmitted from input port $j$ to the detector at output port $i$ when ${N}_{b_j}$ photons are incident and~$b_j\in\{1,\dots,B\}$.
For large enough $B$, the transmission probability converges according to
\begin{equation}
P_{ij} \leftarrow \mathrm{mean}\left\{\frac{N_{ijb_j}}{N_{b_j}}:\, b_j\in \{1,\dots,B\}\right\}.
\end{equation}
Likewise, the amplitudes $\{\alpha_{ij}\}$ are estimated by averaging the single-photon detection counts according to
\begin{align}
\alpha_{ij}
= \sqrt{\frac{P_{11}P_{ij}}{P_{1j}P_{i1}}}
&\leftarrow \mathrm{mean}\left\{ \sqrt{\frac{N_{11b_1}}{N_{b_1}}\frac{N_{ijb_j}}{N_{b_j}}\frac{N_{b_j}}{N_{1jb_j}}\frac{N_{b_1}}{N_{i1b_1}}} :\, b_1,b_j\in \{1,\dots,B\}\right\}\nonumber\\
&= \mathrm{mean}\left\{\sqrt{\frac{N_{11b_1}N_{ijb_j}}{N_{1jb_j}N_{i1b_1}}}:\, b_1,b_j\in \{1,\dots,B\}\right\}.
\label{Eq:AlphaMeasure}
\end{align}
The estimate of $\alpha_{ij}$ relies on single-photon counts measured by impinging photons at the first input port repeatedly (repetition index $b_{1}\in\{1,\dots,B\}$) and independently at the $j$-th input port (with repetitions labelled by a different index $b_{j}\in\{1,\dots,B\}$).

Henceforth, we represent our estimate of any parameter $\bullet$ by $\tilde{\bullet}$.
The estimate $\tilde\alpha_{ij}$ calculated using~(\ref{Eq:AlphaMeasure}) is independent of $N_{b_j}$ and thus resistant to variations in the incident-photon number $N_{b_j}$ over different input modes $j$ and different repetitions $b_j$.
Thus, our estimates~$\{\tilde\alpha_{ij}\}$ are accurate in the realistic case of fluctuating light-source strength and coupling efficiencies.

Finally, the standard deviations $\sigma(\tilde\alpha_{ij})$ of our estimates are calculated according to
\begin{equation}
\sigma(\tilde{\alpha}_{ij})\leftarrow \operatorname{std.~dev.}\left(\sqrt{\frac{N_{11b_1}N_{ijb_j}}{N_{1jb_j}N_{i1b_1}}}:\, b_1,b_j\in \{1,\dots,B\}\right),
\end{equation}
which converges for a large enough $B$.
In line with standard nomenclature, I refer to these standard deviations as error bars.
Algorithm~\ref{Alg:AmplitudeEstimation} details the estimation of $\{\tilde\alpha_{ij}\}$ and error bars on the obtained estimates.

%------------------------------%
\subsection{Calibration to estimate mode-matching parameter~$\gamma$ (Algorithm~\ref{Alg:Calibration})}
\label{Subsec:Calibration}
%------------------------------%

In this subsection, I describe the procedure to calibrate our light sources for imperfect mode matching.
The mode-matching parameter~$\gamma$ is estimated using one- and two-photon interference on an arbitrary beam splitter.
First, the reflectivity of the beam splitter is determined using single-photon counting~\cite{Laing2012}.
Next, controllably delayed photon pairs are incident at the beam splitter inputs and coincidence counting is performed on the beam splitter output .
We introduce a curve-fitting procedure to estimate the value of~$\gamma$ such that~(\ref{Eq:CoincidenceRate}) best fits the measured coincidence counts.

The beam-splitter reflectivity, which is denoted by $\cos\vartheta$, is estimated as follows.
A beam splitter of reflectivity $\cos \vartheta$ effects the $2\times 2$ transformation
\begin{align}
U_{\mathrm{bs}} &= \begin{pmatrix}
\cos\vartheta&\mathrm{i}\sin\vartheta\\\mathrm{i}\sin\vartheta&\cos\vartheta.
\end{pmatrix}\nonumber\\
&= \begin{pmatrix}1&0\\0&\mathrm{i}
\end{pmatrix}
\begin{pmatrix}1&0\\0&\tan\vartheta
\end{pmatrix}
\begin{pmatrix}1&1\\1&-\cot^2\vartheta
\end{pmatrix}
\begin{pmatrix}\cos\vartheta&0\\0&\sin\vartheta
\end{pmatrix}
\begin{pmatrix}1&0\\0&\mathrm{i}
\end{pmatrix},
\label{Eq:BSMatrix}
\end{align}
which is in the form of~(\ref{Eq:RMS0}) with $\alpha_{22} \defeq \cot^2\vartheta$.
The value of $\alpha_{22}$ is estimated using single-photon counting as described in Algorithm~\ref{Alg:AmplitudeEstimation}.
The estimated beam-splitter reflectivity is
\begin{equation}
\cos\tilde{\vartheta} = \sqrt{\frac{\alpha_{22}}{1-\alpha_{22}}}.
\label{Eq:FindVartheta}
\end{equation}
The error bar on $\cos\tilde{\vartheta}$ is estimated by repeating the photon counting along the lines of Algorithm~\ref{Alg:AmplitudeEstimation}.

%=================%
\begin{algorithm}[h]
	\begin{algorithmic}[1]
	\Require{\Statex
	\begin{itemize}	
		\item $k, \Omega = \{\omega_1,\omega_2,\dots\omega_{k-1},\omega_k\} \in \left(\mathds{R}^+\right)^{k}$ \Comment{Frequencies at which $f_1,f_2$ are given.}
		\item $f_1,f_2:\Omega \to \mathds{R}^+$\Comment{Given spectral functions.}
		\item $\ell, T= \{\tau_1,\tau_2,\dots,\tau_\ell\}\in (\mathds{R}\cup 0)^\ell$\Comment{Time delay values coincidence is measured at.}
		\item $C^\mathrm{cal}:T \to \mathds{R}^+$\Comment{Measured coincidence curve.}
		\item $\vartheta \in(-\pi,\pi]$ \Comment{$\cos\vartheta$ is reflectivity of calibrating beam splitter.}
	\end{itemize}
	}
	
	\Ensure{\Statex
	\begin{itemize}
		\item $\tilde{\gamma}\in[0,1]$ \Comment{Estimate of mode-matching parameter
 of photon source.}
	\end{itemize}
	}
	
	\Procedure{Calibration}{$k, \Omega, f_1,f_2, \ell, T,C^\mathrm{cal},\vartheta$}
	\State $A \leftarrow \{\cos\vartheta,\sin\vartheta,\sin\vartheta,\cos\vartheta\}$ \Comment{Beam splitter of reflectivity $R$~(\ref{Eq:BSMatrix})}
	\State $\Phi \leftarrow \{0,\pi/2,\pi/2,0\}$\Comment{Beam splitter of reflectivity $R$~(\ref{Eq:BSMatrix})}
	\State $C(\tau,\gamma)\defeq \textsc{Coincidence}(\Omega, f_1, f_2,T,A,\Phi,\gamma)$	\Comment{The quantities $\Omega, f_1, f_2,R^\mathrm{cal}$ are given.~$\gamma$ is unknown. $\textsc{Coincidence}(\Omega, f_1, f_2,T,R^\mathrm{cal},\gamma)$ depends on~$\gamma$ and $\tau$}
	\State \Return $\tilde{\gamma} \leftarrow$ \textsc{Fit}$(C(\tau,\gamma),C^\mathrm{cal}(\tau),1/C^\mathrm{cal}(\tau),\mathrm{InitGuesses})$
	\Comment{Least-squares curve fitting to obtain the value of~$\gamma$ that minimizes $\frac{\sum_{\tau \in T}|C^\mathrm{cal}(\tau) - C(\tau,\gamma)|^2}{C^\mathrm{cal}(\tau)}$.
The argument $1/C^\mathrm{cal}(\tau)$ is the weight function~\cite{Strutz2010} that accounts for experimental noise, which is assumed to be proportional to $\sqrt{C(\tau)}$.
Ignore values of $\tau$ at which $C(\tau) = 1$.
Appendix~\ref{Sec:CurveFitting} details the choice of initial guesses to the algorithm.}
	\EndProcedure
	\caption{\textsc{Calibration} Calculates the mode-matching parameter~$\gamma$ of source-field using a beam splitter of known reflectivity.}
	\label{Alg:Calibration}
	\end{algorithmic}
\end{algorithm}
%------------------------------%

Next we estimate~$\gamma$ using two-photon coincidence counting.
Controllably delayed pairs of photons are incident at the two input ports of the beam splitter.
Coincidence measurement is performed at the output ports for different values of time delay between the two photons.
A curve-fitting algorithm is employed to find the best-fit value of~$\gamma$, i.e., the value $\tilde{\gamma}$ that minimizes the squared sum of residues between the measured counts and the coincidence counts expected from~(\ref{Eq:CoincidenceRate}) for the beam splitter matrix~(\ref{Eq:BSMatrix}).
Algorithm~\ref{Alg:Calibration} details the calculations of $\tilde{\gamma}$, which is used to estimate $\{\theta_{ij}\}$ values accurately.

%------------------------------%
\subsection{Two-photon interference to estimate $\{\theta_{ij}\}$ (Algorithms~\ref{Alg:PhaseCalc2Channel}-\ref{Alg:PhaseCalcNChannel})}
\label{Subsec:Phases}
%------------------------------%
In this subsection, I describe our procedure to estimate the arguments $\{\theta_{ij}\}$ of the representative matrix $U$~(\ref{Eq:RMS0}).
Our procedure requires the measurement of coincidence counts for $2(m-1)^{2}$ different choices of input and output ports.
Of these measurements, $(m-1)^{2}$ are used to estimate the absolute values $\{|\theta_{ij}|\}$ of the arguments and the remaining $(m-1)^{2}$ are used to estimate the signs $\{\operatorname{sgn}\theta_{ij}\}$.

The absolute values $\{|\theta_{ij}|\}$ are estimated as follows.
Single-photon pairs are incident at input ports $1$ and~$j$ and coincidence measurements are performed at output ports $1$ and~$i$ for $i,j\in\{2,\dots,m\}$.
The state~(\ref{Eq:twophotonstate}) of a photon pair is transformed under the action of the $2\times 2$ submatrix
{\setlength{\arraycolsep}{1pt}
\begin{equation}
U_{i1j1} = \begin{pmatrix}\sqrt{\kappa_1}&0\\0&\sqrt{\kappa_{i}}\end{pmatrix}
\begin{pmatrix}\sqrt{\lambda_1}&0\\0&\sqrt{\lambda_{i}}\end{pmatrix}
\begin{pmatrix}1 &1\\1&\alpha_{ij}\e^{\mathrm{i} \theta_{ij}}\end{pmatrix} \begin{pmatrix}\sqrt{\mu_1}&0\\0&\sqrt{\mu_{j}}\end{pmatrix}
\begin{pmatrix}\sqrt{\nu_1}&0\\0&\sqrt{\nu_{j}}\end{pmatrix}
\end{equation}}
of $U$ labelled by the rows $1$ and~$i$ and columns $1$ and~$j$.
The probability of detecting a coincidence at the output ports $1,i$ is
\begin{align}
C_{i1j1}(\tau)
=&
\kappa_{j}\kappa_1\lambda_{j}\lambda_1\nu_{i}\nu_1\mu_{i}\mu_1\Big[\left\{\alpha_{ij}^2+1\right\}\int \mathrm{d}\omega_1\mathrm{d}\omega_2 |f_{j}(\omega_1)f_{1}(\omega_2)|^2 \nonumber
\\ &+2 \gamma \alpha_{ij} \int \mathrm{d}\omega_1\mathrm{d}\omega_2 f_{j}(\omega_1)f_{1}(\omega_2)f_{j}(\omega_2)f_{1}(\omega_1) \cos\left(\omega_2\tau-\omega_1\tau+\theta_{ij}\right)\Big],
\label{Eq:CoincidenceRateGH}
\end{align}
which is obtained by setting $i' = j'=1$ in~(\ref{Eq:CoincidenceRate}).

The measured coincidence counts are used to estimate the value of $|\theta_{ij}|$ as follows.
The shape of the coincidence-versus-$\tau$ curve~(\ref{Eq:CoincidenceRateGH}) depends on the values of $\alpha_{ij}$ and~$\theta_{ij}$.
The shape does not depend on the parameters $\kappa_1,\kappa_i,\lambda_1,\lambda_i,\mu_1,\mu_j,\nu_1,\nu_j$, which lead to a constant multiplicative factor to the coincidence expression.
Furthermore, the shape is unchanged under the transformation $\theta_{ij} \to -\theta_{ij}$ for $\theta_{ij} \in (-\pi,\pi]$ if the spectral functions are identical.
Hence, $|\theta_{ij}|$ can be estimated using the shape of the coincidence function~(\ref{Eq:CoincidenceRateGH}) and the values $\{\tilde\alpha_{ij}\}$ estimated using Algorithm~\ref{Alg:AmplitudeEstimation}.
A curve-fitting algorithm estimates the value $|\tilde{\theta}_{ij}|\in[0,\pi]$ that best fits the measured coincidence counts.
The calculation of $\{|\tilde{\theta}_{ij}|\}$ is detailed in Algorithm~\ref{Alg:PhaseCalc2Channel}.

%=================%
\begin{algorithm}
	\begin{algorithmic}[1]
	\Require{\Statex
	\begin{itemize}	

		\item $k, \Omega = \{\omega_1,\omega_2,\dots\omega_{k-1},\omega_k\} \in \left(\mathds{R}^+\right)^{k}$ \Comment{$f_1,f_2$ are measured at frequencies $\Omega$.}
	\item $f_1,f_2:\Omega \to \mathds{R}^+$\Comment{measured spectra.}
		\item $\ell, T= \{\tau_1,\tau_2,\dots,\tau_\ell\}\in (\mathds{R}\cup 0)^\ell$\Comment{Time delay values coincidence is measured at.}
	\item $C^{\mathrm{exp}}:T \to \mathds{R}^+$\Comment{Measured coincidence curve.}
 \item $A \leftarrow \{\alpha_{ij},\alpha_{ij'},\alpha_{i'j},\alpha_{i'j'}\}$ \Comment{Complex amplitudes of $2\times 2$ submatrix of $A$~(\ref{Eq:RMS0}).}
	\item $\Theta \leftarrow \{\theta_{ij'},\theta_{i'j},\theta_{i'j'}\in (-\pi,\pi]\}$\Comment{Three complex arguments of submatrix.}
	\item~$\gamma$ \Comment{Mode-matching parameter of photon source.}
	\end{itemize}
	}
	
	\Ensure{\Statex
	\begin{itemize}
	\item $|\tilde{\theta}_{ij}|$\Comment{Estimated magnitude of the unknown complex argument.}
	\end{itemize}
	}
	
	\Procedure{Argument2Port}{$k, \Omega,f_1,f_2,\ell, T,C_\mathrm{exp},A,\Theta,\gamma$}
	\State $\Phi \defeq \{\theta_{ij},\theta_{ij'},\theta_{i'j},\theta_{i'j'}\}$\Comment{Set of three known phases and one unknown phase.}
	\State $C(\tau,\theta_{ij})\defeq \textsc{Coincidence}(\Omega, f_1, f_2,T,A,\Phi,\gamma)$%\Comment{For given $\Omega, f_1, f_2,R^\mathrm{cal}$, $\textsc{Coincidence}(\Omega, f_1, f_2,T,R^\mathrm{cal},\gamma)$ is defined as the $\theta_{ij}$-dependent function of time delay value $C(\tau,\gamma):T\times\mathds{R}^+\rightarrow \mathds{R}^+$.}
	\State \Return $\tilde{\theta}_{ij} \leftarrow \left|\textsc{LM}(C(\tau,\theta_{ij}),C^{\mathrm{exp}}(\tau),1/C^\mathrm{exp}(\tau))\right|$
	\Statex\Comment{Use curve fitting to compute the $\theta_{ij}$ value that minimizes $\frac{\sum_{\tau \in T}|C_\mathrm{exp}(\tau) - C(\tau,\gamma)|^2}{C_\mathrm{exp}(\tau)}$.}
	\EndProcedure
	\caption{\textsc{Argument2Port}: Calculates the unknown complex argument in the entries of a $2\times 2$ transformation using a two-photon coincidence curve.}
	\label{Alg:PhaseCalc2Channel}
	\end{algorithmic}
\end{algorithm}
%------------------------------%

Our procedure computes the signs by using an additional $(m-1)^{2}$ coincidence measurements.
First we arbitrarily set $\theta_{22}$ as positive
\begin{equation}
\operatorname{sgn}\theta_{22} = 1
\label{Eq:Theta22}
\end{equation}
 because of the invariance\footnote{
Expectation values of Fock-state projection measurement with Fock-state inputs are unchanged under $U\to U^*$ if the spectral functions are equal $f_1(\omega) = f_2(\omega)$.
Otherwise, the sign of $-\alpha_{22}$ can be ascertained using the difference in the $\tau>0$ and~$\tau<0$ coincidence counts in $C_{2,2,1,1}(\tau)$.

} of one- and two-photon statistics under complex conjugation $U\to U^*$~\cite{Laing2012}.
The signs of the remaining arguments $\{\theta_{ij}\}$ are set using the coincidence counts between output ports $\{i,i'\}$ when photon pairs are incident at input ports $\{j,j'\}$ for a suitable choice of $\{i',j'\}$ as I describe below.
The coincidence probability at the output ports $i,i'$ is
\begin{align}
C_{ii'jj'}(\tau) =&\,\kappa_{i}\kappa_{i'}\lambda_{i}\lambda_{i'}\mu_{j}\mu_{j'}\nu_{j}\nu_{j'}\Big[\left(\alpha_{ij}^2 \alpha_{i'j'}^2+ \alpha_{ij'}^2 \alpha_{i'j}^2\right)\int \mathrm{d}\omega_1\mathrm{d}\omega_2 |f_{j}(\omega_1)f_{j'}(\omega_2)|^2 \nonumber
\\ &+2 \gamma \alpha_{ij} \alpha_{ij'} \alpha_{i'j} \alpha_{i'j'} \int \mathrm{d}\omega_1\mathrm{d}\omega_2 f_{j}(\omega_1)f_{j'}(\omega_2)f_{j}(\omega_2)f_{j'}(\omega_1)
\nonumber\\
&\times \cos\left(\omega_2\tau-\omega_1\tau+\beta_{ii'jj'}\right)\Big],\label{Eq:Coiniijj}
\end{align}
where
\begin{equation}
\beta_{ii'jj'} \defeq |\theta_{i'j'}-\theta_{ij'}-\theta_{i'j}+\theta_{ij}|\in[0,\pi]\label{Eq:Beta}.
\end{equation}
Curve fitting is employed to estimate the value of $\beta_{ii'jj'}$ that best fits the measured coincidence counts.

%=================%
\begin{algorithm}[h]
	\begin{algorithmic}[1]
	\Require{\Statex
	\begin{itemize}	
		\item $\beta \equiv |\theta_{i'j'}-\theta_{ij'}-\theta_{i'j}+\theta_{ij}|$ \Comment As defined in~(\ref{Eq:SignEquation}).
		\item $\theta_{i'j'}, \theta_{ij'}, \theta_{i'j}, \left|\theta_{ij}\right|$ \Comment Equations~(\ref{Eq:SignEquation}-\ref{Eq:BetaMinus}).
	\end{itemize}
	}
		
	\Ensure{\Statex
	\begin{itemize}
		\item $\operatorname{sgn}\theta_{ij}$\Comment{Sign of $\theta\in(-\pi,\pi]$ is defined in~(\ref{Eq:SignDef})}
	\end{itemize}}
		\Procedure{SignCalc}{$\beta,\theta_{i'j'}, \theta_{ij'}, \theta_{i'j}, \left|\theta_{ij}\right|$}
		\State $\beta^+ \leftarrow |\theta_{i'j'}-\theta_{ij'}-\theta_{i'j}+|\theta_{ij}||$ \Comment If $\theta_{ij}<0$, then $\beta = \beta^{-}$.
		\State $\beta^- \leftarrow |\theta_{i'j'}-\theta_{ij'}-\theta_{i'j}-|\theta_{ij}||$ \Comment If $\theta_{ij}>0$, then $\beta = \beta^{+}$.
		\State $\text{{sgn}}\,\theta_{ij} \leftarrow \operatorname{sgn}\left|\beta-\beta^-\right|-\left|\beta-\beta^+\right|$
 		\State\Return $\text{sgn}\,\theta_{ij}$
		\EndProcedure
		
	\caption{\textsc{SignCalc}: Calculates the complex-phase sign of an element of the $2\times 2$ submatrix of an interferometer transformation matrix.}
	\label{Alg:SignCalc}
	\end{algorithmic}
\end{algorithm}
%------------------------------%

The estimated value of $\beta_{ii'jj'}$ is employed by Algorithm~\ref{Alg:SignCalc} to ascertain the sign of $\theta_{ij}$.
Algorithm~\ref{Alg:SignCalc} relies on the identity
\begin{equation}
\operatorname{sgn}\theta_{ij} = \operatorname{sgn}\left(|\beta_{ii'jj'}-\beta^{-}_{ii'jj'}|-|\beta_{ii'jj'}-\beta^{+}_{ii'jj'}|\right),\label{Eq:SignEquation}
\end{equation}
and on known values of
\begin{equation}
\beta^{\pm}_{ii'jj'}\defeq |\theta_{i'j'}-\theta_{ij'}-\theta_{i'j} \pm |\theta_{ij}||, \beta^{\pm}_{ii'jj'}\in[0,\pi]\label{Eq:BetaMinus}
\end{equation}
to ascertain the sign of $\theta_{ij}$.
If the sign of $\theta_{ij}$ is positive, then $\beta_{ii'jj'} = \beta_{ii'jj'}^{+}$ and~(\ref{Eq:SignEquation}) returns a positive $\operatorname{sgn}\theta_{ij}$. Otherwise, $\beta_{ii'jj'}=\beta_{ii'jj'}^-$, in which case~(\ref{Eq:SignEquation}) gives a negative sign.

%Algorthm~\ref{Alg:PhaseCalcNChannel} finds the arguments of the representative matrix as follows.
%The arguments $\{\theta_{1i}\}$ and~$\{\theta_{i1}\}$ are zero by definition~(\ref{Eq:RMS0}).
%The amplitudes $\{|\theta_{ij}|\}$ of the other arguments are computed using the $(m-1)^{2}$ coincidence measurements $\{C_{i1j1}(\tau): i,j \in\{2,\dots,m\}\}$.
%The sign of $\theta_{22}$ is arbitrarily set positive according to~(\ref{Eq:Theta22}).
%The signs of remaining arguments are computed iteratively.
Algorithm~\ref{Alg:PhaseCalcNChannel} iteratively chooses indices $i,i',j,j'$ such that the signs of $\theta_{ij'}, \theta_{i'j},\theta_{i'j'}$ have already been ascertained before ascertaining the sign of $\theta_{ij}$.
In each iteration, the values of $\beta_{ii'jj'}^{\pm}$ are calculated by substituting $|\theta_{ij}|,|\theta_{ij'}|,|\theta_{i'j}|,|\theta_{i'j'}|,\operatorname{sgn}\theta_{ij'},\operatorname{sgn}\theta_{i'j},\operatorname{sgn}\theta_{i'j'}$.
The algorithm estimates $\beta_{ii'jj'}$ by curve fitting measured coincidence counts to~(\ref{Eq:Coiniijj}).
Algorithm~\ref{Alg:SignCalc} is ascertains the sign of $\theta_{ij}$ using the estimates of $\beta_{ii'jj'}$ and $\beta_{ii'jj'}^{\pm}$.
One suitable ordering of indices $ii'jj'$, which I depict in Figure~\ref{Figure:Ordering}, is
\begin{itemize}
\item set $i'=2,j'=1$ to determine sgn$\theta_{i2}$ for $i\in \{3,\dots,m\}$ (Figure~\ref{Figure:Ordering}b),
\item set $i'=1,j'=2$ to determine sgn$\theta_{2j}$ for $j\in \{3,\dots,m\}$ (Figure~\ref{Figure:Ordering}c),
\item set $i'=2,j'=2$ to determine sgn$\theta_{ij}$ for $(i,j)\in\{3,\dots,m\}\times\{3,\dots,m\}$ (Figure~\ref{Figure:Ordering}d).
\end{itemize}
In summary, $\operatorname{sgn}\theta_{ij}$ is determined using the values of $\beta_{ii'jj'}$, which are estimated by curve fitting and of~$\beta_{ii'jj'}^\pm$, which are computed using the signs and amplitudes of $\theta_{ij'},\theta_{i'j},\theta_{i'j'}$.
Algorithms~\ref{Alg:PhaseCalc2Channel}-\ref{Alg:PhaseCalcNChannel} detail the step-by-step procedure to determine the absolute values and the signs of $\{\theta_{ij}\}$.

For certain interferometers $U$, the ordering of indices $ii'jj'$ depicted in Figure~\ref{Figure:Ordering} can lead to instability in the characterization procedure.
\ref{Sec:Instability} elucidates on this instability and presents strategies to counter the instability.
This completes our procedure to characterize the matrix $A$ for representative matrix $U = LAM$.
In the next subsection, I present a procedure to estimate the matrix that is most likely for the characterized matrix $A$.

\begin{figure}[h]
\centering
\subfloat[]{\includegraphics[width=0.24\textwidth]{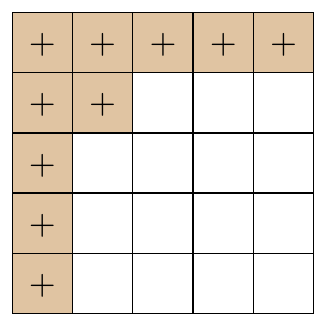}}
\subfloat[]{\includegraphics[width=0.24\textwidth]{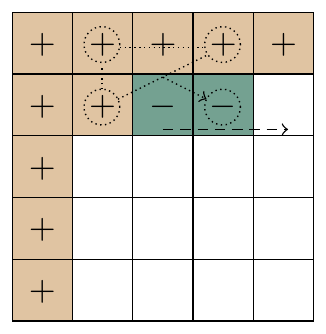}}
\subfloat[]{\includegraphics[width=0.24\textwidth]{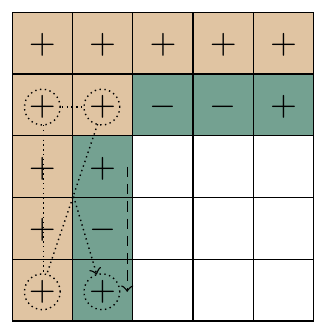}}
\subfloat[]{\includegraphics[width=0.24\textwidth]{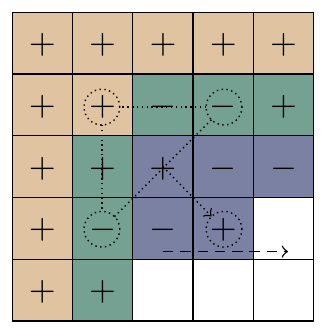}}
\caption{%
A depiction of the sign estimation procedure in Lines~\ref{Algline:StartSigns}--\ref{Algline:EndSigns} of Algorithm~\ref{Alg:PhaseCalcNChannel}.
(a)~The first row and first column arguments $\{\theta_{i1}\},\{\theta_{1j}\}$ are zero, so their signs are arbitrarily set as positive. $\theta_{22}$ is set as positive according to~(\ref{Eq:Theta22}).
(b)~The sign of each second row argument $\theta_{i2}$ is set using the known values $|\theta_{22}|,|\theta_{i2}|$ and coincidence measurement for input ports $1,2$ and output ports $2,i$ as in Line~\ref{Alglin:Rows}.
(c)~The sign of each second column argument $\theta_{2j}$ is set using the known values $|\theta_{22}|,|\theta_{2j}|$ and coincidence measurement for input ports $2,j$ and output ports $1,2$ as in Line~\ref{Alglin:Columns} of Algorithm~\ref{Alg:PhaseCalcNChannel}.
(d)~The signs of each remaining argument $\theta_{ij}$ is set using the known values $|\theta_{22}|,|\theta_{i2}|,|\theta_{2j}|$ and coincidence measurement for input ports $2,j$ and output ports $2,i$ in Line~\ref{Alglin:Rest} of Algorithm~\ref{Alg:PhaseCalcNChannel}.
}\label{Figure:Ordering}
\end{figure}
%------------------------------%
\begin{algorithm}[p]
	\begin{algorithmic}[1]
	\Require{\Statex
	\begin{itemize}	
		\item $k, \Omega = \{\omega_1,\omega_2,\dots\omega_{k-1},\omega_k\} \in \left(\mathds{R}^+\right)^{k}$ \Comment{$f_1,f_2$ are measured at frequencies $\Omega$.}
	\item $f_1,f_2:\Omega \to \mathds{R}^+$\Comment{measured spectra.}
		\item $\ell, T= \{\tau_1,\tau_2,\dots,\tau_\ell\}\in (\mathds{R}\cup 0)^\ell$\Comment{Time delay values coincidence is measured at.}
	\item $C^{\mathrm{exp}}_{ii'jj'}(\tau)$ for $(i,i',j,j')\in\{1,2\}\times\{1,\dots,m\}\times\{1,2\}\times\{1,\dots,m\},\,i\ne i', j\ne j'$
	\Statex \Comment{Measured coincidence at output ports $i',j'$ when photons that have mutual delay $\tau$ are incident at input ports $i,j$.}
	\item $\tilde\alpha:\{2,\dots,m\}\times\{2,\dots,m\}\to\mathds{R}^+$ \Comment{Complex amplitudes~(\ref{Eq:RMS0}).}
	\item $\gamma\in[0,1]$ \Comment{Mode-matching parameter estimated using Algorithm~\ref{Alg:Calibration}.}
	\end{itemize}
	}	
	\Ensure{\Statex
	\begin{itemize}
		\item $\tilde\theta_{ij}:\{1,\dots,m\}\times\{1,\dots,m\}\to (-\pi,\pi]$\Comment{Complex Arguments~(\ref{Eq:RMS0}).}
	\end{itemize}
	}
	\Procedure{ArgumentCalc}{$k, \Omega,f_1,f_2,\ell, T,C^{\mathrm{exp}}_{ii'jj'}(\tau),\alpha_{ij},\gamma$}
	\For{$i$ in $\{1, \dots, m\}$}
		\State $\theta_{i1}, \theta_{1i},\,\operatorname{sgn}\theta_{i1},\,\operatorname{sgn}\theta_{1i} \leftarrow 0$ \Comment The first row, column are real valued.
	\EndFor	

	\For{$(i,j)$ in $\{2, \dots, m\}\times\{2, \dots, m\}$}
		\State $A \leftarrow \{1,1,1,\alpha_{gh}\}$, $\Phi \leftarrow \{0,0,0\}$ \Comment $2\times2$ matrix: rows $1,i$, columns $1,j$.
		\State $|\tilde\theta_{ij}| \leftarrow \textsc{Argument2Port}\left(C^{\mathrm{exp}}_{1i1j}T,\Omega, f_1, f_2,T,A,\Phi,\gamma\right)$
	\EndFor\label{Alglin:EndAmplitudes}
	
	\State $\operatorname{sgn}\theta_{22}\leftarrow 1$ \Comment The sign of $\theta_{22}$ is positive by definition.\label{Algline:StartSigns}
	
	\For{($i,i',j,j')\in \{2\}\times\{3,\dots,m\}\times \{2\}\times\{3,\dots,m\}
	\cup \{1\}\times\{2\}\times \{2\}\times\{3,\dots,m\}
	\cup \{2\}\times\{3,\dots,m\}\times \{1\}\times\{2\}$}
	\State $A \leftarrow \{0,0,0\}, \Phi \leftarrow \{0,0,0\}$
	\State $\beta_{i,i',j,j'} \leftarrow \textsc{Argument2Port}(C^{\mathrm{exp}}_{ii'jj'}(\tau),\Omega, f_1, f_2,T,A,\Phi,\gamma)$
	\EndFor
	\For{$i$ in $\{3, \dots, m\}$} \label{Alglin:StartSecondRowColumn}
		\State $\tilde\theta_{i2} \leftarrow |\tilde\theta_{i2}|${\textsc{SignCalc}}$(\beta_{122i},0,\theta_{22},0,|\theta_{i2}|,));$\label{Alglin:Rows}
		\State $\tilde\theta_{2i} \leftarrow |\tilde\theta_{2i}|${\textsc{SignCalc}}$(\beta_{2i12},0,\theta_{22},0,|\theta_{2i}|,));$\label{Alglin:Columns}
	\EndFor\label{Alglin:EndSecondRowColumn}
	
	\For{$(i,j)$ in $\{3, \dots, m\}\times\{3, \dots, m\}$}
		\State $\tilde\theta_{ij} \leftarrow |\tilde\theta_{ij}|${\textsc{SignCalc}}$(\beta_{ii'jj'},\theta_{22},\theta_{i2},\theta_{2j},|\theta_{ij}|)$\label{Alglin:Rest}
	\EndFor
	\State \Return $\{\theta_{ij}\}$
	\EndProcedure \label{Algline:EndSigns}

	\caption{\textsc{ArgumentCalc}: Calculate $\{\theta_{ij}\}$ using two-photon coincidences}
	\label{Alg:PhaseCalcNChannel}
	\end{algorithmic}
\end{algorithm}
%------------------------------%

%Conservation of total photon numbers requires that the transformation matrix for a lossless linear optical interferometer is unitary,
%\begin{equation}
%U^\dagger U = UU^\dagger = \mathds{1}.
%\end{equation}
%Real-world interferometers can differ from unitarity.
%We assume unitarity and enforce it by requiring the following conditions
%_{i} the first row is normalized,
%(ii)~the first row is orthognal to all other rows,
%(iii)~the first column is normalized, and
%(iv) the first column is orthognal to all other columns.
%These constraints on the rows and columns respectively can be expressed as

%------------------------------%
\subsection{Maximum-likelihood estimation for finding unitary matrix}
\label{Subsec:MaxLikelyUnitary}
%------------------------------%
At this stage, we have estimated the matrix $A$~(\ref{Eq:RMS0}).
The diagonal matrices $L$ and~$M$ can be uniquely determined from $A$ as follows.
The representative matrix $U = LAM$ is unitary so we have
\begin{equation}
U U^\dagger =\mathds{1},
\end{equation}
which, upon substitution $U=LAM$, implies that
\begin{align}
LAMM^*A^\dagger L^* &= \mathds{1}\nonumber\\
\implies AMM^* A^\dagger &= L^{-1}L^{*-1}.\label{Eq:MRRM}
\end{align}
Considering the first columns of the matrices~(\ref{Eq:MRRM}) gives
\begin{equation}
A_{ij}M^*_{jj}M_{jj}A^\dagger_{j1} = \begin{pmatrix}1\\ 0 \\ \vdots\\0\end{pmatrix}
\end{equation}or
\begin{equation}
A\begin{pmatrix}
\mu_1\\\mu_2\\\vdots \\\mu_m\end{pmatrix}
= \begin{pmatrix}1 \\0 \\\vdots \\0\end{pmatrix}.
\label{Eq:MS}
\end{equation}
Similarly, using $U^\dagger U= \mathds{1}$ we obtain
\begin{equation}A^\dagger\begin{pmatrix}1\\\lambda_2\\\vdots \\\lambda_m\end{pmatrix}
=\frac{1}{\mu_1}
\begin{pmatrix}1 \\0 \\\vdots \\0\end{pmatrix}.
\label{Eq:MR}
\end{equation}
Equations~(\ref{Eq:MS}) and~(\ref{Eq:MR}) are systems of linear equations that can be solved for $L$ and~$M$ respectively using standard methods~\cite{Kailath1980}.
The solutions $L$ and~$M$ of the linear systems and the characterized matrix $A$ give us the representative matrix $U = LAM$.

%------------------------------%
\begin{algorithm}[h]
	\begin{algorithmic}[1]
	\Require{\Statex
	\begin{itemize}
		\item $\tilde\alpha:\{1,\dots,m\}\times\{1,\dots,m\}\to\mathds{R}^+\cup 0$\Comment{Estimated amplitudes of $A$~(\ref{Eq:RMS0})}.
		\item $\tilde\theta:\{1,\dots,m\}\times\{1,\dots,m\}\to(-\pi,\pi]$\Comment{Estimated arguments of $A$~(\ref{Eq:RMS0})}.
	\end{itemize}
	}
	\Ensure{\Statex
	\begin{itemize}
		\item $W\in \mathrm{SU}(n)$\Comment{Unitary matrix with maximum likelihood of generating $A$.}
	\end{itemize}
	}
	
	\Procedure{MaxLikelyUnitary}{$\alpha_{ij},\theta_{ij}$}
	\State $\lambda_1\leftarrow 1$
	\State $\{\tilde\mu_i: i\in \{1,\dots,m\}\} \leftarrow$ solution of system~(\ref{Eq:MS}) of linear equations.
	\State $\{\tilde\lambda_i: i\in \{2,\dots,m\}\}\leftarrow$ solution of system~(\ref{Eq:MR}) of linear equations.
	\State $\tilde{U}_{ij}\leftarrow \tilde\lambda_{i}\tilde\alpha_{ij}\e^{\mathrm{i}\tilde\theta_{ij}}\tilde\mu_{j} $
	\State $W \leftarrow \left(\tilde{U}\tilde{U}^\dagger\right)^{-\frac{1}{2}}\tilde{U}$\Comment{Assumption: $U_{ij}-\tilde{U}_{ij}$ is an iid Gaussian random variable with zero mean for all $i,j \in\{1,\dots,m\}$.}
	\EndProcedure
	\caption{\textsc{MaxLikelyUnitary}: Calculates unitary matrix that has maximum likelihood of generating estimated $\{\alpha_{ij}\},\{\theta_{ij}\}$}
	\label{Alg:MaxLikelyUnitary}
	\end{algorithmic}
\end{algorithm}
%------------------------------%

The experimentally determined ${\tilde{A}}$ is different from the actual~$A$ because of random and systematic error in the experiment, where I denote the experimentally determined values of interferometer parameter $\bullet$ by $\tilde{\bullet}$.
Similarly, the $\tilde{L}$ and~$\tilde{M}$ matrices obtained by solving Equations~(\ref{Eq:MS}) and~(\ref{Eq:MR}) for $\tilde{A}$ (rather than $A$) differ from the actual $L$ and $M$ respectively.
The estimated $\tilde{U} = \tilde{L}\tilde{A}\tilde{M}$ is thus a non-unitary matrix and is not equal to $U$ in general.
Furthermore, $\tilde{U}$ is a random matrix, which depends on the random errors in the one- and two-photon experimental data.

We employ maximum-likelihood estimation to calculate the unitary matrix $W$ that best fits the collected data.
First, bootstrapping techniques are used to estimate the probability-density function\gls{pdf} of the entries of the random matrix $\tilde{U}$~\cite{Efron1979,Efron1986}.
Next, standard methods in maximum-likelihood estimation~\cite{Scholz2004} are employed to find the unitary matrix $W$.
Maximum-likelihood estimation simplifies under the assumption that the error on $\tilde{U}$ is a Gaussian random matrix ensemble, i.e, that the matrix entries $\left\{\tilde{U}_{ij}\right\}$ are complex independent and identically distributed (iid) Gaussian random variables centred at the correct matrix entries.
In this case, the most likely unitary matrix $W$ is the one that minimizes the Frobenius distance from $\tilde{U}$~\cite{Tao2012}.
The unitary matrix
\begin{equation}
W = \left(\tilde{U}\tilde{U}^\dagger\right)^{-\frac{1}{2}}\tilde{U},
\label{Eq:W}
\end{equation}
minimizes the~\gls{FND} from $\tilde{U}$~\cite{Keller1975}.
Thus, if the random errors $\{U_{ij} - \tilde{U}_{ij}\}$ in the matrix elements are iid Gaussian random variables with mean zero, then $W$ is the best-fit unitary matrix.
Figure~\ref{Figure:UUV} is a depiction of the actual, the estimated and the most likely transformation matrices.
Algorithm~\ref{Alg:MaxLikelyUnitary} computes $W$.

This completes our procedure to estimate the most-likely unitary matrix $W$ that represents the linear optical interferometer.
In the next section, I present a procedure to estimate the error bars on the entries of the estimated representative matrix $W$ accurately.

\begin{figure}[h]
\begin{centering}
\includegraphics[width=0.5\textwidth]{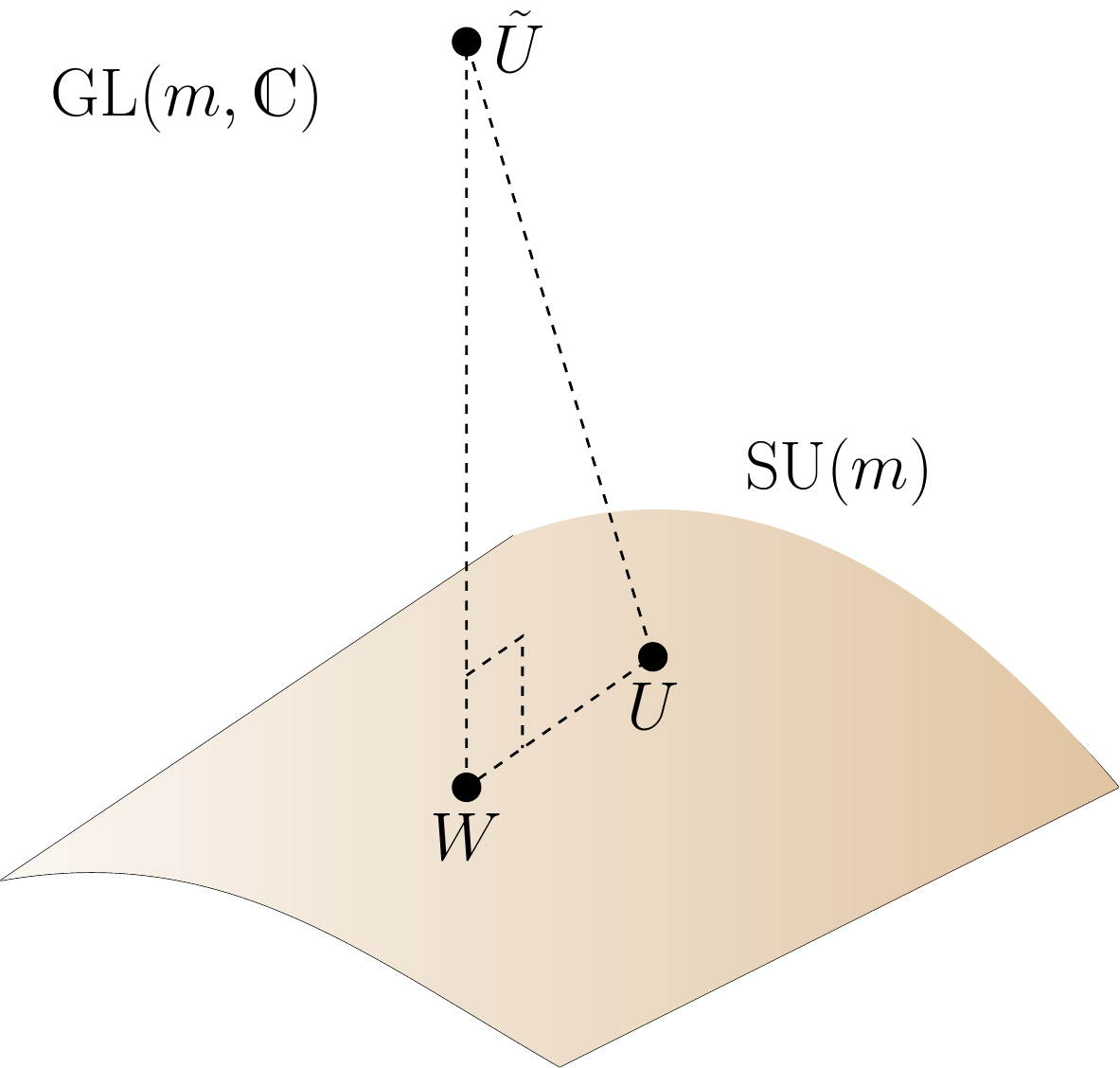}
\caption{
A depiction of the error in reconstruction of the interferometer matrix $U$.
The matrix $U$ represents the unitary transformation effected by the interferometer.
$\tilde U$ is the complex-valued transformation matrix returned by the reconstruction procedure.
Algorithm~\ref{Alg:MaxLikelyUnitary} returns $W$, which represents the unitary matrix that is most likely to have generated the data collected in the characterization experiment.
}
\label{Figure:UUV}
\end{centering}
\end{figure}

%------------------------------%
\subsection{Bootstrapping to estimate error bars (Algorithm~\ref{Alg:Bootstrapping})}
\label{Sec:Bootstrapping}
%------------------------------%

In this section, I present a procedure to estimate the error bars on the matrix entries $\{W_{ij}\}$ of the characterized representative matrix $W$.
The entries $\{W_{ij}\}$ computed by Algorithms~\ref{Alg:Coincidence}--\ref{Alg:MaxLikelyUnitary} are random variables because of random error in experiments.
Obtaining accurate error bars on these random variables is important for using characterized linear optical interferometers in quantum computation and communication.
Current procedures compute error bars under the assumption that Poissonian shot noise is the only source of error in experiment~\cite{Mower2014,Carolan2015}.

We choose to employ bootstrapping on the data determine error bars~\cite{Efron1979,Efron1986,Diaconis1983,Willmott1985,Manly2006}.
Monte-Carlo simulation is widely used but this technique is not applicable here because the Poissonian shot noise assumption is not reliable given the presence of other sources of error some of which are not understood.
Bootstrapping is preferred because the nature of the error need not be characterized and instead relies on random sampling with replacement from the measured data.
Bootstrapping can be employed to yield estimators such as bias, variance and error bars.

Algorithm~\ref{Alg:Bootstrapping} calculates the error bars $\sigma(W_{ij})$ using estimates of the $\{W_{ij}\}$ pdf's, which are obtained using bootstrapping as follows.
The algorithm simulates $N$ characterization experiments using the one- and two-photon data, i.e., the inputs to Algorithms~\ref{Alg:Coincidence}--\ref{Alg:MaxLikelyUnitary}.
In each of the $N$ rounds, the one- and two-photon data are randomly sampled with replacement (resampled) to generate simulated data.
The data thus simulated are given as inputs to Algorithms~\ref{Alg:Coincidence}--\ref{Alg:MaxLikelyUnitary}, which return the simulated representative matrices
\begin{equation}
\left\{ W'^{b}: b\in \{1,\dots,N\},\,N\in\mathds{Z}^{+} \right\}.
\end{equation}
The pdf's of the simulated-matrix entries $\{ W'^{b}_{ij}: b\in \{1,\dots,N\}\}$ converge to the pdf's of the respective elements $\{W_{ij}\}$ for large enough $N$~\cite{Diciccio1996,Davison1997}.

The simulated data are obtained in each round by resampling from the one- and two-photon experimental data as follows.
Single-photon detection counts are simulated by resampling from the set~$\{N_{ijb_{j}}:\, b_{j}\in\{1,\dots,B\}\}$ of experimental detection counts (Line~\ref{Alglin:Sim1Photon} of Algorithm~\ref{Alg:Bootstrapping}).
Two-photon coincidence counts are simulated by shuffling residuals obtained on curve-fitting experimental data.
Specifically, the algorithm (Line~\ref{Alglin:StartBootTheta}) resamples from the set
\begin{equation}
\{r(\tau) = C^{\mathrm{exp}}_{ii'jj'}(\tau) - C_{ii'jj'}(\tau): \tau \in T\}
\end{equation}
of residuals obtained by fitting experimentally measured coincidence counts to function $C_{ii'jj'}(\tau)$~(\ref{Eq:CoincidenceRate}).
The resampled residuals are added to the fitted curve to generate the simulated data (Line~\ref{Alglin:EndBootTheta}).\footnote{
The pdf of the residuals is different for different values of $\tau$.
We assume that the pdf's for different $\tau$ are of the same functional form, albeit with different widths.
The distribution of the residuals for different values of $\tau$ are determined using standard methods for non-parametric estimation of residual distribution~\cite{Akritas2001,Chen2003}.
Algorithm~\ref{Alg:Bootstrapping} normalizes the residuals before resampling from the residual distribution.
}
Algorithms~\ref{Alg:Coincidence}--\ref{Alg:MaxLikelyUnitary} are used to obtain the simulated elements $W_{ij}$ of the representative matrix.
Finally, the error bars on the $\{W_{ij}\}$ are estimated by the standard deviation of the pdf of the elements.

This completes the characterization of representative matrix $W$ and the error bars on its elements.
The next section details a procedure for the scattershot characterization of the interferometer to reduce the experimental time required for characterizing a given interferometer.

%------------------------------%
\begin{algorithm}[p]
	\begin{algorithmic}[1]
	\Require{\Statex
	\begin{itemize}	
	
	\item $k, \Omega, f_1,f_2:\Omega \to \mathds{R}^+$\Comment{Spectral functions: same as Algorithm~\ref{Alg:Calibration}.}
	\item $\ell, T= \{\tau_1,\tau_2,\dots,\tau_\ell\}\in (\mathds{R}\cup 0)^\ell$\Comment{Time delay values.}
 \item $m,C^{\mathrm{cal}}(\tau), C^{\mathrm{exp}}_{ii'jj'}(\tau)$ for $\tau \in T$ and $(i,i',j,j')$ \Comment{Same as Algorithms~\ref{Alg:Calibration} and~\ref{Alg:PhaseCalcNChannel}}.
	\item $B,N_{ijb_{j}}:\{1,\dots,m\}\times\{1,\dots,m\} \times \{1,\dots,B\}\to \mathds{Z}^+$ as in Algorithm~\ref{Alg:AmplitudeEstimation}.
 \item $N$ \Comment Number of bootstrapping samples.

	\end{itemize}
	}	
	\Ensure{\Statex
	\begin{itemize}
		\item $\sigma\left(\operatorname{re}W_{ij}\right),\sigma(\operatorname{im}W_{ij}):\{1,\dots,m\}\times\{1,\dots,m\}\to\mathds{R}^+$ \Comment{Error in $W$ elements.}
	\end{itemize}
	}
	
	\Procedure{Bootstrap}{$k, \Omega, f_1,f_2,\ell,T,\gamma,C^{\mathrm{exp}}_{ii'jj'}(\tau),\theta_{ij},B,N_{ijb},N$}
	\State $A \leftarrow \{\cos\vartheta,\sin\vartheta,\sin\vartheta,\cos\vartheta\}$, $\Phi \leftarrow \{0,\pi/2,\pi/2,0\}$
	\State $\mathrm{Residuals}^{\mathrm{cal}}(\tau) \leftarrow C^\mathrm{cal}(\tau) - \textsc{Coincidence}(k,\Omega,f_1,f_2,\ell,T,A,\Theta,\gamma)$
	\State $\mathrm{NormalResiduals}^\mathrm{cal}(\tau) \leftarrow \frac{\mathrm{Residuals}^{\mathrm{cal}}(\tau)}{C^\mathrm{fit}(\tau)}$\Comment{Assumption: $\mathrm{Residuals}^{\mathrm{cal}}(\tau)$ pdf width $\propto C^{\mathrm{fit}}(\tau)$.}
	\For{$(i,i',j,j')\in\{1,2\}\times\{1,\dots,m\}\times\{1,2\}\times\{1,\dots,m\},\,i\ne i', j\ne j'$}
		\State $A \leftarrow \{\alpha_{i'j'},\alpha_{i'j},\alpha_{ij'},\alpha_{ij}\}$, $\Phi \leftarrow \{\theta_{i'j'},\theta_{i'j},\theta_{ij'},\theta_{ij}\}$
		\State $C^\mathrm{fit}_{ii'jj'}\leftarrow$ \textsc{Coincidence}($k,\Omega,f_1,f_2,\ell,T,A,\Theta,\gamma$)
		\State $\mathrm{Residuals}_{ii'jj'}(\tau) \leftarrow C_{ii'jj'}^\mathrm{exp}(\tau) - C_{ii'jj'}^\mathrm{fit}(\tau)$
		\State $\mathrm{NormalResiduals}_{ii'jj'}(\tau) \leftarrow \frac{\mathrm{Residuals}_{ii'jj'}(\tau) }{C_{ii'jj'}^\mathrm{fit}(\tau)}$
	\EndFor
	\For{$n = 1$ to $N$}
		\State $\mathrm{ShuffledNormalResiduals}^{\mathrm{cal}}(\tau) \leftarrow$ Resample $|T|$ residuals from $\mathrm{NormalResiduals}^{\mathrm{cal}}(\tau)$\label{Alglin:StartBootTheta}
		\State $\mathrm{ShuffledResidual}^{\mathrm{cal}}(\tau) \leftarrow C^\mathrm{fit}(\tau) \times \mathrm{ShuffledNormalResiduals}^{\mathrm{cal}}(\tau) $
		\State $C^{n}(\tau) = C^\mathrm{fit}(\tau) + \mathrm{ShuffledResidual}(\tau)$\label{Alglin:EndBootTheta}
		\State $\gamma^{n} \leftarrow$ \textsc{Calibration}$(k, \Omega, f_1,f_2, \ell, T,C^\mathrm{b},\vartheta)$
	\For{$(i,i',j,j')\in\{1,2\}\times\{1,\dots,m\}\times\{1,2\}\times\{1,\dots,m\},\,i\ne i', j\ne j'$}
			\State $\alpha^{n}_{ij}\leftarrow${\textsc{Mean}}$\sqrt{{N}_{11b_1}{N}_{ijb_j}/{N}_{1jb_j}{N}_{i1b_1}}$ from $B$ values each of $b_1,b_j$ drawn with replacement from $\{1,\dots,B\}$ \label{Alglin:Sim1Photon}
			\State $\mathrm{ShuffledNormalResiduals}_{ii'jj'}(\tau) \leftarrow |T|$ entries in $\mathrm{NormalResiduals}_{ii'jj'}(\tau)$
			\State $\mathrm{ShuffledResidual}_{ii'jj'}(\tau) \leftarrow C_{ii'jj'}^\mathrm{fit}(\tau) \times \mathrm{ShuffledNormalResiduals}_{ii'jj'}(\tau) $
			\State $C_{ii'jj'}^{n}(\tau) = C_{ii'jj'}^\mathrm{fit}(\tau) + \mathrm{ShuffledResidual}_{ii'jj'}(\tau)$	
		\EndFor
	\State $\{\theta_{ij}^{n}\}$ = {\textsc{ArgumentCalc}}($k, \Omega,f_1,f_2,\ell, T,C^{n}_{ii'jj'}(\tau),\alpha_{ij},\gamma$)
	\State $\{W^{\prime b}_{ij}\} =\, ${\textsc{MaxLikelyUnitary}}$\left(\left\{\alpha_{ij}^{n}\},\{\theta_{ij}^{n}\right\}\right)$
	\EndFor
	\For{$(i,j)$ in $\{1, \dots, m\}\times\{1, \dots, m\}$}
		\State $\sigma(\operatorname{re}W_{ij}) = \text{std.~dev.}\left(\{\operatorname{re}W_{ij}\}\right);\,\sigma(\operatorname{im}W_{ij}) = \text{std.~dev.}\left(\{\operatorname{im}W_{ij}\}\right)$
	\EndFor
	\EndProcedure
	
	\caption{\textsc{Bootstrap}: Estimate error bars on $\{W_{ij}\}$.}
	\label{Alg:Bootstrapping}
	\end{algorithmic}
\end{algorithm}
%------------------------------%

%=================%
\section{Scattershot characterization for reduction in experimental time}
\label{Sec:Scattershot}
%=================%
In this section, I present a scattershot-based characterization approach to effect a reduction in the characterization time~\cite{Lund2014,Bentivegna2015}.
Our scattershot approach reduces the time required to characterize an $m$-mode interferometer from $\BigO{m^{4}}$ to $\BigO{m^{2}}$ with constant error in the interferometer-matrix entries.

The straightforward approach of characterization involves coupling and decoupling light sources successively for each one- and two-photon measurement.
In contrast, the scattershot characterization relies on coupling heralded nondeterministic single-photon sources to each of the input ports of the interferometer and detectors to each of the output ports.
Controllable time delays are introduced at two input ports, which are labelled as the first and second ports.
All sources and detectors are switched on and the controllable time-delay values are changed first for the first port and then for the second port.

Single-photon data are collected by selecting the events in which exactly one of the heralding detectors and exactly one of the output detectors register a photon simultaneously.
Two-photon coincidence events at the outputs are counted when two heralding detectors register photons.
The controllable time delays introduced at the first and second input ports ensure that each of the $2(m-1)^2$ coincidence measurements is performed.
Note that our characterization procedure (Algorithms~\ref{Alg:Coincidence}--\ref{Alg:Bootstrapping}) yields accurate estimates of interferometer parameters even when photon sources with different spectral functions are used.
In summary, the required characterization data are collected by selectively recording one- and two-photon events.
The setup for the scattershot characterization of an interferometer is depicted in Figure~\ref{Figure:Scattershot}.

\begin{figure}[h]
\begin{centering}
\includegraphics[width=0.98\textwidth]{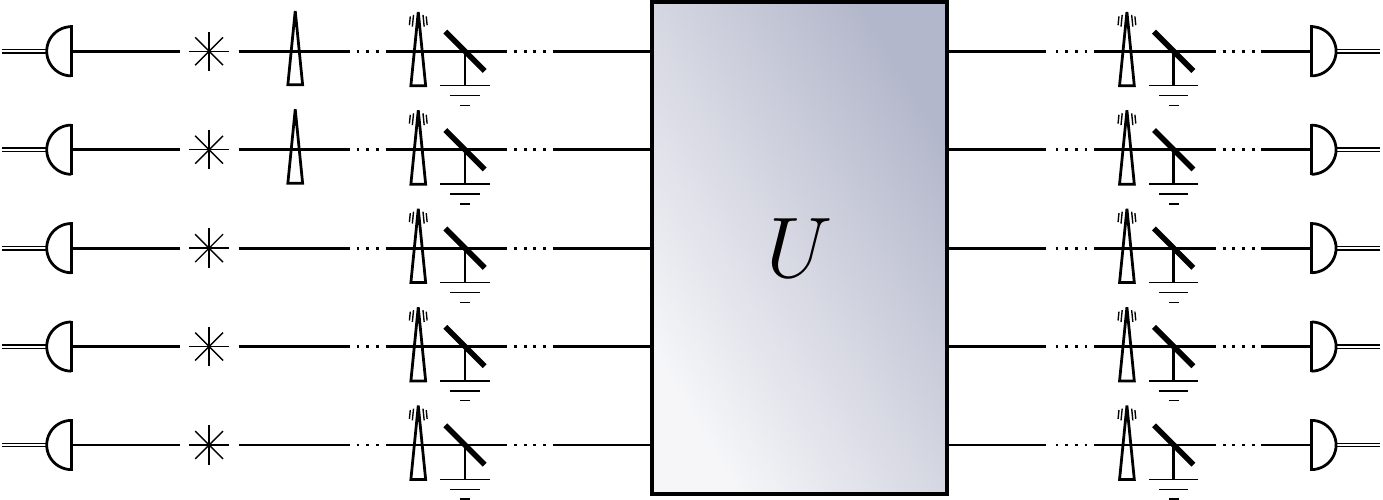}
\caption{Schematic diagram of the procedure for scattershot characterization of a five-mode interferometer $U$.
Heralded single-photon sources are coupled to the inputs of the interferometer and controllable time delays are introduced at the first two ports.
All sources and detectors are switched on and the controllable time delay values are changed for the first port and then for the second port.
The required characterization data are collected by selectively recording one- and two-photon events.
} % data are are. are not is.
\label{Figure:Scattershot}
\end{centering}
\end{figure}

Now I quantify the experimental time required in the characterization of a linear optical interferometer.
Our characterization procedure requires $Bm^2$ single-photon counting measurements and~$2(m-1)^2$ coincidence-counting measurements to characterize an $m$-mode interferometer.
We estimate the time required for each of these measurements such that random errors in the $\{\alpha_{ij}\},\{\theta_{ij}\}$ estimates remain unchanged with increasing $m$.
To ensure constant error in the $\{\alpha_{ij}\},\{\theta_{ij}\}$ estimates, we require that the number of one- and two-photon detection counts remain unchanged with increasing $m$.
The probability of photon detection at the output decreases with increasing $m$ because of the concomitant decrease in the transmission amplitudes $\{\alpha_{ij}\}$.

The amplitudes $\{\alpha_{ij}\}$ drop as $\BigO{1/\sqrt{m}}$ because of the unitarity of $U$~\cite{Dhand2014}.
Hence, one- and two-photon transmission probabilities~(\ref{Eq:PSinglePhotons},\ref{Eq:CoincidenceRate}) decrease as $1/m$ and $1/m^{2}$, respectively.
More photons need to be incident at the interferometer input ports to offset this decrease in transmission probabilities.
Therefore, maintaining a constant standard deviation in the $\{\alpha_{ij}\}$ and~$\{\theta_{ij}\}$ measurements requires $\BigO{m}$ and~$\BigO{m^{2}}$ scaling respectively in the number of incident photons, which amounts to an overall $\BigO{m^4}$ scaling in the experimental time requirement.
Scattershot characterization allows $(m-1)^2$ different sets of the one- and two-photon data to be collected in parallel thereby reducing the time required to characterize the interferometer by a factor of $(m-1)^2$.
The overall time required for the characterization decreases from $\BigO{m^{4}}$ to $\BigO{m^{2}}$ if the scattershot approach is employed.

Our analysis of scattershot characterization assumes that the coupling losses are small and that weak single-photon sources are used, i.e., that the probability of multi-photon emissions from the heralded sources is small as compared to single-photon emission probabilities.
These assumptions are expected to hold for on-chip implementations of linear optics that have integrated single-photon sources and detectors.

Light sources used at each input port in our scattershot-based characterization procedure differ spectrally in generally.
Our characterization procedure is accurate despite this difference because we measure source-field spectra and using these data in the curve-fitting procedure.

We have developed the scattershot approach which has advantages and disadvantages but on balance is a superior experimental approach to consecutive measurement.
The advantage is that the time requirement for characterization if reduced by a factor that scales as~$\BigO{m^{2}}$.
The disadvantage is the overhead of requiring one source at each input port and one detector at each output port.
The disadvantage is not daunting because these requirements are commensurate with other active investigations of QIP such as LOQC and scattershot BosonSampling.
In fact, state of the art implementations~\cite{Bentivegna2015} meet our increased requirements for scattershot characterization.
\section{Removal of instability in $\operatorname{sgn}\theta_{ij}$ estimation}
\label{Sec:Instability}
%------------------------------%
In this section, I describe an instability in our characterization procedure, which can yield large error in the $\{W_{ij}\}$ output for small error in the experimental data $C^{\mathrm{exp}}_{ii'jj'}(\tau)$ in case of certain interferometers $W$.
We present a strategy to circumvent this instability by means of collecting and processing additional experimental data.

\begin{figure}[h]
\centering
\subfloat[]{\includegraphics[width=0.51\textwidth]{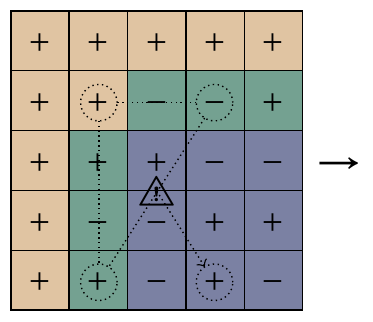}}
\subfloat[]{\includegraphics[width=0.43\textwidth]{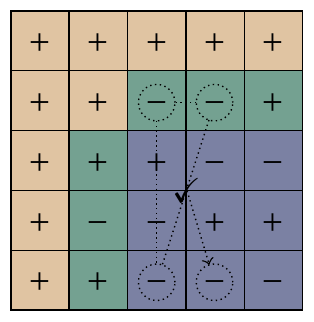}}
\caption{An illustration of the instability in the $\theta_{ij}$ characterization procedure for an interferometer with $m=5$ modes.
(a)~If the value of $|\theta_{22}-\theta_{52}-\theta_{24}|$ is close to $0$ or $\pi$, then small error in $C_{2524}^{\mathrm{exp}}(\tau)$ can lead to an error in the estimation of $\operatorname{sgn}(\theta_{54})$.
(b)~The instability in the $\theta_{54}$ can be removed by collecting two-photon coincidence data for output ports $2,5$ and input ports $3,4$, and using the values of $\theta_{23},\theta_{53},\theta_{24}$ instead of $\theta_{22},\theta_{52},\theta_{24}$ values.}
\label{Figure:Syndrome}
\end{figure}

The instability in the characterization procedure arises because of an instability in estimation of $\{\operatorname{sgn}\theta_{ij}\}$ (Algorithm~\ref{Alg:SignCalc}).
Small error in the measured coincidence counts can lead to the wrong inference of $\operatorname{sgn}\theta_{ij}$, which can lead to a large error $\|W-U\|$ in the characterized matrix $W$.
Recall that Algorithm~\ref{Alg:SignCalc} uses the identity $\operatorname{sgn}\theta_{ij} = \operatorname{sgn}\left(\big|\beta_{ii'jj'}-\beta_{ii'jj'}^{+}\big|-\big|\beta_{ii'jj'}-\beta_{ii'jj'}^{-}\big|\right)$ (\ref{Eq:SignEquation}) to determine the sign of the arguments, where $\beta^{\pm}_{ii'jj'}\defeq |\theta_{i'j'}-\theta_{ij'}-\theta_{i'j}\pm|\theta_{ij}||$ and the values of $\beta_{ii'jj'}, \theta_{i'j'},\theta_{i'j},\theta_{ij'},\left|\theta_{ij}\right|$ are estimated by curve fitting.
%Thus, the procedure checks if the estimated $\beta_{ii'j'}$ value is closer to $\beta_{ii'jj'}^{+}$, in which case it infers that $\operatorname{sgn}\theta_{ij}$ is positive, or to $\beta_{ii'jj'}^{-}$, for which $\operatorname{sgn}\theta_{ij}$ is negative.

Random and systematic error in measured coincidence counts can lead to estimate of variables $\beta_{ii'jj'}, \theta_{i'j'},\theta_{i'j},\theta_{ij'},\left|\theta_{ij}\right|$ differing from their actual values.
The estimation of $\operatorname{sgn}\theta_{ij}$ is unstable if the $\theta_{i'j'}-\theta_{i'j}-\theta_{ij'}$ term (\ref{Eq:SignEquation}) is close to $0$ or $\pi$ because, in this case, a small error in the $\beta_{ii'jj'}$ estimate can lead to an incorrect $\operatorname{sgn}\theta_{ij}$ estimate.
In other words, the sign estimates are unstable if the values of
\begin{equation}
\theta_{ii'jj'}^{\mathrm{ref}} = \mathrm{min}\left[\theta_{i'j'} - \theta_{ij'} - \theta_{i'j},\pi-(\theta_{i'j'} - \theta_{ij'} - \theta_{i'j})\right],
\end{equation}
are small compared to the error in our $\beta_{ii'jj'}, \theta_{i'j'},\theta_{i'j},\theta_{ij'},\left|\theta_{ij}\right|$ estimates.

We mitigate the sign-inference instability by making two modifications to our characterization procedure; the first modification removes instability from the sign-inference of the second row and second column elements whereas the second modification prevents incorrect inference of the remaining signs.
The inference of $\{\operatorname{sgn}\theta_{i2}\},\{\operatorname{sgn}\theta_{2j}\}$ (Lines~\ref{Alglin:StartSecondRowColumn}--\ref{Alglin:EndSecondRowColumn}, Figures~\ref{Figure:Ordering}b,~\ref{Figure:Ordering}c) is unstable if
\begin{equation}
\theta_{i2j2}^{\mathrm{ref}} =\mathrm{min}(\theta_{22},\pi-\theta_{22})
\end{equation}
is small as compared to the error in the $\beta_{i2j2}, \theta_{22},\theta_{2j},\theta_{i2},\left|\theta_{ij}\right|$ estimates.
Hence, we relabel the interferometer ports such that $\theta_{22}$ is as far away from $0$ and $\pi$ as possible.
Specifically, after the amplitudes of the phases have been estimated (Line~\ref{Alglin:EndAmplitudes} of Algorithm~\ref{Alg:PhaseCalcNChannel}), we choose $i,j$ for which $|\theta_{ij}-\pi/2|$ is minimum, and we swap the labels of input ports $2,j$ and output ports $2,i$.
We measure two-photon coincidence counts based on this new labelling and process it using Algorithm~\ref{Alg:PhaseCalcNChannel}.
The instability in the procedure for estimation of the $\{\theta_{i2}\},\{\theta_{2j}\}$ signs is removed as a result of the relabelling.

The second modification is aimed at removing the instability in the remaining signs.
The procedure estimates the remaining signs by using $\{C^{\mathrm{exp}}_{ii'jj'}(\tau)\}$ values for $i'=j'=2$.
The estimation of $\theta_{ij}$ is unstable if $\theta_{i2j2}^{\mathrm{ref}}$ is small as compared to the error in the $\beta_{i2j2}, \theta_{22},\theta_{2j},\theta_{i2},\left|\theta_{ij}\right|$ estimates.
We make a heuristic choice of a threshold angle $\theta^{\mathrm{T}}$ that accounts for the error in these variables, and we reject any $\operatorname{sgn}\theta_{ij}$ inferred using $ \theta_{i2j2}^{\mathrm{ref}} \le \theta^{\mathrm{T}}$.
 Additional two-photon coincidence counting is performed and employed to estimate these values of $\theta_{ij}$, as detailed in the following lines that can be added to the algorithm to remove the instability
\begin{algorithmic}[1]
 \setcounter{ALG@line}{0}
 \algrenewcommand{\alglinenumber}[1]{\footnotesize{\ref{Alglin:EndSecondRowColumn} + #1:}}
 \For{$(i,j)$ in $\{3, \dots, m\}\times\{3, \dots, m\}$}
 \algrenewcommand{\alglinenumber}[1]{\footnotesize{\hspace{2em} #1:}}
 \If{$\theta^r_{i2j2} < \theta_T$}
 \State Choose $i'\ne 1,i$ and $j'\ne 1,j$ such that $|\theta_{i'j'}-\theta_{ij'}-\theta_{j'j}|$ is closest to $\pi/2$.	
 	\State $C^{\mathrm{exp}}_{ii'jj'}(\tau)\leftarrow$ Coincidence counts for input ports $j,j'$ and output ports $i,i'$.
 \State $A \leftarrow \{0,0,0\}, \Phi \leftarrow \{0,0,0\}$
		\State $\beta_{ii'jj'} \leftarrow \textsc{Argument2Port}(C^{\mathrm{exp}}_{ii'jj'}(\tau),\Omega, f_1, f_2,T,A,\Phi,\gamma)$
		\State $\theta_{ij} \leftarrow \theta_{ij}${\textsc{SignCalc}}$(\beta_{ii'jj'},\theta_{i'j'},\theta_{ij'},\theta_{i'j},|\theta_{ij}|)$
 \EndIf
	\EndFor
\end{algorithmic}
Figure~\ref{Figure:Syndrome} illustrates the rejection of those $i,j$ choices for which $\theta_{i2j2}^{\mathrm{ref}} \le \theta^{\mathrm{T}}$ and the use of $ C^{\mathrm{exp}}_{ii'jj'}(\tau), j'\ne 2$ counts to obtain a correct estimate of $\operatorname{sgn}\theta_{ij}.$
We thus remove the instability in the estimation of $\{\theta_{ij}\}$ and in the estimation of the representative matrix $W$.

%------------------------------%
\section{Summary of procedure and discussions}
\label{Sec:CharDiscussions}
%------------------------------%
In this section, I summarize our characterization procedure for a less formally oriented audience.
We describe the processing of the collected experimental data by the various algorithms presented in Section~\ref{Sec:Procedure}.
We compare our procedure with the existing procedure for the characterization of linear optics using one- and two-photon interference~\cite{Laing2012}.
We provide numerical evidence that our characterization procedure promises enhanced accuracy and precision even in the presence of shot noise and mode mismatch.

The experimental data required by our procedure to characterize an $m$-mode interferometer includes the following one- and two-photon measurements.
The number $N_{ijb_{j}}$~(\ref{Eq:Nijbdef}) of single-photon detection events is counted at the $j$-th output port when single photons are incident at the $i$-th input port.
This single-photon counting is repeated $B$ times for each of the input ports and output ports, where $B$ is chosen such that the cumulants of the set $\{N_{ijb_{j}}:b_{j}\in\{1,\dots,B\}\}$ converge.
The single-photon counts $\{N_{ijb_{j}}\}$ are received by Algorithm~\ref{Alg:AmplitudeEstimation}, which returns the $\{\alpha_{ij}\}$~(\ref{Eq:RMS0}) estimates using Equation~(\ref{Eq:AlphaMeasure}).

The spectral function $f_{j}(\omega)$~(\ref{Eq:Single}) of the light incident at each input port $j$ is measured.
This function is used by Algorithm~\ref{Alg:Coincidence} to calculate the expected two-photon coincidence curves using Equation~(\ref{Eq:CoincidenceRate}).
Fitting experimental data to these coincidence curves yields an accurate estimate of the mode-matching parameter during calibration and the arguments $\{\theta_{ij}\}$ in the argument-estimation procedure.
Thus, the spectral function $f_{j}(\omega)$ serves as an input to the algorithms for the estimation of the mode-matching parameter
 and of the arguments $\{\theta_{ij}\}$ (Algorithms~\ref{Alg:Calibration}--\ref{Alg:PhaseCalcNChannel}).

The mode-matching parameter~$\gamma$ is estimated by performing coincidence measurement on a beam splitter that is separate from the interferometer but is constructed using the same material.
First, we use single-photon data to estimate the reflectivity $\cos\vartheta$ of the beam splitter according to Equation~(\ref{Eq:FindVartheta}).
Imperfect mode-matching changes the shape of the coincidence curve, and we find~$\gamma$ by comparing the shapes of (i)~the curve expected for reflectivity $\cos\vartheta$ and (ii)~the curve obtained experimentally.
The estimated beam splitter reflectivity, the measured spectra and the coincidence counts are received as inputs by Algorithm~\ref{Alg:Calibration}, which returns an estimate of~$\gamma$.

Algorithm~\ref{Alg:PhaseCalcNChannel} uses two-photon coincidence counts to estimate the arguments $\{\theta_{ij}\}$.
Coincidence counts are measured for the input ports $j,j'$ and output ports $i,i'$ for the $2(m-1)^{2}$ sets
\begin{equation}
(i,i',j,j') \in \{1,2\}\times\{1,\dots,m\}\times \{1,2\}\times\{1,\dots,m\}, \quad i\ne i',\, j\ne j'
\label{Eq:Choice}
\end{equation}
 of input and output ports.
In other words, coincidence counts are measured for different choices of two input ports and two output ports, such that each of the choices includes (i)~either the first or the second input ports and (ii)~either the first or the second output port.
Algorithm~\ref{Alg:PhaseCalcNChannel} receives as input the measured spectra, the $\{\alpha_{ij}\}$ values estimated by Algorithm~\ref{Alg:AmplitudeEstimation}, the~$\gamma$ value estimated by Algorithm~\ref{Alg:Calibration} and the two-photon coincidence data for the choice~(\ref{Eq:Choice}) of input ports.
The algorithm returns the $\{\theta_{ij}\}$ estimates.
The computed estimates of $\{\alpha_{ij}\}$ and of $\{\theta_{ij}\}$ yield the representative unitary matrix $W$ (\ref{Eq:W}) that has maximum likelihood of describing the characterized interferometer (Algorithm~\ref{Alg:MaxLikelyUnitary}).
This completes a summary of our procedure for characterization of the interferometer.

Algorithm~\ref{Alg:Bootstrapping} employs bootstrapping to find the error bars on the elements $\{W_{ij}\}$ of the characterized unitary matrix.
The bootstrapping procedure uses the experimental data that is received by Algorithms~\ref{Alg:Coincidence}--\ref{Alg:MaxLikelyUnitary} and repeatedly simulates experiments by resampling from the experimental data.
The number $N$ of repetitions is chosen such that the pdf's of the $\{W_{ij}\}$ elements over many rounds of simulation converge.
The error bars on the $\{W_{ij}\}$ elements are computed based on the estimated pdf's of the elements. Our procedure thus enables the estimation of meaningful error bars on the characterized unitary matrix.

Bootstrapping is employed to test the goodness of fit between the experimental curve and expected curves~\cite{Kojadinovic2012}.
Experiments~\cite{Metcalf2013,Tillmann2015} can employ bootstrapping instead of the incorrect $\chi^{2}$-confidence measure to test if the data are consistent with quantum predictions or with the classical theory.

Finally, we recommend a scattershot approach for reducing the experimental time required to characterize interferometers.
The approach involves coupling heralded nondeterministic single-photon sources at each of the input ports and single-photon detectors at each of the output ports.
All the sources and the detectors are switched on in parallel.
Single-photon counts are recorded selectively as two-photon coincidences between the heralding detectors and the output detectors, whereas two-photon events are recorded when two heralding detectors and two output detectors record photons.
Controllable time delays are introduced at the first and second input ports so coincidences between each of the $2(m-1)^2$ choices~(\ref{Eq:Choice}) of input and output ports are recorded.
The scattershot approach reduces the experimental time required to characterize an $m$-mode interferometer from $\BigO{m^{4}}$ to $\BigO{m^{2}}$.

Now we compare and contrast our procedure with the Laing-O'Brien procedure~\cite{Laing2012}.
Our procedure is inspired by the Laing-O'Brien procedure in that it employs
(i)~a `real-bordered' parameterization~(\ref{Eq:RMS0}) of the representative matrix and modelling of linear losses at the interferometer ports,
(ii)~a ratio of single-photon data to estimate the complex amplitudes of the matrix elements and
(iii)~an iterative procedure that uses two-photon data to estimate the amplitudes of the complex arguments and to estimate the signs of the complex arguments.

Our procedure differs from the Laing-O'Brien~\cite{Laing2012} procedure in that we use averaged value~(\ref{Eq:AlphaMeasure}) of the ratio of single-photon detections over many runs rather than the ratio of averaged values.
This difference ensures accuracy of our procedure even under fluctuation in the number of incoming photons.
Such fluctuations might arise from fluctuations in pump strength of the single photon source or in the strength of coupling between photon source and interferometer.

Another advance in our method is the curve-fitting procedure for estimating complex arguments of interferometer matrices.
The Laing-O'Brien procedure requires coincidence-curve visibilities to estimate complex arguments~$\alpha_{ij}$.
Whereas the Laing-O'Brien procedure recommends coincidence probabilities be measured at zero time delay and also at time delays large as compared to the temporal spread of the wave-packet, in practice, current implementations determine the visibilities by fitting experimental data to Gaussian curves~\cite{Hong1987,Peruzzo2011,Crespi2013,Spring2013,Tillmann2013,Carolan2014}.
These implementations are flawed because source spectra differ from Gaussian in general.
Our procedure is accurate because the data are fit to curves computed from spectral functions, rather than fitting to Gaussians.

We introduce the calibration subroutine, which relies on the estimation of the mode mismatch in the source field.
Spatial and polarization mode mismatch is not an issue of major concern in waveguide-based interferometers, which typically operate in the single-photon regime.
In these interferometers, the calibration step of our procedure can be neglected without decreasing accuracy.
The mode-mismatch parameter~$\gamma$, which is an input of the curve-fitting procedure, is set to unity.

In the context of bulk-optics, our calibration step ensures accuracy and precision if (i)~$\gamma$ is identified as the maximum-possible source overlap in the spatial and polarization degrees of freedom and (ii)~the experimentalist adjusts the setup to maximize coincidence visibility for the calibrating beam splitter and for each choice of interferometer inputs ports.
Such an adjustment will ensure that the source overlap acquires its maximum-possible value~$\gamma$ in each of the coincidence-curve measurements.
This maximum value is a property of the sources used and is independent of source alignment and focus so is expected to remain unchanged between different confidence measurements.

Other advances made in our characterization procedure over existing procedures include (i)~a maximum likelihood estimation approach to determine the unitary matrix that best fits the data (ii)~a bootstrapping based procedure to obtain meaningful estimates of precision and (iii)~a scattershot-based procedure to improve the experimental requirements of characterization.

%------------------------------%
\section{Conclusion}
\label{Sec:CharConclusion}
%------------------------------%
In conclusion, we devise a one- and two-photon interference procedure to characterize any linear optical interferometer accurately and precisely.
Our procedure provides an algorithmic method for recording experimental data and computing the representative transformation matrix with known error.

The procedure accounts for systematic errors due to spatiotemporal mode mismatch in the source field by means of a calibration step and corrects these errors using an estimate of the mode-matching parameter.
We measure the spectral function of the incoming light to achieve good fitting between the expected and measured coincidence counts, thereby achieving high precision in characterized matrix elements.
We introduce a scattershot approach to effect a reduction in the experimental requirement for the characterization of interferometer.
The error bars on the characterized parameters are estimated using bootstrapping statistics.

Bootstrapping computes accurate error bars even when the form of experimental error is unknown and is, thus, advantageous over the Monte Carlo method.
Hence, our bootstrapping-based procedure for estimating error bars can replace the Monte Carlo method used in existing linear-optics characterization procedures.
We thus open the possibility of applying bootstrapping statistics for the accurate estimation of error bars in photonic state and process tomography.

%=================%
\chapter{Numerical and experimental verification of accurate and precise characterization of linear optics}
\label{Chap:Verification}
%=================%
This chapter presents experimental and numerical evidence of the accuracy and precision of our characterization procedure that was detailed in the preceding chapter.
Section~\ref{Sec:NumericalVerification} comprises numerical simulations to compare our procedure with existing procedures.
Section~\ref{Sec:Experimental} presents experimental data comparing the accuracy and precision of our procedure with existing procedures.

The majority of the material in Section~\ref{Sec:NumericalVerification} is taken from my article~\cite{Dhand2015a} that I co-authored with Abdullah Khalid, He Lu and Barry C.~Sanders.
Those parts that are reproduced verbatim from our journal paper are listed in ``Thesis content previously published''.
A manuscript reporting the experimental verification presented in Section~\ref{Sec:Experimental} is in preparation.

%------------------------------%
\section{Numerical verification}
\label{Sec:NumericalVerification}
%------------------------------%
This section comprises the methods and results of the numerical simulations performed to verify the accuracy and precision of our characterization procedure as compared to existing procedures.
We simulated one- and two-photon data using experimentally measured spectra and performed characterization many times using different procedures.
Our procedure showed substantial advantage over alternative characterization procedures.
I detail the methods employed in our numerical simulations in Section~\ref{Sec:NumericalVerificationProcedure} before presenting the results of the simulation in Section~\ref{Sec:NumericalVerificationResults}.

%------------------------------%
\subsection{Numerical verification: Methods}
\label{Sec:NumericalVerificationProcedure}
%------------------------------%
Here I detail the numerical methods employed to compare the accuracy of our procedure with existing procedures.
We repeatedly simulated experimental data for randomly sampled interferometer matrices, simulated characterization on the sampled interferometer and computed the error between the expected and characterized interferometer matrices.

Two aspects in which our procedure differs from the Laing-O'Brien procedure are that we fit to curves calculated from experimental spectra rather than to Gaussian, and that we perform calibration to estimate the mode mismatch rather than assuming perfect mode matching.
In these simulations we test the efficacy of our characterization procedure as compared to characterizing without calibration and characterizing by fitting to Gaussian.
Specifically, I repeated following step-by-step procedure 1000 times.
\begin{enumerate}
\item
Generate a random unitary matrix $W$ sampled uniformly from the Haar measure using a \gls{QRD} based procedure~\cite{Mezzadri2007}.
\item
Simulate single-photon detection data by computing single photon transmission probabilities~(\ref{Eq:PSinglePhotons}) and sampling under Poissonian shot noise assumption.
\item
Simulate two-photon coincidence data by computing two-photon coincidence probabilities~(\ref{Eq:CoincidenceRate}) and assuming Poissonian shot noise.
The simulated two-photon coincidence data comprises the coincidence measurements performed for calibration (Algorithm~\ref{Alg:Calibration}) and for argument estimation (Algorithm~\ref{Alg:PhaseCalcNChannel}).
Use experimentally measured spectral functions to simulate coincidence rates for different values of time-delay.
\item
Process simulated experimental data using Algorithms~\ref{Alg:Coincidence}--\ref{Alg:Bootstrapping} to obtain the characterized representative unitary matrix $\tilde{W}_\mathrm{Acc}$.
Use experimentally obtained spectra in the calibration and argument estimation algorithms.
\item
Process same data using Algorithms~\ref{Alg:Coincidence}--\ref{Alg:Bootstrapping} but without calibration, i.e., artificially setting mode-matching parameter~$\gamma = 1$, to obtain the characterized representative unitary matrix $\tilde{W}_{\mathrm{NoCal}}$.
Use experimentally obtained spectra in the calibration and argument estimation algorithms.
\item
Process same data using Algorithms~\ref{Alg:Coincidence}--\ref{Alg:Bootstrapping} but using Gaussian fitting function, i.e., using Gaussian spectra as inputs to the calibration and argument estimation algorithms, to obtain the characterized representative unitary matrix $\tilde{W}_{\mathrm{Gauss}}$.
\item
Find the error
\begin{equation}
\varepsilon_{i} = \operatorname{dist}(W,\tilde{W}_{i})
\label{Eq:CharError}
\end{equation}
in the three procedures that are labelled by index $i$ and are described in steps 4--6 above.
The operator $\operatorname{list}(\bullet,\bullet')$ refers to the trace distance between the two matrices $\bullet$ and $\bullet'$
\end{enumerate}
We compare (i)~the mean error~(\ref{Eq:CharError}) incurred from our characterization procedure and the Gaussian-fitting procedure for different values of Poissonian shot noise and (ii)~the mean error incurred by our procedure with and without calibration for different values of mode-matching parameter~$\gamma$.
The results of these comparisons are presented in the next section.

In summary, we used experimental spectral data to repeatedly generate one- and two-photon measurement data for interferometer matrices sampled randomly from the Haar measure.
Poissonian shot noise was simulated on the generated data.
The data were used as inputs to our characterization procedure and to existing procedures.
The error in the characterization was computed as the trace distance between the actual and the characterized interferometer transformation matrix for each set of simulated data.

%------------------------------%
\subsection{Numerical verification: Result}
\label{Sec:NumericalVerificationResults}
%------------------------------%
\begin{figure}[h]
\begin{centering}
\includegraphics[width=0.7\textwidth]{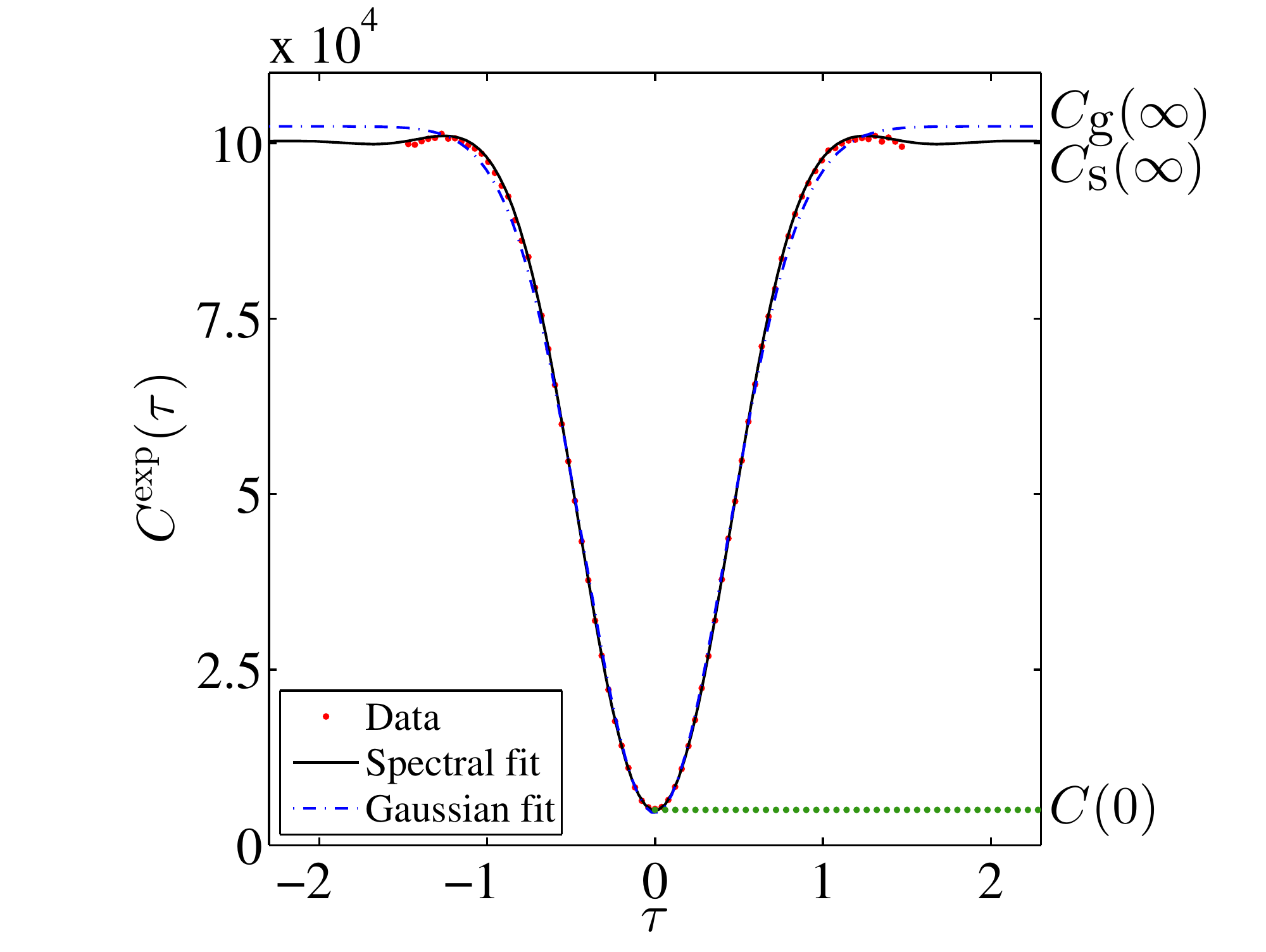}
\caption{
The fitting of coincidence data to curves obtained from spectra using~(\ref{Eq:CoincidenceRate}) and to Gaussian functions.
Coincidence counts are simulated using experimentally measured spectra.}
\label{Figure:CoincidenceCurveFit}
\end{centering}
\end{figure}

Here I present the results of the numerical simulation by comparing the accuracy of our procedure with the accuracy of existing procedures.
First, we compared the error in characterization using our curve-fitting procedure, which relies on experimentally measured spectra for calculating fitting curves, and the standard practice of fitting to Gaussians~\cite{Hong1987,Peruzzo2011,Crespi2013,Spring2013,Tillmann2013,Carolan2014}.
Figure~\ref{Figure:CoincidenceCurveFit} illustrates the distinction between fitting experimental coincidence counts to the coincidence function~(\ref{Eq:CoincidenceRate}) simulated using spectra and fitting to Gaussian functions.
Figure~\ref{Figure:GaussianVersusNongauss} demonstrates the increase in accuracy and precision of characterization by using the correct curve-fitting function.

\begin{figure}[h]
\begin{center}
\includegraphics[width=0.7\textwidth]{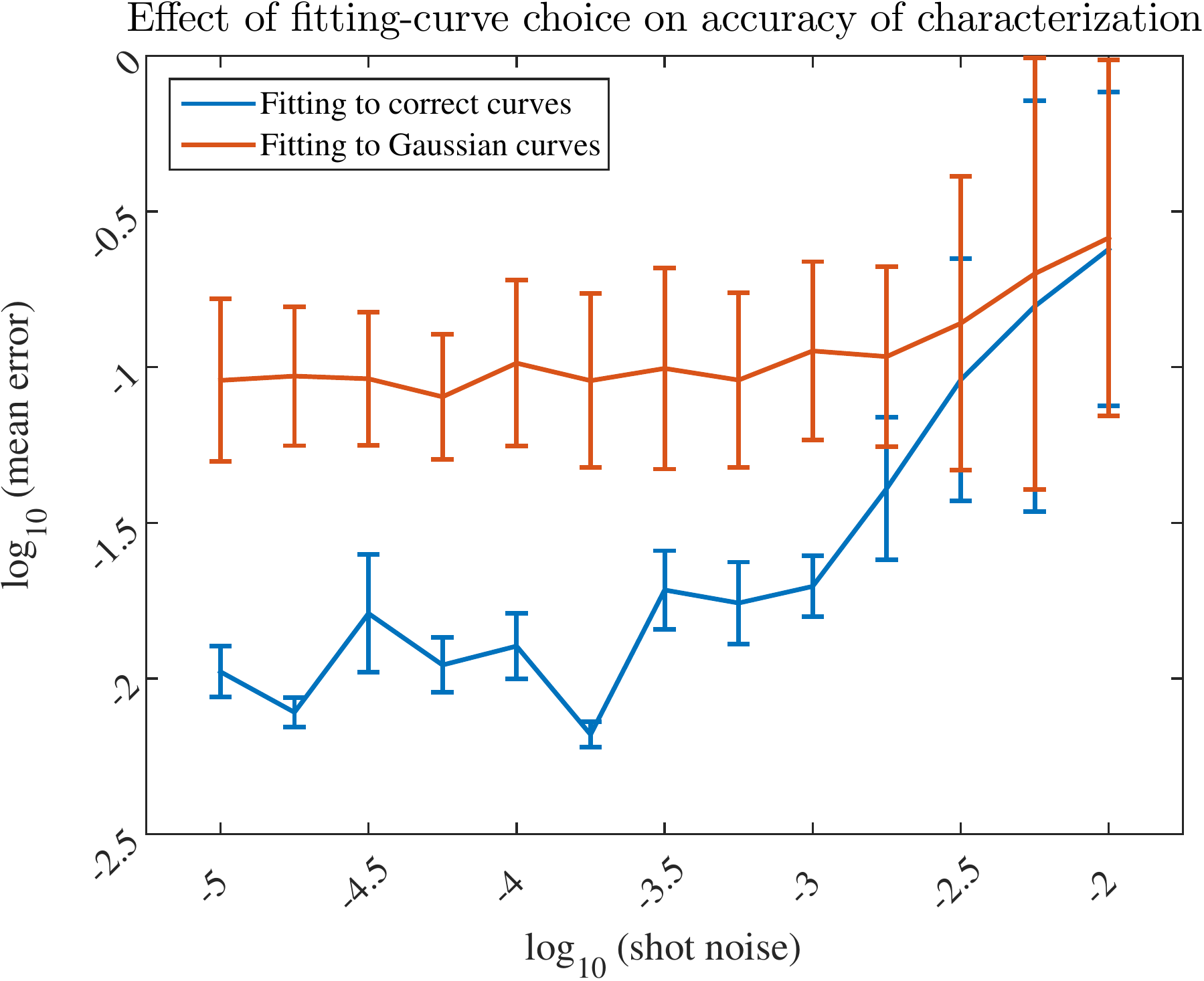}
\caption{A plot showing the effect of fitting-curve choice on the accuracy and precision of the characterization procedure. The two curves depict the mean error for the two different choices of fitting curves, where the error is the trace distance between the expected and the actual unitary transformations and the mean is over 1000 simulated characterization experiments.
One- and two-photon interference data was simulated for a five-channel interferometer using experimentally measured spectra and simulated Poissonian shot-noise.
Characterization was performed by fitting coincidence curves to Gaussians (red curve) and to correct curves according to our procedure (blue curve).
\textsc{matlab} code for the simulations depicted in this figure is available on GitHub.
}
\label{Figure:GaussianVersusNongauss}
\end{center}
\end{figure}

\begin{figure}[h]
\begin{centering}
\includegraphics[width=0.8\textwidth]{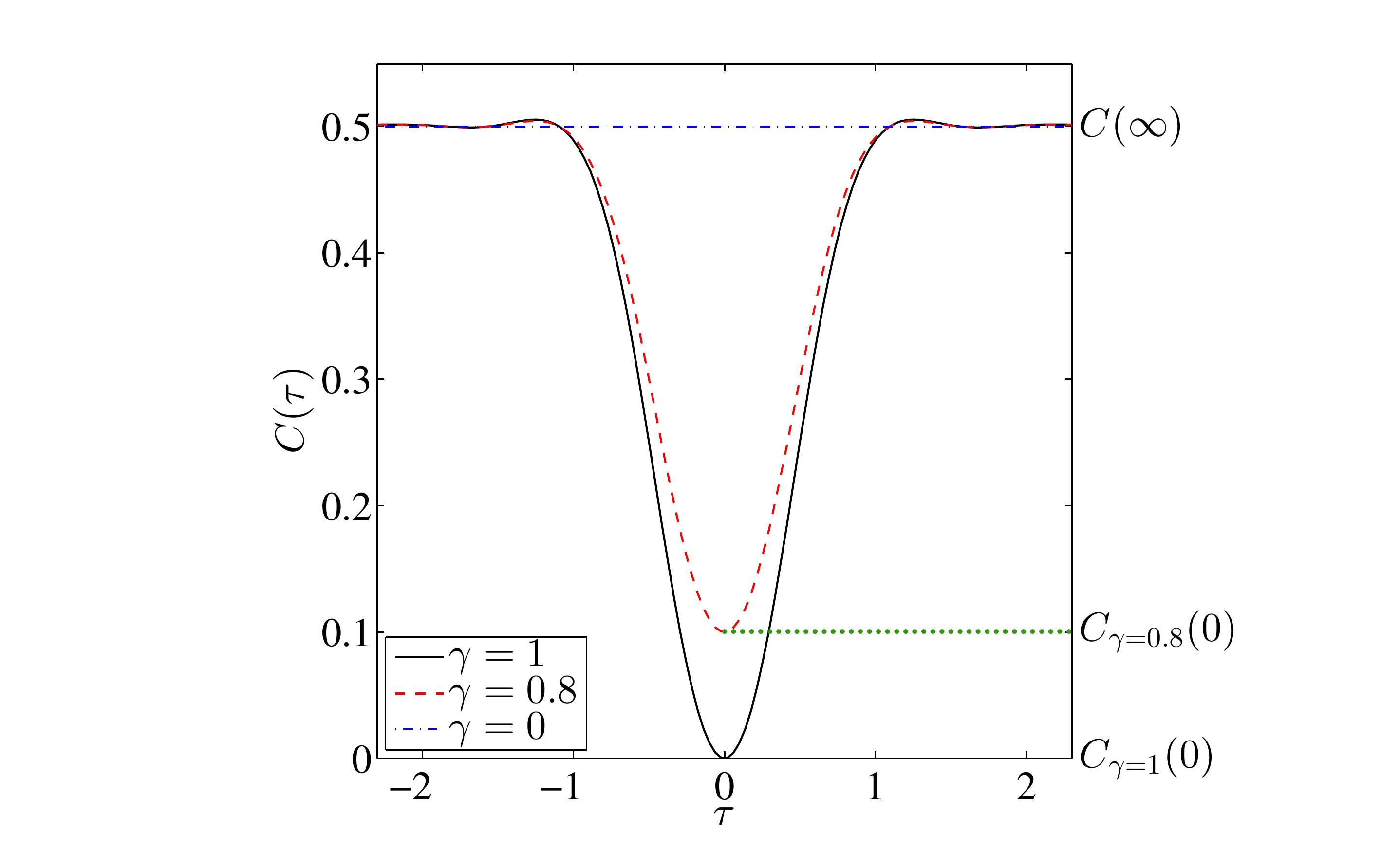}
\caption{
Plots of coincidence probability versus time delay for different values of~$\gamma$ for a lossless balanced beam splitter.
The time delay $\tau$ is in units inverse special width of incoming photons.
}
\label{Figure:ModeMismatch}
\end{centering}
\end{figure}

Next we demonstrate the increase in accuracy by employing the calibration subroutine.
Figure~\ref{Figure:ModeMismatch} depicts how imperfect mode matching, i.e., $\gamma<1$, alters the observed two-photon coincidence counts.
Our calibration procedure estimates and accounts for imperfect mode matching, which is assumed to be constant over the runtime of the characterization experiment.

We simulated the characterization experiment for different values of~$\gamma$ using~(i)~our calibration procedure and~(ii)~our calibration procedure without the calibration subroutine, i.e., by artificially setting~$\gamma = 1$.
Figure~\ref{Figure:IsCalibrationUseful} presents such this comparison for different values of $\gamma$.
Observe that even for almost perfect mode matching~($\gamma = 0.99$), our calibration procedure reduces characterization error by one order of magnitude.
This completes our numerical verification of the advantage offered by our characterization procedure.

\begin{figure}[h]
\begin{center}
\includegraphics[width=0.6\textwidth]{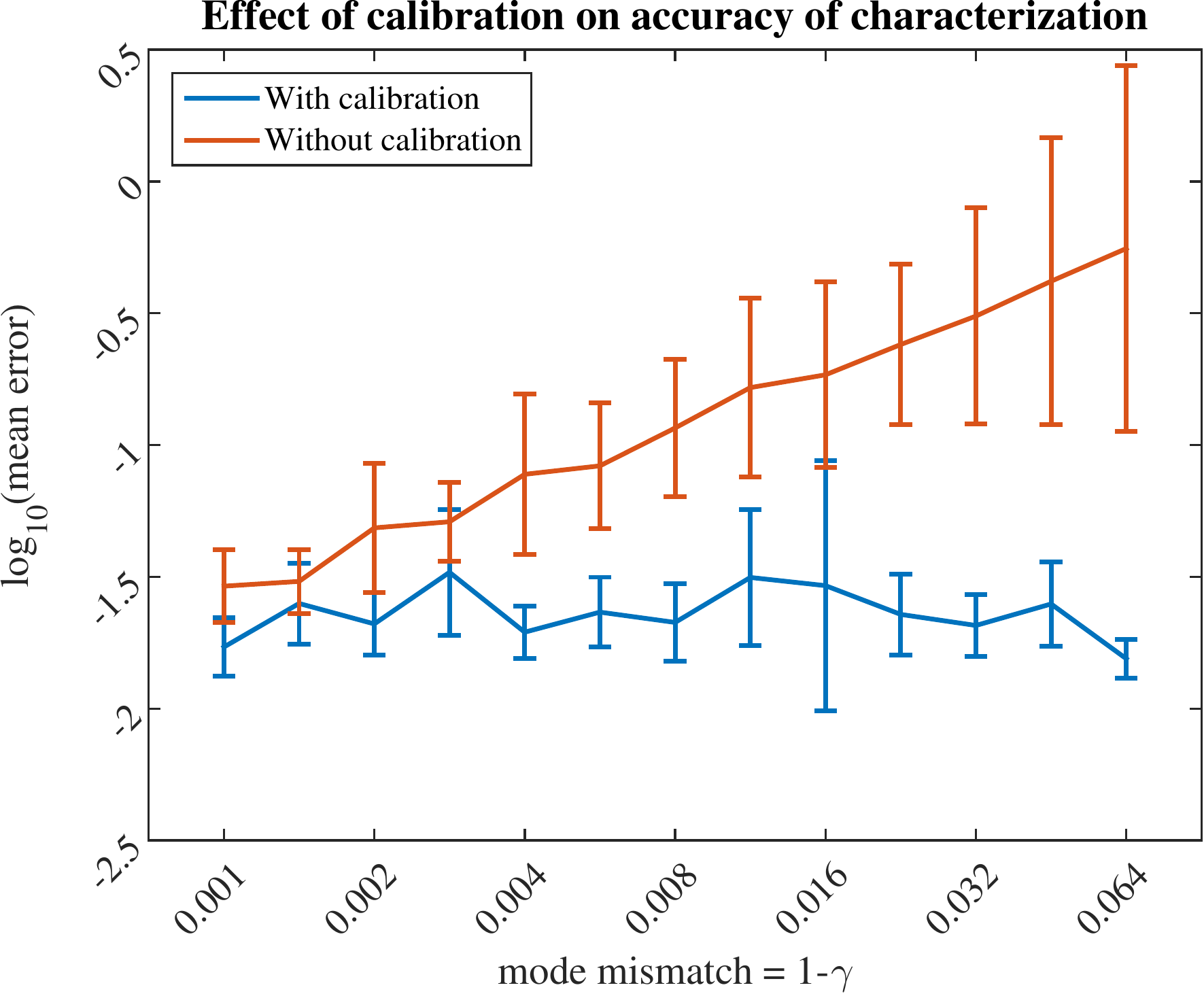}
\caption{A plot showing the effect of calibration on the accuracy and precision of the characterization procedure. The two curves depict the mean error for characterization with (blue curve) and without calibration (red curve), where the error is the trace distance between the expected and the actual unitary transformations and the mean is over 1000 simulated characterization experiment. The simulations comprised generating one- and two-photon interference data based on experimentally measured spectral functions and performing characterization by our procedure.
\textsc{matlab} code for the simulations depicted in this figure is available on GitHub.
}
\label{Figure:IsCalibrationUseful}
\end{center}
\end{figure}

%------------------------------%
\section{Experimental verification of beam splitter characterization}
\label{Sec:Experimental}
%------------------------------%
This section presents experimental results verifying the accuracy and precision of our characterization procedure.
The experiment was performed at the Shanghai branch of the University of Science and Technology of China, and a manuscript reporting these results is in preparation.

The experiments involved the characterization of a two-channel interferometer (beam splitter) using our procedure and existing procedures.
The collected data were used to test (i)~the assumption that the mode-matching parameter~$\gamma$ is constant over different characterization experiments (Section~\ref{Sec:GammaTest}) and (ii)~our claims~\cite{Dhand2015a} that our characterization procedure yields a more accurate characterization of linear optical interferometers (Section~\ref{Sec:ExpVerification}).

%------------------------------%
\subsection{Testing the constant~$\gamma$ assumption}
\label{Sec:GammaTest}
%------------------------------%
Our characterization procedure assumes that the mode-matching parameter is unchanged over different characterization experiments.
Here I present experimental evidence for the validity of this assumption.

We estimate $\gamma$ for four beam splitters using two-photon coincidence curves and reflectivity values ascertained from single-photon transmission data.
Our experimental procedure is as follows.
\begin{enumerate}
\item
Single-photon transmission data and two-photon coincidence data are collected for the four beam splitters labelled by index $i\in\{1,2,3,4\}$.
\item
The single-photon data are employed to estimate the reflectivity $\cos{\vartheta_{i}}$ of each of the beam splitters using Equations~(\ref{Eq:AlphaMeasure}) and~(\ref{Eq:FindVartheta}).
\item
The reflectivity estimate obtained in Step~2 and the measured two-photon coincidence data are used to estimate the mode-matching parameter $\gamma_{i}$ via curve-fitting. Appendix~\ref{Sec:CurveFitting} details the inputs, procedure and outputs of the curve-fitting method.

\item
The error bars $\sigma(\gamma_{i})$ on $\gamma_{i}$ are estimated using a bootstrapping-based procedure.
\end{enumerate}
Table~\ref{Table:Promise} presents the estimated values of $\gamma_{i}$ and $\sigma(\gamma_{i})$, and the ratio of these two estimates.

\begin{table}[h]
\begin{center}
\def\arraystretch{1.5}
\setlength{\tabcolsep}{20pt}
\begin{tabular}{|c|ccc|}
\hline
$i$ & $\gamma_{i}$ &$\sigma(\gamma_{i})$& $\sigma(\gamma_{i})/\gamma_{i}$ \\
\hline
$1$ & $0.960$ & $0.027$ & $0.0281$ \\
$2$ & $0.981$ & $0.007$ & $0.0071$ \\
$3$ & $0.956$ & $0.011$ & $0.0115$ \\
$4$ & $0.985$ & $0.019$ & $0.0193$ \\
\hline
\end{tabular}
\end{center}
\caption{On the promise of constant mode-matching parameter~$\gamma$.
The first column presents the beam splitter labels.
The estimates of $\gamma_{i}$ in the second column are computed using one- and two-photon data.
The error bars $\sigma(\gamma_{i})$ on $\gamma_{i}$ estimates are computed using bootstrapping.}
\label{Table:Promise}
\end{table}

The mode-matching parameters obtained for the four beam splitters agree with each other to within $95\%$ confidence or two standard deviations;
i.e.,
\begin{equation}
|\gamma_{i}-\gamma_{j}|< 2\sigma(\gamma_{i}-\gamma_{j})
\end{equation}
for beam splitter labels $i,j \in \{1,2,3,4\}$ and the standard error of the difference $\gamma_{i}-\gamma_{j}$ estimated according to
\begin{equation}
\sigma(\gamma_{i}-\gamma_{j}) \approx \sqrt{\sigma^{2}(\gamma_{i})+\sigma^{2}(\gamma_{j})}.
\end{equation}
Hence, our assumption of constant mode-mismatching is valid with $95\%$ confidence level.
In the next section, I provide evidence that our characterization procedure offers significant advantage over existing procedures.
%------------------------------%
\subsection{Comparison of beam-splitter reflectivity estimates}
\label{Sec:ExpVerification}
%------------------------------%
Here I compare the accuracy of the beam-splitter reflectivity estimates obtained using various procedures to process two-photon coincidence data.
As in the numerical verification (Section~\ref{Sec:NumericalVerification}), we consider the effect of calibration of source-light for mode mismatch and the effect of using Gaussian functions for fitting the measured coincidence curves.
\begin{table}[h]
\centering
\def\arraystretch{1.5}
\setlength{\tabcolsep}{22pt}
\begin{tabular}{|l|ccc|}
\hline
\hspace{-10pt}\textbf{Beam splitter} & $i = 2$ & $i = 3$ & $i = 4$\\
\hline
 \multicolumn{4}{|l|}{\textbf{\hspace{-10pt}Characterization using single-photon data:}}\\
Reflectivity $R^{\mathrm{sp}}_{i}$ & 0.4712 & 0.4741 & 0.3570 \\
Error bars $\operatorname{\sigma}$$\left(R^{\mathrm{sp}}_{i}\right)$ & 0.0028 & 0.0032 & 0.0033 \\
 \hline
 \multicolumn{4}{|l|}{\hspace{-10pt}\textbf{Two-photon characterization:} ($\gamma=0.960, \sigma(\gamma)=0.027$)}\\
 Reflectivity $R^{\mathrm{cal}}_{i}$& 0.4851 & 0.4632 & 0.3690 \\
 Error bars $\operatorname{\sigma}$$\left(R^{\mathrm{cal}}_{i}\right)$ & 0.0196 & 0.0317 & 0.0170 \\
 $\left.\left|R^{\mathrm{cal}}_{i}-R^{\mathrm{sp}}_{i}\right|\right/\operatorname{\sigma}$$\left(R^{\mathrm{cal}}_{i}-R^{\mathrm{sp}}_{i}\right)$& 0.7021 & 0.3421 & 0.6929 \\
\hline
\multicolumn{4}{|l|}{\hspace{-10pt}\textbf{Two-photon characterization without calibration:} ($\gamma=1, \sigma(\gamma)=0$)}\\

 Reflectivity $R^{\mathrm{nocal}}_{i}$& 0.4422 & 0.4228 & 0.3513 \\
Error bars $\operatorname{\sigma}$$\left(R^{\mathrm{nocal}}_{i}\right)$ & 0.0069 & 0.0106 & 0.0074 \\
$\left.\left|R^{\mathrm{nocal}}_{i}-R^{\mathrm{sp}}_{i}\right|\right/\operatorname{\sigma}$$\left(R^{\mathrm{nocal}}_{i}-R^{\mathrm{sp}}_{i}\right)$ & 3.8945 & 4.6331 & 0.7035 \\
 \hline
 \multicolumn{4}{|l|}{\hspace{-10pt}\textbf{Two-photon. No calibration. Gaussian fitting:} ($\gamma=1, \sigma(\gamma)=0$)}\\

 Reflectivity $R^{\mathrm{g}}_{i}$& 0.4402 & 0.4108 & 0.3473 \\
Error bars $\operatorname{\sigma}$$\left(R^{\mathrm{g}}_{i}\right)$ & 0.0081 & 0.0082 & 0.0068 \\
$\left.\left|R^{\mathrm{g}}_{i}-R^{\mathrm{sp}}_{i}\right|\right/\operatorname{\sigma}$$\left(R^{\mathrm{g}}_{i}-R^{\mathrm{sp}}_{i}\right)$ & 3.6171 & 7.1913 & 1.2833 \\
\hline
\end{tabular}
\caption{Reflectivity values for beam splitters (labelled by index $i$) obtained using different methods. The four section of the table present the reflectivity estimates obtained using (i)~single-photon data, (ii)~our two-photon characterization procedure, (iii)~two-photon characterization without calibration and (iv)~two-photon characterization using Gaussian curve-fitting and without using calibration.}
\label{Table:Reflectivity}
\end{table}

Table~\ref{Table:Reflectivity} presents the reflectivity estimates of beam-splitters $i \in \{2,3,4\}$ obtained using various procedures.
The first section of Table~\ref{Table:Reflectivity} presents the beam splitter reflectivity values that are estimated using single-photon data.
The reflectivity values estimated from different two-photon procedures are compared with these single-photon estimates of reflectivity.
The second section of the table presents the reflectivity values estimated using our characterization procedure, where the mode-matching parameter~$\gamma$ is obtained using two-photon coincidences on beam splitter $i =1$.
The third section presents reflectivity values estimated using our characterization procedure without calibration, i.e., with the assumption that $\gamma = 1$ and $\sigma(\gamma) =0$.
The final section of Table~\ref{Table:Reflectivity} contains reflectivity values estimated without calibration and using Gaussian curve-fitting (instead of fitting to curves calculated from spectral functions).

We compare the accuracy values of different beam splitters by comparing the respective differences between the reflectivity values obtained from the procedure and those obtained from single-photon data.
Specifically, we compare
\begin{equation}
\frac{\left|R^{\mathrm{proc}}_{i}-R^{\mathrm{sp}}_{i}\right|}
{\operatorname{\sigma}\left(R^{\mathrm{proc}}_{i}-R^{\mathrm{sp}}_{i}\right)},
\label{Eq:Distance}
\end{equation}
for beam splitters labelled $i\in\{2,3,4\}$ and reflectivity estimates $R_{i}^{\mathrm{proc}}$ obtained from different procedures.
The denominator in Equation~\eqref{Eq:Distance} normalizes the distances with respect to the error bars on the estimated reflectivity.
Table~\ref{Table:Reflectivity} illustrates that these distances are consistently smaller than unity for our characterization procedure, whereas those for other procedures are not.
We conclude that our procedure is accurate whereas other procedures are not.
This completes the experimental verification of our procedure for the characterization of linear optics.
%------------------------------%
\section{Conclusion}
\label{Sec:VerificationConclusion}
%------------------------------%
In summary, I have presented numerical and experimental evidence for the advantage of our characterization procedure over existing procedures.
Numerically, we sampled 1000 Haar-random unitary interferometer matrices and simulated the characterization of interferometers using different procedures.
Experimentally, we characterized four beam-splitter reflectivities using different characterization procedure.
Both numerically and experimentally, our characterization procedure significantly reduced the error in characterization as compared to existing procedures.

 %=================%
\chapter{ $\mathrm{SU}(n)$ Representation theory for simulating linear~optics}
\label{Chap:Sun}
%=================%

This chapter presents my contribution to the representation theory of $\mathrm{SU}(n)$, namely algorithms for $\mathcal{D}$-functions of $\mathrm{SU}(n)$ and relations between $\mathrm{SU}(n)$ $\mathcal{D}$-functions and immanants.
These contributions enable a deeper analysis of multi-photon multi-channel interferometry and enable a speedup in the simulation of linear optics.
Section~\ref{Sec:SunAlgorithms} details the algorithms to compute $\mathrm{SU}(n)$ states and $\mathcal{D}$-functions.
Section~\ref{Sec:SunImmanantResult} presents our results on the relation between $\mathcal{D}$-functions and immanants of the fundamental representation of $\mathrm{SU}(n)$.

A majority of the material in Section~\ref{Sec:SunAlgorithms} and~\ref{Sec:SunImmanantResult} is taken from two articles.
One article is co-authored with Barry C.~Sanders and Hubert de Duise~\cite{Dhand2015d}.
The other article is coauthored with Hubert de Guise, Dylan Spivak and Justin Kulp~\cite{deGuise2015}.
New material is added or existing material is eliminated to improve coherence and readability.
Those parts that are reproduced verbatim from our articles are listed in ``Thesis content previously published''.

%------------------------------%
\section{Algorithms for boson realizations of $\mathrm{SU}(n)$ states and $\mathcal{D}$-functions}
\label{Sec:SunAlgorithms}
%------------------------------%

This section presents our algorithms for the computation of $\mathcal{D}$-functions via boson realizations.
$\mathcal{D}$-functions of a group element are the entries of irreps of the element.
$\mathcal{D}$-functions of the special unitary group $\mathrm{SU}(2)$ are important in nuclear, atomic, molecular and optical physics~\cite{Varvsalovivc1989,Rose1995,Edmonds1996,Racah1943,Jacob1959,Alder1956,Moshinsky1962}.
$\mathrm{SU}(1,1)$ is the iconic non-compact semi-simple Lie group, and its $\mathcal{D}$-functions appear in connection with Bogolyubov transformations, squeezing and parametric downconversion~\cite{Ui1970,Yurke1986}.
Methods for construction of intelligent states and the analysis of cylindrical Laguerre-Gauss beams employ $\mathcal{D}$-functions of $\mathrm{SU}(1,1)$~\cite{Joanis2010,Karimi2014}.
$\mathcal{D}$-functions of other Lie groups enable exact solutions to problems in quantum optics~\cite{Cervero1996,Wunsche2000}.
Recently, $\mathrm{SU}(n)$ $\mathcal{D}$-functions for arbitrary $n$ have found application to the BosonSampling problem as detailed in Section~\ref{Sec:IntroSimulation}.

Two existing procedures for computing $\mathrm{SU}(n)$ $\mathcal{D}$-function are based on factorization and on exponentiation.
Both procedures have drawbacks, which we describe as follows.
Factorization-based methods, which compute $\mathrm{SU}(n)$ $\mathcal{D}$-functions in terms of $\mathcal{D}$-functions of subtransformations, are well developed for groups of low rank~\cite{Vilenkine1968,Miller1968,Talman1968,Chacon1966,Rowe1999}.
However, generalizing these algorithms to higher $n$ requires $\mathrm{SU}(n-1)$ coupling and recoupling coefficients, which have limited availability for $n>3$, i.e., restricted to certain subgroups of $\mathrm{SU}(3)$~\cite{Draayer1973,Millener1978,Rowe2000}.
Hence, methods for $\mathcal{D}$-functions of higher groups are underdeveloped despite the application of their corresponding algebras to diverse problems~\cite{Beg1965,Slansky1981,Georgi1999,Rowe2010}.

The second approach to computing $\mathrm{SU}(n)$ $\mathcal{D}$-functions involves exponentiating and composing the matrix representations of the algebra~\cite{Cornwell1997,Gilmore2012}.
This approach has three hurdles.
For one, this method requires knowledge of all the matrix elements of each generator to be exponentiated.
Certain applications require closed-form expressions of $\mathcal{D}$-functions in terms of elements of the fundamental representation; exponentiation-based methods are infeasible for these applications because of the difficulty of exponentiating matrices analytically, especially for $n>5$.
Furthermore, if only a limited number of $\mathcal{D}$-functions are required, exponentiation is wasteful because it computes the entire set of $\mathcal{D}$-functions.

We overcome the shortcomings of current approaches by utilizing boson realizations, which map the algebra and its carrier space to bosonic operators and spaces respectively.
Boson realizations arise naturally when considering the groups Sp($2n,\mathds{R}$), SU($n$) and some of their subgroups.
For instance, $\mathrm{SU}(1,1)$, $\mathrm{SU}(2)$ and $\mathrm{SU}(3)$ boson realizations are used to study degeneracies, symmetries and dynamics in quantum systems~\cite{Jauch1940,Schwinger1952,Baker1956,Elliott1958,Fradkin1965,Iachello1987,Klein1991,Kuriyama2000,Bartlett2001}.
A wide class of problems in theoretical physics rely on boson realizations of the symplectic group~\cite{Hwa1966,Kramer1966,Moshinsky1971,Quesne1971,Rowe1984}.

Here we aimed to devise an algorithm to construct the $\mathcal{D}$-functions of arbitrary representations of $\mathrm{SU}(n)$ for arbitrary $n$.
We approach the problem of limited availability of $\mathrm{SU}(n)$ $\mathcal{D}$-functions~\cite{Weyl1950,Chaturvedi2006} by presenting
(i)~a graph-theoretic algorithm to construct boson realizations of basis sets for each weight of a given $\mathfrak{su}(n)$ irrep (Algorithm~\ref{Alg:BasisSet} in Subsections~\ref{Subsec:BasisSet}),
(ii)~an algorithm to compute boson realizations of the canonical basis states of $\mathrm{SU}(n)$ for arbitrary $n$ (Algorithms~\ref{Alg:Main} in Subsection~\ref{Subsec:BasisStates}) and
(iii)~an algorithm that employs the constructed boson realizations to compute expressions for $\mathcal{D}$-functions as polynomials in the matrix elements of the defining representation (Algorithm~\ref{Alg:D} in Subsection~\ref{Subsec:D}).

%------------------------------%
\subsection{Mapping to graphs}
\label{Sec:Definitions}
%------------------------------%
Before delving into the algorithms, I present a mapping from $\mathrm{SU}(n)$ weights and $\mathfrak{su}(n)$ transformations to the vertices and edges of a graph.
Algorithms~\ref{Alg:BasisSet} and~\ref{Alg:Main} rely on mapping the $\mathrm{SU}(n)$ irrep to a graph and systematically traversing the graph to obtain basis states.
The vertices of the irrep graph are identified with the weights of the given irrep of $\mathrm{SU}(n)$ and the edges with the action of the elements of the Lie algebra $\mathfrak{su}(n)$ on the states.
Specifically, the irrep graph $G = (\mathcal{V},\mathcal{E})$ of an $\mathrm{SU}(n)$ irrep is defined as follows.

\begin{defn}[Irrep graph]
The bijection
\begin{equation}
v\colon \{\Lambda_1,\Lambda_2,\dots,\Lambda_{d} \} \to \mathcal{V}
\end{equation}
maps the set $\{\Lambda_1,\Lambda_2,\dots,\Lambda_d\}$ of the $d$ weights in the given irrep to the vertices
\begin{equation}
\mathcal{V} = \{v(\Lambda_1),v(\Lambda_2),\dots,v(\Lambda_d)\}
\end{equation}
of its irrep graph.
Vertices $v(\Lambda_{k})$ and $v(\Lambda_{\ell})$ are connected by an edge $e_j = (v(\Lambda_{k}),v(\Lambda_{\ell})) \in \mathcal{E}$ iff $\exists~c_{i,j}, \Lambda_{k}, \Lambda_{\ell}$ such that
\begin{equation}
c_{i,j\ne i}\ket{\psi_{\Lambda_{k}}} = \ket{\psi_{\Lambda_{\ell}}},
\end{equation}
where $\ket{\psi_{\Lambda_{k}}}$ and $\ket{\psi_{\Lambda_{\ell}}}$ are SU($n$) states that have weights ${\Lambda_{k}}$ and ${\Lambda_{\ell}}$ respectively.
In general, states $\ket{\psi_{\Lambda_{k}}}$ and $\ket{\psi_{\Lambda_{\ell}}}$ are linear combinations of canonical basis states.
Edges~$\mathcal{E}$ together with the vertices $\mathcal{V}$ define the irrep graph $G = (\mathcal{V},\mathcal{E})$.
\end{defn}

More than one basis state can have the same weight.
The number of basis states sharing a weight $\Lambda_i$ is defined as the multiplicity $M(\Lambda_i)$ of the weight.
In other words, each vertex $v(\Lambda_i)$ is identified with an $M(\Lambda_i)$-dimensional space spanned by those canonical basis states that have weight $\Lambda_i$.
The vertex space and vertex basis sets are defined as follows.
\begin{defn}[Vertex spaces]
The vertex space of $v(\Lambda_i)$ is
\begin{equation}
\Psi(\Lambda_i) = \spn{\left(\ket{{\psi^{1}_{\Lambda_i}}},\ket{\psi^{2}_{\Lambda_i}},\dots,\ket{\psi^{M(\Lambda_i)}_{\Lambda_i}}\right)}
\label{Eq:Space}
\end{equation}
of the canonical basis states~(Definition~\ref{Definition:CanonicalBasisStates}) that have the weight $\Lambda_i$.
\end{defn}
\noindent
The set $\left\{\ket{\psi^{(1)}_{\Lambda_i}},\ket{\psi^{(2)}_{\Lambda_i}},\dots, \ket{\psi^{(M(\Lambda_i))}_{\Lambda_i}}\right\}$ of canonical basis states is not the only set that spans the vertex space $\Psi(\Lambda_i)$ of $v(\Lambda_i)$.
In general, basis sets of $\Psi(\Lambda_i)$ can be defined as follows.
\begin{defn}[Vertex basis sets]
The set
\begin{equation}
\left\{\ket{\phi^{(1)}_{\Lambda_i}},\ket{\phi^{(2)}_{\Lambda_i}},\dots, \ket{\phi^{(M(\Lambda_i))}_{\Lambda_i}}\right\}
\end{equation} is called the basis set of a vertex $v(\Lambda_i)$ if it spans the vertex space $\Psi(\Lambda_i)$~(\ref{Eq:Space}) of $v(\Lambda_i)$, i.e.,
\begin{equation}
\spn{\left(\ket{\phi^{1}_{\Lambda_i}},\ket{\phi^{2}_{\Lambda_i}},\dots,\ket{\phi^{M(\Lambda_i)}_{\Lambda_i}}\right)}
= \Psi(\Lambda_i).
\label{Eq:Space2}
\end{equation}
\end{defn}
\noindent
The states $\left\{\ket{\phi^{(1)}_{\Lambda_i}},\ket{\phi^{(2)}_{\Lambda_i}},\dots, \ket{\phi^{(M(\Lambda_i))}_{\Lambda_i}}\right\}$ are linear combinations of the canonical basis states $\left\{\ket{\psi^{(1)}_{\Lambda_i}},\ket{\psi^{(2)}_{\Lambda_i}},\dots, \ket{\psi^{(M(\Lambda_i))}_{\Lambda_i}}\right\}$.
Algorithm~\ref{Alg:BasisSet} computes basis sets of the spaces $\Psi(\Lambda_i)$ for each of the $d$ weights $\Lambda_i$ that occurs in a given irrep.
%
%\begin{align}
%\Psi(\Lambda_i) = \spn{\left(\ket{\phi^{1}_{\Lambda_i}},\ket{\phi^{2}_{\Lambda_i}},\dots,\ket{\phi^{M(\Lambda_i)}_{\Lambda_i}}\right)}
%\end{align}

\begin{algorithm}[p]
\caption{Basis-Set Algorithm}
\label{Alg:BasisSet}
	\begin{algorithmic}[1]
	\Require{
	\Statex
	\begin{itemize} 	
		\item
 	HWS $\ket{\psi^{K}_\mathrm{hws}}$ \Comment{Degree $N_K$~(\ref{Eq:N}) polynomial in bosonic creation operators.}
 	\item
 	$m \in \mathds{Z}^+$ \Comment{$\ket{\psi^{K}_\mathrm{hws}}$ is a HWS of $\mathrm{SU}(m)$ irrep $K$. }
	\end{itemize}
	}
	\Ensure{
	\Statex
	\begin{itemize} 	
		\item
 	$\{\Lambda_1,\Lambda_2,\dots, \Lambda_d \colon \Lambda_i \in (\mathds{Z}^+\cup 0)^{m-1}\}$ \Comment {List of weights in the irrep graph of $K$.}
		\item $d,$ Basis sets~(\ref{Eq:BasisSets})
\begin{align}
	&\left\{\ket{\phi^{(1)}_{\Lambda_1}},\ket{\phi^{(2)}_{\Lambda_1}},\dots, \ket{\phi^{(M(\Lambda_1)}_{\Lambda_1}}\right\},
	\left\{\ket{\phi^{(1)}_{\Lambda_2}},\ket{\phi^{(2)}_{\Lambda_2}},\dots, \ket{\phi^{(M(\Lambda_2))}_{\Lambda_2}}\right\},
	\dots,\nonumber \\
	&\left\{\ket{\phi^{(1)}_{\Lambda_i}},\ket{\phi^{(2)}_{\Lambda_i}},\dots, \ket{\phi^{(M(\Lambda_i))}_{\Lambda_i}}\right\},
	\dots,
	\left\{\ket{\phi^{(1)}_{\Lambda_d}},\ket{\phi^{(2)}_{\Lambda_d}},\dots, \ket{\phi^{(M(\Lambda_d))}_{\Lambda_d}}\right\}\nonumber.
\end{align}
	\end{itemize}
	}

	\Procedure{BasisSet}{$m$, $\ket{\psi^{K}_\mathrm{hws}}$}
	\State
 Initialize empty statesList, empty weightList and currentStateQueue $\gets \ket{\psi^{K}_\mathrm{hws}}$ \;
	\While{currentStateQueue is not empty}
		\State currentState $\gets$ \textsc{Dequeue}(currentStateQueue)
		\For{\textsc{CurrentOperator} $\in$ set of $\mathfrak{su}(m)$ lowering operations}
			\State newState $\gets$ \textsc{CurrentOperator}(currentState) \label{AlgLin:Lowering}
			\If{newState $\ne$ 0}
				\If {weight of currentState is already in stateList}
					\If{currentState is LI of stateList states with same weight}\label{AlgLin:LI}
						\State independentState $\leftarrow$ \textsc{Normalize}(newState)
						\State {Enqueue} independentState in currentStateQueue
						\State Add independentState to stateList
						\State Add weight of independentState to weightList
					\EndIf\Comment{Else, do nothing.}
				\Else
					\State Enqueue newState in currentStateQueue \label{AlgLine:Enqueue}
					\State Add \{weight(newState),newState\} to stateList
				\EndIf
			\EndIf
		\EndFor
	\EndWhile
	\State Return stateList
	\EndProcedure
	\end{algorithmic}
\end{algorithm}

%------------------------------%
\subsection{Basis-set algorithm (Algorithm~\ref{Alg:BasisSet})}
\label{Subsec:BasisSet}
%------------------------------%
%

\begin{figure}[h]
\centering
\includegraphics[width=0.5\textwidth]{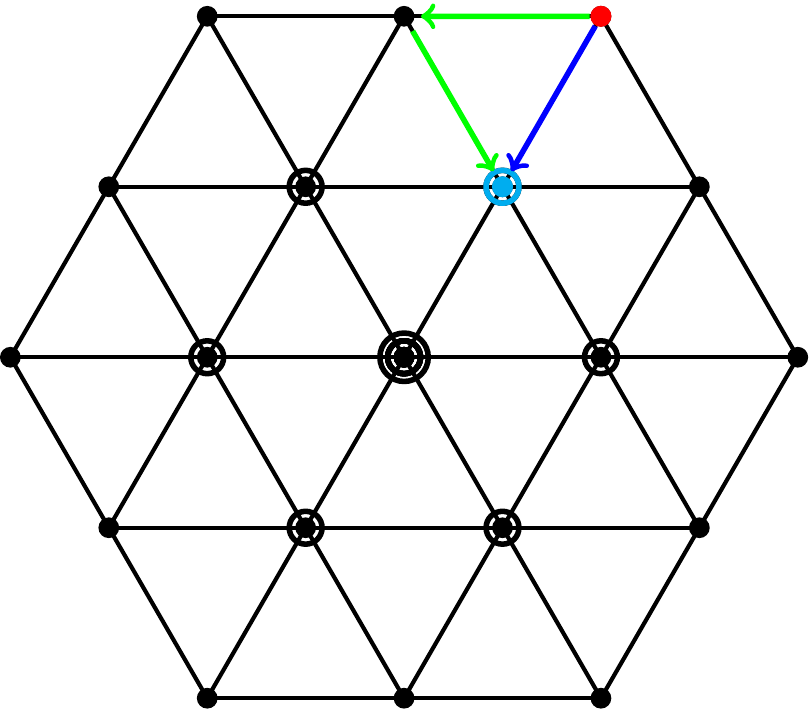}\\
\caption{
The first step of the basis-set computation algorithm~(Algorithm~\ref{Alg:BasisSet}) illustrated for the $(2,2)$ irrep of $\mathrm{SU}(3)$, where the dimension of the space of states at a given vertex is the sum of the number of dots and the number of circles at the vertex.
The algorithm constructs the HWS (occupying the red vertex) using Lemma~\ref{Lemma:HWS}.
The lowering operators can transform states at one vertex to states at another vertex along different paths connecting the starting and the target vertex, for instance the two paths coloured green and blue.
Lowering along the different paths to reach a target vertex will generate the same number of LI as the weight multiplicity.
In our illustration, we obtain a basis set that contains two independent states at the target vertex.
The algorithm traverses the irrep graph systematically until all basis sets are calculated.}%
\label{Figure:Algorithm}%
\end{figure}

This section details the basis-set algorithm%
\footnote{
In Algorithms~\ref{Alg:BasisSet}--\ref{Alg:D}, we denote operations in capital case and \textsc{SmallCaps} font.
Variables are denoted by roman font and are in lower case.
}%
 and presents proofs of termination an correctness of the algorithm.
The basis-set algorithm (Algorithm~\ref{Alg:BasisSet}) finds the basis sets for a given $\mathrm{SU}(m)$ irrep and is the key subroutine of our canonical-basis-state algorithm (Algorithm~\ref{Alg:Main}).

Algorithm~\ref{Alg:BasisSet} requires inputs $\ket{\psi^{K}_\mathrm{hws}}$ and $m$, where $\ket{\psi^{K}_\mathrm{hws}}$ is a HWS of the irrep $K$ of $\mathfrak{su}(m)$ algebra.
The state $\ket{\psi^{K}_\mathrm{hws}}$ is a bosonic state, which is expressed as a summation over products of $N_K$~(\ref{Eq:N}) creation operators $\{a^\dagger_{i,j}\colon i\in \{1,2,\dots,m\}, j\in \{1,2,\dots,m-1\}\}$.
This summation acts on the $m$-site bosonic vacuum state to give an $N_K$-boson state.
The algorithm returns multiple sets
\begin{align}
	&\left\{\ket{\phi^{(1)}_{\Lambda_1}},\ket{\phi^{(2)}_{\Lambda_1}},\dots, \ket{\phi^{(M(\Lambda_1)}_{\Lambda_1}}\right\},
	\left\{\ket{\phi^{(1)}_{\Lambda_2}},\ket{\phi^{(2)}_{\Lambda_2}},\dots, \ket{\phi^{(M(\Lambda_2))}_{\Lambda_2}}\right\},
	\dots,\nonumber \\
	&\left\{\ket{\phi^{(1)}_{\Lambda_i}},\ket{\phi^{(2)}_{\Lambda_i}},\dots, \ket{\phi^{(M(\Lambda_i))}_{\Lambda_i}}\right\},
	\dots,
	\left\{\ket{\phi^{(1)}_{\Lambda_d}},\ket{\phi^{(2)}_{\Lambda_d}},\dots, \ket{\phi^{(M(\Lambda_d))}_{\Lambda_d}}\right\}
\label{Eq:BasisSets}
\end{align}
of $\mathfrak{su}(m)$ states, with each set spanning the space $\Psi(\Lambda_i)$~(\ref{Eq:Space}) at a different vertex $v(\Lambda_i)$ in the $\mathrm{SU}(m)$ irrep $K$.
The states in the output basis sets are represented as polynomials in lowering operators acting on the HWS, or equivalently as polynomials in creation and annihilation operators acting on the $n$-site vacuum state.
Figure~\ref{Figure:Algorithm} is an illustrative example of the algorithm.
%The algorithm receives the two inputs: $m \in \mathds{Z}^+$ and the expression $\ket{\psi^{K}_\mathrm{hws}}$ of an $\mathfrak{su}(m)$ HWS with weights $K = (\kappa_1,\kappa_2,\dots,\kappa_{m-1})$ for $\kappa_i\in \mathds{Z}^+$.
%The output from the algorithm consists of (i)~a list of weights that occur in the given irrep and (ii)~a list of basis sets~(\ref{Eq:BasisSets}), each set spanning the state space at a unique element in the list of weights.

A modified breadth-first search (BFS) graph algorithm~\cite{Moore1961,Lee1961,Knuth1998} is used to traverse the irrep graph for states.
As in usual BFS, we maintain a queue%
\footnote{A queue~\cite{Knuth1998} is a first-in-first-out data structure whose entries are maintained in order.
The two operations allowed on a queue are enqueue, i.e., the addition of entries to the rear and dequeue, which is the removal of entries from the front of the queue.
Both the enqueue and dequeue operations require constant, i.e., O(1) time.},
called currentQueue, of the states that have been constructed but whose neighbourhood is yet to be explored.
The algorithm starts with the given HWS in currentQueue and iteratively dequeues a state from the front of the queue.
States neighbouring the dequeued state are obtained by enacting one-by-one each of the lowering operators of the algebra.
The newly found states are enqueued into the rear of currentQueue, and the current state and its weight are stored.

We modify BFS to handle vertices with weight multiplicity greater than unity as follows.
While traversing the irrep graph, the algorithm directly enqueues the first state that is found at each vertex.
When the same vertex is explored along a different edge, i.e., by enacting different lowering operators, a different state is found in general.
If the newly constructed state is LI of the states already constructed at the vertex, then the new state is enqueued into currentQueue.%
%\footnote{Note that Algorithm 1 does not require knowledge of weight multiplicities: new states of a given weight are kept if they are linearly independent of states of the same weight which already stored. By definition 8, the algorithm cannot produce more or fewer LI states of a given weight than the number of canonical basis states, for otherwise the dimension of this weight sunspace would not equal the number of canonical basis states}.

The algorithm truncates when a state in currentQueue is annihilated by all of the lowering operators and there is no other state in the queue.
This final state must exist because the number of LI states in a given $\mathrm{SU}(n)$ irrep is finite according to the following standard result in representation theory.
\begin{lem}[Dimension of an $\mathrm{SU}(n)$ irrep~\cite{Slansky1981}]
\label{Lemma:Dimension}
The dimension $\Delta_K$ of the carrier space of an $\mathrm{SU}(n)$ irrep $K$ is
\begin{align}
\Delta_{K} \defeq & M(\Lambda_1) + M(\Lambda_2) + \dots + M(\Lambda_d)\nonumber \\
=& \left(1+\kappa_1\right) \left(1+\kappa_2\right)\cdots \left(1+\kappa_{n-1}\right) \left(1+\frac{\kappa_1+\kappa_2}{2}\right)\left(1+\frac{\kappa_2+\kappa_3}{2}\right)\nonumber\\
&\cdots\left(1+\frac{\kappa_{n-2}+\kappa_{n-1}}{2}\right)\left(1+\frac{\kappa_1+\kappa_2+\kappa_3}{3}\right)\left(1+\frac{\kappa_2+\kappa_3+\kappa_4}{3}\right)\nonumber\\
&\cdots\left(1+\frac{\kappa_{n-3}+\kappa_{n-2}+\kappa_{n-1}}{3}\right)\cdots\left(1+\frac{\kappa_1+\kappa_2+\dots+\kappa_{n-1}}{n-1}\right).
\label{Eq:IrrepDimension}
\end{align}
%$D_K$ scales
%\begin{equation}
%N\le \left(1+\mathrm{max}_i\left(\lambda_i\right)\right)^{\frac{n(n-1)}{2}}
%\end{equation}
%no faster than exponentially in $n^2$.
\end{lem}

Now we prove that the basis-set algorithm terminates.
The proof relies on the fact that the carrier space of $\mathrm{SU}(m)$ irrep is finite-dimensional (Lemma~\ref{Lemma:Dimension}).
The algorithm's computational cost is quantified by the number of times the lowering operators are applied on the HWS or on states reached by lowering from the HWS.
We show that the computational cost of Algorithm~\ref{Alg:BasisSet} is linear in the dimension $\Delta_{K}$ of the irrep whose HWS is given as input and polynomial in $n$.
\begin{thm}[Algorithm~\ref{Alg:BasisSet} terminates]
\label{Theorem:1Terminats}
Suppose Algorithm~\ref{Alg:BasisSet} receives as input an HWS $\ket{\psi^{K}_\mathrm{hws}}$ of an $\mathrm{SU}(m)$ irrep $K$.
Then the algorithm terminates after no more than $\Delta_Km(m-1)/2$ applications of lowering operators.
\end{thm}
\begin{proof}
The proof is in two parts.
Firstly, the number of states that enters currentStateQueue is bounded above by the dimension $\Delta_K$~(\ref{Eq:IrrepDimension}) of the irrep space.
Secondly, as each state that enters currentStateQueue is acted upon by no more than $n(n-1)/2$ lowering operators, the number of lowering operations performed is less than or equal to $\Delta_K n(n-1)/2$.

We show that the number of states that enter currentStateQueue is no more than~$\Delta_K$ as follows.
As each currentState is a linear combination of states obtained by acting lowering operators (Line~\ref{AlgLin:Lowering}) on the given HWS, each state that enters currentStateQueue is in the irrep labelled by the HWS.
Moreover, each state entering the queue is tested for linear independence (Line~\ref{AlgLin:LI}) with respect to the states already obtained.
Any state that is not LI is discarded.
Therefore, each enqueued state (Algorithm~\ref{Alg:BasisSet}, Line~\ref{AlgLine:Enqueue}) is in the correct irrep $K$ and is LI of each other enqueued state.
Thus, the number of states that ever enter currentStateQueue is no more than the number~$\Delta_K$ of LI states in irrep $K$.

In each iteration of the algorithm, each of the lowering operators are applied on the states in currentStateQueue.
The number of lowering operations is thus bounded above by the product $\Delta_K n(n-1)/2$ of the number of states that enter currentStateQueue and of the number of lowering operators in the $\mathfrak{su}(n)$ algebra.
The algorithm thus terminates after no more than $\Delta_Kn(n-1)/2$ applications of lowering operators.
\end{proof}

We now prove that the algorithm returns the correct output on termination.
The proof requires the following lemma stating that each canonical basis state can be obtained by enacting only with the lowering operators on the HWS.
\begin{lem}[Every basis-state can be reached by lowering from the HWS~\cite{Humphreys1972}]
\label{Lemma:Cyclic}
No canonical basis-state of a given $\mathrm{SU}(n)$ irrep $K$ is LI of the states obtained by lowering from the HWS by the action
\begin{equation}
c_{i_{k},j_{k}}\cdots c_{i_{2},j_{2}} c_{i_{1},j_{1}}\ket{\psi_\mathrm{hws}}\quad i_\ell \le j_\ell\, \forall 1\le \ell \le k
\label{Eq:Construction}
\end{equation}
of $k\le \sum_i{\kappa_i}$ number of $\mathfrak{su}(n)$ lowering operators on the HWS of the irrep.
\end{lem}
\noindent Lemma~\ref{Lemma:Cyclic} implies that each basis state can be constructed by linearly combining states obtained on lowering from the HWS.
Algorithm~\ref{Alg:BasisSet} leverages from the construction of Equation~(\ref{Eq:Construction}) and from testing linear independence to construct the basis sets.

The correctness of the basis-set algorithm is proved as follows.
We show that each state obtained by enacting any number of lowering operators on the HWS is LD on the states returned by the algorithm.
Each canonical basis state is LD on the states obtained by lowering from the HWS in turn, so each canonical basis state is LD on the algorithm output.
The algorithm only constructs states in the correct irrep so Algorithm~\ref{Alg:BasisSet} returns a complete basis set at each weight of the irrep on truncation.
\begin{thm}[Algorithm~\ref{Alg:BasisSet} is correct]
\label{Theorem:1Correct}
The sets
\begin{align}
	&\left\{\ket{\phi^{(1)}_{\Lambda_1}},\ket{\phi^{(2)}_{\Lambda_1}},\dots, \ket{\phi^{(M(\Lambda_1)}_{\Lambda_1}}\right\},
	\left\{\ket{\phi^{(1)}_{\Lambda_2}},\ket{\phi^{(2)}_{\Lambda_2}},\dots, \ket{\phi^{(M(\Lambda_2))}_{\Lambda_2}}\right\},
	\dots,\nonumber \\
	&\left\{\ket{\phi^{(1)}_{\Lambda_i}},\ket{\phi^{(2)}_{\Lambda_i}},\dots, \ket{\phi^{(M(\Lambda_i))}_{\Lambda_i}}\right\},
	\dots,
	\left\{\ket{\phi^{(1)}_{\Lambda_d}},\ket{\phi^{(2)}_{\Lambda_d}},\dots, \ket{\phi^{(M(\Lambda_d))}_{\Lambda_d}}\right\}\nonumber
\end{align}
of states returned by Algorithm~\ref{Alg:BasisSet} span the respective vertex spaces~$\Psi(\Lambda_i)$~(\ref{Eq:Space}) at each vertex~$\Lambda_i$ of the given irrep $K$.
\end{thm}
\begin{proof}
We first prove by induction that each state in the form of Equation~(\ref{Eq:Construction}) is LD on states in the algorithm output.
Our induction hypothesis is that each state
\begin{equation}
c_{i_{k},j_{k}}\cdots c_{i_{2},j_{2}} c_{i_{1},j_{1}}\ket{\psi_\mathrm{hws}},
\label{Eq:EllState}
\end{equation}
which is obtained by acting $\ell$ lowering operators on the HWS, is LD on the states returned by the algorithm $\forall\ell\in\mathds{Z^+}$.
The proof of the hypothesis follows from mathematical induction over $\ell$.

The induction hypothesis is true for base case $k = 1$.
In the first iteration, the algorithm enacts all the lowering operators on the HWS~(Algorithm~\ref{Alg:BasisSet} Line \ref{AlgLin:Lowering}) and saves each of the obtained states.
No $k = 1$ state~(\ref{Eq:EllState}) is omitted because the vertices neighbouring the HWS vertex are all being explored for the first time.
Hence, all the states that can be reached by lowering once from the HWS are added to currentStateQueue and, eventually, to stateList.

Assume that the induction hypothesis holds for $k = \ell$, i.e., each $k=\ell$ state is LD on the states in stateList.
We prove that the hypothesis holds for $k = \ell + 1$ by contradiction.
Suppose there exists a state that can be reached by enacting $\ell+1$ lowering operators on the HWS but is LI of stateList.
Let $\ket{\psi} = c_{i_{\ell+1},j_{\ell+1}}c_{i_{\ell},j_{\ell}}\cdots c_{i_{2},j_{2}} c_{i_{1},j_{1}}\ket{\psi_\mathrm{hws}}$ be such a state.

Consider now the state $\ket{\varphi} = c_{i_{\ell},j_{\ell}}\cdots c_{i_{2},j_{2}} c_{i_{1},j_{1}}\ket{\psi_\mathrm{hws}}$ obtained by enacting one less lowering operation from the HWS; i.e., $\ket{\psi} = c_{i_{\ell+1},j_{\ell+1}}\ket{\varphi}$.
We have assumed that the induction hypothesis holds for $k = \ell$.
Therefore, $\ket{\varphi}$ is LD on the states constructed by a algorithm.
In other words,
\begin{equation}
\big|{\varphi}\big\rangle = \sum_{j=1}^{J} a_j\ket{\phi_j}
\end{equation}
is LD on the stateList elements $\left\{\ket{\phi_j}:j\in\{1,2,\dots,J\}\right\}$ for complex numbers $a_j$.

The algorithm enacts the lowering operator $c_{i_{\ell+1},j_{\ell+1}}$ on each $\ket{\phi_j}$ and the resulting states are either stored in stateList or are LD on elements in stateList.
Therefore, the elements of the set $\{c_{i_{\ell+1},j_{\ell+1}}\ket{\phi_j}: j \in\{1,2,\dots,J\}\}$ are LD on the elements of stateList.
Hence, the element $c_{i_{\ell+1},j_{\ell+1}}\ket{\varphi}$ is also LD on the elements of stateList.
This dependence contradicts the supposition that $\ket{\psi} = c_{i_{\ell+1},j_{\ell+1}}\ket{\varphi}$ is LI of stateList, thereby proving the induction hypothesis for $k = \ell+1$.

The induction hypothesis is true for $\ell = 1$ and is shown to hold for $k = \ell + 1$ if it holds for $k = \ell$.
Thus, our induction hypothesis is true for all $\ell\in \mathds{Z}^+$.
Every state obtained of irrep $K$ obtained by lowering from the HWS is linearly dependent (\gls{LD}) on the basis sets that are returned by the algorithm.

We know from Lemma~\ref{Lemma:Cyclic} that each canonical basis state is LD on the states obtained by lowering.
Hence, each canonical basis state is LD on the states obtained at the output of the algorithm.
Therefore, the state returned by the algorithm span the space of irrep $K$ states, and the output basis sets span the set of all basis states of the given irrep $K$.
\end{proof}

We have proved that Algorithm~\ref{Alg:BasisSet} terminates and that it returns the correct basis sets on termination.
Now I present our algorithm for the construction of the canonical basis states.
Furthermore, I prove the correctness and termination of the canonical-basis-states algorithm.

%------------------------------%
\subsection{Canonical-basis-states algorithm (Algorithm~\ref{Alg:Main})}
\label{Subsec:BasisStates}
%------------------------------%
The algorithm for constructing the canonical basis-states of $\mathrm{SU}(n)$ requires inputs $n \in \mathds{Z}^+$ and the irrep label $K$.
The algorithm returns expressions for all the canonical basis states in the given irrep.
Figure~\ref{Figure:Main} illustrates $\mathrm{SU}(3)$ basis-state construction using our algorithm.
Algorithm~\ref{Alg:Main} details the step-by-step construction of the canonical basis states.

\begin{figure}[h]
\includegraphics[width=80ex]{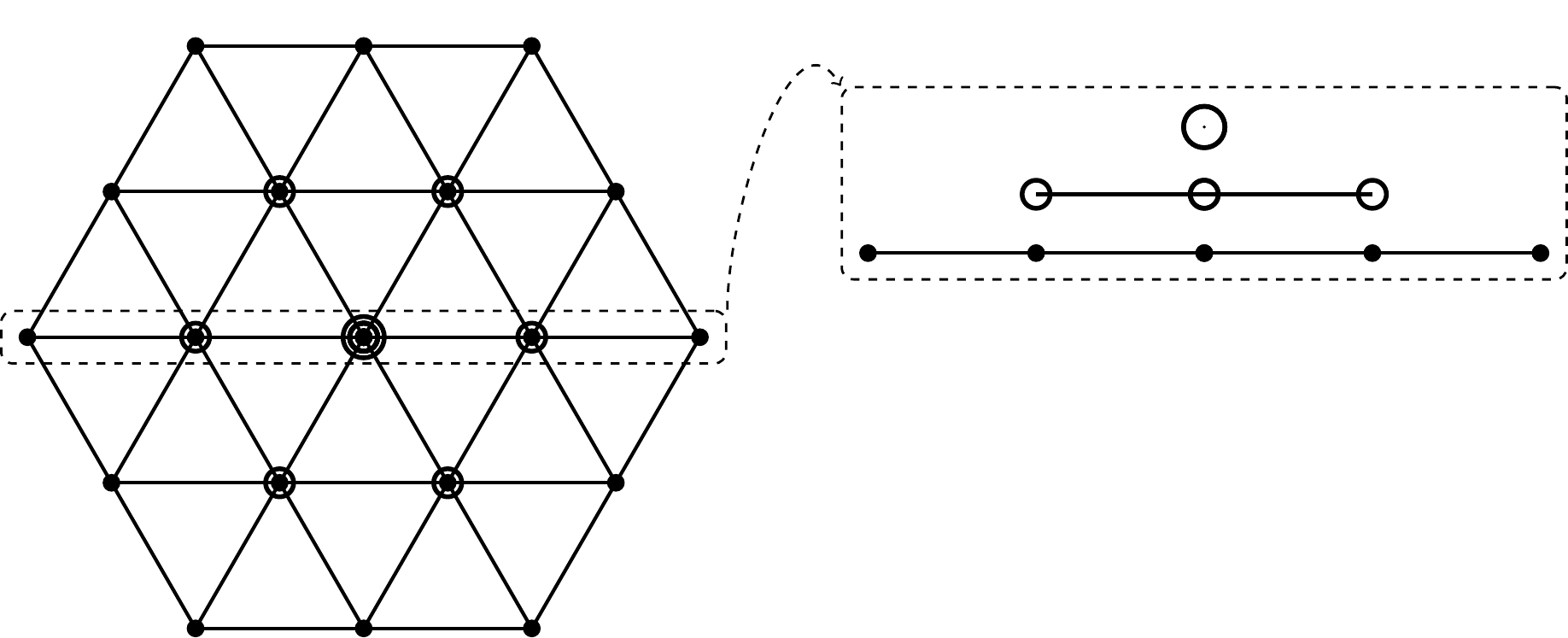}
\caption{Diagrammatic representation of the main algorithm for $n=3$.
The dots and circles represent the canonical basis states.
The dimension of the space of states at a given vertex is the sum of the number of dots and the number of circles at the vertex, for instance weights associated with dimension two are represented by one dot and one circle.
The lines connecting the dots represent the transformation from states of one weight to those of another by the action of $\mathrm{SU}(3)$ raising and lowering operators.
We use Algorithm~\ref{Alg:BasisSet} to construct basis sets for each vertex in the $\mathrm{SU}(n)$ irrep graph.
Once the basis sets for the $\mathrm{SU}(n)$ irreps are computed, the algorithm enacts the $\mathfrak{su}(n-1)$ raising operators on the $(n-1)$-dimensional sub-irreps to find the $\mathfrak{su}(n-1)$ HWS.
Then the algorithm starts with the $\mathfrak{su}(n-1)$ HWS and employs the basis-set construction (Algorithm 1) to find all the states in the $\mathfrak{su}(n-1)$ irrep labelled by the HWS.
The states thus obtained are subtracted from the set of $\mathfrak{su}(n)$ states.
A new state is chosen from the weight of highest multiplicity and the process repeated until all the $\mathfrak{su}(n-1)$ irreps are found.
}%
\label{Figure:Main}%
\end{figure}

\begin{algorithm}[p]
\caption{Canonical-basis-states algorithm}\label{Alg:Main}
	\begin{algorithmic}[1]
	\Require{
	\Statex
	\begin{itemize} 	
 	\item
 	$n \in \mathds{Z}^+$ \Comment{Algorithm constructs basis sets of $\mathfrak{su}(n)$ algebra.}
		\item
		$K = (\kappa_1,\kappa_2,\dots,\kappa_{n-1})\in \left(\mathds{Z}^+ \cup \{0\}\right)^{n-1}$ \Comment{Label of $\mathrm{SU}(n)$ irrep.}
	\end{itemize}
	}
	\Ensure{
	\Statex
	\begin{itemize} 	
		\item $\left\{\left(\Big|{\tensor*{\psi}{*^{K^{(n)}}_{\Lambda^{(n)}}^{,\dots,}_{,\dots,}^{K^{(3)},}_{\Lambda^{(3)},}^{K^{(2)}}_{\Lambda^{(2)}}}}\Big\rangle;K^{(n)},\dots,K^{(2)};\Lambda^{(n)},\dots,\Lambda^{(2)}\right)\right\}$\Comment{List of all canonical basis states and weight labels in the irrep $K^{(n)}=K$.}
	\end{itemize}
	}
	\Statex
	\Procedure{CanonicalBasisStates}{$n$, $K$}
	\State Initialize empty basisStatesList, HWS $\leftarrow \ket{\psi^K_\mathrm{hws}}$ \label{Alglin:HWS}
	\State SUmStates, SUnStates $\leftarrow$ {\textsc{BasisSet}}($n$,HWS) \label{Alglin:FirstStage}
	\While{SUnStates is not empty}
	\For{$m\in \{n,n-1,\dots,2\}$}
		\State $\Lambda_\text{max} \leftarrow \mathfrak{su}(m)$ weight with highest number of states in SUmStates.
		\State $\ket{\psi^{(m)}_\text{max}}\leftarrow$ arbitrary superposition of states at $\Lambda_\text{max}$ in SUmStates.
		\State \label{AlgLin:Raising}Apply $\mathfrak{su}(m-1)$ raising operators on $\ket{\psi^{(m)}_\text{max}}$; reach $\mathfrak{su}(m-1)$ HWS $\ket{\psi^{(m-1)}_\mathrm{hws}}$.
		\State $K^{(m-1)} \leftarrow$ {\textsc{Weight}}$\left(\ket{\psi^{(m-1)}_\mathrm{hws}}\right)$.
		\State SUmStates $\leftarrow$ {\textsc{BasisSet}}$\left(m-1,\ket{\psi^{(m-1)}_\mathrm{hws}}\right)$.
		\If{$m=2$}
		\For{All states $\ket{\psi}$ in SUmStates}
			\State $\{\Lambda^{(n)},\dots,\Lambda^{(2)}\} \leftarrow$ {\textsc{Weights}}$\left(\ket{\psi}\right)$ \Comment{$\mathfrak{su}(m)$ weights $\forall m\le n$.}
			\State Concatenate $\left(\ket{\psi};K^{(n)},\dots,K^{(2)};\Lambda^{(n)},\dots,\Lambda^{(2)} \right)$ basisStatesList
			\State Subtract SUmStates from SUnStates.
		\EndFor
		\EndIf\Comment{Else, do nothing.}	
	\EndFor
	\EndWhile
	\For{All states $\ket{\psi^{(i)}}$ in statelist}
		\State Act $\{C_{1,2},C_{2,3},\dots, C_{n-1,n}\}$ on $\ket{\psi^{(i)}}$ until HWS $\ket{\psi_\mathrm{hws}^{(i)}}$ is reached.
		\State $\ket{\psi^{(i)}}\leftarrow \e^{i\phi^{(i)}}\ket{\psi^{(i)}}$ for $\ket{\psi_\mathrm{hws}^{(i)}} = \e^{i\phi^{(i)}} \ket{\psi^K_\mathrm{hws}}$.\label{Alglin:PhaseConvention}
	\EndFor
	\State Return basisStatesList
	\EndProcedure
	\end{algorithmic}
\end{algorithm}

The canonical-basis-states algorithm proceeds by partitioning $\mathfrak{su}(n)$ basis sets into $\mathfrak{su}(m)$ basis sets for progressively smaller $m$ over $n-1$ stages.
In the first stage, the algorithm employs Lemma~\ref{Lemma:HWS} to construct the HWS of the given irrep $K$ (Algorithm~\ref{Alg:Main}, Line \ref{Alglin:HWS}).
Algorithm~\ref{Alg:BasisSet} is then used to construct the basis sets of the $\mathrm{SU}(n)$ irrep of the constructed HWS (Line \ref{Alglin:FirstStage}).

By the $(n-m)$-th stage, the algorithm has partitioned the entire $\mathfrak{su}(n)$ space into basis states of the $\mathrm{SU}(m+1)$ irreps.
In this stage, each of the $\mathfrak{su}(m+1)$ basis sets is partitioned into $\mathfrak{su}(m)$ basis sets by using $\mathfrak{su}(m)$ operators.
The algorithm searches each $\mathrm{SU}(m+1)$ irrep graph for the vertex that has the highest multiplicity.

An arbitrary linear combination of the basis states at this vertex is chosen.
The algorithm then enacts all the raising operators in the $\mathfrak{su}(m)$ subalgebra on this linear combination until the action of each of the raising operators annihilates the state.
The state thus obtained is the HWS of an $\mathrm{SU}(m)$ irrep, whose label $K^{(m)}$ can be calculated by enacting the Cartan operators on the state.

Next the algorithm performs the basis-set construction algorithm on the $\mathfrak{su}(m)$ HWS employing only the $\mathfrak{su}(m)$ lowering operators.
This procedure gives us sets of basis states that belong to the $\mathrm{SU}(m)$ irrep $K^{(m)}$.
The irrep $K^{(m)}$ basis sets are stored and are then subtracted from the $\mathrm{SU}(m+1)$ states.
The algorithm iteratively (i)~starts from the highest multiplicity vertex of $\mathrm{SU}(m+1)$ irrep graphs, (ii)~constructs a HWS by raising, (iii)~stores the basis sets of $\mathrm{SU}(m)$ irreps corresponding to this HWS and (iv) subtracts them from $\mathrm{SU}(m+1)$ states until all the states in the $\mathfrak{su}(m+1)$ are partitioned.

At the end of $n-1$ stages, we have a list of basis sets of the $\mathrm{SU}(n-1)$ irreps.
We iteratively perform the process of finding basis sets for smaller subgroups until we reach $\mathrm{SU}(2)$ basis sets, which are known to have unit multiplicity.
Hence, the algorithm returns the basis states that are eigenvectors of the Cartan operators of all $\mathrm{SU}(m)\colon m\le n$ groups.

The relative phases between the basis states are fixed by imposing Equation~(\ref{Eq:PhaseConvention}).
Each of the constructed basis states is acted upon by the simple raising operators $\{C_{1,2},C_{2,3},\dots,C_{n-1,n}\}$ until the HWS is reached.
The phase of this HWS obtained by raising is required to be the same for all basis states.
Our algorithm multiplies each of the basis states by a phase factor (Line~\ref{Alglin:PhaseConvention}) to impose the phase convention~ Equation~(\ref{Eq:PhaseConvention}) and returns the set of canonical basis states.

Now we prove that the canonical basis states algorithm terminates.
The proof of termination uses the facts that the number of basis states is equal to the dimension $\Delta_K$ of the irrep and that each basis state is added to currentStateQueue no more than once.
\begin{thm}[Algorithm~\ref{Alg:Main} terminates]
\label{Theorem:2Terminates}
Algorithm~\ref{Alg:Main} terminates after the action of no more than $\Delta_K n(n-1)^2/2$ lowering operators.
\end{thm}
\begin{proof}
In each of the $n-1$ stages of Algorithm~\ref{Alg:Main}, the states that are added to currentStateQueue are LI of each other because of the conditions imposed in the algorithm.
There are no more LI states in the given $\mathrm{SU}(n)$ irrep than the dimension $\Delta_K$ of the irrep space.
Thus, the total number of states that are added to currentQueue in each of the $n-1$ stages is no more than $\Delta_K$.
No more that $n(n-1)/2$ lowering operators are applied on the states that enter currentQueue.
Thus, each stage terminates after the application of $\Delta_K n(n-1)/2$ lowering operations.
Furthermore, the algorithm terminates after $n-1$ stages and the application of no more than $\Delta_K n(n-1)^2/2$ lowering operations.
\end{proof}

Finally, we prove that the canonical-basis-states algorithm returns the correct output when it terminates.
\begin{thm}[Algorithm~\ref{Alg:Main} is correct]
\label{Theorem:2Correct}
The $\mathrm{SU}(n)$ states
\begin{equation}
\Big|{\tensor*{\psi}{*^{K^{(z)}}_{\Lambda^{(n)}}^{,\dots,}_{,\dots,}^{K^{(3)},}_{\Lambda^{(3)},}^{K^{(2)}}_{\Lambda^{(2)}}}}\Big\rangle
\end{equation}
yielded by Algorithm~\ref{Alg:Main} are the canonical states of Definition~\ref{Definition:CanonicalBasisStates}.\end{thm}
\begin{proof}
The theorem holds if the states yielded by Algorithm~\ref{Alg:Main} have well defined weights and have well defined irrep labels.
First we show that the weight of each state in the output of the algorithm is well defined.
Each state in the output is obtained either by enacting lowering operators on the HWS or by taking linear combinations of states that have the same weight.
Linear combination of states with the same weights have well defined weights themselves.
Thus, all the output states have well defined weights for $\mathrm{SU}(m)$ irreps for all $2\le m\le n$.

We prove that the states have well defined $\mathrm{SU}(m)$ irrep label separately for $m = n$ and for $2\le m\le n-1$.
The correctness of the $\mathfrak{su}(m)$ HWS follows from Lemma~\ref{Lemma:HWS}.
Every state in the output is a linear combination of states obtained by lowering from the constructed $\mathfrak{su}(m)$ HWS.
Thus, every state is in the correct $\mathrm{SU}(n)$ irrep $K^{(n)}$.

The algorithm (Line~\ref{AlgLin:Raising}) enacts raising operators on linear combinations of $\mathfrak{su}(m+1)$ basis states at one weight until each of the raising operators annihilates the raised state.
The $\mathfrak{su}(m)$ state thus obtained are legitimate $\mathfrak{su}(m)$ HWS's or possibly linear combinations of $\mathfrak{su}(m)$ HWS's by construction.
The uniqueness of the HWS is guaranteed by the existence of the canonical basis~\cite{Humphreys1972}.
Each of the canonical basis states is obtained by lowering from these $\mathfrak{su}(m)$ HWS's using $\mathfrak{su}(m)$ lowering operators and thus have well defined irrep labels for all $2\le m \le n-1$.

We have shown that the states yielded by the algorithm have well defined values of irrep labels $K^{(m)}$ for $\mathfrak{su}(m)$ algebras for all $\{m:2\le m\le n\}$ and of $\mathfrak{su}(m)$ weights $\Lambda^{(m)}$ for all $\{m:2\le m\le n\}$.
Thus, these states are the canonical $\mathrm{SU}(n)$ basis states.
This completes the proof of correctness of Algorithm~\ref{Alg:Main}.
\end{proof}

We have proved that Algorithm~\ref{Alg:Main} terminates and returns the canonical basis states on termination.
The states constructed by the canonical-basis-states algorithm are employed to compute arbitrary $\mathrm{SU}(n)$ $\mathcal{D}$-functions using an algorithm presented in the next subsection.

%------------------------------%
\subsection{$\mathcal{D}$-function algorithm}
\label{Subsec:D}
%------------------------------%
\begin{algorithm}[p]
\caption{$\mathcal{D}$-function Algorithm}\label{Alg:D}
	\begin{algorithmic}[1]
	\Require{
	\Statex
	\begin{itemize} 	
 	\item
 	$n \in \mathds{Z}^+$ \Comment{Algorithm constructs $\mathcal{D}$-functions of $\mathrm{SU}(n)$ elements.}
		\item
		$\Omega = \{\omega_1,\omega_2,\dots,\omega_{n^2-1}\} \in \mathds{R}^{n^2-1}$ \Comment{Parametrization of $\mathrm{SU}(n)$ transformation.}
		\item
		$K^{(n)},\dots,K^{(2)}$ and ${\Lambda^{(n)},\dots,\Lambda^{(2)}}$ \Comment{Row Label.}
		\item
		$K^{\prime(n)},\dots,K^{\prime(2)}$ and ${\Lambda'^{(n)},\dots,\Lambda'^{(2)}}$ \Comment{Column Label.}
	\end{itemize}
	}
	\Ensure{
	\Statex
	\begin{itemize} 	
		\item $\tensor*{\D}{*^{K^{(n)}}_{\Lambda^{(n)}}^{,\dots,}_{,\dots,}^{K^{(3)},}_{\Lambda^{(3)},}^{K^{(2)}}_{\Lambda^{(2)}}^;_;^{K^{\prime(n)}}_{\Lambda^{\prime(n)}}^{,\dots,}_{,\dots,}^{K^{\prime(3)},}_{\Lambda^{\prime(3)},}^{K^{\prime(2)}}_{\Lambda^{\prime(2)}}}(\Omega)$	
	\end{itemize}
	}
	\Statex
	\Procedure{D}{$n,\Omega, K^{(n)},\dots,K^{(2)}, K'^{(n)},\dots,K'^{(2)}, {\Lambda^{(n)},\dots,\Lambda^{(2)}},\Lambda^{\prime (n)},\dots,\Lambda^{\prime (2)}$}
	\State Construct $V \in GL(n,\mathds{C})$ from parametrization $\Omega$~\cite{Reck1994}
	\If{$K^{(n)} = K^{\prime(n)}$}
	\State $\Big|{\tensor*{\psi}{*^{K^{(n)}}_{\Lambda^{(n)}}^{,\dots,}_{,\dots,}^{K^{(3)},}_{\Lambda^{(3)},}^{K^{(2)}}_{\Lambda^{(2)}}}}\Big\rangle \leftarrow$ {using \textsc{CanonicalBasisStates}}($n,K(n)$).	
	\State $\Big|{\tensor*{\psi}{*^{K^{\prime(n)}}_{\Lambda^{\prime(n)}}^{,\dots,}_{,\dots,}^{K^{\prime(3)},}_{\Lambda^{\prime(3)},}^{K^{\prime(2)}}_{\Lambda^{\prime(2)}}}}\Big\rangle \leftarrow$ using {\textsc{CanonicalBasisStates}}($n,K'(n)$).	
	\State Construct $\Big\langle{\tensor*{\psi}{*^{K^{(n)}}_{\Lambda^{(n)}}^{,\dots,}_{,\dots,}^{K^{(3)},}_{\Lambda^{(3)},}^{K^{(2)}}_{\Lambda^{(2)}}}}\Big|$ from $\Big|{\tensor*{\psi}{*^{K^{(n)}}_{\Lambda^{(n)}}^{,\dots,}_{,\dots,}^{K^{(3)},}_{\Lambda^{(3)},}^{K^{(2)}}_{\Lambda^{(2)}}}}\Big\rangle$ by complex conjugation.
	\State Construct $V(\Omega)\Big|{\tensor*{\psi}{*^{K^{\prime(n)}}_{\Lambda^{\prime(n)}}^{,\dots,}_{,\dots,}^{K^{\prime(3)},}_{\Lambda^{\prime(3)},}^{K^{\prime(2)}}_{\Lambda^{\prime(2)}}}}\Big\rangle$ using $a^\dagger_{i,k} \to \sum_j V_{i,j}(\Omega) a^\dagger_{j,k}\,\forall\,a_{i,k},a^\dagger_{i,k}$.	
	\State Return $\D =\Big\langle{\tensor*{\psi}{*^{K^{(n)}}_{\Lambda^{(n)}}^{,\dots,}_{,\dots,}^{K^{(3)},}_{\Lambda^{(3)},}^{K^{(2)}}_{\Lambda^{(2)}}}}\Big|V(\Omega) \Big|{\tensor*{\psi}{*^{K^{\prime(n)}}_{\Lambda^{\prime(n)}}^{,\dots,}_{,\dots,}^{K^{\prime(3)},}_{\Lambda^{\prime(3)},}^{K^{\prime(2)}}_{\Lambda^{\prime(2)}}}}\Big\rangle$
	\Else
	\State Return $\D = 0$
	\EndIf
	\EndProcedure
	\end{algorithmic}
\end{algorithm}

Our task is to construct the $\mathcal{D}$-function
\begin{equation}
\tensor*{\D}{*^{K^{(n)}}_{\Lambda^{(n)}}^{,\dots,}_{,\dots,}^{K^{(3)},}_{\Lambda^{(3)},}^{K^{(2)}}_{\Lambda^{(2)}}^;_;^{K^{\prime(n)}}_{\Lambda^{\prime(n)}}^{,\dots,}_{,\dots,}^{K^{\prime(3)},}_{\Lambda^{\prime(3)},}^{K^{\prime(2)}}_{\Lambda^{\prime(2)}}}(\Omega)
\label{Eq:DDef2}
\end{equation}
for given labels $\{K^{(m)}\}, \{\Lambda^{(m)}\}, \{K'^{(m)}\}, \{\Lambda'^{(m)}\}$ of the $\mathrm{SU}(n)$ element $V(\Omega)$ given by the parametrization $\Omega$.
The $\mathcal{D}$-function~(\ref{Eq:DDef2}) is computed as the inner product between the state
\begin{equation}
\Big|{\tensor*{\psi}{*^{K^{(n)}}_{\Lambda^{(n)}}^{,\dots,}_{,\dots,}^{K^{(3)},}_{\Lambda^{(3)},}^{K^{(2)}}_{\Lambda^{(2)}}}}\Big\rangle
\end{equation}
of Equation~(\ref{Eq:States}) and the transformed state
\begin{equation}
V(\Omega)\Big|{\tensor*{\psi}{*^{K^{\prime(n)}}_{\Lambda^{\prime(n)}}^{,\dots,}_{,\dots,}^{K^{\prime(3)},}_{\Lambda^{\prime(3)},}^{K^{\prime(2)}}_{\Lambda^{\prime(2)}}}\Big\rangle}.
\end{equation}
Algorithm~\ref{Alg:D} constructs the fundamental representation, i.e., the $n\times n$ matrix, $V_{ij}$ of the $\mathrm{SU}(n)$ element $V(\Omega)$~\cite{Reck1994}.
Then, the expressions for the basis states
\begin{equation}
\Big|{\tensor*{\psi}{*^{K^{(n)}}_{\Lambda^{(n)}}^{,\dots,}_{,\dots,}^{K^{(3)},}_{\Lambda^{(3)},}^{K^{(2)}}_{\Lambda^{(2)}}}}\Big\rangle, \Big|{\tensor*{\psi}{*^{K^{\prime(n)}}_{\Lambda^{\prime(n)}}^{,\dots,}_{,\dots,}^{K^{\prime(3)},}_{\Lambda^{\prime(3)},}^{K^{\prime(2)}}_{\Lambda^{\prime(2)}}}\Big\rangle
}
\label{Eq:States}
\end{equation}
corresponding to the given labels are computed using the canonical-basis-states algorithm.
The basis states thus obtained are expressed as summations over products of creation and annihilation operators.
$V(\Omega)$ acts on the boson realization by transforming each boson independently according to
\begin{equation}
a^\dagger_{i,j} \rightarrow a^{\dagger\prime}_{i,j} = \sum_{k} V_{ik}(\Omega)a^\dagger_{k,j},
\label{Eq:TransformationU}
\end{equation}
where $\{V_{ik}(\Omega)\}$ are the matrix elements of the $n\times n$ representation of $V(\Omega)$.
The algorithm transforms the second basis state of Equation~(\ref{Eq:States}) under the action of $V(\Omega)$ by replacing each of the creation and annihilation operators of the state according to Equation~(\ref{Eq:TransformationU}).

The $\mathcal{D}$-function is evaluated as the inner product using the commutation relations~(\ref{Eq:ccr}) or equivalently by using the Wick's theorem~\cite{Bogoljubov1959}.
The correctness and termination of Algorithm~\ref{Alg:D} follows directly from Theorems~\ref{Theorem:2Terminates} and~\ref{Theorem:2Correct}, which completes our algorithms for the computation of boson realizations of $\mathrm{SU}(n)$ states and of $\mathcal{D}$-functions.

%------------------------------%
\subsection{Conclusion}
%------------------------------%
In summary, we have devised an algorithm to compute expressions for boson realizations of the canonical basis states of $\mathrm{SU}(n)$ irreps.
Boson realizations are ideally suited for analyzing the physics of single photons, providing a tractable interpretation to basis states as multi-photon states and to transformations on these states as optical transformations.
Furthermore, we have devised an algorithm to compute expressions for $\mathrm{SU}(n)$ $\mathcal{D}$-functions in terms of elements of the fundamental representation of the group.
Our algorithm offers significant advantage over competing algorithms to construct $\mathcal{D}$-functions.
Furthermore, our $\mathcal{D}$-function algorithm lays the groundwork for generalizing the analysis of optical interferometry beyond the three-photon level~\cite{Tan2013,Guise2014,Tillmann2015}.

This work is the first known application of graph-theoretic algorithms to $\mathrm{SU}(n)$ representation theory.
We overcome the problem of $\mathrm{SU}(n)$ weight multiplicity greater than unity by modifying the breadth-first graph-search algorithm.
Our procedure for generating a basis set can be extended to subgroups of $\mathrm{SU}(n)$.
In particular, the boson realization of the HWS of O$(2k)$ and O$(2k+1)$ irreps can be constructed along the lines of Lemma~\ref{Lemma:HWS}~\cite{Pang1967,Wong1969,Lohe1971}.
Graph-search algorithms can be employed to construct $\BigO{n}$ basis states and $\mathcal{D}$-functions if the problem of labelling $\BigO{n}$ basis states can be overcome.
Our approach opens the possibility of exploiting the diverse graph-algorithms toolkit for solving problems in representation theory of Lie groups.

%=================%
\section{$\mathcal{D}$-functions and immanants of unitary matrices and submatrices}
\label{Sec:SunImmanantResult}
%=================%

%------------------------------%
\subsection{Introduction and basic result}
%------------------------------%
In this section I present result on the connection between immanants and group functions (or $\mathcal{D}$-functions) for the unitary groups.
We extend a result of Kostant~\cite{Kostant1995} to submatrices of the fundamental representation of these groups.

Immanants of totally non-negative and of Hermitian matrices have been studied in~\cite{Pate1992,Pate1994,Rhoades2005};
our results instead are applicable to unitary matrices and depend on the well-known
duality between representations of the unitary and of the symmetric groups~\cite{Weyl1950,Rowe2012}.
This duality identifies some irreps of {$\mathrm{SU}(m)$ with irreps of $S_N$ with $N\le m$.}
If $\{\lambda\}=\{\lambda_1,\lambda_2,\ldots,\lambda_N\}$ is a partition of $\lambda=\sum_k\lambda_k$ labelling an irrep of $S_N$, we {choose to} label irreps of $\mathrm{SU}(m)$ using the round brackets
$(\lambda)$ with {$m-1$} entries defined by {$(\lambda_1-\lambda_2,\ldots,\lambda_{N-1}-\lambda_N, \lambda_N)$ and trailing zeroes omitted.} Thus, the irrep $\{21\}$ of {$S_3$} corresponds to the $\mathrm{SU}(4)$ irrep $(110)\sim(11)$, the $\mathrm{SU}(5)$ irrep $(1100)\sim(11)$ etc.
%
%Group functions (or Wigner $\mathcal{D}$-functions) are defined as the overlap between two basis states of the same irrep of $\mathrm{SU}(m)$, one of which
%has been transformed by an element $\Omega\in \mathfrak{su}{(m)}$ .
%{If $\ket{\psi^{(\lambda)}_\ell}$, $\ket{\psi^{(\tau)}_t}$ are any two
%basis states in irreps $(\lambda)$ and $(\tau)$ respectively and $T^{(\lambda)}(\Omega)$ is the matrix representing element
%$\Omega\in\operatorname{SU}(m)$ in the irrep $(\lambda)$}, then
%\begin{equation}
%\mathcal{D}^{(\lambda)}_{\ell t}(\Omega)\defeq\bra{\psi^{(\lambda)}_\ell}T^{(\lambda)}(\Omega)\ket{\psi^{(\tau)}_t}\delta_{\lambda,\tau}\, .
%\end{equation}

Kostant~\cite{Kostant1995} has shown a simple connection between immanants~(\ref{defineimmanant}) of the fundamental representation
$T$ of $\mathrm{SU}(m)$ group elements and group functions $\mathcal{D}^{(\lambda)} _{tt}$ of $\mathrm{SU}(m)$
with $t$ running over each of the zero-weight states in irrep $(\lambda)$.
Specifically, let $\Omega\in\operatorname{SU}(m)$ and $T(\Omega)$ {(no superscript)} is the defining $m\times m$ representation of $\Omega$.
{Further define the matrix $\mathcal{D}^{(\tau)}(\Omega)$ by}
\begin{equation}
\left(\mathcal{D}^{(\tau)}(\Omega)\right)_{rt}=\mathcal{D}^{(\tau)}_{rt}(\Omega) \label{calD}
\end{equation}
with $r,t$ {restricted to} labelling zero-weight states in the irrep $(\tau)$. Then we have~\cite{Kostant1995}
\begin{equation}
\operatorname{imm}^{\{\tau\}}(T(\Omega))
=\operatorname{Tr}\left[\mathcal{D}^{(\tau)}(\Omega)\right].\label{kostantD}
\end{equation}

For $\mathrm{SU}(2)$, this result simply states that the permanent of the matrix
\begin{equation}
T(\Omega)=\left(
\begin{array}{cc}
 \e^{-(\frac{1}{2}) i (\alpha +\gamma )} \cos \left(\frac{\beta }{2}\right) & -\e^{-(\frac{1}{2}) i (\alpha -\gamma )} \sin \left(\frac{\beta }{2}\right) \\
 \e^{(\frac{1}{2}) i (\alpha -\gamma )} \sin \left(\frac{\beta }{2}\right) & \e^{(\frac{1}{2}) i (\alpha +\gamma )} \cos \left(\frac{\beta }{2}\right) \\
\end{array}
\right)\, ,
\end{equation}
where $\Omega=(\alpha,\beta,\gamma)\in \operatorname{SU}(2)$,
is the $\mathrm{SU}(2)$-function $\operatorname{imm}^{\{2\}}(T(\Omega))=\mathcal{D}^1_{00}(\alpha,\beta,\gamma)=\cos\beta$ whereas
the determinant $\operatorname{imm}^{\{1,1\}}(T(\Omega))=\mathcal{D}^0_{00}(\alpha,\beta,\gamma)=1$.
The trace of Equation~(\ref{kostantD}) contains a single term in both $\mathrm{SU}(2)$ cases as the zero-weight subspaces in irreps $J=1$ and $J=0$ are both one-dimensional.
Here and henceforth we follow the physics notation of
labelling $\mathrm{SU}(2)$ irreps using the angular momentum label $J=\frac{1}{2}\lambda$, such that $2J$ is an integer. Thus, $\mathcal{D}^{1}_{00}(\Omega)$ is an
$\mathrm{SU}(2)$ $\mathcal{D}$-function in the three-dimensional irrep $J=1$.

%------------------------------%
\subsection{Recap of notation and an illustration}
%------------------------------%
Here I recall the relevant notation of boson realization of $\mathrm{SU}(n)$ states and $\mathfrak{su}(n)$ operators (Section~\ref{Sec:BosonRealizations}) and present and example of the connection between $\mathcal{D}$-functions and immanants.
We first introduce a basis for $\mathds{H}^{(1)}_p$, which is the $p$-th copy of the
carrier space for fundamental irrep $\{1\} \equiv (1)$ of $\mathrm{SU}(m)$. I write this basis in terms of harmonic oscillator states according to
\begin{equation}
\mathds{H}^{(1)}_p=\operatorname{span}\{a_k^\dagger(\omega_p)\ket{0}\, , k=1,\ldots,m\}\, .
\end{equation}
The label $\omega_p$ can be thought of as an internal DOF, say the frequency, of the $p$-th oscillator.
We introduce the (reducible) Hilbert space
$\mathds{H}^{(N)} \defeq \mathds{H}^{(1)}_1\otimes \mathds{H}^{(1)}_2\ldots\otimes \mathds{H}^{(1)}_N$, which is spanned by the set of harmonic oscillator states of the type
\begin{equation}
a_k^\dagger(\omega_1)a_r^\dagger(\omega_2)\ldots a_s^\dagger(\omega_N)\ket{0}\, , \quad
k =1,\ldots,m\, ;\quad r=1,\ldots,m\, ,{\emph{etc}}\label{Hpstates}
\end{equation}

Now I detail the action of the permutation group $S_N$ on our basis states.% $\mathds{H}^{(N)}$.
 The action of $P(\sigma)$ is defined as
\begin{align}
P(\sigma)a_k^\dagger(\omega_1)a_r^\dagger(\omega_2)\ldots a_s^\dagger(\omega_N)\ket{0} =a_k^\dagger(\omega_{\sigma^{-1}(1)})a_r^\dagger(\omega_{\sigma^{-1}(2)})\ldots a_s^\dagger(\omega_{\sigma^{-1}(N)})\ket{0}.
\label{leftaction}
\end{align}
Alternatively, one may consider each of the sets $\{a^\dagger_k(\omega_p); k=1,\ldots,m\}$, which are labelled by $p$, as a tensor operator that carries the defining irrep $(1)$ of $\mathfrak{su}(m)$.
$\ket{0}$ is invariant under the action of $S_N$ and $\mathfrak{su}(m)$ elements.

Specifically, algebra u$(m)$ is spanned by the $S_N$-invariant operators
\begin{equation}
\hat C_{ij}=\sum_{k=1}^N a_i^\dagger(\omega_k)a_j (\omega_k)\quad i,j=1,\ldots, m.
\end{equation}
The $\mathfrak{su}(m)$ subalgebra is obtained from u$(m)$ by removing the diagonal operator $\sum_{k=1}^m\hat C_{kk}$,
so the Cartan subalgebra of $\mathrm{su}(m)$ is spanned by the traceless diagonal operators
\begin{equation}
\hat h_i \defeq \hat C_{ii}-\hat C_{i+1,i+1},\quad i=1,\ldots, m-1.
\end{equation}

Furthermore, a basis for the irrep $(\lambda)$ of $\mathfrak{su}(m)$ is given in terms of the harmonic oscillator occupation number $n$ according to
\begin{equation}
\ket{\psi^{(\lambda^{(n)})}n;\Lambda}=
\ket{(\lambda)n_1n_2,\ldots,n_m;(\lambda^{(n-1)})\ldots (\Lambda^{(2)})},
\end{equation}
where $n\defeq(n_1,n_2,\ldots,n_m)$ and $n_k$ indicates the number of excitations in mode $k\le m$.
The weight of this state is equals the $(m-1)$-tuple $[n_1-n_2,n_2-n_3,\ldots, n_{m-1}-n_m]$.
Finally, the multi-index $\Lambda: =(\lambda')\ldots (J)$ refers to a collection of indices, each of which labels irreps in
the subalgebra chain
\begin{equation}
\begin{array}{cccccc}
\mathfrak{su}(m)&\supset& \mathfrak{su}(m-1)&\supset&\ldots &\supset \mathfrak{su}(2)\\
(\lambda^{(n)}) && (\lambda^{(n-1)}) && & (\lambda^{(2)})
\end{array},
\label{subalgebrachain}
\end{equation}
and is needed to fully distinguish states having the same weight. The representation labels are all integers.
We take the subalgebra
$\mathfrak{su}(k-1)\subset \mathfrak{su}(k)$ to be spanned by the subset of the $k\times k$ hermitian traceless matrices of the form
\begin{equation}
\left(\begin{array}{cccc}
0&0&\ldots&0\\
0&*&*&*\\
\vdots &*&*&*\\
0&*&*&*\end{array}\right),
\end{equation}
where $*$ denote possible non-zero entries in $\mathfrak{su}(k-1)$.

As an illustration of the connection between $\mathcal{D}$-function and immanants, consider $n=3$.
The matrix representation of $\Omega\in\operatorname{SU}(3)$ in the fundamental representation
(which is denoted by (1) once the trailing $0$ has been eliminated) of $\mathrm{SU}(3)$ as
\begin{equation}
T(\Omega)=\left(\begin{array}{ccc}
\mathcal{D}^{(1)}_{100(0);100(0)}(\Omega)&\mathcal{D}^{(1)}_{100(0);010(1)}(\Omega)&\mathcal{D}^{(1)}_{100(0);001(1)}(\Omega)\\
\mathcal{D}^{(1)}_{010(1);100(0)}(\Omega)&\mathcal{D}^{(1)}_{010(1);010(1)}(\Omega)&\mathcal{D}^{(1)}_{010(1);001(1)}(\Omega)\\
\mathcal{D}^{(1)}_{001(1);100(0)}(\Omega)&\mathcal{D}^{(1)}_{001(1);010(1)}(\Omega)&\mathcal{D}^{(1)}_{001(1);001(1)}(\Omega)
\end{array}\right)\, ,
\end{equation}
with $\Omega\in \operatorname{SU}(3)$. The result of Kostant~\cite{Kostant1995} applied to $\mathfrak{su}{(3)}$ then states that
\begin{align}
\operatorname{per}(T(\Omega))=\operatorname{imm}^{\{3\}}(T(\Omega))=&\mathcal{D}^{(3)}_{111(1);111(1)}(\Omega)\, ,\nonumber \\
\operatorname{imm}^{\{21\}}(T(\Omega))=&\mathcal{D}^{(11)}_{111(1);111(1)}(\Omega)+\mathcal{D}^{(11)}_{111(0);111(0)}(\Omega)\, ,\\
\operatorname{det}(T(\Omega))=\operatorname{imm}^{\{111\}}(T(\Omega))=&\mathcal{D}^{(0)}_{000(0);000(0)}(\Omega)=1\, ,\nonumber
\end{align}
For convenience, we use the symbols $T$ and $\Omega$ to respectively denote matrices and elements in different $\mathrm{SU}(m)$ without
indicating $m$; this does not affect our conclusions as our results apply to any $m$.

Our novel contribution is to extend result of~\cite{Kostant1995} encapsulated in Equation~(\ref{kostantD})
to submatrices of the fundamental representations.
Our results enable application of $\mathrm{SU}(n)$ methods to the $m$-photon $n$-channel interferometry for $n>m$.
I present a proof of our result in the next section.

%------------------------------%
\subsection{Proving the theorem: {the case $N=m$}}
%------------------------------%

We consider the state $\ket{\Psi_{123\ldots m}}=a_1^\dagger(\omega_1)a_2^\dagger(\omega_2)\ldots a_m^\dagger(\omega_m)\ket{0}$, which lives in the (reducible) tensor product space $\mathds{H}^{(m)}=\mathds{H}^{(1)}_1\otimes \mathds{H}^{(1)}_2\ldots\otimes \mathds{H}^{(1)}_m$.
The first lemma deals with the weight of this state.

\begin{lem}\label{Lem:SUmWeightZero}
The $\mathrm{SU}(m)$ weight of $\ket{\Psi_{123\ldots m}}$ is 0.
 This is immediate since every mode is occupied once, so $n_i=1 \forall\, i$. Since the component $k$ of the
 weight is $n_k-n_{k+1}$, $\hat h_i\ket{\Psi_{123\ldots m}}=0\ \forall\, i$.\hfill $\blacksquare$
\end{lem}

From Lemma~\ref{Lem:SUmWeightZero}, we write $\ket{\Psi_{123\ldots m}}$ as an expansion over zero-weight states in all irrep occurring in $\mathds{H}^{(m)}$ according to
\begin{equation}
\ket{\Psi_{123\ldots m}}=\sum_{\alpha\lambda\ell}\tilde c^{(\lambda)_\alpha}_{\ell} \ket{\psi^{(\lambda)_\alpha}_{\ell}}\, ,\qquad
\tilde c^{(\lambda)_\alpha}_{\ell}=\left\langle\psi^{(\lambda)_\alpha}_{\ell}\big\vert\Psi_{123\ldots m}\right\rangle\, ,
\end{equation}
where $(\lambda)_\alpha$ is the $\alpha$-th copy of the irrep {$(\lambda)_\alpha$ of $\mathrm{SU}(m)$,}
and $\ell$ labels those basis states that have $0$-weight in the irrep $(\lambda)_\alpha$ of $\mathrm{SU}(m)$.

\begin{lem}
With the notation above:
\begin{equation}
\sum_{\sigma}\chi^{\{\tau\}}(\sigma)P(\sigma)\ket{\Psi_{123\ldots m}}=\frac{m!}{\dim(\tau)}
\sum_{\alpha t}\tilde c^{(\tau)_\alpha}_{t}\, \ket{\psi^{(\tau)_\alpha}_{t}}.
\end{equation}
\end{lem}
\begin{proof}
The proof of Lemma~2 relies on the duality between representations of the symmetric and the unitary groups.
From duality, the basis states $\{\ket{\psi^{(\tau)_\alpha}_{t}}\}$ are also basis states for the irrep $\{\tau\}$ of $S_m$. Hence, using
 Equation~(\ref{leftaction}) we obtain
\begin{align}
P(\sigma)\ket{\Psi_{123\ldots m}}=&
\sum_{\alpha\lambda\ell}
\ket{\psi^{(\lambda)_\alpha}_{\ell}}\bra{\psi^{(\lambda)_\alpha}_{\ell}}P(\sigma)\ket{\Psi_{123\ldots m}}\, ,\\
=&\sum_{\alpha\ell\lambda k} \ket{\psi^{(\lambda)_\alpha}_{\ell}}
\Gamma^{\{\lambda\}}_{\ell k}(\sigma)
\left\langle{\psi^{(\lambda)_\alpha}_{k}}|{\Psi_{123\ldots m}}\right\rangle\, ,\\
=&\sum_{\alpha\ell\lambda k} \ket{\psi^{(\lambda)_\alpha}_{\ell}}
\Gamma^{\{\lambda\}}_{k\ell}(\sigma^{-1})\,\tilde c^{(\lambda)_\alpha}_{ k},
\end{align}
where $\Gamma^{\{\lambda\}}$ is the unitary irrep $\{\lambda\}$ of $S_m$.
Writing $\chi^{\{\tau\}}(\sigma)=\sum_{t}\Gamma^{\{\tau\}}_{tt}(\sigma)$ gives us
\begin{align}
\sum_{\sigma}\chi^{\{\tau\}}(\sigma)P(\sigma)\ket{\Psi_{123\ldots m}}=&
\sum_{\alpha k\lambda\ell}\tilde c^{(\lambda)_\alpha}_{\ell} \ket{\psi_\ell^{(\lambda)_\alpha}}
\left[\sum_{\sigma t}\Gamma^{\{\tau\}}_{tt}(\sigma)\Gamma^{\{\lambda\}}_{\ell k}(\sigma^{-1})\right],\\
=&\frac{m!}{\dim(\tau)}\sum_{\alpha t}\tilde c^{(\tau)_\alpha}_{t} \ket{\psi^{(\tau)_\alpha}_{t}}\, ,\label{Eq:mFactDimensions}
\end{align}
where we have used the orthogonality of characters to arrive at Equation~(\ref{Eq:mFactDimensions}).

Because the action of $\Omega\in\operatorname{SU}(m)$ commutes with the action of
$\sigma\in S_m$, we have
\begin{align}
\operatorname{imm}^{\{\tau\}}(T(\Omega))=&\sum_{\sigma}\,\chi^{\{\tau\}}(\sigma)P(\sigma)
\left[T_{11}(\Omega)T_{22}(\Omega)\ldots T_{mm}(\Omega)\right]\, ,\\
=&\bra{\Psi_{123\ldots m}}{\left[T(\Omega)\otimes T(\Omega)\ldots\otimes T(\Omega)\right]} \nonumber \\
&\qquad\qquad\left[\sum_{\sigma}\,\chi^{\{\tau\}}(\sigma)P(\sigma)\right] \ket{\Psi_{123\ldots m}}\, ,\\
=&\sum_{\alpha rt} (\tilde c^{\tau_\alpha}_{r})^*\,\tilde c^{\tau_\alpha}_{t}\,\frac{m!}{\operatorname{dim}(\tau)}\mathcal{D}^{(\tau)}_{rt}(\Omega)\, .
\end{align}
Introducing the scaled coefficients
$c^{(\tau)_\alpha}_{t}= \tilde c^{(\tau)_\alpha}_{t} \sqrt{\displaystyle\frac{m!}{\operatorname{dim}(\tau)}}$,
we finally obtain
\begin{equation}
\operatorname{imm}^{\{\tau\}}(T(\Omega))=\sum_{rt} \left[\sum_\alpha (c^{(\tau)_\alpha }_r)^*\, c^{(\tau)_\alpha}_{t}\right]\,\mathcal{D}^{(\tau)}_{rt}(\Omega)\, ,
\label{mixedsum}
\end{equation}
where the sums over $t$ and $r$ is a sum over zero-weight states in $(\tau)_\alpha.$
\end{proof}

This result is not unexpected as the operator
\begin{equation}
\hat \Pi^{\{\tau\}}=\left[\sum_{\sigma}\chi^{\{\tau\}}(\sigma)P(\sigma)\right] \, ,\qquad \sigma\in S_m\label{immanantoperator}
\end{equation}
is a projector to that subspace of $S_m$ which has permutation symmetry $\{\tau\}$, and
hence (by duality) is a projector to a subspace that carries (possibly multiple copies of) the irrep $(\tau)$ of SU($m$)
in the $m$-fold product $(1)^{\otimes m}$.
\begin{thm}
(Kostant~\cite{Kostant1995})
\begin{equation}
\operatorname{imm}^{\{\tau\}}(T(\Omega))=\sum_{t}{\mathcal{D}}^{(\tau)}_{tt}(\Omega)
\end{equation}
\end{thm}
\begin{proof} We present a proof that will {eventually} allow us to dispense with the requirements that
{$N=m$ and that} states have zero-weight. Construct the matrix
\begin{equation}
W^{\{\tau\}}_{rt}= \sum_{\alpha} c^{(\tau)}_{\alpha t}\, (c^{(\tau)}_{\alpha r})^*\, .
\end{equation}
 Equation~(\ref{mixedsum}) then becomes
\begin{equation}
\operatorname{imm}^{\{\tau\}}(T(\Omega))=\sum_{rt} W^{\{\tau\}}_{rt}\,\mathcal{D}^{(\tau)}_{rt}(\Omega)=\operatorname{Tr}\left[W^{\{\tau\}}\mathcal{D}^{(\tau)}(\Omega)\right]\, ,
\label{tracesum}
\end{equation}
with $\mathcal{D}^{(\tau)}(\Omega)$ defined in Equation~(\ref{calD}). Our objective is to prove that $W^{\{\tau\}}$ is the unit matrix.

Any immanant has the property of invariance
under conjugation by elements {in $S_m$} {\emph{i.e.}}, the immanant of any matrix satisfies
\begin{align}
\hspace{-1cm}\operatorname{imm}^{\{\tau\}}(T(\Omega))&=\sum_\sigma \chi^{\{\tau\}}(\sigma) P(\sigma) \,\left[T_{11}(\Omega)T_{22}(\Omega)\ldots T_{mm}(\Omega)\right]
\, \nonumber \\
\qquad &=\sum_\sigma \chi^{\{\tau\}}(\sigma) P^{-1}(\bar\sigma)P(\sigma) P(\bar\sigma)\,\left[T_{11}(\Omega)T_{22}(\Omega)\ldots T_{mm}(\Omega)\right]
\, ,
\end{align}
with $\sigma,\bar\sigma \in S_m$. Under conjugation by $\bar\sigma$, Equation~(\ref{mixedsum}) becomes
\begin{equation}
\begin{array}{rcl}
\operatorname{imm}^{\{\tau\}}(T(\Omega))
&=& {\operatorname{Tr}\left[\Gamma^{\{\tau\}}(\bar \sigma) W^{\{\tau\}} \Gamma^{\{\tau\}}(\bar \sigma^{-1})\mathcal{D}^{(\tau)}(\Omega)\right]}\, , \\
&=&{\operatorname{Tr}\left[W^{\{\tau\}} \mathcal{D}^{(\tau)}(\Omega)\right] }\, .
\end{array}\, .
\end{equation}
{Since $\mathcal{D}^{(\tau)}(\Omega)$ is certainly \emph{not} the unit matrix for arbitrary $\Omega$}, it follows that
\begin{equation}
{\Gamma^{\{\tau\}}(\bar \sigma) W^{\{\tau\}} \Gamma^{\{\tau\}}(\bar \sigma^{-1})= W^{\{\tau\}},}
\end{equation}
{\emph{i.e.}}~the matrix $W^{\{\tau\}}$ is invariant under any permutation. By Schur's lemma
$W^{\{\tau\}}$ must therefore be proportional to the unit matrix, {\emph{i.e.}}~we have
$W^{\{\tau\}}_{ts}=\xi\, \delta_{ts}$ with $\xi $ the relevant constant of proportionality.
The immanant thus takes the form
\begin{equation}
\operatorname{imm}^{\{\tau\}}(T(\Omega))=\xi\left(\displaystyle\sum_{t}\mathcal{D}^{(\tau)}_{tt}(\Omega)\right)\, .
\end{equation}
To determine $\xi$, choose $\Omega=\mathds{1}$. Then $T(\mathds{1})$ is the $m\times m$ unit matrix, and
$T_{k,\sigma(k)}(\mathds{1})$ is zero unless $\sigma=\mathds{1}\in S_m$. The immanant for $\Omega=\mathds{1}$ is then just the dimension
of the irrep $\{\tau\}$ and we have
\begin{equation}
\operatorname{imm}^{\{\tau\}}(T(\mathds{1}))=\chi^{\{\tau\}}(\mathds{1})=
\operatorname{dim}(\tau)=\xi,\quad\left(\sum_{t} 1\right)=\xi\,\operatorname{dim}(\{\tau\})
\end{equation}
since $\mathcal{D}^{(\tau)}_{tt}(\mathds{1})=1$. Hence, $\xi=1$ and the theorem is proved.
\end{proof}
This completes our results on $\mathcal{D}$-functions and immanants of the fundamental matrix representation.
The next section generalizes these results to submatrices of the fundamental representation.

%------------------------------%
\subsection{Results on submatrices: {the case $N < m$.}}
%------------------------------%
We now consider the submatrices of $T$. In multi-photon interferometry, such submatrices describe the unitary scattering from an input state of the form
\begin{equation}
\ket{\Psi_{k_1\ldots k_p}}=a^\dagger_{k_1}(\omega_1)a^\dagger_{k_2}(\omega_2)\ldots a^\dagger_{k_p}(\omega_p)\ket{0}\, ,
\qquad p< m\, , \label{substate}
\end{equation}
to an output state $\ket{\Psi_{\ell_1\ldots \ell_p}}$, which need not the identical to $\ket{\Psi_{k_1\ldots k_p}}$.
{Both input and output
live in the reducible Hilbert space $\mathds{H}^{(p)}$ and have expansions of the form
\begin{equation}
\ket{\Psi_{k_1\ldots k_p}}=\sum_{\alpha \lambda \ell} \tilde {\mathcal{D}}^{(\lambda)_{\alpha}}_\ell \ket{\psi^{(\lambda)_{\alpha}}_\ell}
\label{expandgeneral}
\end{equation}
where $\ket{\psi^{(\lambda)_{\alpha}}_\ell}$ has weight $[k_1-k_2,k_2-k_3,\ldots,k_{p-1}-k_p]$.}

First we select from $T(\Omega)$ a principal submatrix $\bar T(\Omega)_k$, {\emph{i.e.},}~$\bar T(\Omega)_k$ is obtained by keeping rows and columns $k=(k_1,k_2,\ldots,k_p)$
with $p< m$. In such a case, input and output states are identical.
The permutation group $S_p$ shuffles the $p$ indices $k_1,k_2,\ldots,k_p$ amongst themselves.
Although the submatrix
$\bar T(\Omega)_k$ is not unitary, the proof of Theorem 3 does not depend on the unitarity of
$T(\Omega)$ and so can be copied to show

\begin{thm}
\label{Thm:MainThm}
The immanant
$\operatorname{imm}^{\{\lambda\}}_k(T(\Omega))$ of a submatrix $\bar T(\Omega)_k$, which is a principal submatrix of $T$, is given by
\begin{equation}
\operatorname{imm}^{\{\lambda\}}_k(T(\Omega))=\sum_r \mathcal{D}^{(\lambda)}_{rr}(\Omega) \label{corollary1}
\end{equation}
where $(\lambda)$ is the irrep of $\mathrm{SU}(m)$ corresponding to the partition
$\{\lambda\}$ and where the sum over $r$ is a sum over all the states in $(\lambda)$ with weight $[k_1-k_2,k_2-k_3,\ldots,k_{p-1}-k_p]$;
{following Equation(\ref{expandgeneral}) this is the weight of $\ket{\Psi_{k_1\ldots k_p}}$ in Equation~(\ref{substate})}
and need not be zero.
\end{thm}

As an illustration, if the third and fifth rows and columns are removed from the $5\times 5$ fundamental matrix representation of $\mathrm{SU}(5)$, then
the states {entering} in the sum of Equation~(\ref{corollary1}) are {linear combinations of terms of} the form
\begin{equation}
P(\sigma)\left[a^\dagger_{1}(\omega_1)a^\dagger_{2}(\omega_2) a^\dagger_{4}(\omega_3)\ket{0}\right]
\end{equation}
with weight $[0,1,-1,1]$.
Using the $\mathfrak{su}(k)\downarrow\mathfrak{su}(k-1)$ branching rules~\cite{Slansky1981,Whippman1965} to label basis states,
the $\{2,1\}$ immanant of this submatrix is the sum
\begin{equation}
\operatorname{imm}^{\{2,1\}}_{124}(T(\Omega))=\mathcal{D}^{(1,1)}_{11010(2)(1)(1);11010(2)(1)(1)}(\Omega) +\mathcal{D}^{(1,1)}_{11010(0,1)(1)(1);11010(0,1)(1)(1)}(\Omega)\, ,
\end{equation}
where the labels $(2)(1)(1) $ and $(0,1)(1)(1)$ refer to the $\mathfrak{su}(4)\supset\mathfrak{su}(3)\supset\mathfrak{su}(2)$ chains of irreps
(recall that trailing 0s are omitted).

Now we consider our results applied to the most general submatrices of a matrix.
To fix ideas, we start with the $4\times 4$ matrix $T$ and remove row $1$ and column $2$ to obtain the submatrix~$\bar T$:
\begin{equation}
T(\Omega)\to \bar T(\Omega)=
\left(
\begin{array}{cccc}
T_{21}(\Omega)&T_{23}(\Omega)&T_{24}(\Omega)\\
T_{31}(\Omega)&T_{33}(\Omega)&T_{34}(\Omega)\\
T_{41}(\Omega)&T_{43}(\Omega)&T_{44}(\Omega)
\end{array}
\right)
\label{su4submatrix}
\end{equation}
The immanants of $3\times 3$ submatrix $\bar T(\Omega)$ are in the form
\begin{align}
\operatorname{imm}^{\{\lambda\}}(\bar T(\Omega))=&\sum_{\sigma}\chi^{\{\lambda\}}(\sigma)P(\sigma)\left[T_{11}(\Omega)T_{22}(\Omega)T_{34}(\Omega)\right]\, \\
=&\hat \Pi^{\{\lambda\}} \left[T_{11}(\Omega)T_{22}(\Omega)T_{34}(\Omega)\right]\, ,
\end{align}
where $\hat \Pi^{\{\lambda\}}$ is the immanant projector of Equation~(\ref{immanantoperator}) and
$\sigma$ permutes the triple $(124)$.

Let $\{a_k^\dagger(\omega_k)\ket{0},k=1,\ldots,4\}$ be a basis for the fundamental irrep of SU($4$), and define
\begin{align}
\ket{\Psi_{134}}\defeq&a_1^\dagger(\omega_1)a_3^\dagger(\omega_2)a_4^\dagger(\omega_3)\ket{0}\, ,\\
\ket{\Phi_{234}}\defeq&a_2^\dagger(\omega_1)a_3^\dagger(\omega_2)a_4^\dagger(\omega_3)\ket{0}\,
\end{align}
as three-particle states elements of $\mathds{H}^{\{1\}\otimes\{1\}\otimes\{1\}}$.
Clearly there is $\sigma'\in S_4$ such that
\begin{equation}
\ket{\Psi_{234}}=P(\sigma')\ket{\Phi_{134}}\, .
\end{equation}
Indeed by inspection this element is given by $P(\sigma')=P_{12}$. More generally, if
\begin{align}
\ket{\Phi_k}=&a_{k_1}^\dagger(\omega)a_{k_2}^\dagger(\omega_2)a_{k_3}^\dagger(\omega_3)\ket{0}\,,k=(k_1,k_2,k_3)\, ,\\
\ket{\Psi_q}=&a_{q_1}^\dagger(\omega)a_{q_2}^\dagger(\omega_2)a_{q_3}^\dagger(\omega_3)\ket{0}\, ,q=(q_1,q_2,q_3)\, ,
\end{align}
then there is $\sigma_{qk}$ exists such that $\ket{\Psi_{q}}=P(\sigma_{qk})\ket{\Phi_{k}}$.
As the action of the permutation group commutes with the action of the unitary group:
\begin{align}
\operatorname{imm}^{\{\lambda\}}(\bar T(\Omega))_{kq}=&\bra{\Phi_{k}}\,\left[T(\Omega)\otimes T(\Omega)\ldots\otimes T(\Omega)\right] \,\hat \Pi^{\{\lambda\}}\,P(\sigma_{qk})\ket{\Phi_{k}}\,\\
=&\bra{\Phi_{k}}\,\hat \Pi^{\{\lambda\}}\,\left[T(\Omega)\otimes T(\Omega)\ldots\otimes T(\Omega)\right]P(\sigma_{qk})\ket{\Phi_{k}}\, \\
=&\sum_{rs\alpha}\bra{\Phi_{k}}\hat \Pi^{\{\lambda\}}\ket{\psi^{(\lambda)_{\alpha}}_r} \times \bra{\psi^{(\lambda)_{\alpha}}_r}T^{(\lambda)}(\Omega)P(\sigma_{qk})
\ket{\psi^{(\lambda)_\alpha}_s}\left\langle{\psi^{(\lambda)_\alpha}_s}|{\Phi_{k}}\right\rangle \label{immanantnondiagonal}
\end{align}
Now, the permutation $P(\sigma_{qk})$ is represented by a unitary matrix in the carrier space $(\lambda)_\alpha$.
Thus, there exist $\Omega_{qk}\in \mathfrak{su}{(4)}$ and a phase $\zeta$ such that
$P(\sigma_{qk})\ket{\psi^{(\lambda)_\alpha}_s}=e^{i\zeta}\,T(\Omega_{qk})\ket{\psi^{(\lambda)_\alpha}_s}$.
This transforms our original problem back to the case of principal submatrices, but with now an element
$\Omega\cdot \Omega_{qk}$ {\emph{i.e.}},
\begin{equation}
\operatorname{imm}^{\{\lambda\}}(\bar T(\Omega))=\sum_{t}\mathcal{D}^{(\lambda)}_{tt}(\Omega\cdot \Omega_{qk})\, .
\end{equation}
Unfortunately, the action $P(\sigma_{qk})\ket{\psi^{(\lambda)_\alpha}_s}$ is in general non-trivial~\cite{Kramer1966,Moshinsky1968,Rowe1999}~and it is not obvious
how to find $\Omega'$, much less $(\Omega\cdot \Omega')$. Nevertheless, we found that the sum of $\mathcal{D}$-functions that occur on the right hand side of
 Equation~(\ref{immanantnondiagonal}) always contains the same number of $D$ as the dimension of the dual irrep $\{\tau\}$, and that the coefficients of these
$D$'s is always one. This result relied on (i)~evaluating the appropriate group functions using the algorithm~\cite{Dhand2015d}, (ii) explicitly constructing each of
the immanants of all possible $4\times 4$ submatrices and of all possible $3\times 3$ submatrices of the fundamental irrep of $\mathfrak{su}{(5)}$ and (iii) explicitly
constructing the immanants of $3\times 3$ submatrices of the fundamental irrep of $\mathfrak{su}{(4)}$ or $\mathfrak{su}{(5)}$.

Thus, in the specific case of the submatrix given in Equation~(\ref{su4submatrix}), we have
\begin{equation}
\hspace{-1.2cm}\operatorname{imm}^{\{21\}}(\bar T(\Omega))_{(234)(134)}=\mathcal{D}^{(1,1)}_{0111(2)(1);1011(2)(1)}(\Omega)
+ \mathcal{D}^{(1,1)}_{0111(11)(1);1011(11)(1)}(\Omega)\, .
\label{Eq:IshVerified}
\end{equation}
We also verified that a similar identity holds for all $3\times 3$ submatrices of $T(\Omega)\in\mathfrak{su}{(4)}$. For instance,
\begin{align}
\operatorname{imm}^{\{21\}}(\bar T(\Omega))_{(234)(124)}=&\mathcal{D}^{(1,1)}_{0111,(2)(1);1101(2)(1)}(\Omega) + \mathcal{D}^{(1,1)}_{0111(11)(1);1101(01)(1)}(\Omega)\, ,\\
\operatorname{imm}^{\{21\}}(\bar T(\Omega))_{(134)(124)}=& \mathcal{D}^{(1,1)}_{1011(2)(1);1110(2)(1)}(\Omega) + \mathcal{D}^{(1,1)}_{1011(11)(1);1110(01)(1)}(\Omega)\, .
\end{align}
Likewise, we have, for $T(\Omega)\in \mathfrak{su}(5)$,
\begin{align}
\operatorname{imm}^{\{21\}}(\bar T(\Omega))_{(1345)(1235)}\nonumber =&\mathcal{D}^{(1,1)}_{01101(11)(2)(1),10110(2)(2)(1)}(\Omega)
+\mathcal{D}^{(1,1)}_{01101(11)(01)(1),10110(01)(01)(1)}(\Omega)\, ,\\
\operatorname{imm}^{\{31\}}(\bar T(\Omega))_{(235)(134)}=&\mathcal{D}^{(2,1)}_{10111(3)(3)(1),11101(3)(2)(1)}(\Omega)+\mathcal{D}^{2,1}_{10111(11)(11)(1),11101(11)(2)(1)}(\Omega)\nonumber\\
&\quad+\mathcal{D}^{(2,1)}_{10111(11)(11)(0),11101(11)(01)(1)}(\Omega)\, ,
\label{Eq:IshVerifiedLast}
\end{align}
this last being an example of a $4\times 4$ submatrix not principal coaxial.
We thus conjecture that, even for generic submatrices, $\operatorname{imm}^{\{\lambda\}}(\bar T(\Omega))_{kq}$
is a sum of $\operatorname{dim}{\{\lambda\}}$ distinct $D$'s with coefficients equal to $+1$.

\subsection{An application: Relations between $\mathcal{D}$-functions}
Here we present a relation between $\mathcal{D}$-functions of $\mathrm{SU}(3)$ and those of $\mathrm{SU}(4)$.
This relation is obtained using Theorem~\ref{Thm:MainThm} and a theorem due to Littlewood~\cite{Littlewood1950}.

Littlewood~\cite{Littlewood1950} has established relations between immanants of a matrix and sums of products of immanants of
principal coaxial submatrices.
For instance, {the equality for Schur functions $\{3\}\{1\}=\{3,1\}+\{4\}$ yields the immanant relation}
\begin{equation}
\sum_{ijk\ell} \left(\operatorname{imm}^{\{3\}}_{ijk}(T(\Omega))\right)
\left(\operatorname{imm}^{\{1\}}_{\ell}(T(\Omega))\right)\nonumber =\operatorname{imm}^{\{3,1\}}(T(\Omega))+\operatorname{imm}^{\{4\}}(T(\Omega))
\end{equation}
where the sum over $ijk\ell$ is a sum over complementary coaxial submatrices, i.e.
\begin{equation}
\begin{array}{cc||cc}
ijk&\ell &ijk&\ell \\
\hline
123&4&124&3\\
134&2&234&1\\
\end{array}.
\end{equation}
This expands to a sum of products of immanants of submatrices given explicitly by
\begin{align}
&\operatorname{imm}^{\{3\}}_{123}(T(\Omega))\operatorname{imm}^{\{1\}}_{4}(T(\Omega))
+ \operatorname{imm}^{\{3\}}_{124}(T(\Omega)\operatorname{imm}^{\{1\}}_{3}(T(\Omega))\nonumber \\
& + \operatorname{imm}^{\{3\}}_{134}(T(\Omega))\operatorname{imm}^{\{1\}}_{2}(T(\Omega))
+\operatorname{imm}^{\{3\}}_{234}(T(\Omega))\operatorname{imm}^{\{1\}}_{1}(T(\Omega))\nonumber \\
& =\operatorname{imm}^{\{3,1\}}(T(\Omega))+\operatorname{imm}^{\{4\}}(T(\Omega)),
\end{align}
which becomes an equality on the corresponding products of sum of $\mathrm{SU}(4)$ $\mathcal{D}$-functions:
\begin{align}
&\mathcal{D}^{(3)}_{1110(2)(1);1110(2)(1)}(\Omega)\mathcal{D}^{(1)}_{0001(1)(1);0001(1)(1)}(\Omega)\nonumber \\
&\quad +\mathcal{D}^{(3)}_{1101(2)(1);1101(2)(1)}(\Omega)\mathcal{D}^{(1)}_{0010(1)(1);0010(1)(1)}(\Omega)\nonumber \\
&\quad +\mathcal{D}^{(3)}_{1011(2)(2);1011(2)(2)}(\Omega)\mathcal{D}^{(1)}_{0100(1)(0);0100(1)(0)}(\Omega)\nonumber \\
&\quad +\mathcal{D}^{(3)}_{0111(3)(2);0111(3)(2)}(\Omega)\mathcal{D}^{(1)}_{1000(0)(0);1000(0)(0)}(\Omega)\nonumber \\
&=\mathcal{D}^{(21)}_{1111(3)(2);1111(3)(2)}(\Omega)+\mathcal{D}^{(21)}_{1111(11)(2);1111(11)(2)}(\Omega) \nonumber \\
&\quad+\mathcal{D}^{(21)}_{1111(11)(0);1111(11)(0)}(\Omega) +\mathcal{D}^{(4)}_{1111(3)(2);1111(3)(2)}(\Omega).
\end{align}
The subgroup labels are obtained by systematically using the $\mathfrak{su}(k)\downarrow\mathfrak{su}(k-1)$ branching rules~\cite{Dhand2015d}.

\subsection{Conclusion}

%The relation between characters of $S_N$ and U$(n)$ is well known and is a rich source of results in mathematical physics.
%Our work expands these beyond characters to novel connections between immanants and the group functions proper.
%Results on group functions are comparatively less common than those available for, say, the calculation of generator matrix elements or Clebsch-Gordan coefficients~\cite{Rowe1999,Louck1970,Dhand2015d}.
%We hope some of the results given here might be useful in providing impetus or remedy to this relative paucity of results on group functions.

Immanants are connected to the interferometry of partially distinguishable pulses~\cite{Guise2014,Tan2013,Tillmann2015};
the associated permutation symmetries lead to novel interpretations of
immanants as a type of normal coordinates describing lossless passive interferometers \cite{Tillmann2015}.
This connection immediately provides a physical interpretation to the
appropriate combinations of group functions corresponding to these immanants and should stimulate further development of toolkits to compute group functions.

Conjectures in complexity theory regarding the behaviour of permanents of large unitary matrices may also provide an entry point towards understanding the behaviour of
$\mathcal{D}$-functions in similar asymptotic regimes.
It remains to see if this line of thought can also be turned around it might be possible to use results on the asymptotic
behaviour of $\mathcal{D}$-functions to establish some conjectures on the behaviour of immanants of large matrices.

Finally, although the Schur-Weyl duality is not directly applicable to subgroups of the unitary groups, the permutation group retains its deep connection with representations of the classical groups, which are considered as subgroups of the unitary groups~\cite{Butler1969,Wybourne1974}.
Hence, it might be possible to extend the results of this section to
functions of the orthogonal or symplectic groups, thus generalizing the result of Section~5 on immanants associated with plethysms of representations.

%------------------------------%
\section{Conclusion}
%------------------------------%
In summary, I have advanced group theoretic methods for linear optics by means of algorithms for computing $\mathrm{SU}(n)$ $\mathcal{D}$-functions and by finding relations between $\mathcal{D}$-functions and immanants of the interferometer transformation.
Our algorithm for $\mathcal{D}$-functions enables the expression of interferometry outputs in terms of $\mathrm{SU}(n)$ $\mathcal{D}$-functions similar to the three-photon case~(\ref{eq:23symmetry}).
Furthermore, my results on the connection between $\mathcal{D}$-functions and immanants allow for the computation of interferometer outputs using immanants similar to the case of three-photons outputs~(\ref{eq:ABC}).
Thus, I contribute to group-theoretic methods for analyzing and simulating multi-photon multi-channel interferometry along the lines of the three-photon three-channel treatment of Section~\ref{Sec:ThreePhotons}.

%=================%
\chapter{Summary}
\label{Chap:Summary}
%=================%
This chapter summarizes my contribution to the theory of design, characterization and simulation of multi-photon multi-channel interferometry (Section~\ref{Sec:Summary}).
I conclude the chapter and the thesis with a list of open problems related to the contribution reported herein (Section~\ref{Sec:OpenProbs}).

%------------------------------%
\section{Summary of results}
\label{Sec:Summary}
%------------------------------%
In summary, I contribute to the theory of design, characterization and simulation of multi-photon multi-channel interferometry.
The advances that I have reported in this thesis contribute to making linear optics a viable candidate for QIP.

In the design of linear optics, we devised a procedure that enables the realization of arbitrary discrete unitary transformations on the spatial and internal degrees of freedom of light.
Our procedure receives as input the dimensions $n_{s}$ and $n_{p}$ of spatial and internal DOFs respectively and an $n_{s}n_{p}\times n_{s}n_{p}$ unitary matrix and yields as output a sequence of matrices that correspond to either $2n_{p}\times 2n_{p}$ beam-splitter transformations or to arbitrary $n_{p}\times n_{p}$ internal transformations.
By exploiting the $n_{p}$-dimensional internal DOF, the required number of beam splitters is reduced by a factor of $n_{p}^{2}/2$ as compared to realizing the same transformation on spatial modes alone.
Our procedure thus enables realizing larger unitary transformations that are required for implementing QIP tasks.

My contribution to the characterization of multi-channel interferometers includes an accurate and precise procedure that uses one- and two-photon interference.
Our procedure is advantageous to the existing procedures as it accounts for and corrects systematic errors due to spatiotemporal mode mismatch in the source field.
Our procedure employs experimentally measured source spectra to achieve accurate curve-fitting between measured and theoretical coincidence counts and thus yields accurate transformation matrix elements.
We use maximum-likelihood estimation to find the unitary matrix that best represents the measured data.
A scattershot approach is recommended for reducing the required characterization time.
A bootstrapping procedure is introduced to obtain meaningful error bars on the characterized parameters even when the form of experimental error is unknown.
The efficacy of the characterization procedure is verified numerically and experimentally.

I advance the theory of simulation of multi-photon multi-channel interferometry by developing $\mathrm{SU}(n)$ group-theoretic methods for simulation of linear optics.
I devise an algorithm for computing boson realization of canonical $\mathrm{SU}(n)$ basis states and to compute $\mathrm{SU}(n)$ $\mathcal{D}$-functions.
I find relations between these $\mathcal{D}$-functions and immanants of the matrices and submatrices of the fundamental $\mathrm{SU}(n)$ representations.
These results open the possibility speeding up the computation of arbitrary multi-photon multi-channel measurement probabilities.
%------------------------------%
\section{Open problems}
\label{Sec:OpenProbs}
%------------------------------%
Here I list the problems that have been opened by the advances reported in this thesis.
Section~\ref{Sec:OpenDesign} presents two problems that deal with improvements in the linear optical realization of discrete unitary transformations.
A thorough experimental verification of our characterization procedure is recommended in Section~\ref{Sec:OpenCharacterization}.
In Section~\ref{Sec:OpenSimulation}, I suggest open problems regarding analyzing and improving the speed and the complexity of our $\mathrm{SU}(n)$ methods and applying them to optimally compute multi-photon multi-channel measurement probabilities.

%------------------------------%
\subsection{Improved realization of linear optics on multiple degrees of freedom}
\label{Sec:OpenDesign}
%------------------------------%
Our design procedure (Chapter~\ref{Chap:Design}) enables realizing unitary matrices on the combined state of light in the spatial and one internal DOF but there is no known procedure to transform the composite state of light in more than one internal DOF.
The key challenge to realizing such a transformation is the realization of beam-splitter-like transformations that mix light among different DOFs.
This challenge can be overcome on a case-by-case basis for different internal DOFs.

Another direction related to the design of interferometers is to devise a realization that requires the minimum number of beam splitters and internal elements.
We conjecture that our decomposition is optimal in its beam-splitter-number requirement.
However, we think that other decompositions might reduce the requirement of optical elements acting on internal modes and experimental implementations would gain from such a decomposition.

%------------------------------%
\subsection{Experimental evidence for efficacy of characterization procedure}
\label{Sec:OpenCharacterization}
%------------------------------%
Chapter~\ref{Chap:Verification} presents experimental evidence for the efficacy of our procedure for beam splitter characterization $(m=2)$.
A verification of the accuracy and precision of our procedure on bigger $(m>2)$ interferometers is appealing.
A comparison of the accuracy and precision of our procedure with respect to classical-light procedures~\cite{Keshari2013} would also experimentally relevant.

%------------------------------%
\subsection{Group-theoretic methods for simulation of linear optics}
\label{Sec:OpenSimulation}
%------------------------------%

One open problem related to our boson-realization algorithm (Section~\ref{Sec:SunAlgorithms}) is to devise algorithm that computes only a specific $\mathrm{SU}(n)$ state rather than the entire set of $\mathrm{SU}(n)$ states of an irrep.
Such an algorithm is expected to be faster than our algorithm.
One approach to constructing a specific $\mathrm{SU}(n)$ state to construct the HWS of a given $\mathrm{SU}(n)$ irrep and systematically lower from this $\mathrm{SU}(n)$ HWS via the correct $\mathrm{SU}(m)$ HWSs for $2<m<n$ to the given $\mathrm{SU}(n)$ weight.
A faster algorithm to construct specific basic states would also enable a faster $\mathcal{D}$-function computation algorithm.

Finally, the problem of exploiting our group-theoretic methods to speedup the computation the outputs of multi-photon multi-channel interferometry remains open.
Faster classical algorithms for simulating linear optics would make feasible the benchmarking and simulations of this candidate system for QIP.

\appendix
\chapter{Constructive proof of the CSD}
\label{Appendix:Construction}
In this appendix, I present a constructive proof of Theorem~\ref{Thm:CSD} and a procedure to construct the CSD.
Recall that our CSD is a building block of our main decomposition procedure, which is discussed chapter~\ref{Chap:Design}.

The output of our constructive proof matches the output of existing procedures~\cite{Stewart1977,Stewart1982} but our proof emphasizes the key role of the singular value decomposition in the CSD.
Furthermore, numerical implementations of this proof are expected to be more efficient and stable as compared to existing procedures because of the efficiency and stability of established singular-value-decomposition algorithms~\cite{Golub1965,Klema1980}.
Note that efficiency of numerical implementations refers to the computational cost of performing the decomposition and differs from the requirement of efficient realization, which deals with the number of optical elements required to experimentally realize the matrices.

Recall that the singular value decomposition factorizes any $m\times n$ complex matrix $M$ into the form
\begin{equation}
M = W\Lambda^{M} V^{\dagger}
\end{equation}
for $m\times m$ unitary matrix $W$, $n\times n$ unitary matrix $V$ and real non-negative diagonal matrix $\Lambda^{M}$.
The matrices $W$ and $V$ diagonalize $M\,M^\dagger$ and $M^\dagger M$ respectively.
In other words, the rows of $W$ and $V$ are the eigenvectors of $M\,M^{\dagger}$ and $M^{\dagger}M$.
These rows are called the left- and right-singular vectors of $M$.

Now I describe the construction of metrics $\mathds{L}_{m+n},\mathds{S}_{m+n},\mathds{R}_{m+n}$ in the CSD of a given $(m+n) \times (m+n)$ unitary matrix $U$.

\noindent
\textit{Proof of Theorem 7.}~In order to perform CSD of an arbitrary unitary matrix $U$, it is expressed as a $2\times 2$ block matrix
\begin{equation}
 \label{Eq:U-Block}
 U =\left(
 \begin{array}{c|c}
 A & B\\
 \hline
 C & D
 \end{array}\right),
\end{equation}
for complex matrices $A$, $B$, $C$ and $D$ of dimensions $m\times m, n\times m, m\times n$ and $n\times n$ respectively.
The constructive proof is in three parts.
First, I show the matrices $A$ and $C$ have the same left- and right-singular vectors and that $B$ and $D$, too, have the same left- and right-singular vectors.
The next step is to show that these common singular vectors can diagonalize each of the matrices $\{A,B,C,D\}$.
Finally, I show that diagonal form of the matrix $U$~\eqref{Eq:U-Block} is a CS matrix.

The unitarity of $U$ implies the relations
\begin{align}
 U\,U^\dagger \equiv \left(\begin{array}{c|c}
 A\,A^\dagger + B\,B^\dagger &A\,C^\dagger + B\,\mathcal{D}^\dagger\\\hline
 C\,A^\dagger + D\,B^\dagger& C\,C^\dagger+D\,\mathcal{D}^\dagger
 \end{array}\right) &= \mathds{1}_{m+n},\label{Eq:UUDBlocks}\\
U^\dagger U \equiv \left(\begin{array}{c|c}
 A^\dagger A + C^\dagger C&A^\dagger B+ C^\dagger D\\\hline
 B^\dagger A+ \mathcal{D}^\dagger C& B^\dagger B+\mathcal{D}^\dagger D
 \end{array}\right) &= \mathds{1}_{m+n}.\label{Eq:UDUBlocks}
\end{align}
Considering the blocks on the diagonals of Equations~\eqref{Eq:UUDBlocks} yields the matrix equations
\begin{align}
A\,A^\dagger + B\,B^\dagger &= \mathds{1}_m, \label{Eq:07}\\
C\,C^\dagger+D\,\mathcal{D}^\dagger &= \mathds{1}_n.\label{Eq:08}
\end{align}
Equations~\eqref{Eq:07} and~\eqref{Eq:08} imply that
\begin{align}
[A\,A^\dagger, B\,B^\dagger] &= 0,\label{Eq:Commutation1}\\
[C\,C^\dagger, D\,\mathcal{D}^\dagger] &= 0,\label{Eq:CDCommute}
\end{align}
i.e., $A\,A^\dagger$ commutes with $B\,B^\dagger$ and $C\,C^\dagger$ commutes with $D\,\mathcal{D}^\dagger$.
Furthermore, $A\,A^\dagger$ and $B\,B^\dagger$ are normal matrices.
Hence, $A\,A^\dagger$ and $B\,B^\dagger$ are diagonalized by the same matrix; or $A$ and $B$ have the same (up to a phase) left-singular vectors, denoted by the unitary matrix $L_m$.
From Equation~\eqref{Eq:CDCommute}, $C$ and $D$ have the same left-singular vectors, denoted by $L_n'$.

From Equation~\eqref{Eq:UDUBlocks}, we have
\begin{align}
A^\dagger A + C^\dagger C &= \mathds{1}_m,\label{Eq:09}\\
B^\dagger B+\mathcal{D}^\dagger D&= \mathds{1}_n.\label{Eq:10}
\end{align}
Following the same line of reasoning as the one used for obtaining common left-singular vectors, we observe that matrices $A$ and $C$ have the same right-singular vectors, say $R_m$, and $B$ and $D$ have the same right-singular vectors $R'_n$.

% From the commutation relation~\eqref{Eq:Commutation1}, we can show that $A$ and $B$ have the same left singular vectors as follows.
% we can diagonalize
% \begin{align}
% A\,A^\dagger = L_m (\Lambda^{A})^{2} L^{\dagger}_m,\\
% B\,B^\dagger = L_m (\Lambda^{B})^{2} L^{\dagger}_m,
% \end{align}
% the commuting normal matrices $A\,A^\dagger$ and $B\,B^\dagger$ simultaneously, where $L_m$ is a unitary matrix, and $\{|\Lambda^{A}_{ii}|\}$ and $\{|\Lambda^{B}_{jj}|\}$ are the singular value of $A$ and $B$~\cite{Horn1985}.
% Each row of $L_m$ is a left singular vector of $A$ and of $B$.
% Thus $A$ and $B$ share the set of left singular vectors.
% Similarly, matrices $C$ and $D$ share the same left singular vectors, which are denoted by the unitary matrix $L'_n$.

% Equations~\eqref{Eq:09} and~\eqref{Eq:10} imply that $A$ and $C$ share right singular vectors, say $R_m^\dagger$, and that $B$ and $D$ share right singular vector, say $R_n^{\prime\dagger}$.
The left- and right-singular vectors of the matrices $\{A,\,B,\,C,\,D\}$ can be employed to diagonalize these matrices according to
\begin{align}
A &= L_m\Lambda^{A}R^{\dagger}_m,\label{Eq:A}\\
B &= L_m\Lambda^{B}R^{\prime\dagger}_n,\label{Eq:B}\\
C &= L'_n\Lambda^{C}R^{\dagger}_m,\label{Eq:C}\\
D &= L'_n\Lambda^{D}R^{\prime\dagger}_n,\label{Eq:D1}
\end{align}
for diagonal complex matrices $\{\Lambda^{A},\Lambda^{B},\Lambda^{C},\Lambda^{D}\}$.
The matrices consisting of the absolute values of the corresponding complex elements of $\{\Lambda^{A},\Lambda^{B},\Lambda^{C},\Lambda^{D}\}$ matrices are denoted by $|\Lambda^A|,\,|\Lambda^B|,\,|\Lambda^C|$ and $|\Lambda^D|$ and comprise the singular values of $A,\,B,\,C$ and $D$ matrices respectively.
Equations~\eqref{Eq:A} to \eqref{Eq:D1} can be combined into a single $(m+n)\times (m+n)$ matrix equation
\begin{align}
\left( \begin{array}{c|c}
 A&B\\
 \hline
 C&D
 \end{array}\right)
 &= \left(
\begin{array}{c|c}
 L_m & \\
\hline
& L'_n
\end{array}\right)
\left(\begin{array}{c|c}
 \Lambda^A & \Lambda^B \\
\hline
 \Lambda^C & \Lambda^D
\end{array}\right)
\left(
\begin{array}{c|c}
 R_m^\dagger & \\
\hline
& R^{\prime\dagger}_n
\end{array}\right)\nonumber\\
\implies
U&= \tilde{\mathds{L}}_{m+n}
{\tilde{\Lambda}}_{m+n}
\tilde{\mathds{R}}_{m+n}.
\label{Eq:CombineEqn}
\end{align}
Factorization~\eqref{Eq:CombineEqn} is similar to the CSD because $\tilde{\mathds{L}}_{m+n}$ and $\tilde{\mathds{R}}_{m+n}$ block-diagonal unitary matrices and $\tilde{\Lambda}_{m+n}$ comprises diagonal blocks.
In the remainder of this appendix, we show that $\tilde{\Lambda}_{m+n}$ can be brought into the form of a CS matrix~\eqref{Eq:CSMatrix}, thereby completing the construction of the CSD.

If the matrices $L_m$ ($L_n'$) and $R_m$ ($R_n'$) are calculated from the singular value decomposition of $A$ ($D$), then $\Lambda^A$ ($\Lambda^D$) is a real and non-negative diagonal matrix.
The matrices $L_m$, $L_n'$, $R_m$ and $R_n'$ also diagonalize the matrices $C$ and $D$ resulting in $\Lambda^B$ and $\Lambda^C$.
Unlike $\Lambda^A$ and $\Lambda^D$, which consist of real elements, these matrices $\Lambda^B$ and $\Lambda^C$ are complex matrices in general.
In other words, the diagonal matrices $\Lambda^B$ and $\Lambda^C$ are of the form
\begin{equation}
\begin{aligned}
 \Lambda^B &= P|\Lambda^B|\\
 \Lambda^C &= -|\Lambda^C|P^\dagger,
 \label{Eq:Lambda-BC}
\end{aligned}
\end{equation}
where $P$ is an $m\times m$ diagonal unitary matrix.
The phases $P_{jj}$ in Equation~\eqref{Eq:Lambda-BC} for $C$ are complex conjugates of the phases for $B$ because of the unitarity of $\Lambda$.

We can remove the matrix $P$ from $\Lambda^B$ and $\Lambda^C$ by redefining $L_m$ and $R_m$ as
\begin{align}
 \tilde L_m &= L_m P,\\
 \tilde R_m &= R_m P.
\end{align}
Thus Equation~\eqref{Eq:CombineEqn} can be rewritten as:
\begin{equation}
U = \left(
\begin{array}{c|c}
 L_m P& \\
\hline
& L'_n
\end{array}\right)
\left(\begin{array}{r|r}
 \Lambda^A & |\Lambda^B| \\
\hline
 -|\Lambda^C| & \Lambda^D
\end{array}\right)
\left(
\begin{array}{c|c}
P^{\dagger}R_m^\dagger & \\
\hline
& R^{\prime\dagger}_n
\end{array}\right)
\end{equation}
or
\begin{equation}
U = \mathds{L}_{m+n} \Lambda_{m+n} \mathds{R}_{m+n}.
\end{equation}
Note that the matrix $\Lambda_{m+n}$ comprises only real elements.
Furthermore, $\Lambda_{m+n}$ is unitary because it is a product $\Lambda_{m+n} = \mathds{L}_{m+n}^{\dagger} U \mathds{R}^{\dagger}_{m+n}$.
Hence, $\lambda_{m+n}$ is an orthogonal matrix.
%From the unitarity of $U$, we see that each column (and every row) in this matrix contains at the most two non-zero elements with the sum of the square of the two elements equal to unity.
%Therefore, the magnitude of these two elements can be represented by $\cos\theta$ and $\sin \theta$.

The orthogonality of the $\Lambda$ implies that any two rows and any two columns of the matrix are orthogonal.
Therefore, the $2\times 2$ block matrices
\begin{align}
 \Lambda_i &= \begin{pmatrix}
 \Lambda_{i,i} & \Lambda_{i,i+m}\\
 \Lambda_{i+m,i} & \Lambda_{i+m,i+m}
\end{pmatrix}
\end{align}
is also an orthogonal matrix.
Any $2\times 2$ orthogonal matrix is of the form
\begin{align}
 \Lambda_i &= \begin{pmatrix}
 \cos\theta_i & \sin\theta_i\\
 -\sin\theta_i & \cos\theta_i
\end{pmatrix}
\end{align}
for $1\le i \le m$.

Next we consider the case of $i>m$.
For the matrix $\Lambda^{B}$ all the columns with the index $i > m$ are zero.
Similarly, for the matrix $\Lambda^{C}$ all the rows with the index $i>m$ are zero.
From the unitarity of $\Lambda_{m+n}$, we see that each of the diagonal elements in the last $n-m$ columns and rows of the matrix $\Lambda^{D}$ is unity.
In summary, the matrix $\Lambda_{m+n}$ is of the form
\begin{align}
\Lambda_{m+n} &= \mathds{S}_{2m} \oplus \mathds{1}_{n-m}
\end{align}
for $\mathds{S}_{2m}$ a CS matrix in the form of Equation~\eqref{Eq:CSMatrix}.\hfill$\blacksquare$

This completes our procedure for factorizing a given unitary matrix using the CSD.
\textsc{matlab} code for our CSD procedure is available online~\cite{Dhand2015}.

%=================%
\chapter{Curve-fitting subroutine}
\label{Sec:CurveFitting}
%=================%
Here I detail the inputs and outputs of the curve-fitting algorithm employed in our accurate and precise characterization procedure (Chapter~\ref{Chap:Procedure}).
The accurate and precise characterization procedure (Chapter~\ref{Chap:Procedure}) employs curve fitting in Algorithm~\ref{Alg:Calibration} to estimate the mode-matching parameter~$\gamma$ and in Algorithms~\ref{Alg:PhaseCalc2Channel}--\ref{Alg:PhaseCalcNChannel} to estimate the interferometer-matrix arguments $\{\theta_{ij}\}$.
The curve-fitting algorithm uses the experimental data and determines those values of unknown parameters that yield the best between experimental and theoretical coincidence rates.

\begin{figure}[h]
\centering
\subfloat{\includegraphics[width=0.49\textwidth]{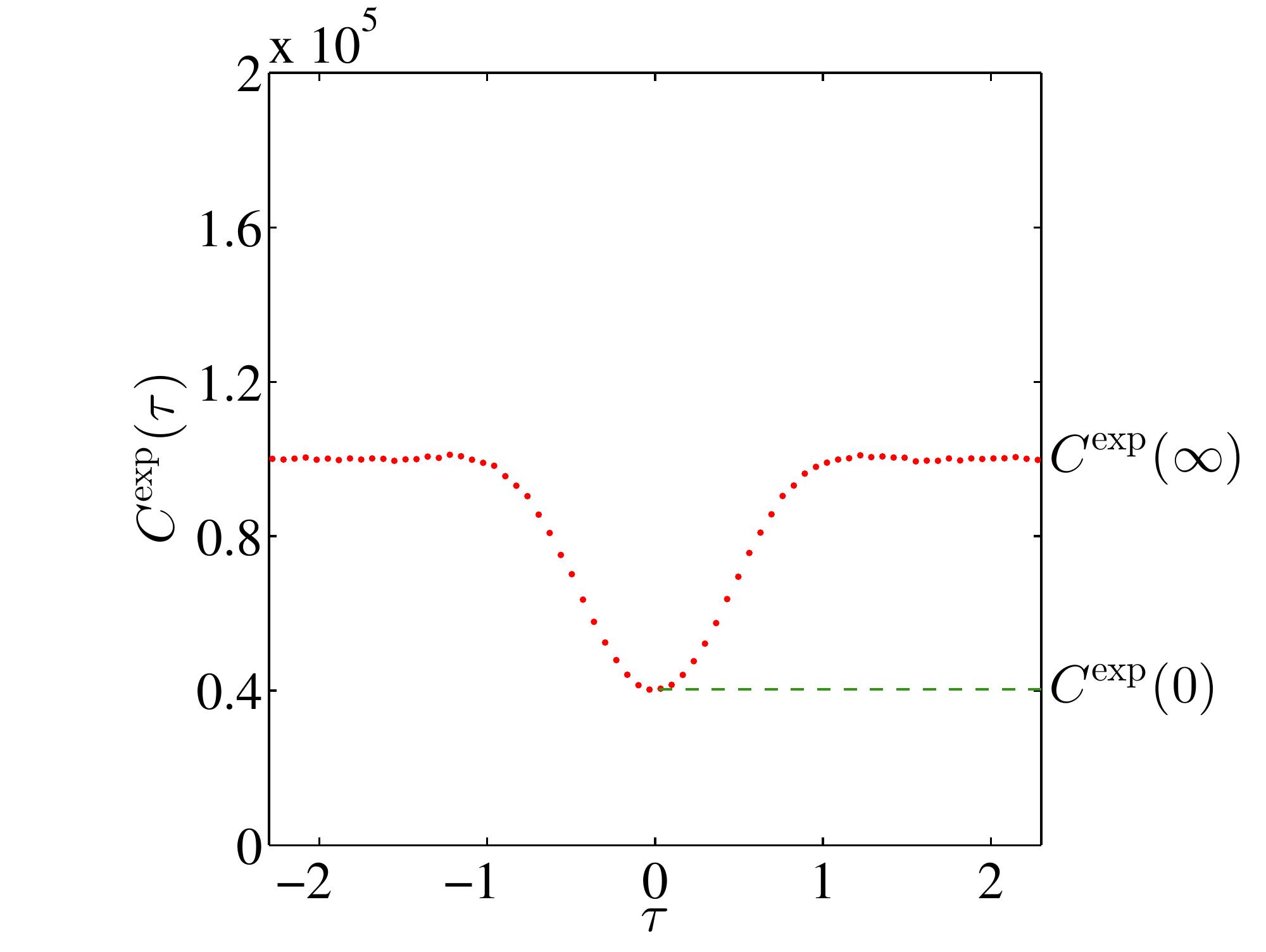}\label{Figure:CoincidenceShapeSU2}}
\subfloat{\includegraphics[width=0.49\textwidth]{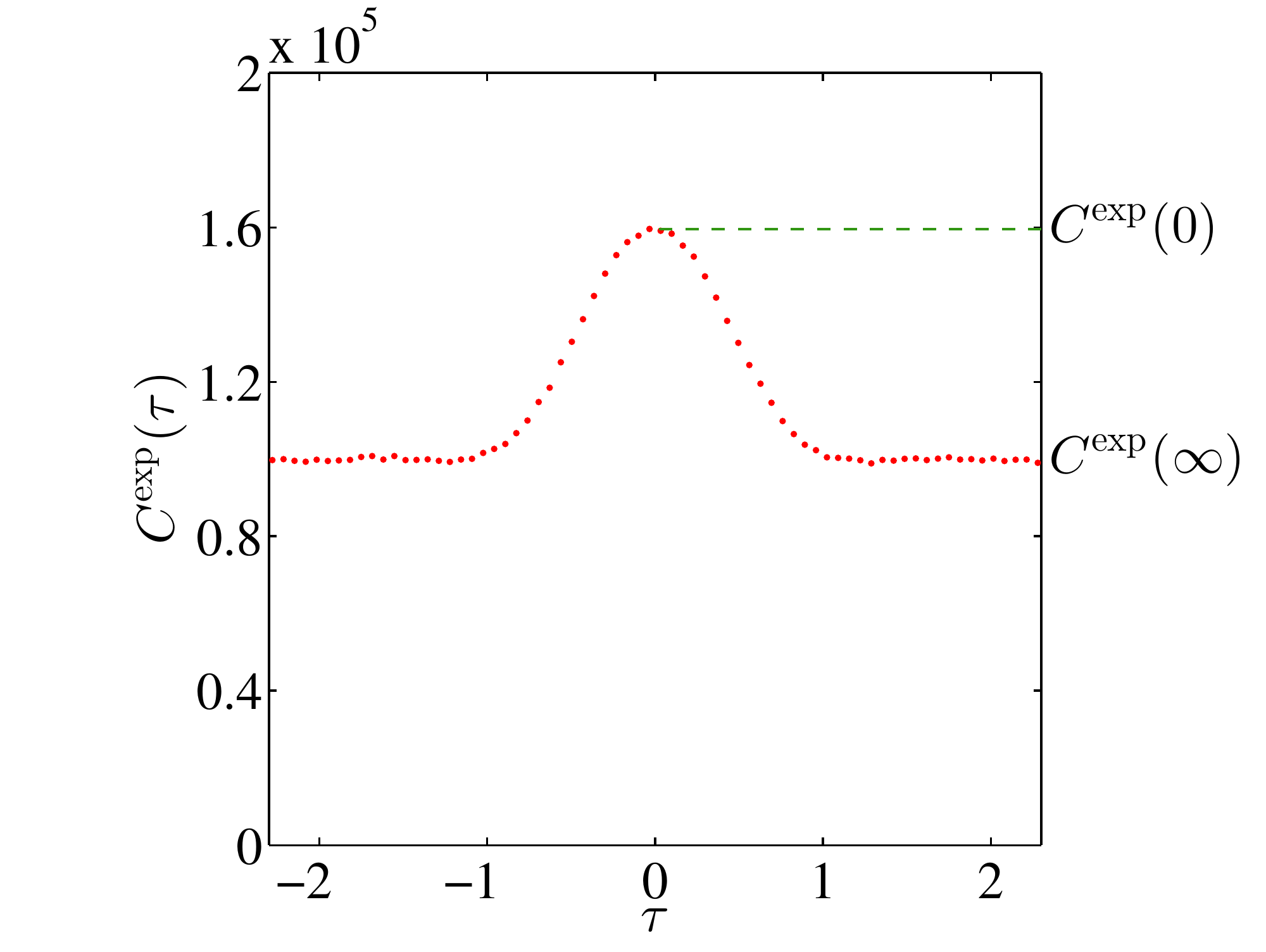}\label{Figure:CoincidenceShapeBetaZero}}\\
\subfloat{\includegraphics[width=0.49\textwidth]{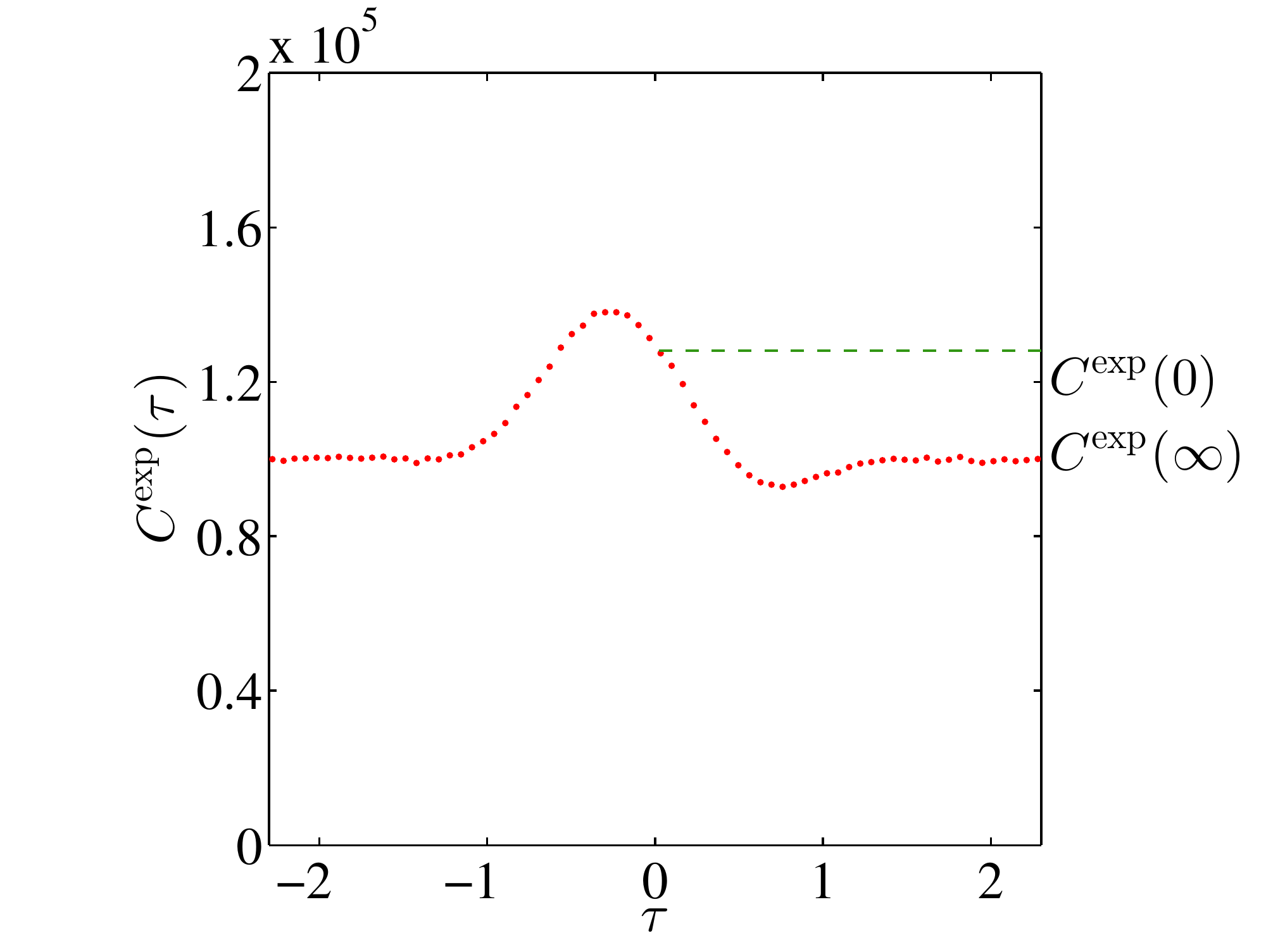}\label{Figure:CoincidenceShapeGeneral1}}
\subfloat{\includegraphics[width=0.49\textwidth]{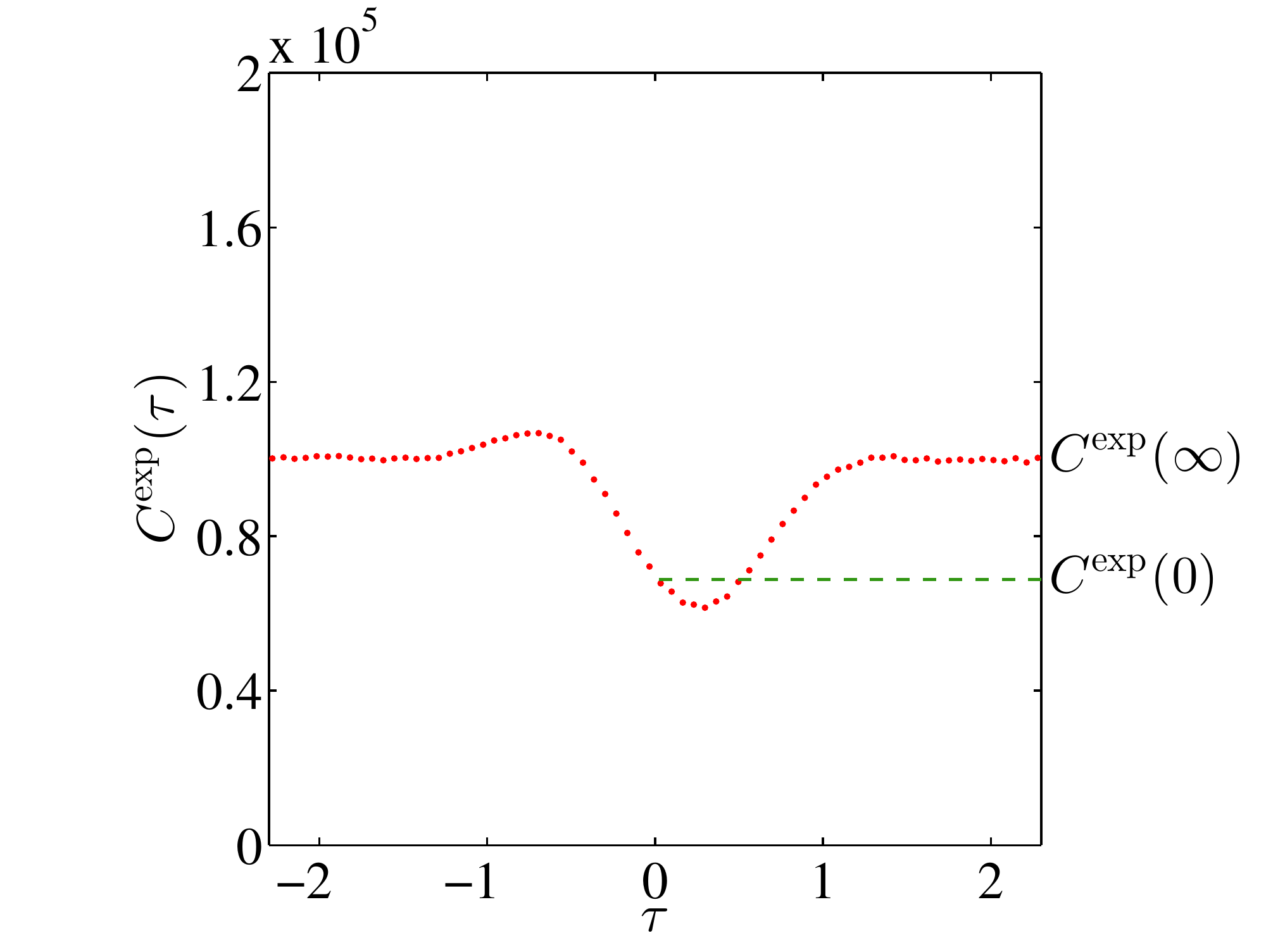}\label{Figure:CoincidenceShapeGeneral2}}
\caption{Simulated coincidence counts for output ports $i,i'$ and input ports $j,j'$ of interferometer with $\alpha_{ii}=\alpha_{i'j'}=\sqrt{3}/4$ and $\alpha_{ii'}=\alpha_{ij'}=1/4$ and for different values of $\beta_{ii'jj'}$. The value of $\beta_{ii'jj'}$ in each respective figure is (a) $\pi$, (b) $0$, (c) ${\pi}/{3}$ and (d) ${2\pi}/{3}$. The coincidence counts corresponding to $\tau=0$ and $\tau\to\infty$ are marked on each plot by $C^{\mathrm{exp}}(0)$ and $C^\mathrm{\text{exp}}(\infty)$ respectively.
}
\label{Figure:CoincidenceCurveShapes}
\end{figure}

The curve-fitting algorithm receives the following inputs:
(i)~the choice of parameters to be fitted;
(ii)~the coincidence counts $\{C_{ii'jj'}^{\text{exp}}(\tau)\}$;
(iii)~an objective function, which characterizes the least-square error between expected and experimental counts; and
(iv)~the initial guesses for each of the fitted parameters. The output of the curve-fitting subroutine is the set of parameter values that optimize the objective function.

%The main text already describes how to collect the experimental data, which is the first input. In the rest of this section I discuss the other inputs to the curve-fitting subroutine. The parameters for Algorithms~\ref{Alg:Calibration} and \ref{Alg:PhaseCalc2Channel} are different. I explain the parameters of each algorithm in turn. Next I discuss the objective function. Last I give a procedure to find a guess for each parameter.

The first input to the subroutine is the choice of the parameters to be fit.
The curve-fitting subroutine fits three parameters.
One of these three (namely the mode-matching parameter~$\gamma$ in Algorithm~\ref{Alg:Calibration} or the $|\theta_{ij}|$ or $\beta_{ii'jj'}$ value in Algorithm~\ref{Alg:PhaseCalcNChannel}) is related to the shape of the curve, whereas the other two are related to the ordinate scaling and the abscissa shift of the curve respectively.
The ordinate scaling factor comprises the unknown losses $\{\kappa_{i},\nu_{j}\}$, transmission factors $\{\lambda_{i},\mu_{j}\}$ and the incident photon-pair count.
The horizontal shift factor accounts for the unknown zero of the time delay between the incident photons.
The algorithm returns the values of the shape parameter, the abscissa shift and the ordinate scaling that best fit the given coincidence curves.

The second input is the experimental data that are fit to the theoretical coincidence curves, which are described in the third input: the objective function.
The objective function quantifies the goodness of fit between the experimental data and the parameterized curve.
We use a weighted sum
\begin{equation}
\sum_{\tau \in T} w(\tau)|C^\mathrm{exp}(\tau) - C'(\tau)|^2
\end{equation}of squares between the experimental data and the fitted curve as the objective function~\cite{Strutz2010} for weighs $w(\tau)$.
We assume that the pdfs of the coincidence counts are proportional to $\sqrt{C^{\mathrm{exp}}(\tau)}$ and we assign the weights
\begin{equation}
 w(\tau) =
 \begin{cases}
 {1}/{C^{\mathrm{exp}}(\tau)} & \text{if } C^{\mathrm{exp}}(\tau) \ne 0\\
 \quad 1 & \text{if } C^{\mathrm{exp}}(\tau) = 0
 \end{cases}
\end{equation}
to the squared sum of residues.
In case the pdf's of the residuals for different values of $\tau$ is not known, standard methods for non-parametric estimation of residual distribution can be employed to estimate the pdf's~\cite{Akritas2001,Chen2003}.
Thus, the curve fitting algorithm returns those values of the fitting parameters that minimize weighted sum of squared residues between experimental and fitted data.

The curve-fitting procedure optimizes the fitness function over the domain of the fitting parameter values.
Like other optimization procedures, the convergence of curve fitting is sensitive to the initial guesses of the fitting parameters.
The following heuristics give good guesses for the three fitting parameters.
We guess the ordinate scaling as the ratio
\begin{equation}
\frac{C^{\mathrm{exp}}(\infty)}{C_{ii'jj'}(\infty)}
\end{equation}
 of the experimental coincidence counts
\begin{equation}
C^{\mathrm{exp}}(\infty) \defeq \frac{C^{\mathrm{exp}}(\tau_1) + C^{\mathrm{exp}}(\tau_\ell)}{2},
\end{equation}
to the coincidence probability $C_{ii'jj'}(\infty)$ for large (compared to the temporal length of the photon) time-delay values.
The~$\gamma$ value is guessed for Algorithm~\ref{Alg:Calibration} as the ratio of the visibility of the experimental curve to the expected visibility in the curve.
The initial guesses for $\vartheta \equiv |\theta_{ij}|$ and $\vartheta \equiv \beta_{ii'jj'}$ are based on the known estimate of~$\gamma$ and the visibility
\begin{equation}
 V = \frac{2\gamma\cos^2\vartheta\sin^2\vartheta}{\cos^4\vartheta + \sin^4\vartheta}.
\end{equation}
of the curve.
As there are four kinds of curves (see Figure~\ref{Figure:CoincidenceCurveShapes}) possible for different values of the shape parameter ($\gamma, |\theta_{ij}|, \beta_{ii'j'}$), another approach is to perform curve fitting four times, each time with a value from the set $\pi/4,3\pi/4,5\pi/4,7\pi/4$ of initial guesses and choose the fitted parameters that optimize the objective function.
Finally, the initial value of the abscissa shift parameter is guessed such that the global maxima or minima (whichever is further from the mean of the coincidence-count values over $\tau$) of the coincidence curve is at zero time delay.

In summary, the curve fitting procedure uses the measured coincidence counts, the objective function and the initial guesses to compute the parameters that yield he best fit between theoretical and measured coincidence counts.
This completes our description of the curve-fitting procedure and of heuristics that can be employed to computed the initial guesses for the fitted parameters.

%----------------------------------------%
\chapter{Choice of subalgebra chain}
\label{Appendix:SubAlgebraChoice}
%----------------------------------------%
This appendix elaborates on the different choices of sub algebra chain that can be employed in the labelling of the $\mathrm{SU}(n)$ states and $\mathcal{D}$-functions.
Our algorithms construct canonical basis states that reduce the subalgebra chain~(\ref{Eq:SubalgebraChain}).
Other $\mathfrak{su}(n)\supset\mathfrak{su}(n-1)\supset\dots\supset\mathfrak{su}(2)$ subalgebra chains are possible and our algorithm can be generalized to construct canonical basis states that reduce other chains, as I discuss in this appendix.

Each $\mathfrak{su}(m)$ subalgebra of $\mathfrak{su}(n),\,m<n$ is specified by the sets of raising, lowering and Cartan operators that generate it.
For a given sequence
\begin{equation}
I^{(m)} = \left\{i^{(m)}_1,i^{(m)}_2,\dots, i^{(m)}_{m}\right\}
\end{equation}
of $m$ increasing integers, we can define the corresponding set of raising, lowering and Cartan operators
\begin{align}
&\left\{C_{i_1,i_2},C_{i_1,i_3},\dots ,C_{i_1,i_m},C_{i_2,i_3},\dots ,C_{i_2,i_m},\dots,C_{i_{m-1},i_m}\right\}&\text{(Raising)}\\
&\left\{C_{i_2,i_1},C_{i_3,i_1},\dots ,C_{i_m,i_1},C_{i_3,i_2},\dots ,C_{i_m,i_2},\dots,C_{i_{m},i_{m-1}}\right\}&\text{(Lowering)}\\
&\left\{C_{i_2,i_2}-C_{i_1,i_1},C_{i_3,i_3}-C_{i_2,i_2},\dots,C_{i_m,i_m}-C_{i_{m-1},i_{m-1}}\right\}&\text{(Cartan)}
\end{align}
that generate the algebra.
Thus, each $\mathfrak{su}(n)\supset\mathfrak{su}(n-1)\supset\dots\supset\mathfrak{su}(2)$ subalgebra chain is uniquely specified by the ordered sequences $I^{(m)}\colon m<n$ of integers, where
\begin{align}
\begin{split}
I^{(n-1)} &= \{i^{(n-1)}_1,i^{(n-1)}_2,\dots, i^{(n-1)}_{n-1}\} \subset \{1,2,\dots,n\}\\
I^{(n-2)} &= \{i^{(n-2)}_1,i^{(n-2)}_2,\dots, i^{(n-2)}_{n-1}\} \subset I^{(n-1)}\\
&\dots\\
I^{(m-1)} &= \{i^{(m-1)}_1,i^{(m-1)}_2,\dots, i^{(m-1)}_{m}\} \subset I^{(m)}\\
&\dots\\
I^{(2)} &= \{i^{(2)}_1,i^{(2)}_2\} \subset I^{(3)}.
\label{Eq:Chain}
\end{split}
\end{align}

Consider the example of $\mathfrak{su}(2)$ subalgebras of $\mathfrak{su}(3)$.
The three subsets
\begin{align}
\{C_{1,2},C_{2,1},C_{1,1}-C_{2,2}\}\label{Eq:A1}\\
\{C_{1,3},C_{3,1},C_{1,1}-C_{3,3}\}\\
\{C_{2,3},C_{3,2},C_{2,2}-C_{3,3}\}\label{Eq:A3}
\end{align}
of the generators $\{C_{i,j}\colon i,j \in \{1,2,3\}\}$ of $\mathfrak{su}(3)$ generate three distinct $\mathfrak{su}(2)$ algebras.
Each of the three subsets~(\ref{Eq:A1})-(\ref{Eq:A3}) can be labelled with a two-element subset of the $\{1,2,3\}$ and can be employed to define canonical basis states of $\mathrm{SU}(n)$.
For instance, consider $(\lambda,\kappa) = (1,1)$ irrep of $\mathrm{SU}(3)$.
The weight $(\lambda_2,\lambda_1) = (0,0)$ is associated with a two-dimensional space.
We can identify two basis states of this space by specifying the following:
\begin{enumerate}
\item
choice of $\mathfrak{su}(2)$ algebra. For instance $I^{(2)} = \{1,2\}$, which corresponds to the algebra generated by $\{C_{1,2},C_{2,1},C_{1,1}-C_{2,2}\}$,
\item
$\mathfrak{su}_{1,2,3}(3)$ irreps label: $K^{(3)} = (1,1)$, $\mathfrak{su}_{(1,2)}(2)$ irreps label: $K^{(2)} = (0)$ and $(1)$ for the two basis states.
\item
$\mathfrak{su}_{1,2,3}(3)$ weights: $(0,0)$, $\mathfrak{su}_{(1,2)}(2)$ weights: $(0)$.
\end{enumerate}
Another basis set of the $\Lambda =(0,0)$ space of $\mathfrak{su}(3)$ irrep $K = (1,1)$ is specified by choosing a different $\mathfrak{su}(2)$ subalgebra as follows.
\begin{enumerate}
\item
choice of $\mathfrak{su}(2)$ algebra. For instance $I^{(2)} = \{1,3\}$, which corresponds to the algebra generated by $\{C_{1,3},C_{3,1},C_{1,1}-C_{3,3}\}$,
\item
$\mathfrak{su}_{1,2,3}(3)$ irreps label: $K^{(3)} = (1,1)$, $\mathfrak{su}_{(1,3)}(2)$ irreps label: $K^{(2)} = (0)$ and $(1)$ for the two basis states.
\item
$\mathfrak{su}_{1,2,3}(3)$ weights: $(0,0)$, $\mathfrak{su}_{(1,3)}(2)$ weights: $(0)$.
\end{enumerate}
Thus, different choices of subalgebra chain give us different basis states.

In the main text, we have chosen the subalgebra chain~(\ref{Eq:SubalgebraChain}).
Our algorithms can be modified to account for other choices of subalgbra chain by choosing a different set of lowering operators in the basis-set subroutine.
Thus our algorithms can be used to construct states and $\mathcal{D}$-functions in any of the bases that reduce $\mathfrak{su}(m)$ subalgebra chains.

%=================%
\chapter{Connection to Gelfand-Tsetlin basis}
\label{Appendix:Connection}
%=================%
In this appendix, I detail the mapping between our $\mathrm{SU}(n)$ basis states and the canonical Gelfand-Tsetlin (GT) basis.
The GT basis identifies each $\mathrm{SU}(n)$ irrep with a sequence of $n$ numbers
\begin{align}
S_n &= (m_{1,n},\dots,m_{n,n})\\
m_{k,n} &\ge m_{k+1,n}~\forall 1\le k\le n-1,
\end{align}
where the first label in the subscript is the sequence index and the second label identifies the algebra.
The carrier space of every $\mathfrak{su}(m)$ subalgebra is composed of disjoint $\mathfrak{su}(m-1)$ carrier spaces
\begin{equation}
\left\{(m_{1,n-1},\dots,m_{n-1,n-1})\right\}
\end{equation}
that obey the betweenness condition
\begin{equation}
m_{k,n}\ge m_{k,n-1}\ge m_{k+1,n}.
\end{equation}
Thus, each $\mathfrak{su}(n)$ basis state $\ket{M}$ can be labelled by the GT pattern
\begin{equation}
\ket{M} \equiv \begin{pmatrix}
\multicolumn{2}{c}{m_{1,N}} & \multicolumn{2}{c}{m_{2,N}} &
\multicolumn{2}{c}{\ldots} & \multicolumn{2}{c}{m_{N,N}} \\
& \multicolumn{2}{c}{m_{1,N-1}} & \multicolumn{2}{c}{\ldots}
& \multicolumn{2}{c}{m_{N-1,N-1}} & \\
&& \ddots &&& \reflectbox{\(\ddots\)} && \\
&& \multicolumn{2}{r}{m_{1,2}} & \multicolumn{2}{l}{m_{2,2}} && \\
&&& \multicolumn{2}{c}{m_{1,1}} &&&
\end{pmatrix},
\end{equation}
where
\begin{equation}
m_{k,\ell}\ge m_{k,n-1}\ge m_{k+1,\ell}\,, \,1\le k < \ell \le n.
\end{equation}

The canonical basis states are eigenstates of the Cartan operators $\{H_i\}$~(\ref{Eq:Cartan}) as detailed in the following lemma.
\begin{lem}[Connection to Gelfand-Tsetlin basis~\cite{Alex2011}] The canonical basis states are connected to the GT basis according to
\begin{equation}
\label{eq:gtpattern}
\Big|{\tensor*{\psi}{*^{K^{(z)}}_{\Lambda^{(n)}}^{,\dots,}_{,\dots,}^{K^{(3)},}_{\Lambda^{(3)},}^{K^{(2)}}_{\Lambda^{(2)}}}}\Big\rangle =
 \begin{pmatrix}
\multicolumn{2}{c}{m_{1,N}} & \multicolumn{2}{c}{m_{2,N}} &
\multicolumn{2}{c}{\ldots} & \multicolumn{2}{c}{m_{N,N}} \\
& \multicolumn{2}{c}{m_{1,N-1}} & \multicolumn{2}{c}{\ldots}
& \multicolumn{2}{c}{m_{N-1,N-1}} & \\
&& \ddots &&& \reflectbox{\(\ddots\)} && \\
&& \multicolumn{2}{r}{m_{1,2}} & \multicolumn{2}{l}{m_{2,2}} && \\
&&& \multicolumn{2}{c}{m_{1,1}} &&&
\end{pmatrix}
\end{equation}
Every state $\ket{M}$ in the GT-labeling scheme is a simultaneous eigenstate of all $\mathfrak{su}(n)$ Cartan operators,
\begin{equation}
H_\ell \ket{M} = \lambda^M_\ell \ket{M}, \quad (1 \le \ell \le N-1),
\end{equation}
with eigenvalues
\begin{equation}
\label{eq:jzelement}
\lambda_\ell = \sum_{k=1}^\ell m_{k,\ell} - \frac{1}{2} \left(\sum_{k=1}^{\ell+1} m_{k,\ell+1}+ \sum_{k=1}^{\ell-1} m_{k,\ell-1}\right), \, 1 \leq \ell \leq N-1.
\end{equation}
\end{lem}
\noindent
Thus, the canonical basis states of Def.~\ref{Definition:CanonicalBasisStates} is uniquely mapped to the GT basis.

Furthermore, the weights $\lambda_{\ell}$ are also mapped via the boson realizations to differences in number of bosons at sites $\ell$ and $\ell+1$.
Hence, the difference
\begin{equation}
\nu_{\ell+1}- \nu_{\ell} = \sum_{k=1}^l m_{k,\ell} - \frac{1}{2} \left(\sum_{k=1}^{\ell+1} m_{k,\ell+1}+ \sum_{k=1}^{\ell-1} m_{k,\ell-1}\right), \,1 \leq \ell \leq N-1
\end{equation}
in the number of bosons at sites $\ell+1$ and $\ell$ of the boson realization of a basis state is also connected to its GT pattern.
Once we recall the total number of bosons in the system is $\nu_{1}+\nu_{2}+\dots+\nu_{n}=N_k$, one can then invert the differences and recover $\nu_p$ in term of the $m_{k,\ell-1}.$
Thus the canonical GT basis states are connected to our $\mathrm{SU}(n)$ basis states.

\bibliographystyle{good_unsrt}
\bibliography{Design_Characterization_Simulation_Arxiv}

\begin{thebibliography}{100}
\bibitem{Pearson1900}
K.~Pearson.
\newblock {On the criterion that a given system of deviations from the probable
  in the case of a correlated system of variables is such that it can be
  reasonably supposed to have arisen from random sampling}.
\newblock {\em Philos. Mag.}, 50(302):157--175, 1900.
\end{thebibliography}

\end{document}